\documentclass[aps,prx,twocolumn,notitlepage,superscriptaddress, preprintnumbers]{revtex4-2}

\usepackage{amsmath,amssymb, times}
\usepackage{mathtools}
\usepackage{color}
\usepackage{bm}
\usepackage{bbm}
\usepackage{graphicx}
\usepackage[normalem]{ulem}
\usepackage{enumitem}
\usepackage{makecell}
\usepackage{amsthm}
\usepackage{thmtools}

\usepackage{here}
\usepackage{amsthm}
\usepackage{thmtools}
\usepackage{blindtext}
\usepackage{array}
\usepackage{bbding}
\usepackage[colorlinks=true,citecolor=blue,linkcolor=blue]{hyperref}

\usepackage{tikz}
\usepackage{tikz-3dplot}
\usetikzlibrary{external}
\usetikzlibrary{matrix,positioning, decorations.pathreplacing,calligraphy,decorations.markings, knots, arrows.meta,decorations.markings}
\definecolor{tensorcolor}{rgb}{1., 1., 1.}
\definecolor{btensorcolor}{rgb}{0.65,0.77,0.95}

\tikzset{
    mid arrow/.style args={#1}{
        decoration={
            markings,
            mark=at position #1 with {\arrow{Stealth}}
        },
        postaction={decorate}
    }
}

\declaretheoremstyle[
shaded={bgcolor=\color{rgb}{0.9,0.9,0.9}}
]{definition}
\declaretheorem[style=theorem]{definition}
\newtheorem{proposition}{Proposition}
\newtheorem{properties}{Properties}
\newtheorem{lemma}[proposition]{Lemma}

\begin{document}
\title{Unitary network: Tensor network unitaries with local unitarity}
\author{Wenqing Xie}
\affiliation{Department of Physics, Hong Kong University of Science and Technology, Clear Water Bay, Hong Kong SAR, China}

\author{Seishiro Ono}
%\affiliation{Institute for Solid State Physics, University of Tokyo, Kashiwa 277-8581, Japan}
\affiliation{RIKEN Center for Interdisciplinary Theoretical and Mathematical Sciences (iTHEMS), RIKEN, Wako 351-0198, Japan}
\affiliation{Department of Physics, Hong Kong University of Science and Technology, Clear Water Bay, Hong Kong SAR, China}

\author{Hoi Chun Po}
\email{hcpo@ust.hk}
\affiliation{Department of Physics, Hong Kong University of Science and Technology, Clear Water Bay, Hong Kong  SAR, China}

\preprint{RIKEN-iTHEMS-Report-25}

\begin{abstract}
We introduce \textit{unitary network}, an oriented architecture for tensor network unitaries. Compared to existing architectures, in a unitary network each local tensor is required to be a unitary matrix upon suitable reshaping. Global unitarity is ensured when the network obeys a suitable ordering property.
Unitary operators represented by unitary networks need not preserve locality. In particular, we show that the class of unitary networks encompasses global unitaries which preserve locality up to exponentially suppressed tails, as in those that naturally arise from the finite-time evolution of local Hamiltonians. 
Non-invertible symmetries, as exemplified by the non-local Kramers-Wannier duality in one dimension, can also be represented using unitary networks. We also show that information flow in a unitary network can be characterized by a flow index, which matches the known index for quantum cellular automata as a special case.
\end{abstract}
\maketitle

% This can be deleted
\tableofcontents

\section{Introduction}
Tensor network states are quantum many-body states defined through the contraction of local tensors. As the number of parameters in the ansatz grows only linearly in the system size, tensor networks provide a possible route to bypass the prohibitive exponential growth of the Hilbert space dimension and enable classical simulation of quantum problems. As a tradeoff, however, it is generally a challenging problem to determine the representability of various quantum states using a specific tensor network architecture. Nevertheless, the \textit{matrix product states} (MPSs) have proven remarkably effective for both representing and determining the ground and low-lying excited states of one-dimensional (1D) local Hamiltonians~\cite{White_DMRG, verstraete2006matrix, hastings2007area}. On the conceptual side, the MPS architecture has also enabled the classification of symmetry-protected topological phases in quantum spin systems~\cite{Chen_MPS_SPT_2011, Schuch_MPS_SPT_2011, Cirac_RMP}.

Recent studies on quantum many-body dynamics have also highlighted the importance of obtaining efficient representations of quantum many-body operators. For instance, the generalization of the equilibrium quantum Hall effect to the non-equilibrium Floquet setting revealed that the shift operator, known to be not realizable in a 1D system using \textit{finite-depth local unitary} (FDLU) circuit, can be understood as the dynamical analog of the anomalous chiral edge state~\cite{Harper_chiral_Floquet, Po_2016}. On another front, the advancement in our general understanding of topological phases~\cite{chen20122d, sahinoglu2150characterizing, bultinck2017anyons} had also inspired the generalization of symmetries---a fundamental concept in physics---to include both higher-form and non-invertible variants \cite{Okada2024, Lootens2023, Lootens2024, Lootens2025, Fukusumi2021, Ashkenazi2022, Tantivasadakarn2024, Aasen2022, bravyi2022adaptiveconstantdepthcircuitsmanipulating, Tantivasadakarn2023, Fechisin2025, Okuda2024, aasen2016topological, cao2023subsystem}. 

How can we represent quantum operators using the tensor network ansatz? In 1D, the \textit{matrix-product operators} (MPO) \cite{Verstraete_MPDO_2004, pirvu2010matrix} provide a natural generalization of MPS to operators which preserve the entanglement area law of quantum states. In the context of quantum dynamics and generalized symmetries, however, one is often interested in the restricted class of unitary operators. Demanding unitarity on the MPO is a nontrivial problem, and this has been tackled through the framework of \textit{matrix product unitaries} (MPUs)~\cite{chen20122d, cirac2017matrix, _ahino_lu_2018}. Restricted to the case of an infinite 1D chain with uniform local tensors, i.e., a manifestly translation invariant ansatz, MPUs (with finite bond dimension) have been shown to be equivalent to \textit{quantum cellular automata} (QCAs) ~\cite{cirac2017matrix,Piroli_2020}. QCAs~\cite{schumacher2004reversiblequantumcellularautomata, Farrelly_2020, arrighi2019overview} generalize FDLU circuits in that only the locality preserving property, but not the locally generated requirement, of the unitary map is required. As such, QCAs naturally include the shift operators, and are in fact classified by an index~\cite{Gross_2012} closely tied to the presence of the shifts (henceforth referred to as the GNVW index). 

Yet, such equivalence of MPUs with QCAs also implies their limit in representing unitary maps which are only approximately locality preserving, or those that preserve the area-law property but not locality. The former appears naturally in the study of quantum dynamics, as the finite-time evolution of a local Hamiltonian will generally transform an operator into one which is only approximately contained within the Lieb-Robinson lightcone~\cite{lieb1972finite} and admits an exponentially decaying tail outside~\cite{ranard2022converse}. The latter scenario arises naturally in non-invertible symmetries which map between area-law states in different phases. Can the MPU framework be generalized to include these important cases of unitary maps which are not QCAs?

In this work, we answer this question in the affirmative through a simple construction: instead of first considering a general MPO and then imposing conditions on the local tensors to realize unitarity, we consider MPOs built explicitly using local unitary building blocks~\cite{ferris2012perfect, pollmann2016efficient, wahl2017efficient, liu2019machine, haghshenas2021optimization, haghshenas2022variational, Styliaris_2025}. Importantly, the local unitary operators act on both the physical and auxiliary Hilbert spaces, which differentiates the construction from FDLU circuits and, as we will show, provides a way to represent non-QCA but area-law preserving \textit{translationally invariant} (TI) unitaries on an infinite chain simply by repeating the same local tensor, a feature that is often called ``uniform.’’

\begin{figure}[t]
    \centering
    \includegraphics[width=0.99\linewidth]{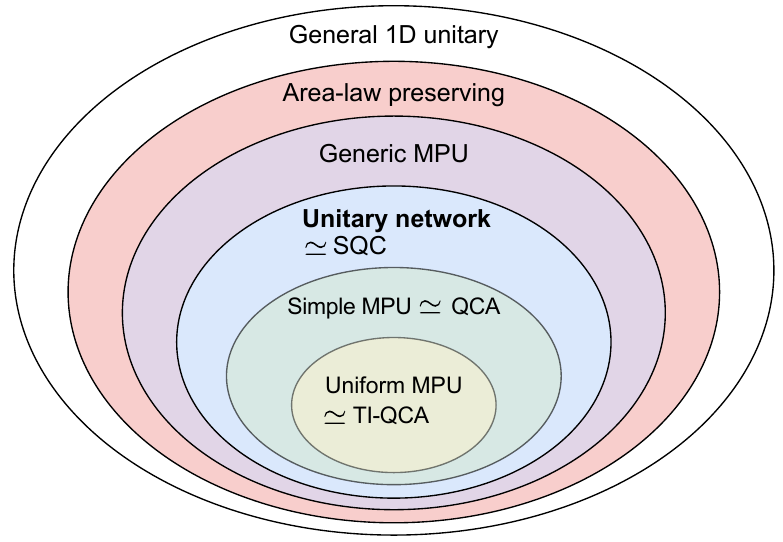}
    \caption{
    \textbf{Hierarchy of one-dimensional unitaries.}
    An area-law preserving unitary transforms an area-law state to another area-law state.
    MPU, (TI-)QCA, and SQC stand for \textit{matrix product unitary}, \textit{(translational invariant) quantum cellular automaton}, and \textit{sequential quantum circuit}, respectively.
    The class of unitary networks is equivalent to that of SQCs.
    }
    \label{fig:UN include MPU}
\end{figure}

We remark that non-uniform MPUs go beyond QCA \cite{Styliaris_2025} and can be realized with a quantum circuit of depth $O(\text{poly }N)$ and $O(N)$ auxiliary qudits~\cite{styliaris2025MPU} for $N$ physical qudits. On the other hand, we show that unitary networks have equivalent representability to \textit{sequential quantum circuits} (SQCs)~\cite{schon2005sequential, schon2007sequential, banuls2008sequentially}, which can represent all QCAs and various non-local transformations to connect states of distinct gapped phases~\cite{Chen_2024} with only $O(N)$ depth and no auxiliary qudit. A hierarchy of classes of unitary operators is shown in Fig.~\ref{fig:UN include MPU}. While generic (non-uniform) MPUs might have superior representability compared to unitary networks, constructing a unitary network for a specific unitary operator may be simpler due to its inherent global unitarity. Some earlier works have also utilized unitarity at the local tensor level~\cite{ferris2012perfect, pollmann2016efficient, wahl2017efficient, liu2019machine, haghshenas2021optimization, haghshenas2022variational, Styliaris_2025}. Their proposed tensor networks are restricted by fixed architectures or specific local gate sets, preventing the representation of certain unitaries with finite bond dimensions. Unitary networks, as we will show, correspond to a sufficiently versatile architecture and are expected to represent a broader range of global unitary operators.

This paper is organized as follows.
Section~\ref{Sec: Quantum cellular automata and the GNVW index} provides a review of QCAs and the GNVW index. Section~\ref{Sec: UN} outlines the structure of the unitary network and its key properties. We show that a directed acyclic graph structure combined with local unitary tensors guarantees global unitarity. Section~\ref{Sec: Unitary network with different boundary conditions} covers unitary networks in infinite OBC systems and PBC systems. In Section~\ref{Sec: QCA and unitary networks}, we show that any 1D QCA can be represented by a unitary network. We then define locality-preserving unitary networks. Section~\ref{Sec: Non-local unitaries and Unitary networks} provides examples demonstrating that unitary networks of finite bond dimension are capable of representing non-local unitaries. In Section~\ref{Sec: Information flow and the GNVW index}, we introduce the net information flow for unitary networks. With a locality-preserving unitary network representation, the net information flow becomes an intrinsic characteristic, aligned with the GNVW index for the QCAs being represented. Section~\ref{Sec: UN and SQC} examines the connection between unitary networks and SQCs. Unitary networks with finite bond dimension are shown to be equivalent to SQCs in terms of representability. Section~\ref{Sec: reduce bond dimension} investigates methods for decomposing a global unitary into an efficient unitary network with reduced bond dimensions.

\section{Quantum cellular automata and the GNVW index}
\label{Sec: Quantum cellular automata and the GNVW index}

Given the extensive use of QCA properties in this paper, a brief review of QCAs is provided in this section. We review the Margolus partitioning scheme of a QCA and the definition of the GNVW index. We also review the definition of approximately locality preserving unitaries (ALPUs) and their indices. Readers who are familiar with these notions can skip this section.

\subsection{Operator Algebra}

 Consider a cubic lattice $\Gamma = \mathbb{Z}^s$ with spatial dimension $s$. Lattice sites are described by integer vectors $\vec{n} \in \Gamma$. At each site, there is a corresponding Hilbert space $\mathcal{H}_{\vec{n}}$. The general Hilbert space is the tensor product of Hilbert spaces for all sites: $\mathcal{H} = \bigotimes_{\vec{n}} \mathcal{H}_{\vec{n}}$. However, when dealing with a system that has an infinite number of sites, the tensor product structure becomes unclear. That's why we take the quasi-local algebra approach.

Instead of the Hilbert space, we will consider the algebra of local observables $\mathcal{A}_{X} = \bigotimes_{\vec{n} \in X} \mathcal{A}_{\vec{n}}$, which is defined on a finite set of sites $X \in \mathbb{Z}^s$. For two subsets $X \subseteq X^\prime$ there is a natural inclusion $\mathcal{A}_{X} \subseteq \mathcal{A}_{X^{\prime}}$, where $\mathcal{A}_{X}$ is identified with $\mathcal{A}_{X} \otimes I_{X^{\prime} \backslash X}$ (tensoring with the identity in $X^{\prime} \backslash X$). We may define the algebra of all strictly local operators as follows: 
\begin{equation}
    \mathcal{A}_{\Gamma}^{\text{strict}} = \bigcup_{X \subseteq \Gamma \ \text{finite} }  \mathcal{A}_{X},
\end{equation}
whose norm completion is called the quasi-local  algebra $\mathcal{A}_{\Gamma}$.  For a more complete discussion of operator algebra, refer to Ref.~\cite{bratteli2012operator}.

When examining the algebra supported on multiple sites, it is useful to define the support algebra \cite{zanardi2002stabilizationquantuminformationunified, schumacher2004reversiblequantumcellularautomata}:
\begin{definition}
   For an algebra $\mathcal{A} \subseteq \bigotimes_{i} \mathcal{B}_{i}$ across multiple sites, the support algebra $\boldsymbol{S}(\mathcal{A}, \mathcal{B}_{i})$ is the minimum subalgebra in $\mathcal{B}_{i}$ required to construct elements of $\mathcal{A}$.:
    \begin{equation}
        \mathcal{A} \subseteq \bigotimes_{i} \boldsymbol{S}(\mathcal{A}, \mathcal{B}_{i})\subseteq \bigotimes_{i} \mathcal{B}_{i}
    \end{equation}
\end{definition}

\subsection{Automorphism}

An automorphism $u$ can be considered as an extension of unitary operations from finite to infinite systems. In a finite system, an automorphism $u$ can always be represented by
\begin{equation}
    u(O) = UOU^{\dagger}
\end{equation}
where $U$ is a unitary operator. However, within an infinite system like the one-dimensional chain $\mathbb{Z}$, a unitary $U$ might not exist for a given automorphism. Throughout the paper, we will still construct the automorphism $u$ as $u(O) = UOU^{\dagger}$, where $U$ could represent a unitary network $U_{Net}$ on the infinite system. For $U$ on an infinite system and its associated automorphism $u$, we sometimes still call it a unitary operator for convenience.

\subsection{Locality preserving and QCAs}

We provide the definitions of locality preservation and quantum cellular automata (QCAs) here. For a more comprehensive review on QCAs, refer to Ref.~\cite{Farrelly_2020}.

\begin{definition} [Locality preserving and QCAs]
    An automorphism $u: \mathcal{A}_{\Gamma} \rightarrow \mathcal{A}_{\Gamma}$ is locality preserving if there is some $R>0$, such that
     \begin{equation}
         u(O_{X})  \in  \mathcal{A}_{\bar{B}(X,R)} \quad \text{for} \ O_{X} \in \mathcal{A}_{X},
     \end{equation}
     where $\bar{B}(X,R)$ is the union of closed ball:
     \begin{equation}
     \begin{aligned}
         \bar{B}(X,R) &= \{ n\in \Gamma: d(n,X) \le R\} \\
         &= \cup_{x \in X} \{  n\in \Gamma: d(n,x) \le R\}
     \end{aligned}
     \end{equation}
     A quantum cellular automaton (QCA) \cite{schumacher2004reversiblequantumcellularautomata, perez2005models, perez2007local} is a locality-preserving automorphism defined on a lattice $\Gamma$.  $R$ is called the radius of the QCA.
\end{definition}

Some literature \cite{schumacher2004reversiblequantumcellularautomata, perez2007local} also requires that a QCA exhibit translational invariance, expressed as:
\begin{equation}
    u \cdot \tau_{x} = \tau_{x} \cdot u,
\end{equation}
where $\tau_{x}$ is the shift operation $\tau_{x}: \mathcal{A}_{Y} \rightarrow \mathcal{A}_{Y+x}$, such that:
\begin{equation}
    \tau_{x}(O_{y}) = T_{x} O_{y} T_{-x} = O_{x+y}.
\end{equation}
Throughout the paper, QCAs that meet the TI condition will be termed TI-QCAs.

\subsection{Margolus partitioning} \label{subsection: Margolus partitioning}

The Margolus neighborhood scheme of (classical) Cellular Automata was introduced in \cite{toffoli1990invertible}. The Margolus partitioning construction was first introduced in \cite{schumacher2004reversiblequantumcellularautomata} for TI QCAs. However, the discussions in \cite{Gross_2012, ranard2022converse} establish that translational invariance is not essential.

Margolus partitioning dictates that a nearest-neighbor QCA $u$ can be constructed in two steps:
\begin{equation}
    u(O) = (v \cdot w)(O),
\end{equation}
\begin{figure}
    \centering
\includegraphics[width=\linewidth]{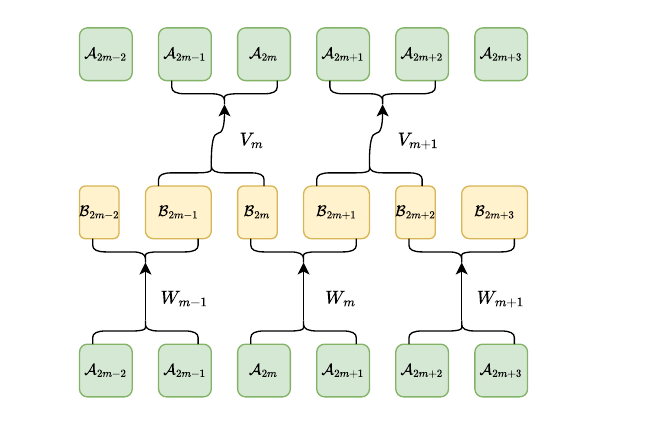}
    \caption{Any nearest-neighbor QCA $u$ can be constructed by a two-step Margolus partitioning scheme. This diagram is adapted from \cite{Farrelly_2020}.}
    \label{fig:Margolus partitioning}
\end{figure}
where $v$ and $w$ each can be decomposed into an automorphism on supercells that contains an alternate Margolus neighborhood. Fig.~\ref{fig:Margolus partitioning} illustrates a Margolus partitioning scheme for QCA construction in a one-dimensional chain. The supercells defined for $w$, $\{2m, 2m+1\}$, differ from those for $v$: $\{2m-1, 2m\}$:
\begin{equation}
\begin{aligned}
    w &= \bigotimes_{m} w_{m}, \quad  v = \bigotimes_{m} v_{m}, \\
   &  w_{m}: \mathcal{A}_{2m} \otimes \mathcal{A}_{2m+1} \rightarrow \mathcal{B}_{2m} \otimes \mathcal{B}_{2m+1}, \\ 
     &   v_{m}: \mathcal{B}_{2m-1} \otimes \mathcal{B}_{2m} \rightarrow \mathcal{A}_{2m-1} \otimes \mathcal{A}_{2m}.
\end{aligned}
\end{equation}
The intermediate algebras $\mathcal{B}_{n}$ usually vary from the observable algebra $\mathcal{A}_{n}$, yet they are proved to be isomorphic to the algebra of $b_n \times b_n$ complex matrices, i.e., $\mathcal{B}_{n} \cong \mathcal{M}_{b_{n}}$ \cite{schumacher2004reversiblequantumcellularautomata}. In general $\mathcal{B}_{2m} \ncong \mathcal{B}_{2m+1}$.

\subsection{GNVW index}
Ref.~\cite{Gross_2012} introduced the GNVW index of QCAs to quantify the ``net flow of quantum information" across the system. Two QCAs $u_{1}$ and $u_{2}$ can be concatenated or continuously deformed into each other only when their GNVW indices are identical.

\begin{definition} [GNVW index]
    Following the Margolus partitioning scheme of QCAs, suppose that the observable algebra $\mathcal{A}_{n}$ is isomorphic to the algebra for $a_n \times a_{n}$ complex matrices $\mathcal{M}_{a_{n}}$, and the intermediate algebra $\mathcal{B}_{n}$ is isomorphic to the algebra for $b_n \times b_{n}$ complex matrices $\mathcal{M}_{b_{n}}$. We denote
\begin{equation}
    \dim \mathcal{A}_{n} = a_{n}^{2}, \ \dim \mathcal{B}_{n} = b_{n}^{2}.
\end{equation}

The GNVW index \cite{Gross_2012} of a QCA is defined as 
\begin{equation}
    I_{GNVW}(u) \coloneq \sqrt{\frac{\dim \mathcal{B}_{2m}}{\dim \mathcal{A}_{2m}}} = \frac{b_{2m}}{a_{2m}}.
\end{equation}
\end{definition}
Focusing instead on the Hilbert spaces $\mathcal{H}_{\mathcal{A}_{n}}$ and $\mathcal{H}_{\mathcal{B}_{n}}$ where operator algebras $\mathcal{A}_{n}$ and $\mathcal{B}_{n}$ operate. Notice that $a_{2m} = \dim \mathcal{H}_{\mathcal{A}_{2m}}$ and $b_{2m} = \dim \mathcal{H}_{\mathcal{B}_{2m}}$ actually denote Hilbert space dimensions. We take the logarithm of the GNVW index to obey the additive nature of a flow:
\begin{equation}
\begin{aligned}
       \log_{d} I_{GNVW} (u)
       = \log_{d}{\dim \mathcal{H}_{\mathcal{B}_{2m}}} - \log_{d}{\dim \mathcal{H}_{\mathcal{A}_{2m}}} \\
\end{aligned}
\end{equation}

\subsection{Approximately locality preserving unitaries}
Among global unitary operators that violate the locality preserving condition, some unitary operators have approximate causal cones, and are called \textit{approximately locality preservation unitaries} (ALPU) \cite{ranard2022converse}. ALPUs are an important class of unitary operators, because time evolution by a local Hamiltonian will not have a strict causal cone;
instead, the dynamic given by local Hamiltonian evolutions satisfies the Lieb-Robinson bounds \cite{lieb1972finite}.

\begin{definition} [Near inclusion]
    An operator $a$ is nearly included \cite{ranard2022converse} in the algebra $\mathcal{B}$, denoted as $a \overset{\epsilon}{\in} \mathcal{B}$, if 
    \begin{equation}
        \exists b \in \mathcal{B}, \|a-b\| \le \epsilon \|a\|.
    \end{equation}
A algebra $\mathcal{A}$ is nearly included in the algebra $\mathcal{B}$ if all its elements are nearly included in $\mathcal{B}$:
\begin{equation}
    \forall a \in \mathcal{A}, a \overset{\epsilon}{\in} \mathcal{B}.
\end{equation}
\end{definition}

\begin{definition} [ALPU]
    An ALPU \cite{ranard2022converse} $\alpha$ with $f(r)$-tails is
    an automorphism of the quasi-local algebra such that
    \begin{equation}
        \alpha(O_{X}) \overset{f(r)}{\in} \mathcal{A}_{B(X,r)}.
    \end{equation}
where $f (r)$ is some positive function that satisfied $\lim_{r\rightarrow\infty} f(r)=0$.
\end{definition}

Ref.~\cite{ranard2022converse} shows that any 1D ALPU can be approximated by a sequence of QCAs:
\begin{proposition} [QCA approximations of an ALPU]
\label{prop: QCA approximations of an ALPU}
    An 1D ALPU $\alpha$ with $f (r)$-tails can be approximated by a sequence of QCAs $\{\beta_{j}\}_{j=1}^{\infty}$ \cite{ranard2022converse}, such that $\beta_{j}$ has radius $2j$,  and for any finite subset $X \subseteq \Gamma$:
    \begin{equation}
        \| (\alpha - \beta_{j})|_{\mathcal{A}_{X}}\| \le C_{f} \cdot f(j) \cdot \min \{|X|, \frac{ \left\lceil\text{diam}(X) \right \rceil}{j}\},
    \end{equation}
    where $C_{f}$ is a constant determined by the tails $f(r)$, $|X|$ indicates the number of sites in subset $X$, and $\text{diam}(X)$ represents the diameter of $X$.
\end{proposition}
The GNVW index applicable to QCAs extends to ALPUs as follows:
\begin{definition}[Index for ALPUs]
\label{def: Index for ALPUs}
    The index of an ALPU $\alpha$ is defined in Ref.~\cite{ranard2022converse} defines as:
    \begin{equation}
        I(\alpha) \coloneq \lim_{j \rightarrow \infty} I_{GNVW}(\beta_{j}).
    \end{equation}
\end{definition}
Ref.~\cite{ranard2022converse} showed that $I_{GNVW}(\beta_{j})$ stabilizes at large $j$ such that
\begin{equation}
    I_{GNVW}(\beta_{j_{1}}) = I_{GNVW}(\beta_{j_{2}}) \quad \text{for} \  j_{1}, j_{2} \ge j_{0}.
\end{equation}

\section{Unitary Network}
\label{Sec: UN}

This section outlines the architecture of unitary networks and highlights their distinctions from MPUs. Within a unitary network, the edges are oriented and the local tensors are unitary. Employing local unitary tensors in combination with a directed acyclic graph ensures the unitarity of the global tensor. We proposed an architecture named bilayer unitary networks. We demonstrate that without constraints on bond dimensions, they can represent any unitary on finite systems. We demonstrate that a unitary network generates a tensor network state when applied to a product state.

\subsection{Local unitary tensor}

The building block of a unitary network is a local unitary tensor. A unitary tensor is similar to a standard tensor, but its legs are directed, either incoming or outgoing. 
\begin{definition}[unitary tensor]
    A unitary tensor $U$ can be transformed into a unitary matrix $\tilde{U}$ by merging all its incoming legs into one incoming leg, and merging all its outgoing legs into one outgoing leg \cite{ferris2012perfect}:
    \begin{align} 
    \label{Eq: unitary tensor}
    \begin{array}{c}
    \tikzsetnextfilename{unitary_tensor}
        \begin{tikzpicture}
            \node[draw, rectangle, rounded corners=2pt, minimum width=0.6cm, minimum height=0.6cm, fill=tensorcolor, thick] (U)  {$U$};
            \draw[mid arrow=0.65] (0.3, 0.1) -- node[above] {$n$} (0.9, 0.1) ;
            \draw[mid arrow=0.65] (-0.9, 0.1) -- node[above] {$i$}(-0.3, 0.1);
            \draw[mid arrow=0.65] (0.9, -0.1) -- node[below]{$k$} (0.3, -0.1);
            \draw[mid arrow=0.65] (-0.3, -0.1) -- node[below]{$l$} (-0.9, -0.1);
            \draw[mid arrow=0.65] (0, 0.3) -- node[right]{$m$}(0, 0.9);
            \draw[mid arrow=0.65] (0, -0.9) -- node[right]{$j$}(0, -0.3);
        \end{tikzpicture}     
    \end{array}
        = 
    \begin{array}{c}
        \tikzsetnextfilename{unitary_matrix}
        \begin{tikzpicture}
             \node[draw, rectangle, rounded corners=2pt, minimum width=1.2cm, minimum height=0.6cm, fill=tensorcolor, thick] (a){$\Tilde{U}$};
             \draw[mid arrow=0.65] (-0.3, 0.3) -- (-0.3, 0.9) node[draw=none,fill=none,font=\scriptsize] at (-0.3, 1.1) {$l$};
             \draw[mid arrow=0.65] (0, 0.3) -- (0, 0.9) node[draw=none,fill=none,font=\scriptsize] at (0, 1.1) {$m$};
            \draw[mid arrow=0.65] (0.3, 0.3) -- (0.3, 0.9) node[draw=none,fill=none,font=\scriptsize] at (0.3, 1.1) {$n$};
            \draw[mid arrow=0.65] (-0.3, -0.9) -- (-0.3, -0.3) node[draw=none,fill=none,font=\scriptsize] at (-0.3, -1.1) {$i$};
             \draw[mid arrow=0.65] (0, -0.9) -- (0, -0.3) node[draw=none,fill=none,font=\scriptsize] at (0, -1.1) {$j$};
            \draw[mid arrow=0.65] (0.3, -0.9) -- (0.3, -0.3) node[draw=none,fill=none,font=\scriptsize] at (0.3, -1.1) {$k$};
        \end{tikzpicture} 
    \end{array} =
    \begin{array}{c}
        \tikzsetnextfilename{unitary_matrix_legs_merged}
        \begin{tikzpicture}
             \node[draw, rectangle, rounded corners=2pt, minimum width=0.6cm, minimum height=0.6cm, fill=tensorcolor, thick] (a){$\Tilde{U}$};
             \draw[mid arrow=0.65,thick] (0, 0.3) -- (0, 0.9) node[draw=none,fill=none,font=\scriptsize] at (0, 1.1) {$lmn$};
             \draw[mid arrow=0.65, thick] (0, -0.9) -- (0, -0.3) node[draw=none,fill=none,font=\scriptsize] at (0, -1.1) {$ijk$};
        \end{tikzpicture} 
    \end{array}
    ,
    \end{align}
satisfying
\begin{equation}
    \tilde{U}\tilde{U}^{\dagger} = \tilde{U}^{\dagger}\tilde{U} = I.
\end{equation}
\end{definition}
In a unitary tensor, each leg $e$ corresponds to a Hilbert space $\mathcal{H}_{e}$, and we recognize the dimension of the Hilbert space $\dim \mathcal{H}_{e}$ as that of the legs. Later, in a unitary network, contracting two tensor legs $i,j$ to create a bond results in a bond dimension of
\begin{equation}
    D = \dim \mathcal{H}_{i} = \dim \mathcal{H}_{j},
\end{equation}
aligning with conventional tensor network terminology.

The outgoing and incoming legs are denoted by upper and lower indices, respectively. For example, the local unitary tensor in (\ref{Eq: unitary tensor}) is denoted as:
\begin{equation}
    U^{lmn}_{ijk}.
\end{equation}
As with any tensor, the legs in unitary tensors can be split or merged. For a unitary tensor, the legs to be merged must share the same direction. Appendix \ref{Append: Splitting or merging legs} provides a brief review of the splitting or merging legs.

\subsection{Unitary network}

The contraction of unitary tensors adheres to the same principles as the contraction of general tensors. However, an extra constraint exists for contracting unitary tensors: An incoming leg must contract with an outgoing leg. A brief overview of unitary tensor contraction can be found in Appendix \ref{Append: Unitary tensor contraction}. A unitary network results from the contraction of a set of unitary tensors. 

\begin{definition} [Unitary network]
    A unitary network, denoted as $U_{Net}$, is constructed by connecting a set of unitary tensors. Upon evaluating this unitary network, all unitary tensors are contracted, producing a global tensor $U$ represented by: \begin{equation}
        U = \text{Eval}(U_{Net}),
    \end{equation}
    which may not be a unitary tensor in general.
    
    As suggested by the term network, each unitary network $U_{Net}$ corresponds to a directed graph $G[U_{Net}] = (V,E)$. The vertex set $V$ includes all interior nodes as well as sources and sinks:
    \begin{equation}
        V = \{ \text{internal nodes}\} \cup \{\text{sources}\} \cup \{\text{sinks} \}.
    \end{equation}
    An internal node symbolizes a local unitary tensor, featuring both incoming and outgoing connections. In a unitary network, the external outgoing edges point to sinks, which are nodes without outgoing edges. In contrast, sources provide the external incoming legs. The edge set $E \subseteq V \times V$ consists of all connections in a unitary network:
\begin{equation}
    \begin{aligned}
        E = &\{\text{directed bond legs between local tensors}\} \\
        &\cup \{ \text{external outgoing legs pointing to sinks}\} \\
        &\cup \{ \text{all incoming external legs originating from sources}\}.
    \end{aligned}
\end{equation}
\end{definition}
An example of a unitary network $U_{Net}$ and its graph $G[U_{Net}]$ is depicted below:
\begin{align}
\begin{array}{c}
\tikzsetnextfilename{UN_graph}
    \begin{tikzpicture}
    \node[draw, circle, minimum width=0.6cm, minimum height=0.6cm, fill=tensorcolor, thick] (in) at (0,-1.2) {$i$};
    \node[draw, circle, minimum width=0.6cm, minimum height=0.6cm, fill=tensorcolor, thick] (out1) at (-1,2.4) {$o_{1}$};
    \node[draw, circle, minimum width=0.6cm, minimum height=0.6cm, fill=tensorcolor, thick] (out2) at (1,2.4) {$o_{2}$};
   \node[draw, rectangle, rounded corners=2pt, minimum width=0.6cm, minimum height=0.6cm, fill=tensorcolor, thick] (A) at (0,0) {$A$};
   \node[draw, rectangle, rounded corners=2pt, minimum width=0.6cm, minimum height=0.6cm, fill=tensorcolor, thick] (B) at (1,1) {$B$};
   \node[draw, rectangle, rounded corners=2pt, minimum width=0.6cm, minimum height=0.6cm, fill=tensorcolor, thick] (C) at (-1,1) {$C$};
    \draw[mid arrow=0.65] (C.90) -- (out1.270);
    \draw[mid arrow=0.65] (B.90) -- (out2.270);
    \draw[mid arrow=0.65] (in.90) to [out = 90,in=270]  (A.270);
    \draw[mid arrow=0.65] (A.0) to [out = 0,in=270]  (B.270);
    \draw[mid arrow=0.65] (C.270) to [out = 270,in=180]  (A.180);
    \draw[mid arrow=0.65] ([yshift = 0.1cm] B.180) to [out = 180,in=0]  ([yshift = 0.1cm] C.0);
    \draw[mid arrow=0.65] ([yshift = -0.1cm] C.0) to [out = 0,in=180]  ([yshift = -0.1cm] B.180);
\end{tikzpicture}
\end{array},
\end{align}
whose vertex and edge sets are
\begin{equation}
\begin{aligned}
    V =& \{ A,B,C,i, o_{1}, o_{2}  \}, \\
    E =& \{ (i, A), (A,B), (B,C), (B, o_{2}), \\
    &(C, A), (C,B), (C,o_{1})  \}. 
\end{aligned}
\end{equation}

More often, we omit the source and sink vertices, resulting in dangling external edges: 
\begin{align} \label{eq: unitary network example}
\begin{array}{c}
\tikzsetnextfilename{Unitary_network}
    \begin{tikzpicture}
   \node[draw, rectangle, rounded corners=2pt, minimum width=0.6cm, minimum height=0.6cm, fill=tensorcolor, thick] (A) at (0,0) {$A$};
   \node[draw, rectangle, rounded corners=2pt, minimum width=0.6cm, minimum height=0.6cm, fill=tensorcolor, thick] (B) at (1,1) {$B$};
   \node[draw, rectangle, rounded corners=2pt, minimum width=0.6cm, minimum height=0.6cm, fill=tensorcolor, thick] (C) at (-1,1) {$C$};
    \draw[mid arrow=0.65] (A.0) to [out = 0,in=270]  (B.270);
    \draw[mid arrow=0.65] (C.270) to [out = 270,in=180]  (A.180);
    \draw[mid arrow=0.65] ([yshift = 0.1cm] B.180) to [out = 180,in=0]  ([yshift = 0.1cm] C.0);
    \draw[mid arrow=0.65] ([yshift = -0.1cm] C.0) to [out = 0,in=180]  ([yshift = -0.1cm] B.180);
    \draw[mid arrow=0.65] (A.270)+(0,-0.6) --  (A.270);
    \draw[mid arrow=0.65] (B.90) --  ++(0,0.6);
    \draw[mid arrow=0.65] (C.90) --  ++(0,0.6);
\end{tikzpicture}
\end{array}.
\end{align}

By removing and adding sources and sinks when necessary, we can define the sub-network of a unitary network as follows: 
\begin{definition} [Sub-network of a unitary network]
Consider a unitary network $U_{Net}$ with an associated graph $G[U_{Net}]=(V,E)$. A sub-network $U_{Net}^{S}$ consists of the subgraph $G(U_{Net}^{S}) = (V^{S}, E^{S}) \subseteq G[U_{Net}]$, made up of certain internal nodes (local unitary tensors) and edges (legs) from $U_{Net}$.
\end{definition}
An example of a sub-network for the unitary network in \eqref{eq: unitary network example} is illustrated below.

\begin{align} \label{eq: sub unitary network example}
\begin{array}{c}
\tikzsetnextfilename{sub_UN}
    \begin{tikzpicture}
   \node[draw, rectangle, rounded corners=2pt, minimum width=0.6cm, minimum height=0.6cm, fill=tensorcolor, thick] (A) at (0,0) {$A$};
   \node[draw, rectangle, rounded corners=2pt, minimum width=0.6cm, minimum height=0.6cm, fill=tensorcolor, thick] (B) at (1,1) {$B$};
    \draw[mid arrow=0.65] (A.0) to [out = 0,in=270]  (B.270);
    \draw[mid arrow=0.65]   ([xshift=-0.6cm]A.180) -- (A.180);
    \draw[mid arrow=0.65] ([yshift = 0.1cm] B.180) -- ++(-0.6, 0);
    \draw[mid arrow=0.65] ([xshift=-0.6cm, yshift = -0.1cm] B.180) to [out = 0,in=180]  ([yshift = -0.1cm] B.180);
    \draw[mid arrow=0.65] (A.270)+(0,-0.6) --  (A.270);
    \draw[mid arrow=0.65] (B.90) --  ++(0,0.6);
\end{tikzpicture}
\end{array}.
\end{align}
Sub-networks feature external legs disconnected from physical sites, but associated with bond degrees of freedom. The Hilbert spaces of these external legs can be labeled by the vertices they attach to. For example, $\mathcal{H}_{v,i}^{in}$ and $\mathcal{H}_{v,j}^{out}$ denote the Hilbert spaces for the $i$-th input leg and $j$-th output leg at vertex $v$. The corresponding observable algebras are denoted by $\mathcal{A}_{v,i}^{in}$ and $\mathcal{A}_{v,j}^{out}$.

\subsection{Global unitarity}

It may be tempting to assume that the unitarity of local unitary tensors automatically ensures the unitarity of the global tensor. However, this is not the case. It can be readily shown that certain unitary networks do not form a global unitary operator. Consider the following example:
\begin{align}
    \tikzsetnextfilename{Self_contraction}
    \begin{tikzpicture}
        \node[draw, rectangle, rounded corners=2pt, minimum width=0.8cm, minimum height=0.6cm, fill=tensorcolor, thick] (I) {$I$};
        \draw[mid arrow=0.6] ([xshift=0.2cm]I.90) -- ++(0,0.4) -- ++(0.4,0) --  ++(0,-1.4) -- ++(-0.4,0) -- ([xshift=0.2cm]I.270);
        \draw[mid arrow=0.65]([xshift=-0.2cm]I.90) -- ++(0,0.4);
        \draw[mid arrow=0.65] ([xshift=-0.2cm, yshift=-0.4cm]I.270) -- ([xshift=-0.2cm]I.270);
    \end{tikzpicture}
\end{align}
The identity matrix for the combined system $A + B$ is shown above, followed by the partial trace performed on subsystem $B$. The resulting tensor can be expressed as:
\begin{equation}
    \text{Tr}_{B}{I} = \dim \mathcal{H}_{B} \cdot I_{A},
\end{equation}
where $\dim \mathcal{H}_{B}$ represents the dimension of the Hilbert space linked to subsystem B. Evidently, this is not a unitary tensor. As will be discussed, a unitary network can generate a non-unitary global tensor if its directed graph has loops.

\begin{definition} [Directed Acyclic Graph]
    A Directed Acyclic Graph (DAG) \cite{manber1989introduction} is a directed graph in which no directed paths form a cycle.
\end{definition}
DAGs are widely used across various domains and possess the following nice property:
\begin{lemma} [Topologically sorting for DAG]
For any DAG, we can find a total ordering $<$ of vertices in a directed cyclic graph (DAG) such that for every directed edge $(A,B) \in E$, we have $A < B$ in the ordering. The proof and the algorithm for topological sorting are provided in Ref.~\cite{manber1989introduction}. 
\end{lemma}

\begin{proposition} [Unitarity of unitary network] 
\label{prop: Unitarity of unitary network}
A unitary network $U_{Net}$ forms a global unitary tensor if its corresponding graph $G[U_{Net}] = (V,E)$ is a DAG. Additionally, all its sub-networks $U^{S}_{Net}$ are global unitaries.

Conversely, if the graph of a unitary network contains directed loops, the overall tensor typically loses its unitary property.
\end{proposition}

\begin{proof}
We begin with a topological sorting of the DAG to establish a strict total order $<$ of the vertices (local unitary tensors). This ordering is then treated as a causal sequence of unitary operations: beginning with sources $in$, symbolizing the incoming Hilbert space, we perform the local unitary operations in sequence. Observe that edge propagation merely involves permutations within the Hilbert space, which are unitary operations. This ensures that the global tensor remains unitary. 

For a DAG, any subgraph remains a DAG, ensuring the sub-network forms a global unitary.
\end{proof}
Illustrated here is an example of a unitary network structured as a DAG. Following the topological sort, the sequence of each local unitary tensor is shown in parentheses. Unitarity is preserved at each step of applying a local unitary tensor.
\begin{align}
\label{Eq:DAG=circuit}
\begin{array}{c}
    \tikzsetnextfilename{UN_DAG}
    \begin{tikzpicture}
    \node[draw, circle, minimum width=0.6cm, minimum height=0.6cm, fill=tensorcolor, thick] (in) at (0,-1.2) {$i$};
    \node[draw, circle, minimum width=0.6cm, minimum height=0.6cm, fill=tensorcolor, thick] (out1) at (-1,2.4) {$o_{1}$};
    \node[draw, circle, minimum width=0.6cm, minimum height=0.6cm, fill=tensorcolor, thick] (out2) at (1,2.4) {$o_{2}$};
   \node[draw, rectangle, rounded corners=2pt, minimum width=0.6cm, minimum height=0.6cm, fill=tensorcolor, thick] (A) at (0,0) {$A (1)$};
   \node[draw, rectangle, rounded corners=2pt, minimum width=0.6cm, minimum height=0.6cm, fill=tensorcolor, thick] (B) at (1,1) {$B (2)$};
   \node[draw, rectangle, rounded corners=2pt, minimum width=0.6cm, minimum height=0.6cm, fill=tensorcolor, thick] (C) at (-1,1) {$C (3)$};
    \draw[mid arrow=0.65] (C.90) to [out = 90,in=270]  (out1.270);
    \draw[mid arrow=0.65] (B.90) to [out = 90,in=270]  (out2.270);
    \draw[mid arrow=0.65] (in.90) to [out = 90,in=270]  (A.270);
    \draw[mid arrow=0.65] (A.0) to [out = 0,in=270]  (B.270);
    \draw[mid arrow=0.65] (A.180) to [out = 180,in=270]  (C.270);
    \draw[mid arrow=0.65] (B.180) to [out = 180,in=0]  (C.0);
\end{tikzpicture}
\end{array} = 
\begin{array}{c}
\tikzsetnextfilename{UN_circuit}
\begin{tikzpicture}
    \node[draw, circle, minimum width=0.6cm, minimum height=0.6cm, fill=tensorcolor, thick] (in) at (0,-1.2) {$i$};
    \node[draw, circle, minimum width=0.6cm, minimum height=0.6cm, fill=tensorcolor, thick] (out1) at (0,3.6) {$o_{1}$};
    \node[draw, circle, minimum width=0.6cm, minimum height=0.6cm, fill=tensorcolor, thick] (out2) at (0.9,3.6) {$o_{2}$};
   \node[draw, rectangle, rounded corners=2pt, minimum width=1cm, minimum height=0.6cm, fill=tensorcolor, thick] (A) at (0,0) {$A (1)$};
   \node[draw, rectangle, rounded corners=2pt, minimum width=0.6cm, minimum height=0.6cm, fill=tensorcolor, thick] (B) at (0.6,1.2) {$B (2)$};
   \node[draw, rectangle, rounded corners=2pt, minimum width=1cm, minimum height=0.6cm, fill=tensorcolor, thick] (C) at (0,2.4) {$C (3)$};
    \draw[mid arrow=0.65] (C.90) to [out = 90,in=270]  (out1.270);
    \draw[mid arrow=0.65] ([xshift=0.3cm]B.90) to [out = 90,in=270]  (out2.270);
    \draw[mid arrow=0.65] (in.90) to [out = 90,in=270]  (A.270);
    \draw[mid arrow=0.65] ([xshift=0.3cm]A.90) --  ([xshift=-0.3cm]B.270);
    \draw[mid arrow=0.65] ([xshift=-0.3cm]A.90) to  ([xshift=-0.3cm]C.270);
    \draw[mid arrow=0.65] ([xshift=-0.3cm]B.90) --  ([xshift=0.3cm]C.270);
\end{tikzpicture}
\end{array}
.
\end{align}
The inverse of proposition \ref{prop: Unitarity of unitary network} is not true. With directed loops, a unitary network may still constitute a global unitary. For example, starting from (\ref{Eq:DAG=circuit}), by first contracting $C$ with $A$ into $CA$, we obtain the following digram,
\begin{align}
\label{Eq:contract CA}
\begin{array}{c}
    \tikzsetnextfilename{UN_DAG_contracted}
    \begin{tikzpicture}
    \node[draw, circle, minimum width=0.6cm, minimum height=0.6cm, fill=tensorcolor, thick] (in) at (-1,-0.2) {$i$};
    \node[draw, circle, minimum width=0.6cm, minimum height=0.6cm, fill=tensorcolor, thick] (out1) at (-1,2.2) {$o_{1}$};
    \node[draw, circle, minimum width=0.6cm, minimum height=0.6cm, fill=tensorcolor, thick] (out2) at (1,2.2) {$o_{2}$};
   \node[draw, rectangle, rounded corners=2pt, minimum width=0.6cm, minimum height=0.6cm, fill=tensorcolor, thick] (B) at (1,1) {$B$};
   \node[draw, rectangle, rounded corners=2pt, minimum width=0.6cm, minimum height=0.6cm, fill=tensorcolor, thick] (CA) at (-1,1) {$CA$};
    \draw[mid arrow=0.65] (CA.90) to [out = 90,in=270]  (out1.270);
    \draw[mid arrow=0.65] (B.90) to [out = 90,in=270]  (out2.270);
    \draw[mid arrow=0.65] (in.90) to [out = 90,in=270]  (CA.270);
    \draw[mid arrow=0.65] ([yshift=-0.1cm]CA.0) to [out = 0,in=180]  ([yshift=-0.1cm]B.180);
    \draw[mid arrow=0.65] ([yshift=0.1cm]B.180) to [out = 180,in=0]  ([yshift=0.1cm]CA.0);
\end{tikzpicture}
\end{array}
\end{align}
which clearly contains a directed loop $CA \rightarrow B, B\rightarrow CA$. Since it originates from a DAG, it inherently possesses global unitarity. In some sense, this directed loop in the graph is not intrinsic.  This situation occurs when later a unitary network is utilized to model a Quantum Cellular Automata on a PBC system.

Although the topological sorting of local unitary tensors is employed to demonstrate global unitarity and aids in converting the unitary network into a quantum circuit, it is not mandatory to adhere to this order during contraction. In evaluating a unitary network $U = \text{Eval}(U_{Net})$, the contraction sequence can be freely selected to reduce computational cost. In the preceding example (\ref{Eq:contract CA}), C is first contracted with A. Although this contraction creates directed loops, it does not impact global unitarity.

\subsection{Bilayer unitary network}
\label{subsection: Bilayer}

In general, the graph of a unitary network can be arbitrary. However, when describing a global unitary on the lattice system $\mathbb{Z}^{D}$, it is natural for the unitary network to adopt a lattice structure, essentially forming a non-uniform MPU \cite{Styliaris_2025} (or 1D PEPU \cite{Piroli_2020}) with directed legs. The following is a one-dimensional example:
\begin{align}
\tikzsetnextfilename{UN_graph_with_loops2}
    \begin{tikzpicture}
        \node[draw, rectangle, rounded corners=2pt, minimum width=0.6cm, minimum height=0.6cm, fill=tensorcolor, thick](A) at (-1.2, 0) {};
        \node[draw, rectangle, rounded corners=2pt, minimum width=0.6cm, minimum height=0.6cm, fill=tensorcolor, thick](B) {};
        \node[draw, rectangle, rounded corners=2pt, minimum width=0.6cm, minimum height=0.6cm, fill=tensorcolor, thick](C) at (1.2, 0) {};
        \draw[mid arrow=0.65] (A.90) -- ++(0,0.6);
        \draw[mid arrow=0.65] (B.90) -- ++(0,0.6);
        \draw[mid arrow=0.65] (C.90) -- ++(0,0.6);
        \draw[mid arrow=0.65] ([yshift=-0.6cm]A.270) -- (A.270);
        \draw[mid arrow=0.65] ([yshift=-0.6cm]B.270) -- (B.270);
        \draw[mid arrow=0.65] ([yshift=-0.6cm]C.270) -- (C.270);
        \draw[mid arrow=0.65] ([yshift=0.1cm]A.0) -- ([yshift=0.1cm]B.180);
        \draw[mid arrow=0.65] ([yshift=0.1cm]B.0) -- ([yshift=0.1cm]C.180);
        \draw[mid arrow=0.65] ([yshift=-0.1cm]C.180) -- ([yshift=-0.1cm]B.0);
        \draw[mid arrow=0.65] ([yshift=-0.1cm]B.180) -- ([yshift=-0.1cm]A.0);
    \end{tikzpicture}
\end{align}
However, the single-layer architecture above features a directed loop, thus the unitarity of the global tensor may not be guaranteed. To tackle this issue, we propose an efficient architecture for a unitary network, incorporating two layers of local unitary tensors.
\begin{definition} [Bilayer unitary network]
A bilayer unitary network consists of two layers of local unitary tensors. Two layers are stacked in the temporal direction. 

This paper consistently employs the right canonical form without explicit specification: Within the bottom layer, the horizontal legs are uniformly pointed in the positive direction, while within the upper layer, the horizontal legs are uniformly pointed in the negative direction:

\begin{align} \label{Eq: bilayer unitary network}
\tikzsetnextfilename{Bilayer_UN_DAG2}
    \begin{tikzpicture}
        \node[draw, rectangle, rounded corners=2pt, minimum width=0.6cm, minimum height=0.6cm, fill=tensorcolor, thick](A) at (-1.2, 0) {1};
        \node[draw, rectangle, rounded corners=2pt, minimum width=0.6cm, minimum height=0.6cm, fill=tensorcolor, thick](B) {2};
        \node[draw, rectangle, rounded corners=2pt, minimum width=0.6cm, minimum height=0.6cm, fill=tensorcolor, thick](C) at (1.2, 0) {3};
        \node[draw, rectangle, rounded corners=2pt, minimum width=0.6cm, minimum height=0.6cm, fill=tensorcolor, thick](D) at (-1.2, 1.2) {6};
        \node[draw, rectangle, rounded corners=2pt, minimum width=0.6cm, minimum height=0.6cm, fill=tensorcolor, thick](E) at (0, 1.2) {5};
        \node[draw, rectangle, rounded corners=2pt, minimum width=0.6cm, minimum height=0.6cm, fill=tensorcolor, thick](F) at (1.2, 1.2) {4};
        \draw[mid arrow=0.65] (D.90) -- ++(0,0.6);
        \draw[mid arrow=0.65] (E.90) -- ++(0,0.6);
        \draw[mid arrow=0.65] (F.90) -- ++(0,0.6);
        \draw[mid arrow=0.65] ([yshift=-0.6cm]A.270) -- (A.270);
        \draw[mid arrow=0.65] ([yshift=-0.6cm]B.270) -- (B.270);
        \draw[mid arrow=0.65] ([yshift=-0.6cm]C.270) -- (C.270);
        \draw[mid arrow=0.65] (A.0) -- (B.180);
        \draw[mid arrow=0.65] (B.0) -- (C.180);
        \draw[mid arrow=0.65] (F.180) -- (E.0);
        \draw[mid arrow=0.65] (E.180) -- (D.0);
        \draw[mid arrow=0.65] (A.90) -- (D.270);
        \draw[mid arrow=0.65] (B.90) -- (E.270);
        \draw[mid arrow=0.65] (C.90) -- (F.270);
    \end{tikzpicture}.
\end{align}
\end{definition}
The bilayer unitary network above forms a DAG, where the labels of each local unitary tensor correspond to a sequence derived from a topological sort:
\begin{equation}
    1 < 2 < 3 < 4 < 5 < 6.
\end{equation}
Nevertheless, the right canonical form is not the only option; alternative canonical forms may also be employed for a bilayer unitary network.
\begin{definition} [Canonical form of tensor network]
\label{def: Canonical form}

A tensor network $U_{Net}$ is said to be in canonical form with center $C = (V_{C}, E_{C}) \subseteq G[U_{Net}]$ a subgraph of it, if its graph $G[U_{Net}] = (V,E)$ satisfies
\begin{enumerate} [label=(\roman*)]
    \item $G[U_{Net}]$ is a DAG.
    \item Following a topological sort, the local unitaries within $C$ can align consecutively.
\end{enumerate}
The center $C$ constitutes a sub-network, denoted by $U^{C}_{Net} \subseteq U_{Net}$.
\end{definition}
An example of a bilayer unitary network with the center $C$ colored blue is shown below. Observe that the arrow directions differ from the right canonical form of bilayer unitary networks.
\begin{align}
\tikzsetnextfilename{canonical_form}
       \begin{tikzpicture}
        \node[draw, rectangle, rounded corners=2pt, minimum width=0.6cm, minimum height=0.6cm, fill=tensorcolor, thick](A1) at (-1.2, 0) {1};
        \node[draw, rectangle, rounded corners=2pt, minimum width=0.6cm, minimum height=0.6cm, fill=btensorcolor, thick](B1) at (0,0) {3};
        \node[draw, rectangle, rounded corners=2pt, minimum width=0.6cm, minimum height=0.6cm, fill=tensorcolor, thick](C1) at (1.2, 0) {2};
        \node[draw, rectangle, rounded corners=2pt, minimum width=0.6cm, minimum height=0.6cm, fill=tensorcolor, thick](A2) at (-1.2, 1.2) {5};
        \node[draw, rectangle, rounded corners=2pt, minimum width=0.6cm, minimum height=0.6cm, fill=btensorcolor, thick](B2) at (0, 1.2) {4};
        \node[draw, rectangle, rounded corners=2pt, minimum width=0.6cm, minimum height=0.6cm, fill=tensorcolor, thick](C2) at (1.2, 1.2) {6};
        \draw[mid arrow=0.65] (A2.90) -- ++(0,0.6);
        \draw[mid arrow=0.65] (B2.90) -- ++(0,0.6);
        \draw[mid arrow=0.65] (C2.90) -- ++(0,0.6);
        \draw[mid arrow=0.65] ([yshift=-0.6cm]A1.270) -- (A1.270);
        \draw[mid arrow=0.65] ([yshift=-0.6cm]B1.270) -- (B1.270);
        \draw[mid arrow=0.65] ([yshift=-0.6cm]C1.270) -- (C1.270);
        \draw[mid arrow=0.65] ([xshift=-0.6cm]A1.180) -- (A1.180);
        \draw[mid arrow=0.65] (A1.0) -- (B1.180);
        \draw[mid arrow=0.65] (C1.180) -- (B1.0);
        \draw[mid arrow=0.65] ([xshift=0.6cm]C1.0) -- (C1.0);
        \draw[mid arrow=0.65] (C2.0) -- ++(0.6,0);
        \draw[mid arrow=0.65]  (B2.0) -- (C2.180);
        \draw[mid arrow=0.65] (B2.180) -- (A2.0);
        \draw[mid arrow=0.65] (A2.180) -- ++(-0.6,0);
        \draw[mid arrow=0.65] (A1.90) -- (A2.270);
        \draw[mid arrow=0.65] (B1.90) -- (B2.270);
        \draw[mid arrow=0.65] (C1.90) -- (C2.270);
    \end{tikzpicture}
\end{align}

Bilayer unitary networks are applicable to higher-dimensional systems. Illustrated below is a bilayer unitary network for a two-dimensional system, with each local unitary tensor represented by a dot:
\begin{align}
\label{Eq: BUN_2D}
\tikzsetnextfilename{Bilayer_UN_2D}
    \begin{tikzpicture}
    \tdplotsetmaincoords{60}{120}
    \begin{scope}[tdplot_main_coords, shift={(0,-3,0)}]
        \draw[-to] (0,0,0) -- (1,0,0) node[anchor=west]{$x$};
	\draw[-to] (0,0,0) -- (0,1,0)         node[anchor=west]{$y$};
	\draw[-to] (0,0,0) -- (0,0,1) node[anchor=west]{$t$};
    \end{scope}
  \begin{scope}[tdplot_main_coords]
         \node[draw, circle, scale=0.5, fill=tensorcolor] (AA1) at (-1, -1, 0) {};
        \node[draw, circle, scale=0.5, fill=tensorcolor] (BA1) at (0, -1, 0) {};
        \node[draw, circle, scale=0.5, fill=tensorcolor] (CA1) at (1, -1, 0) {};
        \node[draw, circle, scale=0.5, fill=tensorcolor] (AB1) at (-1, 0, 0) {};
        \node[draw, circle, scale=0.5, fill=tensorcolor] (BB1) at (0, 0, 0) {};
        \node[draw, circle, scale=0.5, fill=tensorcolor] (CB1) at (1, 0, 0) {};
        \node[draw, circle, scale=0.5, fill=tensorcolor] (AC1) at (-1, 1, 0) {};
        \node[draw, circle, scale=0.5, fill=tensorcolor] (BC1) at (0, 1, 0) {};
        \node[draw, circle, scale=0.5, fill=tensorcolor] (CC1) at (1, 1, 0) {};
         \node[draw, circle, scale=0.5, fill=tensorcolor] (AA2) at (-1, -1, 3) {};
        \node[draw, circle, scale=0.5, fill=tensorcolor] (BA2) at (0, -1, 3) {};
        \node[draw, circle, scale=0.5, fill=tensorcolor] (CA2) at (1, -1, 3) {};
        \node[draw, circle, scale=0.5, fill=tensorcolor] (AB2) at (-1, 0, 3) {};
        \node[draw, circle, scale=0.5, fill=tensorcolor] (BB2) at (0, 0, 3) {};
        \node[draw, circle, scale=0.5, fill=tensorcolor] (CB2) at (1, 0, 3) {};
        \node[draw, circle, scale=0.5, fill=tensorcolor] (AC2) at (-1, 1, 3) {};
        \node[draw, circle, scale=0.5, fill=tensorcolor] (BC2) at (0, 1, 3) {};
        \node[draw, circle, scale=0.5, fill=tensorcolor] (CC2) at (1, 1, 3) {};
        \draw[mid arrow=0.65] (AA1) -- (AA2);
        \draw[mid arrow=0.65] (BA1) -- (BA2);
        \draw[mid arrow=0.65] (CA1) -- (CA2);
        \draw[mid arrow=0.65] (AB1) -- (AB2);
        \draw[mid arrow=0.65] (BB1) -- (BB2);
        \draw[mid arrow=0.65] (CB1) -- (CB2);
        \draw[mid arrow=0.65] (AC1) -- (AC2);
        \draw[mid arrow=0.65] (BC1) -- (BC2);
        \draw[mid arrow=0.65] (CC1) -- (CC2);
        \draw[mid arrow=0.65] (AA1) -- (BA1);
        \draw[mid arrow=0.65] (BA1) -- (CA1);
        \draw[mid arrow=0.65] (AB1) -- (BB1);
        \draw[mid arrow=0.65] (BB1) -- (CB1);
        \draw[mid arrow=0.65] (AC1) -- (BC1);
        \draw[mid arrow=0.65] (BC1) -- (CC1);
        \draw[mid arrow=0.65] (AA1) -- (AB1);
        \draw[mid arrow=0.65] (AB1) -- (AC1);
        \draw[mid arrow=0.65] (BA1) -- (BB1);
        \draw[mid arrow=0.65] (BB1) -- (BC1);
        \draw[mid arrow=0.65] (CA1) -- (CB1);
        \draw[mid arrow=0.65] (CB1) -- (CC1);
        \draw[mid arrow=0.65] (BA2) -- (AA2);
        \draw[mid arrow=0.65] (CA2) -- (BA2);
        \draw[mid arrow=0.65] (BB2) -- (AB2);
        \draw[mid arrow=0.65] (CB2) -- (BB2);
        \draw[mid arrow=0.65] (BC2) -- (AC2);
        \draw[mid arrow=0.65] (CC2) -- (BC2);
        \draw[mid arrow=0.65] (AB2) -- (AA2);
        \draw[mid arrow=0.65] (AC2) -- (AB2);
        \draw[mid arrow=0.65] (BB2) -- (BA2);
        \draw[mid arrow=0.65] (BC2) -- (BB2);
        \draw[mid arrow=0.65] (CB2) -- (CA2);
        \draw[mid arrow=0.65] (CC2) -- (CB2);
        \draw[mid arrow=0.65] (AA2) -- ++(0,0,0.5);
        \draw[mid arrow=0.65] (BA2) -- ++(0,0,0.5);
        \draw[mid arrow=0.65] (CA2) -- ++(0,0,0.5);
        \draw[mid arrow=0.65] (AB2) -- ++(0,0,0.5);
        \draw[mid arrow=0.65] (BB2) -- ++(0,0,0.5);
        \draw[mid arrow=0.65] (CB2) -- ++(0,0,0.5);
        \draw[mid arrow=0.65] (AC2) -- ++(0,0,0.5);
        \draw[mid arrow=0.65] (BC2) -- ++(0,0,0.5);
        \draw[mid arrow=0.65] (CC2) -- ++(0,0,0.5);
        \draw[mid arrow=0.65] ($(AA1)+(0,0,-0.5)$) -- ++(AA1);
        \draw[mid arrow=0.65] ($(BA1)+(0,0,-0.5)$) -- ++(BA1);
        \draw[mid arrow=0.65] ($(CA1)+(0,0,-0.5)$) -- ++(CA1);
        \draw[mid arrow=0.65] ($(AB1)+(0,0,-0.5)$) -- ++(AB1);
        \draw[mid arrow=0.65] ($(BB1)+(0,0,-0.5)$) -- ++(BB1);
        \draw[mid arrow=0.65] ($(CB1)+(0,0,-0.5)$) -- ++(CB1);
        \draw[mid arrow=0.65] ($(AC1)+(0,0,-0.5)$) -- ++(AC1);
        \draw[mid arrow=0.65] ($(BC1)+(0,0,-0.5)$) -- ++(BC1);
        \draw[mid arrow=0.65] ($(CC1)+(0,0,-0.5)$) -- ++(CC1);
    \end{scope}
    \end{tikzpicture}
\end{align}
The bottom layer features horizontal legs directed positively along $X+$ and $Y+$, while in the bottom layer, the horizontal legs have negative directions, designated as $X-$ and $Y-$. The higher-dimensional bilayer unitary network can be verified to form a DAG for higher-dimensional systems. This study will concentrate on the 1D unitary network, although much of the reasoning applies to higher dimensions.

Note that the bilayer unitary network design resembles the quantum circuit tensor networks proposed in Ref.~\cite{haghshenas2022variational}. The distinction lies in employing two layers with opposing directions, ensuring the bilayer unitary network exhibits universality. 

Ref.~\cite{Styliaris_2025} introduces a 2-floor staircase quantum circuit construction for MPU, whose architecture is similar to our bilayer unitary networks. They showed that a 2-floor staircase circuit with two-qubit gates is insufficient for representing all MPUs \cite{Styliaris_2025}. In bilayer unitary networks, local unitary tensors are not restricted to two-qubit gates. With a larger bond dimension $D$, a local unitary tensor can be regarded as a wider gate. Infinite OBC scenarios further differentiate 2-floor staircase quantum circuits from unitary networks. Later in Section \ref{Sec: UN and SQC}, we will illustrate that unitary networks enable non-zero net information flow, sometimes offering a more efficient representation than quantum circuits.

\begin{properties} [Properties of bilayer unitary networks] 
A bilayer unitary network has the following properties:
\begin{enumerate}[label=(\roman*)]
    \item Unitary (acyclic): A bilayer unitary network is devoid of directed loops and constitutes a global unitary.
    \item Spatially connected: There exists a directed path linking every incoming external leg to each outgoing external leg pair.
    \item Universal: With sufficient bond dimensions, a bilayer unitary network can represent any global unitary on a finite-size system.
\end{enumerate}
\end{properties}

\begin{proof}
    (i) Acyclic and (ii) Spatial connected properties are illustrated in the diagrams (\ref{Eq: bilayer unitary network}) and (\ref{Eq: BUN_2D}). Our primary focus here is to provide a visual demonstration for (iii), universal property. We present a global unitary operation applied to a one-dimensional finite chain of qudits, where each site features a Hilbert space of dimension $d$:
\begin{equation}
    \mathcal{H}_{n} = \mathcal{H}_{\text{qudit}} = d \quad \text{for } n=1,2,3 ,
\end{equation}
    \begin{align}
    \label{Eq: decompose into BUN}
    \begin{array}{c}
    \tikzsetnextfilename{global_unitary}
        \begin{tikzpicture}
             \node[draw, rectangle, rounded corners=2pt, minimum width=1.2cm, minimum height=0.6cm, fill=tensorcolor, thick]{$U$};
             \draw[mid arrow=0.65] (-0.3, 0.3) -- (-0.3, 0.9);
             \draw[mid arrow=0.65] (0, 0.3) -- (0, 0.9);
            \draw[mid arrow=0.65] (0.3, 0.3) -- (0.3, 0.9);
            \draw[mid arrow=0.65] (-0.3, -0.9) -- (-0.3, -0.3);
             \draw[mid arrow=0.65] (0, -0.9) -- (0, -0.3);
            \draw[mid arrow=0.65] (0.3, -0.9) -- (0.3, -0.3);
        \end{tikzpicture} 
    \end{array} = 
    \begin{array}{c}
    \tikzsetnextfilename{UN_rep}
\begin{tikzpicture}
        \node[draw, rectangle, rounded corners=2pt, minimum width=0.6cm, minimum height=0.6cm, fill=tensorcolor, thick](A) at (-1.2, 0) {$I_{A}$};
        \node[draw, rectangle, rounded corners=2pt, minimum width=0.6cm, minimum height=0.6cm, fill=tensorcolor, thick](B) {$I_{B}$};
        \node[draw, rectangle, rounded corners=2pt, minimum width=0.6cm, minimum height=0.6cm, fill=tensorcolor, thick](C) at (1.2, 0) {$I_{C}$};
        \node[draw, rectangle, rounded corners=2pt, minimum width=0.6cm, minimum height=0.6cm, fill=tensorcolor, thick](D) at (-1.2, 1.2) {$I_{D}$};
        \node[draw, rectangle, rounded corners=2pt, minimum width=0.6cm, minimum height=0.6cm, fill=tensorcolor, thick](E) at (0, 1.2) {$I_{E}$};
        \node[draw, rectangle, rounded corners=2pt, minimum width=0.6cm, minimum height=0.6cm, fill=tensorcolor, thick](F) at (1.2, 1.2) {$U_{F}$};
        \draw[mid arrow=0.65] (D.90) -- ++(0,0.6);
        \draw[mid arrow=0.65] (E.90) -- ++(0,0.6);
        \draw[mid arrow=0.65] (F.90) -- ++(0,0.6);
        \draw[mid arrow=0.65] ([yshift=-0.6cm]A.270) -- (A.270);
        \draw[mid arrow=0.65] ([yshift=-0.6cm]B.270) -- (B.270);
        \draw[mid arrow=0.65] ([yshift=-0.6cm]C.270) -- (C.270);
        \draw[mid arrow=0.65] (A.0) -- (B.180);
        \draw[mid arrow=0.65] ([yshift=-0.1cm]B.0) -- ([yshift=-0.1cm]C.180);
        \draw[mid arrow=0.65] ([yshift=0.1cm]B.0) -- ([yshift=0.1cm]C.180);
        \draw[mid arrow=0.65] ([yshift=-0.1cm]F.180) -- ([yshift=-0.1cm]E.0);
        \draw[mid arrow=0.65] ([yshift=0.1cm]F.180) -- ([yshift=0.1cm]E.0);
        \draw[mid arrow=0.65] (E.180) -- (D.0);
        \draw[mid arrow=0.65] (C.90) -- (F.270);
        \draw[mid arrow=0.65] ([xshift=-0.2cm]C.90) -- ([xshift=-0.2cm]F.270);
        \draw[mid arrow=0.65] ([xshift=0.2cm]C.90) -- ([xshift=0.2cm]F.270);
    \end{tikzpicture}
    \end{array}.
    \end{align}
Here and in later sections, some instances occur where a leg of dimension $d^{n}$ is split into $n$ legs, each with a dimension of $d$. This notation helps to demonstrate more clearly the unitarity of local tensors within the diagram. Moreover, aligning the Hilbert space of dimension $d$ with the Hilbert space of a qudit $\mathcal{H}_{\text{qudit}}$ allows us to equate a unitary network with a quantum circuit.

Placing the target global unitary in the top right and transforming specific vertical physical legs into horizontal bond legs allows us to treat it as a local unitary tensor, $U_{F} = U$. By adding local identity tensors as padding, we develop a bilayer unitary network for the target global unitary. Observe that $I_{A}$ and $I_{D}$ are unconnected. In principle, we can introduce a vertical leg with a one-dimensional Hilbert space that links them. 

\end{proof}

Although the construction approach in the prior proof ensures that any unitary operator on a finite-sized system can be represented as a bilayer unitary network, this representation is inefficient.   We will discuss how to achieve a more efficient unitary network with a smaller bond dimension later in Section \ref{Sec: reduce bond dimension}.

\subsection{Unitary networks and tensor network states}

As stated in Ref.~\cite{haghshenas2022variational}, employing a tensor network unitary on a product state such as $|00\cdots 0\rangle$ generates a \textit{tensor network state} (TNS) \cite{evenbly2011tensor}. This is also true for unitary networks. 

Examine a one-dimensional bilayer unitary network with uniform bulk tensors and arbitrary boundary tensors, operating on the initial state $|00\cdots 0\rangle$:
\begin{align}
\begin{array}{c}
\tikzsetnextfilename{UN_act_on_state}
    \begin{tikzpicture}
     \node[draw, circle, minimum width=0.6cm, fill=tensorcolor, thick](L0) at (-2.4, -1.2) {$0$};
        \node[draw, circle, minimum width=0.6cm, fill=tensorcolor, thick](A0) at (-1.2, -1.2) {$0$};
        \node[draw, circle, minimum width=0.6cm, fill=tensorcolor, thick](B0) at (0, -1.2) {$0$};
        \node[draw, circle, minimum width=0.6cm, fill=tensorcolor, thick](C0) at (1.2, -1.2) {$0$};
        \node[draw, circle, minimum width=0.6cm, fill=tensorcolor, thick](R0) at (2.4, -1.2) {$0$};
        \node[draw, rectangle, rounded corners=2pt, minimum width=0.6cm, minimum height=0.6cm, fill=tensorcolor, thick](L1) at (-2.4, 0) {$L_{B}$};
        \node[draw, rectangle, rounded corners=2pt, minimum width=0.6cm, minimum height=0.6cm, fill=tensorcolor, thick](A1) at (-1.2, 0) {$A$};
        \node[draw, rectangle, rounded corners=2pt, minimum width=0.6cm, minimum height=0.6cm, fill=tensorcolor, thick](B1) at (0,0) {$A$};
        \node[draw, rectangle, rounded corners=2pt, minimum width=0.6cm, minimum height=0.6cm, fill=tensorcolor, thick](C1) at (1.2, 0) {$A$};
        \node[draw, rectangle, rounded corners=2pt, minimum width=0.6cm, minimum height=0.6cm, fill=tensorcolor, thick](R1) at (2.4, 0) {$R_{B}$};
        \node[draw, rectangle, rounded corners=2pt, minimum width=0.6cm, minimum height=0.6cm, fill=tensorcolor, thick](L2) at (-2.4, 1.2) {$L_{T}$};
        \node[draw, rectangle, rounded corners=2pt, minimum width=0.6cm, minimum height=0.6cm, fill=tensorcolor, thick](A2) at (-1.2, 1.2) {$B$};
        \node[draw, rectangle, rounded corners=2pt, minimum width=0.6cm, minimum height=0.6cm, fill=tensorcolor, thick](B2) at (0, 1.2) {$B$};
        \node[draw, rectangle, rounded corners=2pt, minimum width=0.6cm, minimum height=0.6cm, fill=tensorcolor, thick](C2) at (1.2, 1.2) {$B$};
         \node[draw, rectangle, rounded corners=2pt, minimum width=0.6cm, minimum height=0.6cm, fill=tensorcolor, thick](R2) at (2.4, 1.2) {$R_{T}$};
        \draw[mid arrow=0.65] (A2.90) -- ++(0,0.6);
        \draw[mid arrow=0.65] (B2.90) -- ++(0,0.6);
        \draw[mid arrow=0.65] (C2.90) -- ++(0,0.6);
        \draw[mid arrow=0.65] (L2.90) -- ++(0,0.6);
        \draw[mid arrow=0.65] (R2.90) -- ++(0,0.6);
        \draw[mid arrow=0.65] ([yshift=-0.6cm]A1.270) -- (A1.270);
        \draw[mid arrow=0.65] ([yshift=-0.6cm]B1.270) -- (B1.270);
        \draw[mid arrow=0.65] ([yshift=-0.6cm]C1.270) -- (C1.270);
        \draw[mid arrow=0.65] (L1.0) -- (A1.180);
        \draw[mid arrow=0.65] (A1.0) -- (B1.180);
        \draw[mid arrow=0.65] (B1.0) -- (C1.180);
        \draw[mid arrow=0.65] (C1.0) -- (R1.180);
        \draw[mid arrow=0.65] (R2.180) -- (C2.0);
        \draw[mid arrow=0.65] (C2.180) -- (B2.0);
        \draw[mid arrow=0.65] (B2.180) -- (A2.0);
        \draw[mid arrow=0.65] (A2.180) -- (L2.0);
        \draw[mid arrow=0.65] (A1.90) -- (A2.270);
        \draw[mid arrow=0.65] (B1.90) -- (B2.270);
        \draw[mid arrow=0.65] (C1.90) -- (C2.270);
        \draw[mid arrow=0.65] (L1.90) -- (L2.270);
        \draw[mid arrow=0.65] (R1.90) -- (R2.270);
        \draw[mid arrow=0.65] (A0.90) -- (A1.270);
        \draw[mid arrow=0.65] (B0.90) -- (B1.270);
        \draw[mid arrow=0.65] (C0.90) -- (C1.270);
        \draw[mid arrow=0.65] (L0.90) -- (L1.270);
        \draw[mid arrow=0.65] (R0.90) -- (R1.270);
    \end{tikzpicture}
\end{array}
\end{align}
An MPS can be obtained by vertically contracting $\{ A, B, |0\rangle \}$ into the tensor $C$, followed by ignoring the leg orientations:
\begin{align}
        \begin{array}{c}
\tikzsetnextfilename{TNS}
\begin{tikzpicture}
    \node[draw, circle, minimum width=0.6cm, fill=tensorcolor, thick](L0) at (-2, 0) {$L$};
    \node[draw, circle, minimum width=0.6cm, fill=tensorcolor, thick](A0) at (-1, 0) {$C$};
    \node[draw, circle, minimum width=0.6cm, fill=tensorcolor, thick](B0) at (0, 0) {$C$};
    \node[draw, circle, minimum width=0.6cm, fill=tensorcolor, thick](C0) at (1, 0) {$C$};
    \node[draw, circle, minimum width=0.6cm, fill=tensorcolor, thick](R0) at (2, 0) {$R$};
        \draw (A0.90) -- ++(0,0.4);
        \draw (B0.90) -- ++(0,0.4);
        \draw (C0.90) -- ++(0,0.4);
        \draw (L0.90) -- ++(0,0.4);
        \draw (R0.90) -- ++(0,0.4);
        \draw (L0.0) -- (A0.180);
        \draw (A0.0) -- (B0.180);
        \draw (B0.0) -- (C0.180);
        \draw (C0.0) -- (R0.180);
\end{tikzpicture}
\end{array}
\end{align}
This example demonstrates that in a unitary network, contractions are not limited by the sequence resulting from topological sorting. The flexibility in choosing the contraction order grants unitary networks an advantage over quantum circuit simulation, where gate operations must be performed sequentially. Once the TNS is acquired, the correlations or entanglement entropy can be calculated \cite{evenbly2011tensor}.

\section{Unitary network with different boundary conditions}
\label{Sec: Unitary network with different boundary conditions}

The prior section focused on the unitary network in a finite OBC system. Here, we address unitary networks in both infinite OBC and PBC systems, noting how these boundary conditions introduce subtle variations.

\subsection{Unitary networks in infinite OBC systems}

In employing a unitary network to represent a global unitary $U$ for an infinite system, we typically focus on a finite subsystem. Consider, for example, a subsystem $S \subseteq \Gamma$ consisting of three sites taken from an infinite chain $\Gamma$. The remaining part of the lattice $\Gamma$ can be considered as the environment $E = \Gamma \backslash S$. Within the larger unitary network, we identify a sub-network $U_{Net}^{S} \subseteq U_{Net}$ such that $\text{Eval}(U_{Net}^{S}) = U_{S}$:
\begin{align}
\label{Fig: UN_OBC_boundary}
\begin{array}{c}
\tikzsetnextfilename{UN_OBC_boundary}
    \begin{tikzpicture}
        \node[draw, rectangle, rounded corners=2pt, minimum width=0.6cm, minimum height=0.6cm, fill=tensorcolor, thick](A1) at (-1.2, 0) {};
        \node[draw, rectangle, rounded corners=2pt, minimum width=0.6cm, minimum height=0.6cm, fill=tensorcolor, thick](B1) at (0,0) {};
        \node[draw, rectangle, rounded corners=2pt, minimum width=0.6cm, minimum height=0.6cm, fill=tensorcolor, thick](C1) at (1.2, 0) {};
        \node[draw, rectangle, rounded corners=2pt, minimum width=0.6cm, minimum height=0.6cm, fill=tensorcolor, thick](A2) at (-1.2, 1.2) {};
        \node[draw, rectangle, rounded corners=2pt, minimum width=0.6cm, minimum height=0.6cm, fill=tensorcolor, thick](B2) at (0, 1.2) {};
        \node[draw, rectangle, rounded corners=2pt, minimum width=0.6cm, minimum height=0.6cm, fill=tensorcolor, thick](C2) at (1.2, 1.2) {};
        \draw[mid arrow=0.65] (A2.90) -- ++(0,0.6);
        \draw[mid arrow=0.65] (B2.90) -- ++(0,0.6);
        \draw[mid arrow=0.65] (C2.90) -- ++(0,0.6);
        \draw[mid arrow=0.65] ([yshift=-0.6cm]A1.270) -- (A1.270);
        \draw[mid arrow=0.65] ([yshift=-0.6cm]B1.270) -- (B1.270);
        \draw[mid arrow=0.65] ([yshift=-0.6cm]C1.270) -- (C1.270);
        \draw[mid arrow=0.65] ([xshift=-0.6cm]A1.180) -- (A1.180);
        \draw[mid arrow=0.65] (A1.0) -- (B1.180);
        \draw[mid arrow=0.65] (B1.0) -- (C1.180);
        \draw[mid arrow=0.65] (C1.0) -- ++(0.6,0);
        \draw[mid arrow=0.65] ([xshift=0.6cm]C2.0) -- (C2.0);
        \draw[mid arrow=0.65] (C2.180) -- (B2.0);
        \draw[mid arrow=0.65] (B2.180) -- (A2.0);
        \draw[mid arrow=0.65] (A2.180) -- ++(-0.6,0);
        \draw[mid arrow=0.65] (A1.90) -- (A2.270);
        \draw[mid arrow=0.65] (B1.90) -- (B2.270);
        \draw[mid arrow=0.65] (C1.90) -- (C2.270);
    \end{tikzpicture}
\end{array}
    =
    \begin{array}{c}
\tikzsetnextfilename{global_unitary_OBC2}
\begin{tikzpicture}
     \node[draw, rectangle, rounded corners=2pt, minimum width=1.2cm, minimum height=0.8cm, fill=tensorcolor, thick](A) at (-1.2, 0) {$U_{S}$};
        \draw[mid arrow=0.65] ([xshift=-0.4cm,yshift=-0.6cm]A.270) -- ([xshift=-0.4cm]A.270);
        \draw[mid arrow=0.65] ([yshift=-0.6cm]A.270) -- (A.270);
        \draw[mid arrow=0.65] ([xshift=0.4cm,yshift=-0.6cm]A.270) -- ([xshift=0.4cm]A.270);
        \draw[mid arrow=0.65] (A.90) -- ++(0,0.6);
        \draw[mid arrow=0.65] ([xshift=-0.4cm]A.90) -- ++(0,0.6);
        \draw[mid arrow=0.65] ([xshift=0.4cm]A.90) -- ++(0,0.6);
        \draw[mid arrow=0.65] ([yshift=-0.2cm]A.0) -- ++(0.6,0);
        \draw[mid arrow=0.65] ([xshift=0.6cm, yshift=0.2cm]A.0) -- ++(-0.6,0);
        \draw[mid arrow=0.65] ([xshift=-0.6cm,yshift=-0.2cm]A.180) -- ++(0.6,0);
         \draw[mid arrow=0.65] ([yshift=0.2cm]A.180) -- ++(-0.6,0);
\end{tikzpicture}
\end{array},
\end{align}
Contracting the unitary network in the subsystem results in a tensor $U_{S}$ whose unitarity is associated with the combination of external vertical physical legs and external horizontal bond legs. The external horizontal bond legs can be perceived as environmental degrees of freedom. Often, we are concerned solely with the physical part of the transformation. Thus, we would like to trace out the environmental degree of freedom accordingly.
\begin{definition} [Reduced unitary]
\label{def: Reduced unitary}
Considering a unitary $U$ and the automorphism defined by $u(O) = UOU^{\dagger}$, we define the reduced unitary $u_{A}$ for subsystem $A$ as illustrated in the diagram below:
\begin{align}
u_{A}(O_{A}) = 
    \frac{1}{\dim \mathcal{H}_{B}}\begin{array}{c}
    \tikzsetnextfilename{reduced_automorphism}
    \begin{tikzpicture}
         \node[draw, rectangle, rounded corners=2pt, minimum width=0.6cm, minimum height=0.6cm, fill=tensorcolor, thick] (OA) at (-0.4, 0) {$O_{A}$};
        \node[draw, rectangle, rounded corners=2pt, minimum width=1.2cm, minimum height=0.6cm, fill=tensorcolor, thick] (U) at (0, 1.2) {$U$};
        \node[draw, rectangle, rounded corners=2pt, minimum width=1.2cm, minimum height=0.6cm, fill=tensorcolor, thick] (Ud) at (0, -1.2) {$U^{\dagger}$};
        \draw[mid arrow=0.65] ([xshift=-0.4cm,yshift=-0.6cm]U.270) -- ++(0,0.6);
        \draw[mid arrow=0.65] ([xshift=-0.4cm]Ud.90) -- ++(0,0.6);
        \draw[mid arrow=0.65] ([xshift=-0.4cm,yshift=-0.6cm]Ud.270) -- ++(0,0.6);
         \draw[mid arrow=0.65] ([xshift=-0.4cm]U.90) -- ++(0,0.6);
         \draw[mid arrow=0.65] ([xshift=0.4cm]Ud.90) -- ([xshift=0.4cm]U.270);
        \draw[mid arrow=0.5] ([xshift=0.4cm]U.90) -- ++(0,0.6) -- ++(0.4,0) -- ++(0,-4.2) -- ++ (-0.4, 0) -- ([xshift=0.4cm]Ud.270);
    \end{tikzpicture}
    \end{array} 
    =
     \begin{array}{c}
     \tikzsetnextfilename{reduced_automorphism_normalized}
    \begin{tikzpicture}
         \node[draw, rectangle, rounded corners=2pt, minimum width=0.6cm, minimum height=0.6cm, fill=tensorcolor, thick] (OA) at (-0.4, 0) {$O_{A}$};
        \node[draw, rectangle, rounded corners=2pt, minimum width=1.2cm, minimum height=0.6cm, fill=tensorcolor, thick] (U) at (0, 1.2) {$U$};
        \node[draw, rectangle, rounded corners=2pt, minimum width=1.2cm, minimum height=0.6cm, fill=tensorcolor, thick] (Ud) at (0, -1.2) {$U^{\dagger}$};
        \node[draw, rectangle, rounded corners=2pt, minimum width=0.6cm, minimum height=0.6cm, fill=tensorcolor, thick] (IB1) at (0.4, 0) {$\hat{I}_{B}$};
        \node[draw, rectangle, rounded corners=2pt, minimum width=0.6cm, minimum height=0.6cm, fill=tensorcolor, thick] (IB2) at (1.2, 0) {$\hat{I}_{B}$};
        \draw[mid arrow=0.65] ([xshift=-0.4cm,yshift=-0.6cm]U.270) -- ++(0,0.6);
        \draw[mid arrow=0.65] ([xshift=-0.4cm]Ud.90) -- ++(0,0.6);
        \draw[mid arrow=0.65] ([xshift=-0.4cm,yshift=-0.6cm]Ud.270) -- ++(0,0.6);
         \draw[mid arrow=0.65] ([xshift=-0.4cm]U.90) -- ++(0,0.6);
         \draw[mid arrow=0.65] ([xshift=0.4cm]Ud.90) -- (IB1.270);
         \draw[mid arrow=0.65] (IB1.90) --([xshift=0.4cm]U.270);
        \draw[mid arrow=0.5] ([xshift=0.4cm]U.90) -- ++(0,0.6) -- ++(0.8,0) -- (IB2.90);
        \draw[mid arrow=0.5] (IB2.270) -- ++(0,-1.8) -- ++(-0.8,0) -- ([xshift=0.4cm]Ud.270);
    \end{tikzpicture}
    \end{array}, 
\end{align}
where $\hat{I}_{B} = \frac{1}{\sqrt{\dim \mathcal{H}_{B}}}$ is the normalized operator $I_{B}$ (The norm $||I_{B}||= \sqrt{\dim \mathcal{H}_{B}}$).
\end{definition}
Typically, a reduced unitary $u_{A}$ does not serve as an automorphism in $\mathcal{A}_{A}$. However, it functions as a superoperator, meaning it is a linear transformation from $\mathcal{A}_{A}$ to itself, fulfilling
\begin{equation}
   u(O_{A} \otimes I_{B}) =  u_{A} (O_{A}) \otimes I_{B} + Q_{A} \otimes Q_{B}, 
\end{equation}
where $Q_{B} \ne I_{B}$ is non-trivial on subsystem B. Back to the unitary network, we may construct the reduced unitary $u_{S}[U_{S}]$ acting on the physical degree of freedom only by tracing out bond degrees of freedom:
\begin{align}
u_{S}[U_{S}] = 
\frac{1}{\dim \mathcal{H}_{\text{bond}}}
\begin{array}{c}
    \tikzsetnextfilename{reduced_UN}
\begin{tikzpicture}
     \node[draw, rectangle, rounded corners=2pt, minimum width=1.2cm, minimum height=0.8cm, fill=tensorcolor, thick](A) at (0, 1.3) {$U_{S}$};
     \node[draw, rectangle, rounded corners=2pt, minimum width=1.2cm, minimum height=0.8cm, fill=tensorcolor, thick](B) at (0, -1.3) {$U_{S}^{\dagger}$};
        \draw[mid arrow=0.65] ([xshift=-0.4cm,yshift=-0.6cm]A.270) -- ([xshift=-0.4cm]A.270);
        \draw[mid arrow=0.65] ([yshift=-0.6cm]A.270) -- (A.270);
        \draw[mid arrow=0.65] ([xshift=0.4cm,yshift=-0.6cm]A.270) -- ([xshift=0.4cm]A.270);
        \draw[mid arrow=0.65] (A.90) -- ++(0,0.6);
        \draw[mid arrow=0.65] ([xshift=-0.4cm]A.90) -- ++(0,0.6);
        \draw[mid arrow=0.65] ([xshift=0.4cm]A.90) -- ++(0,0.6);
        \draw[mid arrow=0.5] ([yshift=-0.2cm]A.0) -- ++(0.6,0) -- ++(0,-2.2) -- ([yshift=0.2cm]B.0);
         \draw[mid arrow=0.5] ([yshift=0.2cm]A.180) -- ++(-0.8,0) -- ++(0,-3) -- ([yshift=-0.2cm]B.180);
        \draw[mid arrow=0.65] ([yshift=-0.6cm]B.270) -- (B.270);
        \draw[mid arrow=0.65] ([xshift=0.4cm,yshift=-0.6cm]B.270) -- ([xshift=0.4cm]B.270);
        \draw[mid arrow=0.65] (B.90) -- ++(0,0.6);
        \draw[mid arrow=0.65] ([xshift=-0.4cm]B.90) -- ++(0,0.6);
        \draw[mid arrow=0.65] ([xshift=0.4cm]B.90) -- ++(0,0.6);
        \draw[mid arrow=0.5] ([yshift=-0.2cm]B.0) -- ++(0.8,0) --++(0,3) -- ([yshift=0.2cm]A.0);
        \draw[mid arrow=0.65] ([xshift=-0.4cm, yshift=-0.6cm]B.270) -- ++(0,0.6);
         \draw[mid arrow=0.5] ([yshift=0.2cm]B.180) -- ++(-0.6,0) -- ++(0,2.2) -- ([yshift=-0.2cm]A.180);
\end{tikzpicture}
\end{array}
\end{align}
We hope that the reduced unitaries $u_{S}[U_{S}]$ to accurately represent the local transformation of the global unitary $U$, such that $u_{S}[U_{S}] = u_{S}[U]$:
\begin{equation}
\frac{1}{\dim \mathcal{H}_{\text{bond}}}
\begin{array}{c}
      \tikzsetnextfilename{reduced_UN2}
\begin{tikzpicture}
     \node[draw, rectangle, rounded corners=2pt, minimum width=1.2cm, minimum height=0.8cm, fill=tensorcolor, thick](A) at (0, 1.3) {$U_{S}$};
     \node[draw, rectangle, rounded corners=2pt, minimum width=1.2cm, minimum height=0.8cm, fill=tensorcolor, thick](B) at (0, -1.3) {$U_{S}^{\dagger}$};
        \draw[mid arrow=0.65] ([xshift=-0.4cm,yshift=-0.6cm]A.270) -- ([xshift=-0.4cm]A.270);
        \draw[mid arrow=0.65] ([yshift=-0.6cm]A.270) -- (A.270);
        \draw[mid arrow=0.65] ([xshift=0.4cm,yshift=-0.6cm]A.270) -- ([xshift=0.4cm]A.270);
        \draw[mid arrow=0.65] (A.90) -- ++(0,0.6);
        \draw[mid arrow=0.65] ([xshift=-0.4cm]A.90) -- ++(0,0.6);
        \draw[mid arrow=0.65] ([xshift=0.4cm]A.90) -- ++(0,0.6);
        \draw[mid arrow=0.5] ([yshift=-0.2cm]A.0) -- ++(0.2,0) -- ++(0,-2.2) -- ([yshift=0.2cm]B.0);
         \draw[mid arrow=0.5] ([yshift=0.2cm]A.180) -- ++(-0.4,0) -- ++(0,-3) -- ([yshift=-0.2cm]B.180);
        \draw[mid arrow=0.65] ([yshift=-0.6cm]B.270) -- (B.270);
        \draw[mid arrow=0.65] ([xshift=0.4cm,yshift=-0.6cm]B.270) -- ([xshift=0.4cm]B.270);
        \draw[mid arrow=0.65] (B.90) -- ++(0,0.6);
        \draw[mid arrow=0.65] ([xshift=-0.4cm]B.90) -- ++(0,0.6);
        \draw[mid arrow=0.65] ([xshift=0.4cm]B.90) -- ++(0,0.6);
        \draw[mid arrow=0.5] ([yshift=-0.2cm]B.0) -- ++(0.4,0) --++(0,3) -- ([yshift=0.2cm]A.0);
        \draw[mid arrow=0.65] ([xshift=-0.4cm, yshift=-0.6cm]B.270) -- ++(0,0.6);
         \draw[mid arrow=0.5] ([yshift=0.2cm]B.180) -- ++(-0.2,0) -- ++(0,2.2) -- ([yshift=-0.2cm]A.180);
\end{tikzpicture}
\end{array}
= 
\frac{1}{\dim \mathcal{H}_{E}}
\begin{array}{c}
\tikzsetnextfilename{reduced_U}
\begin{tikzpicture}
     \node[draw, rectangle, rounded corners=2pt, minimum width=2cm, minimum height=0.8cm, fill=tensorcolor, thick](A) at (0, 1.3) {$U$};
     \node[draw, rectangle, rounded corners=2pt, minimum width=2cm, minimum height=0.8cm, fill=tensorcolor, thick](B) at (0, -1.3) {$U^{\dagger}$};
        \draw[mid arrow=0.65] ([xshift=-0.4cm,yshift=-0.6cm]A.270) -- ([xshift=-0.4cm]A.270);
        \draw[mid arrow=0.65] ([yshift=-0.6cm]A.270) -- (A.270);
        \draw[mid arrow=0.65] ([xshift=0.4cm,yshift=-0.6cm]A.270) -- ([xshift=0.4cm]A.270);
        \draw[mid arrow=0.65] (A.90) -- ++(0,0.6);
        \draw[mid arrow=0.65] ([xshift=-0.4cm]A.90) -- ++(0,0.6);
        \draw[mid arrow=0.65] ([xshift=0.4cm]A.90) -- ++(0,0.6);
        \draw[mid arrow=0.65] ([yshift=-0.6cm]B.270) -- (B.270);
        \draw[mid arrow=0.65] ([xshift=0.4cm,yshift=-0.6cm]B.270) -- ([xshift=0.4cm]B.270);
        \draw[mid arrow=0.65] (B.90) -- ++(0,0.6);
        \draw[mid arrow=0.65] ([xshift=-0.4cm]B.90) -- ++(0,0.6);
        \draw[mid arrow=0.65] ([xshift=0.4cm]B.90) -- ++(0,0.6);
        \draw[mid arrow=0.65] ([xshift=-0.4cm, yshift=-0.6cm]B.270) -- ++(0,0.6);
        \draw[mid arrow=0.5] ([xshift=-0.8cm]B.90) -- ([xshift=-0.8cm]A.270);
        \draw[mid arrow=0.5] ([xshift=0.8cm]B.90) -- ([xshift=0.8cm]A.270);
        \draw[mid arrow=0.5] ([xshift=0.8cm]A.90) -- ++(0,0.6) -- ++(0.4,0) -- ++(0,-4.6) -- ++(-0.4,0)-- ([xshift=0.8cm]B.270);
        \draw[mid arrow=0.5] ([xshift=-0.8cm]A.90) -- ++(0,0.6) -- ++(-0.4,0) -- ++(0,-4.6) -- ++(0.4,0) -- ([xshift=-0.8cm]B.270);
\end{tikzpicture}
\end{array}
\end{equation}
However, this applies exclusively when $U_{Net}$ is in its canonical form, with the local tensors supported on the subsystem $S$ serving as the center $C$.

\begin{proposition}
    For a unitary network $\text{Eval}(U_{Net}) = U$ in its canonical form with center $C$ located in a subsystem $S \in \Gamma$, the evaluation $U_{C} = \text{Eval}(U^{C}_{Net})$ results in reduced unitaries that are identical on subsystem $S$:
    \begin{equation}
        u_{S}[U_{C}] = u_{S}[U]
    \end{equation}
\end{proposition}
\begin{proof}
    After topological sorting of $U_{Net}$, we get an order of local unitary tensors:
    \begin{equation}
        U_{1} <  \cdots < U_{m} < U^{C}_{m+1} < \cdots < U^{C}_{n} < U_{n+1} <  \cdots < U_{N}.
    \end{equation}
    Beginning with $U_{1}$, as all incoming legs to $U_{1}$ are necessarily external, the contraction of $U_{1}$ with $U_{1}^{\dagger}$ leads to an identity. This process continues with contracts involving $U_{2}$ until $U^{C}_{m+1}$ is reached. Moreover, since all outgoing legs for $U_{N}$ are external, the contraction of $U_{N}$ with $U_{N}^{\dagger}$ yields the identity. This process of contracting the highest-order unitary with its Hermitian conjugate is continued until only the central tensors are left.
\end{proof}

For an infinite system, imposing the canonical form condition for a unitary network with center $C$ located in a finite region $S \subset \Gamma$ ensures that $U_{Net}^{C}$ faithfully represents the local transformation. However, imposing TI on the unitary network requires a right- (or left-) canonical form, which makes $u_{S}[U^{Net}_{S}]$ an approximation rather than a faithful superoperator of $u_{S}[U]$.

\subsection{Unitary networks in PBC system}

Up to this point, we have focused on unitary networks in OBC systems. In this subsection, we examine unitary networks in systems with PBC.

A potential design for PBC systems is the PBC-bilayer unitary network, wrapping an OBC-bilayer unitary network \eqref{Fig: UN_OBC_boundary} and contracting its external horizontal legs:
\begin{align}
\begin{array}{c}
\tikzsetnextfilename{bilayer_UN_PBC}
    \begin{tikzpicture}
        \node[draw, rectangle, rounded corners=2pt, minimum width=0.6cm, minimum height=0.6cm, fill=tensorcolor, thick](A) at (-1.2, 0) {};
        \node[draw, rectangle, rounded corners=2pt, minimum width=0.6cm, minimum height=0.6cm, fill=tensorcolor, thick](B) at (0,0) {};
        \node[draw, rectangle, rounded corners=2pt, minimum width=0.6cm, minimum height=0.6cm, fill=tensorcolor, thick](C) at (1.2, 0) {};
        \node[draw, rectangle, rounded corners=2pt, minimum width=0.6cm, minimum height=0.6cm, fill=tensorcolor, thick](D) at (-1.2, 1.2) {};
        \node[draw, rectangle, rounded corners=2pt, minimum width=0.6cm, minimum height=0.6cm, fill=tensorcolor, thick](E) at (0, 1.2) {};
        \node[draw, rectangle, rounded corners=2pt, minimum width=0.6cm, minimum height=0.6cm, fill=tensorcolor, thick](F) at (1.2, 1.2) {};
        \draw[mid arrow=0.65] (D.90) -- ++(0,0.6);
        \draw[mid arrow=0.65] (E.90) -- ++(0,0.6);
        \draw[mid arrow=0.65] (F.90) -- ++(0,0.6);
        \draw[mid arrow=0.65] ([yshift=-0.6cm]A.270) -- (A.270);
        \draw[mid arrow=0.65] ([yshift=-0.6cm]B.270) -- (B.270);
        \draw[mid arrow=0.65] ([yshift=-0.6cm]C.270) -- (C.270);
        \draw[mid arrow=0.65] (A.0) -- (B.180);
        \draw[mid arrow=0.65] (B.0) -- (C.180);
        \draw[mid arrow=0.65] (F.180) -- (E.0);
        \draw[mid arrow=0.65] (E.180) -- (D.0);
        \draw[mid arrow=0.65] (A.90) -- (D.270);
        \draw[mid arrow=0.65] (B.90) -- (E.270);
        \draw[mid arrow=0.65] (C.90) -- (F.270);
        \draw[mid arrow=0.5] (C.0) -- ++(0.3,0) -- ++(0,-0.5) -- ++(-3.6,0) -- ++(0,0.5) -- (A.180);
        \draw[mid arrow=0.5] (D.180) -- ++(-0.3,0) -- ++(0,0.5) -- ++(3.6,0) -- ++(0,-0.5) -- (F.0);
    \end{tikzpicture}
\end{array}
\end{align}
However, due to the directed loops in the architecture, the global tensor may not be unitary (see Proposition \ref{prop: Unitarity of unitary network}). Fortunately, wrapping a locality-preserving OBC unitary network (with a sufficient number of sites) guarantees unitarity, as demonstrated in the following: 
\begin{proposition} [Unitarity of PBC-bilayer unitary network]
For an OBC-unitary network $u_{Net}^{OBC}$, if the algebra $\mathcal{A}_{h,in}$ associated with the incoming horizontal leg, once transformed by $u_{OBC}$, will not be supported on the outgoing horizontal leg:
\begin{equation} \label{Eq:unitary condition for PBC-UN}
\begin{aligned}
&\forall O_{h,in} \in \mathcal{A}_{h,in}, \\ 
 &  u_{Net}^{OBC}(O_{h,in}) = O_{v,out} \otimes I_{h,out} \in \mathcal{A}_{v,out},
\end{aligned}
\end{equation}
then, wrapping the unitary network by contracting its external horizontal legs will result in a PBC-unitary network that is a global unitary.
\end{proposition}
\begin{proof}
First, let us examine the bottom layer of the OBC bilayer unitary network. Upon satisfying the required condition, we can decompose it into:

\begin{align}
\begin{array}{c}
\tikzsetnextfilename{unitary_OBC}
\begin{tikzpicture}
     \node[draw, rectangle, rounded corners=2pt, minimum width=1.2cm, minimum height=0.6cm, fill=tensorcolor, thick](A) at (-1.2, 0) {};
        \draw[mid arrow=0.65] ([xshift=-0.4cm,yshift=-0.6cm]A.270) -- ([xshift=-0.4cm]A.270);
        \draw[mid arrow=0.65] ([yshift=-0.6cm]A.270) -- (A.270);
        \draw[mid arrow=0.65] ([xshift=0.4cm,yshift=-0.6cm]A.270) -- ([xshift=0.4cm]A.270);
        \draw[mid arrow=0.65] (A.90) -- ++(0,0.6);
        \draw[mid arrow=0.65] ([xshift=-0.4cm]A.90) -- ++(0,0.6);
        \draw[mid arrow=0.65] ([xshift=0.4cm]A.90) -- ++(0,0.6);
        \draw[mid arrow=0.65] (A.0) -- ++(0.6,0);
        \draw[mid arrow=0.65] ([xshift=-0.6cm]A.180) -- ++(0.6,0);
\end{tikzpicture}
\end{array}
=
\begin{array}{c}
\tikzsetnextfilename{bilayer_global_unitary_OBC}
\begin{tikzpicture}
     \node[draw, rectangle, rounded corners=2pt, minimum width=1.2cm, minimum height=0.6cm, fill=tensorcolor, thick](A) at (0, 0) {};
     \node[draw, rectangle, rounded corners=2pt, minimum width=1.2cm, minimum height=0.6cm, fill=tensorcolor, thick](B) at (0, 1.2) {};
        \draw[mid arrow=0.65] ([xshift=-0.4cm,yshift=-0.6cm]A.270) -- ([xshift=-0.4cm]A.270);
        \draw[mid arrow=0.65] ([yshift=-0.6cm]A.270) -- (A.270);
        \draw[mid arrow=0.65] ([xshift=0.4cm,yshift=-0.6cm]A.270) -- ([xshift=0.4cm]A.270);
        \draw[mid arrow=0.65] (B.90) -- ++(0,0.6);
        \draw[mid arrow=0.65] ([xshift=-0.4cm]B.90) -- ++(0,0.6);
        \draw[mid arrow=0.65] ([xshift=0.4cm]B.90) -- ++(0,0.6);
        \draw[mid arrow=0.65] ([xshift=-0.4cm]A.90) -- ([xshift=-0.4cm]B.270);
        \draw[mid arrow=0.65] (A.90) -- (B.270);
        \draw[mid arrow=0.65] ([xshift=0.4cm]A.90) -- ([xshift=0.4cm]B.270);
        \draw[mid arrow=0.65] (A.0) -- ++(0.6,0);
        \draw[mid arrow=0.65] ([xshift=-0.6cm]B.180) -- ++(0.6,0);
\end{tikzpicture}
\end{array}
\end{align}
Then the corresponding PBC-unitary network can be decomposed into two local unitary tensors with strict causal order:
\begin{align}
\begin{array}{c}
\tikzsetnextfilename{unitary_PBC}
\begin{tikzpicture}
     \node[draw, rectangle, rounded corners=2pt, minimum width=1.2cm, minimum height=0.6cm, fill=tensorcolor, thick](A) at (-1.2, 0) {};
        \draw[mid arrow=0.65] ([xshift=-0.4cm,yshift=-0.6cm]A.270) -- ([xshift=-0.4cm]A.270);
        \draw[mid arrow=0.65] ([yshift=-0.6cm]A.270) -- (A.270);
        \draw[mid arrow=0.65] ([xshift=0.4cm,yshift=-0.6cm]A.270) -- ([xshift=0.4cm]A.270);
        \draw[mid arrow=0.65] (A.90) -- ++(0,0.6);
        \draw[mid arrow=0.65] ([xshift=-0.4cm]A.90) -- ++(0,0.6);
        \draw[mid arrow=0.65] ([xshift=0.4cm]A.90) -- ++(0,0.6);
         \draw[mid arrow=0.5] (A.0) -- ++(0.3,0) -- ++(0,-0.4) -- ++(-1.8,0) -- ++(0,0.4) -- (A.180);
\end{tikzpicture}
\end{array}
=
\begin{array}{c}
\tikzsetnextfilename{bilayer_global_unitary_PBC}
\begin{tikzpicture}
     \node[draw, rectangle, rounded corners=2pt, minimum width=1.2cm, minimum height=0.6cm, fill=tensorcolor, thick](A) at (0, 0) {1};
     \node[draw, rectangle, rounded corners=2pt, minimum width=1.2cm, minimum height=0.6cm, fill=tensorcolor, thick](B) at (0, 1.2) {2};
        \draw[mid arrow=0.65] ([xshift=-0.4cm,yshift=-0.6cm]A.270) -- ([xshift=-0.4cm]A.270);
        \draw[mid arrow=0.65] ([yshift=-0.6cm]A.270) -- (A.270);
        \draw[mid arrow=0.65] ([xshift=0.4cm,yshift=-0.6cm]A.270) -- ([xshift=0.4cm]A.270);
        \draw[mid arrow=0.65] (B.90) -- ++(0,0.6);
        \draw[mid arrow=0.65] ([xshift=-0.4cm]B.90) -- ++(0,0.6);
        \draw[mid arrow=0.65] ([xshift=0.4cm]B.90) -- ++(0,0.6);
        \draw[mid arrow=0.65] ([xshift=-0.4cm]A.90) -- ([xshift=-0.4cm]B.270);
        \draw[mid arrow=0.65] (A.90) -- (B.270);
        \draw[mid arrow=0.65] ([xshift=0.4cm]A.90) -- ([xshift=0.4cm]B.270);
        \draw[mid arrow=0.5] (A.0) -- ++(0.3,0) -- ++(0,0.4) -- ++(-1.8,0) -- ++(0,0.8) -- (B.180);
\end{tikzpicture}
\end{array}.
\end{align}
This also applies to the top layer. Then according to Proposition \ref{prop: Unitarity of unitary network}, the PBC-unitary network forms a global unitary tensor.
\end{proof}
This echoes the wrapping lemma \cite{schumacher2004reversiblequantumcellularautomata}: when neighborhoods in a finite PBC system overlap similarly to in an infinite setting, one can directly correlate a QCA on an infinite lattice with one on a finite lattice with PBC.

For a general 1D OBC unitary network that does not satisfy the locality-preserving condition, contracting its external horizontal legs will produce a non-unitary tensor. For ALPU which is approximately locality-preserving, contracting horizontal legs to create a finite-size PBC unitary network will yield an approximately unitary result. To see that, we can consider the horizontal legs as boundary degrees of freedom located on the left and right boundaries, respectively:
\begin{align}
\begin{array}{c}
\tikzsetnextfilename{ALPU_tensor}
\begin{tikzpicture}
     \node[draw, rectangle, rounded corners=2pt, minimum width=1.2cm, minimum height=0.6cm, fill=tensorcolor, thick](A) at (-1.2, 0) {$U_{OBC}$};
        \draw[mid arrow=0.65] ([xshift=-0.4cm,yshift=-0.6cm]A.270) -- ([xshift=-0.4cm]A.270);
        \draw[mid arrow=0.65] ([yshift=-0.6cm]A.270) -- (A.270);
        \draw[mid arrow=0.65] ([xshift=0.4cm,yshift=-0.6cm]A.270) -- ([xshift=0.4cm]A.270);
        \draw[mid arrow=0.65] (A.90) -- ++(0,0.6);
        \draw[mid arrow=0.65] ([xshift=-0.4cm]A.90) -- ++(0,0.6);
        \draw[mid arrow=0.65] ([xshift=0.4cm]A.90) -- ++(0,0.6);
        \draw[mid arrow=0.65, red!70] (A.0) -- ++(0.6,0);
        \draw[mid arrow=0.65, red!70] ([xshift=-0.6cm]A.180) -- ++(0.6,0);
\end{tikzpicture}
\end{array}
=
\begin{array}{c}
\tikzsetnextfilename{ALPU_gate}
\begin{tikzpicture}
     \node[draw, rectangle, rounded corners=2pt, minimum width=1.6cm, minimum height=0.6cm, fill=tensorcolor, thick](A) at (-1.2, 0) {$U_{OBC}$};
         \draw[mid arrow=0.65, red!70] ([xshift=-0.6cm,yshift=-0.6cm]A.270) -- ([xshift=-0.6cm]A.270);
        \draw[mid arrow=0.65] ([xshift=-0.2cm,yshift=-0.6cm]A.270) -- ([xshift=-0.2cm]A.270);
        \draw[mid arrow=0.65] ([xshift=0.2cm,yshift=-0.6cm]A.270) -- ([xshift=0.2cm]A.270);
        \draw[mid arrow=0.65] ([xshift=0.6cm,yshift=-0.6cm]A.270) -- ([xshift=0.6cm]A.270);
        \draw[mid arrow=0.65] ([xshift=-0.6cm]A.90) -- ++(0,0.6);
        \draw[mid arrow=0.65] ([xshift=-0.2cm]A.90) -- ++(0,0.6);
        \draw[mid arrow=0.65] ([xshift=0.2cm]A.90) -- ++(0,0.6);
        \draw[mid arrow=0.65, red!70] ([xshift=0.6cm]A.90) -- ++(0,0.6); 
\end{tikzpicture}
\end{array}
\end{align}
If the OBC-unitary network is approximately locality-preserving with $f(r)$-tails (a definition of approximately locality-preserving unitary network will be given later in Def.~\ref{def: (Approximate) locality-preserving unitary network}), then 
\begin{equation}
    u_{OBC}(O_{h,in}) \overset{f(L)}{\in} \mathcal{A}_{v,out},
\end{equation}
where $L$ represents the length of the chain, defined as the total number of sites. Given that $\lim_{r\rightarrow\infty} f(r)=0$,  Eq.~(\ref{Eq:unitary condition for PBC-UN})  holds true in the thermodynamic limit. This means that when the chain is extended to greater lengths, $U_{PBC}$ progressively approaches a genuine global unitary tensor.

Apart from PBC-bilayer unitary networks, alternative unitary network architectures can be utilized. These alternative architectures can eliminate directed circuits and ensure global unitarity. For example, for a nearest-neighbor QCA, we may use the following 4-layer structure:
\begin{align} \label{fig: nearest-neighbour un pbc}
\tikzsetnextfilename{4_layer_UN_PBC}
    \begin{tikzpicture}
          \node[draw, rectangle, rounded corners=2pt, minimum width=0.6cm, minimum height=0.6cm, fill=tensorcolor, thick] (a) at (-1.2,-1.8) {};
          \node[draw, rectangle, rounded corners=2pt, minimum width=0.6cm, minimum height=0.6cm, fill=tensorcolor, thick] (b) at (0,-1.8) {};
          \node[draw, rectangle, rounded corners=2pt, minimum width=0.6cm, minimum height=0.6cm, fill=tensorcolor, thick] (c) at (1.2,-1.8) {};
          \node[draw, rectangle, rounded corners=2pt, minimum width=0.6cm, minimum height=0.6cm, fill=tensorcolor, thick] (d) at (-1.2,-0.6) {};
          \node[draw, rectangle, rounded corners=2pt, minimum width=0.6cm, minimum height=0.6cm, fill=tensorcolor, thick] (e) at (0,-0.6) {};
          \node[draw, rectangle, rounded corners=2pt, minimum width=0.6cm, minimum height=0.6cm, fill=tensorcolor, thick] (f) at (1.2,-0.6) {};
          \node[draw, rectangle, rounded corners=2pt, minimum width=0.6cm, minimum height=0.6cm, fill=tensorcolor, thick] (g) at (-1.2,0.6) {};
          \node[draw, rectangle, rounded corners=2pt, minimum width=0.6cm, minimum height=0.6cm, fill=tensorcolor, thick] (h) at (0,0.6) {};
           \node[draw, rectangle, rounded corners=2pt, minimum width=0.6cm, minimum height=0.6cm, fill=tensorcolor, thick] (i) at (1.2,0.6) {};
           \node[draw, rectangle, rounded corners=2pt, minimum width=0.6cm, minimum height=0.6cm, fill=tensorcolor, thick] (j) at (-1.2,1.8) {};
          \node[draw, rectangle, rounded corners=2pt, minimum width=0.6cm, minimum height=0.6cm, fill=tensorcolor, thick] (k) at (0,1.8) {};
          \node[draw, rectangle, rounded corners=2pt, minimum width=0.6cm, minimum height=0.6cm, fill=tensorcolor, thick] (l) at (1.2,1.8) {};
          \draw[mid arrow=0.65] ([yshift=-0.6cm]a.270) -- ++(0,0.6);
          \draw[mid arrow=0.65] ([yshift=-0.6cm]b.270) -- ++(0,0.6);
          \draw[mid arrow=0.65] ([yshift=-0.6cm]c.270) -- ++(0,0.6);
          \draw[mid arrow=0.65] (a.90) -- ++(0,0.6);
          \draw[mid arrow=0.65] (b.90) -- ++(0,0.6);
          \draw[mid arrow=0.65] (c.90) -- ++(0,0.6);
          \draw[mid arrow=0.65] (d.90) -- ++(0,0.6);
          \draw[mid arrow=0.65] (e.90) -- ++(0,0.6);
          \draw[mid arrow=0.65] (f.90) -- ++(0,0.6);
          \draw[mid arrow=0.65] (g.90) -- ++(0,0.6);
          \draw[mid arrow=0.65] (h.90) -- ++(0,0.6);
          \draw[mid arrow=0.65] (i.90) -- ++(0,0.6);
          \draw[mid arrow=0.65] (j.90) -- ++(0,0.6);
          \draw[mid arrow=0.65] (k.90) -- ++(0,0.6);
          \draw[mid arrow=0.65] (l.90) -- ++(0,0.6);
          \draw[mid arrow=0.65] (a.0) to [out = 0,in=180]  (e.180);
          \draw[mid arrow=0.65] (b.0) to [out = 0,in=180]  (f.180);
          \draw[mid arrow =0.5] (c.0) -- ++(0.3,0) -- ++(0,0.4) -- ++(-3.6,0) -- ++(0,0.8) -- (d.180);
          \draw[mid arrow =0.5] (g.180) -- ++(-0.3,0) -- ++(0,0.4) -- ++(3.6,0) -- ++(0,0.8) -- (l.0);
          \draw[mid arrow=0.65] (h.180) to [out = 180,in=0]  (j.0);
          \draw[mid arrow=0.65] (i.180) to [out = 180,in=0]  (k.0);
    \end{tikzpicture}
\end{align}
The design shown in the diagram ensures that the global tensor functions as a nearest-neighbor QCA. 

More general PBC unitary operators can be expressed using SQC architecture \cite{Chen_2024}. SQCs can represent all QCAs and various non-local unitary transformations.  Characterized by local unitary gates and a sequential structure, an SQC can be viewed as a unitary network of finite bond dimensions. We will demonstrate this point more explicitly later in Section \ref{Sec: UN and SQC}.

\section{Quantum cellular automata and unitary networks}
\label{Sec: QCA and unitary networks}

Using the Margolus partitioning scheme, we demonstrate that any 1D QCA can be represented as a bilayer unitary network. We define the locality-preserving criterion for a unitary network, distinct from that of a global unitary operator. We also explore imposing a locality-preserving condition for unitary networks by selecting specific architectures.

\subsection{Margolus partitioning scheme as unitary networks}

In this subsection, we demonstrate that any 1D QCA can be expressed as a bilayer unitary network with finite bond dimensions. 

We consider QCAs to be nearest-neighbor, which can always be ensured by regrouping sites. All 1D nearest-neighbor QCAs can be encapsulated by a Margolus partitioning scheme \cite{schumacher2004reversiblequantumcellularautomata, Gross_2012, ranard2022converse}, which can be regarded as a unitary network:
\begin{align} \label{fig:margolus partitioning}
\tikzsetnextfilename{margolus_partitioning_UN}
    \begin{tikzpicture}
      \node[draw, rectangle, rounded corners=2pt, minimum width=1.8cm, minimum height=0.6cm, fill=tensorcolor, thick] (Wm-1) at (-2.4, -0.6) {$W_{m-1}$};
         \node[draw, rectangle, rounded corners=2pt, minimum width=1.8cm, minimum height=0.6cm, fill=tensorcolor, thick] (Wm) at (0, -0.6) {$W_{m}$};
         \node[draw, rectangle, rounded corners=2pt, minimum width=1.8cm, minimum height=0.6cm, fill=tensorcolor, thick] (Wm+1) at (2.4, -0.6) {$W_{m+1}$};
         \node[draw, rectangle, rounded corners=2pt, minimum width=1.8cm, minimum height=0.6cm, fill=tensorcolor, thick] (Vm-1) at (-3.6, 0.6) {$V_{m-1}$};
         \node[draw, rectangle, rounded corners=2pt, minimum width=1.8cm, minimum height=0.6cm, fill=tensorcolor, thick] (Vm) at (-1.2, 0.6) {$V_{m}$};
         \node[draw, rectangle, rounded corners=2pt, minimum width=1.8cm, minimum height=0.6cm, fill=tensorcolor, thick] (Vm+1) at (1.2, 0.6) {$V_{m+1}$};
         \draw[mid arrow=0.65] 
         ([xshift=-0.6cm,yshift=-0.6cm]Wm-1.270) -- ++(0,0.6);
         \draw[mid arrow=0.65] 
         ([xshift=0.6cm,yshift=-0.6cm]Wm-1.270) -- ++(0,0.6);
         \draw[mid arrow=0.65] ([xshift=-0.6cm,yshift=-0.6cm]Wm.270) -- node[yshift=-0.2cm, below,font=\scriptsize]{$\mathcal{A}_{2m}$}  ++(0,0.6);
         \draw[mid arrow=0.65] 
         ([xshift=0.6cm,yshift=-0.6cm]Wm.270) -- node[yshift=-0.2cm,below,font=\scriptsize]{$\mathcal{A}_{2m+1}$} ++(0,0.6);
         \draw[mid arrow=0.65] 
         ([xshift=-0.6cm,yshift=-0.6cm]Wm+1.270) -- ++(0,0.6);
         \draw[mid arrow=0.65] 
         ([xshift=0.6cm,yshift=-0.6cm]Wm+1.270) -- ++(0,0.6);
         \draw[mid arrow=0.65] ([xshift=-0.6cm]Vm-1.90) -- ++(0,0.6);
          \draw[mid arrow=0.65] ([xshift=+0.6cm]Vm-1.90) -- ++(0,0.6);
         \draw[mid arrow=0.65] ([xshift=-0.6cm]Vm.90) -- node[yshift=0.2cm,above, font=\scriptsize] {$\mathcal{A}_{2m-1}$} ++(0,0.6);
          \draw[mid arrow=0.65] ([xshift=+0.6cm]Vm.90) -- node[yshift=0.2cm,above, font=\scriptsize] {$\mathcal{A}_{2m}$} ++(0,0.6);
          \draw[mid arrow=0.65] ([xshift=-0.6cm]Vm+1.90) -- ++(0,0.6);
          \draw[mid arrow=0.65] ([xshift=+0.6cm]Vm+1.90) -- ++(0,0.6);
          \draw[mid arrow=0.65, very thin] ([xshift=-0.6cm]Wm-1.90) -- ++(0,0.6);
          \draw[mid arrow=0.65,semithick] ([xshift=+0.6cm]Wm-1.90) -- ++(0,0.6);
        \draw[mid arrow=0.65, very thin] ([xshift=-0.6cm]Wm.90) -- node[right, font=\scriptsize] {$\mathcal{B}_{2m}$} ++(0,0.6);
          \draw[mid arrow=0.65,semithick] ([xshift=+0.6cm]Wm.90) -- node[right, font=\scriptsize] {$\mathcal{B}_{2m+1}$}++(0,0.6);
          \draw[mid arrow=0.65, very thin] ([xshift=-0.6cm]Wm+1.90) -- ++(0,0.6);
          \draw[mid arrow=0.65,semithick] ([xshift=+0.6cm]Wm+1.90) -- ++(0,0.6);
          \draw[mid arrow=0.65,very thin] ([xshift=-0.6cm,yshift=-0.6cm]Vm-1.270) -- ++(0,0.6);
    \end{tikzpicture}
\end{align}
where local unitary operators $W_{m}$ and $V_{m}$ are such that
\begin{equation}
\begin{aligned}
  & w_{m}(O) = W_{m} OW_{m}^{\dagger},\\
  &v_{m}(O) = V_{m} OV_{m}^{\dagger}.
\end{aligned}
\end{equation}
This Margolus partitioning scheme for a QCA in \eqref{fig:margolus partitioning} is generally not an FDLU circuit. Unlike quantum circuit wires, the Hilbert space dimensions of various legs (such as $\mathcal{B}_{2m}$ and $\mathcal{B}_{2m+1}$) can vary. A more detailed discussion can be found in Appendix \ref{Append: Margolus Partitioning is not a quantum circuit}.

As we show below, the unitary network in  \eqref{fig:margolus partitioning} can always be transformed into a bilayer unitary network. Starting from the Margolus partitioning scheme ($W_{m}$ and $V_{m}$), we do the following decomposition of $W_{m}$:
\begin{align} \label{fig:Wm}
    \begin{array}{c}
    \tikzsetnextfilename{decompose_Wm_global}
    \begin{tikzpicture}
         \node[draw, rectangle, rounded corners=2pt, minimum width=1cm, minimum height=0.6cm, fill=tensorcolor, thick] (Wm) {$W_{m}$} ;
         \draw[mid arrow=0.65] ([xshift=-0.3cm]Wm.90) -- ++(0,0.6);
         \draw[mid arrow=0.65] ([xshift=0.3cm]Wm.90) -- ++(0,0.6);
         \draw[mid arrow=0.65] ([xshift=-0.3cm,yshift=-0.6cm]Wm.270) -- ++(0,0.6);
         \draw[mid arrow=0.65] ([xshift=0.3cm,yshift=-0.6cm]Wm.270) -- ++(0,0.6);
    \end{tikzpicture}
    \end{array}
    =
     \begin{array}{c}
     \tikzsetnextfilename{decompose_Wm_bilayer}
    \begin{tikzpicture}
         \node[draw, rectangle, rounded corners=2pt, minimum width=0.6cm, minimum height=0.6cm, fill=tensorcolor, thick] (A)  at (-0.6,-0.6) {$A_{m}$};
         \node[draw, rectangle, rounded corners=2pt, minimum width=0.6cm, minimum height=0.6cm, fill=tensorcolor, thick] (B)  at (0.6,-0.6) {$B_{m}$};
         \node[draw, rectangle, rounded corners=2pt, minimum width=0.6cm, minimum height=0.6cm, fill=tensorcolor, thick] (C)  at (-0.6, 0.6) {$C_{m}$};
         \node[draw, rectangle, rounded corners=2pt, minimum width=0.6cm, minimum height=0.6cm, fill=tensorcolor, thick] (D)  at (0.6, 0.6) {$D_{m}$};
         \draw[mid arrow=0.65]  ([yshift=-0.6cm]A.270) -- ++(0,0.6);
         \draw[mid arrow=0.65] ([yshift=-0.6cm]B.270) -- ++(0,0.6);
         \draw[mid arrow=0.65] (C.90) -- ++(0,0.6);
         \draw[mid arrow=0.65] (D.90) -- ++(0,0.6);
         \draw[mid arrow=0.65] (A.0) -- (B.180);
        \draw[mid arrow=0.65] (A.90) -- (C.270);
        \draw[mid arrow=0.65] (B.90) -- (D.270);
        \draw[mid arrow=0.65] (D.180) -- (C.0);
    \end{tikzpicture}
    \end{array}
    =  
    \begin{array}{c}
    \tikzsetnextfilename{decompose_Wm_LR}
    \begin{tikzpicture}
         \node[draw, rectangle, rounded corners=2pt, minimum width=0.6cm, minimum height=1.2cm, fill=tensorcolor, thick] (A)  at (-0.6,0) {$L_{m}$};
         \node[draw, rectangle, rounded corners=2pt, minimum width=0.6cm, minimum height=1.2cm, fill=tensorcolor, thick] (B)  at (0.6, 0) {$R_{m}$};
         \draw[mid arrow=0.65] (A.90) -- ++(0,0.6);
         \draw[mid arrow=0.65] (B.90) -- ++(0,0.6);
         \draw[mid arrow=0.65] ([yshift=-0.6cm]A.270) -- ++(0,0.6);
         \draw[mid arrow=0.65] ([yshift=-0.6cm]B.270) -- ++(0,0.6);
         \draw[mid arrow=0.65] ([yshift=-0.3cm]A.0) -- ([yshift=-0.3cm]B.180);
         \draw[mid arrow=0.65] ([yshift=0.3cm]B.180) -- ([yshift=0.3cm]A.0);
    \end{tikzpicture}
    \end{array}
\end{align}
$W_{m}$ can be broken down into a bilayer unitary network, as illustrated from the left panel to the middle, using the method described in (\ref{Eq: decompose into BUN}). From the middle to the right panel, we contract $A_{m}$ with $C_{m}$ to form $L_{m}$, and $B_{m}$ with $D_{m}$ to produce $R_{m}$. This is for later use when calculating the GNVW index. All tensors represented in the diagram are ensured to be unitary.

\begin{align} \label{fig:QCA Bilayer}
\begin{array}{c}
\tikzsetnextfilename{margolus:VW}
\begin{tikzpicture}
    \node[draw, rectangle, rounded corners=2pt, minimum width=1cm, minimum height=0.6cm, fill=tensorcolor, thick] (Wm) at (0,-0.6) {$W_{m}$} ;
    \node[draw, rectangle, rounded corners=2pt, minimum width=1cm, minimum height=0.6cm, fill=tensorcolor, thick] (Vm) at (-0.6, 0.6) {$V_{m}$} ;
    \node[draw, rectangle, rounded corners=2pt, minimum width=1cm, minimum height=0.6cm, fill=tensorcolor, thick] (Wm-1) at (-1.2,-0.6){$W_{m-1}$} ;
         \draw[mid arrow=0.65] ([xshift=-0.3cm]Wm.90) -- ++(0,0.6);
         \draw[mid arrow=0.65] ([xshift=0.3cm]Wm.90) -- ++(0,0.6);
         \draw[mid arrow=0.65] ([xshift=-0.3cm,yshift=-0.6cm]Wm.270) -- ++(0,0.6);
         \draw[mid arrow=0.65] ([xshift=0.3cm,yshift=-0.6cm]Wm.270) -- ++(0,0.6);
         \draw[mid arrow=0.65] ([xshift=-0.3cm]Vm.90) -- ++(0,0.6);
         \draw[mid arrow=0.65] ([xshift=0.3cm]Vm.90) -- ++(0,0.6);
          \draw[mid arrow=0.65] ([xshift=-0.3cm]Wm-1.90) -- ++(0,0.6);
         \draw[mid arrow=0.65] ([xshift=-0.3cm,yshift=-0.6cm]Vm.270) -- ++(0,0.6);
         \draw[mid arrow=0.65] ([xshift=0.3cm,yshift=-0.6cm]Wm-1.270) -- ++(0,0.6);
         \draw[mid arrow=0.65] ([xshift=-0.3cm,yshift=-0.6cm]Wm-1.270) -- ++(0,0.6);
\end{tikzpicture}
\end{array}
= 
\begin{array}{c}
\tikzsetnextfilename{margolus__bilayer_UN}
\begin{tikzpicture}
        \node[draw, rectangle, rounded corners=2pt, minimum width=1.6cm, minimum height=0.6cm, fill=tensorcolor, font=\scriptsize, thick] (Vm)  at (-1.2,1.8) {$V_{m}$};
        \draw[mid arrow=0.65] ([xshift=-0.6cm]Vm.90) -- ++(0,0.6);
        \draw[mid arrow=0.65] ([xshift=0.6cm]Vm.90) -- ++(0,0.6);
        \node[draw, rectangle, rounded corners=2pt, minimum width=0.6cm, minimum height=0.6cm, fill=tensorcolor, font=\scriptsize, thick] (A)  at (-0.6,-0.6) {$A_{m}$};
         \node[draw, rectangle, rounded corners=2pt, minimum width=0.6cm, minimum height=0.6cm, fill=tensorcolor, font=\scriptsize, thick] (B)  at (0.6,-0.6) {$B_{m}$};
         \node[draw, rectangle, rounded corners=2pt, minimum width=0.6cm, minimum height=0.6cm, font=\scriptsize,fill=tensorcolor, font=\scriptsize, thick] (C)  at (-0.6, 0.6) {$C_{m}$};
         \node[draw, rectangle, rounded corners=2pt, minimum width=0.6cm, minimum height=0.6cm, fill=tensorcolor, font=\scriptsize, thick] (D)  at (0.6, 0.6) {$D_{m}$};
         \node[draw, rectangle, rounded corners=2pt, minimum width=0.6cm, minimum height=0.6cm, fill=tensorcolor, font=\scriptsize, thick] (A-1)  at (-3.2,-0.6) {$A_{m-1}$};
         \node[draw, rectangle, rounded corners=2pt, minimum width=0.6cm, minimum height=0.6cm, fill=tensorcolor, font=\scriptsize, thick] (B-1)  at (-1.8,-0.6) {$B_{m-1}$};
         \node[draw, rectangle, rounded corners=2pt, minimum width=0.6cm, minimum height=0.6cm, font=\scriptsize,fill=tensorcolor, font=\scriptsize, thick] (C-1)  at (-3.2, 0.6) {$C_{m-1}$};
         \node[draw, rectangle, rounded corners=2pt, minimum width=0.6cm, minimum height=0.6cm, fill=tensorcolor, font=\scriptsize, thick] (D-1)  at (-1.8, 0.6) {$D_{m-1}$};
          \draw[mid arrow=0.65]  ([yshift=-0.6cm]A.270) -- ++(0,0.6);
         \draw[mid arrow=0.65] ([yshift=-0.6cm]B.270) -- ++(0,0.6);
         \draw[mid arrow=0.65] (C.90) -- ++(0,0.6);
         \draw[mid arrow=0.65] (D.90) -- ++(0,0.6);
         \draw[mid arrow=0.65] (A.0) -- (B.180);
        \draw[mid arrow=0.65] (A.90) -- (C.270);
        \draw[mid arrow=0.65] (B.90) -- (D.270);
        \draw[mid arrow=0.65] (D.180) -- (C.0);
        \draw[mid arrow=0.65]  ([yshift=-0.6cm]A-1.270) -- ++(0,0.6);
         \draw[mid arrow=0.65] ([yshift=-0.6cm]B-1.270) -- ++(0,0.6);
         \draw[mid arrow=0.65] (C-1.90) -- ++(0,0.6);
         \draw[mid arrow=0.65] (D-1.90) -- ++(0,0.6);
         \draw[mid arrow=0.65] (A-1.0) -- (B-1.180);
        \draw[mid arrow=0.65] (A-1.90) -- (C-1.270);
        \draw[mid arrow=0.65] (B-1.90) -- (D-1.270);
        \draw[mid arrow=0.65] (D-1.180) -- (C-1.0);
        \draw[dashed, blue] ([shift={(-0.3,-0.1)}]B-1.225) -- ([shift={(0.1,-0.1)}]A.315) -- ([shift={(0.1,0.1)}]A.45) -- ([shift={(-0.3,0.1)}]B-1.135) -- cycle;
        \draw[dashed, blue] ([shift={(-0.3,-0.1)}]D-1.225) -- ([shift={(0.1,-0.1)}]C.315) -- ([shift={(0.2,0.4)}]Vm.0) -- ([shift={(-0.4,0.4)}]Vm.180) -- cycle;
\end{tikzpicture}
\end{array} 
\end{align}
By reincorporating $V_{m}$ into the unitary network, we can then contract $\{V_{m}, D_{m-1}, C_{m}\}$ and integrate $\{B_{m-1}, A_{m}\}$. This results in a bilayer unitary network where each local unitary tensor is supported on a supercell of two sites $\{ m-1, m\}$. The bond dimension $D < d^2$, where $d$ denotes the dimension of the physical Hilbert space for an individual site. If we disregard the orientation of the horizontal legs, we retrieve an MPU representation for the QCA \cite{cirac2017matrix,_ahino_lu_2018,Farrelly_2020}.

In fact, this bilayer unitary network representation can be used for any QCA that admits a Margolus partitioning scheme. However, for dimensions greater than 1, there are QCAs that do not allow for a Margolus partitioning scheme \cite{arrighi2008one}.

\subsection{Locality-preserving unitary network}

In representing global unitary operators with unitary networks, we differentiate the physical unitary operators $U$ from their corresponding unitary network implementations $U_{Net}$. It is sometimes useful to discuss whether the unitary network implementation $U_{Net}$ of a QCA is also locality-preserving. We begin by defining the distance in a unitary network.
\begin{definition}[Distance in a unitary network graph]
\label{def: UN distance}
Let $U_{Net}$ be a unitary network associated with the graph $G[U_{Net}] = (V,E)$. Each directed edge $e_{ij} \in E$ can have a defined length $|e_{ij}|$. The distance between vertices $s$ and $t$ in the network is defined as:
\begin{equation}
\begin{aligned}
     d_{U_{Net}} (s, t) \coloneq  \min_{P \in \mathcal{P}(s,t)}  \sum_{(u,v) \in P} |e_{uv}|,
\end{aligned}
\end{equation}
where $\mathcal{P}(s, t) $ is the set of all directed paths from $s$ to $t$.
\end{definition}
In general unitary networks, only the locations of the sinks $Out = \{o_{1}, o_{2}, \cdots \}$ and the sources $In = \{i_{1}, i_{2}, \cdots \}$ are specified, as they correspond to physical sites. The edge lengths $|e_{ij}|$ between the local tensors are rather arbitrary. For a bilayer unitary network $U_{Net}$ defined in a lattice system with lattice constant $a$, each local tensor is associated with a site, denoted by position $x$.
\begin{align}
\tikzsetnextfilename{info_path_length_UN}
    \begin{tikzpicture}
        \node[draw, circle, minimum width=0.6cm, minimum height=0.6cm, fill=tensorcolor, thick](o1) at (-1.2, 2.4) {$o_{1}$};
        \node[draw, circle, minimum width=0.6cm, minimum height=0.6cm, fill=tensorcolor, thick](o2) at (0, 2.4) {$o_{2}$};
        \node[draw, circle, minimum width=0.6cm, minimum height=0.6cm, fill=tensorcolor, thick](o3) at (1.2, 2.4) {$o_{3}$};
        \node[draw, circle, minimum width=0.6cm, minimum height=0.6cm, fill=tensorcolor, thick](i1) at (-1.2, -1.2) {$i_{1}$};
        \node[draw, circle, minimum width=0.6cm, minimum height=0.6cm, fill=tensorcolor, thick](i2) at (0, -1.2) {$i_{2}$};
        \node[draw, circle, minimum width=0.6cm, minimum height=0.6cm, fill=tensorcolor, thick](i3) at (1.2, -1.2) {$i_{3}$};
        \node[draw, rectangle, rounded corners=2pt, minimum width=0.6cm, minimum height=0.6cm, fill=tensorcolor, thick](A) at (-1.2, 0) {1};
        \node[draw, rectangle, rounded corners=2pt, minimum width=0.6cm, minimum height=0.6cm, fill=tensorcolor, thick](B) {2};
        \node[draw, rectangle, rounded corners=2pt, minimum width=0.6cm, minimum height=0.6cm, fill=tensorcolor, thick](C) at (1.2, 0) {3};
        \node[draw, rectangle, rounded corners=2pt, minimum width=0.6cm, minimum height=0.6cm, fill=tensorcolor, thick](D) at (-1.2, 1.2) {6};
        \node[draw, rectangle, rounded corners=2pt, minimum width=0.6cm, minimum height=0.6cm, fill=tensorcolor, thick](E) at (0, 1.2) {5};
        \node[draw, rectangle, rounded corners=2pt, minimum width=0.6cm, minimum height=0.6cm, fill=tensorcolor, thick](F) at (1.2, 1.2) {4};
        \draw[mid arrow=0.65] (D.90) -- (o1.270);
        \draw[mid arrow=0.65] (E.90) -- (o2.270);
        \draw[mid arrow=0.65] (F.90) --(o3.270);
        \draw[mid arrow=0.65] (i1.90) -- (A.270);
        \draw[mid arrow=0.65] (i2.90) -- (B.270);
        \draw[mid arrow=0.65] (i3.90) -- (C.270);
        \draw[mid arrow=0.65] (A.0) -- (B.180);
        \draw[mid arrow=0.65] (B.0) -- (C.180);
        \draw[mid arrow=0.65] (F.180) -- (E.0);
        \draw[mid arrow=0.65] (E.180) -- (D.0);
        \draw[mid arrow=0.65] (A.90) -- (D.270);
        \draw[mid arrow=0.65] (B.90) -- (E.270);
        \draw[mid arrow=0.65] (C.90) -- (F.270);
    \end{tikzpicture}
\end{align}
Therefore, the length $|e_{ij}|$ for edges can be naturally defined below:
\begin{enumerate} [label=(\roman*)]
    \item For every horizontal leg $e^{h}_{ij}$, its length $|e^{h}_{ij}| = a$, where $a$ is the lattice constant. This is the distance over which information travels to adjacent sites.
    \item For every vertical leg $e^{v}_{ij}$,  its length $|e^{v}_{ij}| = 0$. This can be justified by the fact that the temporal extension is artificial and can be adjusted arbitrarily.
\end{enumerate}

We also define a causal cone  within a unitary network:
\begin{definition} [Causal cone in unitary network \cite{ferris2012perfect}]
The causal cone of a vertex $v \in V$ within $U_{Net}$ is a sub-network $U_{Net}^{cc}(v)$ containing all vertices (local unitary tensors) accessible from $v$ and the edges (legs) connected to these vertices:
\begin{equation}
    G[U_{Net}^{cc}(v)] = (V^{cc}(v), E^{cc}(v)) \in G[U_{Net}]
\end{equation}
The base of the causal cone denoted as $B[U_{Net}^{cc}(v)] \subseteq \Gamma$, is a subset of lattice points $\Gamma$, defined as the support of all sinks contained in the subgraph $G[U_{Net}^{cc}(v)]$:
\begin{equation}
    B[U_{Net}^{cc}(v)] = \{ x(o_{i})\in \Gamma:  o_{i} \in (V^{cc}(v) \cap Out)   \} 
\end{equation}
\end{definition}
To assess the transformation of a local operator $O_{x}$, it suffices to consider the local tensors within the causal cone from source vertex $i_{x}$, thereby reducing computational cost \cite{ferris2012perfect}. For example, take the unitary network $U_{Net}$ obtained from the Margolus partitioning of a QCA; the causal cone connected to a local operator $O_{2m}$ located at lattice site $2m$ is highlighted in blue:
\begin{align} \label{fig:causal cone}
\tikzsetnextfilename{UN_causal_cone}
    \begin{tikzpicture}
    \node[draw, rectangle, rounded corners=2pt, minimum width=0.6cm, minimum height=0.6cm, fill=btensorcolor, thick]  at (-0.6, -1.8) {$i_{2m}$};
    \node[draw, rectangle, rounded corners=2pt, minimum width=0.6cm, minimum height=0.6cm, fill=btensorcolor, thick]  at (-0.6, 1.8) {$o_{2m}$};
    \node[draw, rectangle, rounded corners=2pt, minimum width=0.6cm, minimum height=0.6cm, fill=btensorcolor, thick]  at (-1.8, 1.8) {$o_{2m-1}$};
    \node[draw, rectangle, rounded corners=2pt, minimum width=0.6cm, minimum height=0.6cm, fill=btensorcolor, thick]  at (0.6, 1.8) {$o_{2m+1}$};
    \node[draw, rectangle, rounded corners=2pt, minimum width=0.6cm, minimum height=0.6cm, fill=btensorcolor, thick]  at (1.8, 1.8) {$o_{2m+2}$};
      \node[draw, rectangle, rounded corners=2pt, minimum width=1.8cm, minimum height=0.6cm, fill=tensorcolor, thick] (Wm-1) at (-2.4, -0.6) {$W_{m-1}$};
         \node[draw, rectangle, rounded corners=2pt, minimum width=1.8cm, minimum height=0.6cm, fill=btensorcolor, thick] (Wm) at (0, -0.6) {$W_{m}$};
         \node[draw, rectangle, rounded corners=2pt, minimum width=1.8cm, minimum height=0.6cm, fill=tensorcolor, thick] (Wm+1) at (2.4, -0.6) {$W_{m+1}$};
         \node[draw, rectangle, rounded corners=2pt, minimum width=1.8cm, minimum height=0.6cm, fill=tensorcolor, thick] (Vm-1) at (-3.6, 0.6) {$V_{m-1}$};
         \node[draw, rectangle, rounded corners=2pt, minimum width=1.8cm, minimum height=0.6cm, fill=btensorcolor, thick] (Vm) at (-1.2, 0.6) {$V_{m}$};
         \node[draw, rectangle, rounded corners=2pt, minimum width=1.8cm, minimum height=0.6cm, fill=btensorcolor, thick] (Vm+1) at (1.2, 0.6) {$V_{m+1}$};
         \draw[mid arrow=0.65] 
         ([xshift=-0.6cm,yshift=-0.6cm]Wm-1.270) -- ++(0,0.6);
         \draw[mid arrow=0.65] 
         ([xshift=0.6cm,yshift=-0.6cm]Wm-1.270) -- ++(0,0.6);
         \draw[mid arrow=0.65, btensorcolor] ([xshift=-0.6cm,yshift=-0.6cm]Wm.270) --   ++(0,0.6);
         \draw[mid arrow=0.65, btensorcolor] 
         ([xshift=0.6cm,yshift=-0.6cm]Wm.270) --  ++(0,0.6);
         \draw[mid arrow=0.65] 
         ([xshift=-0.6cm,yshift=-0.6cm]Wm+1.270) -- ++(0,0.6);
         \draw[mid arrow=0.65] 
         ([xshift=0.6cm,yshift=-0.6cm]Wm+1.270) -- ++(0,0.6);
         \draw[mid arrow=0.65] ([xshift=-0.6cm]Vm-1.90) -- ++(0,0.6);
          \draw[mid arrow=0.65] ([xshift=+0.6cm]Vm-1.90) -- ++(0,0.6);
         \draw[mid arrow=0.65, btensorcolor] ([xshift=-0.6cm]Vm.90) --  ++(0,0.6);
          \draw[mid arrow=0.65, btensorcolor] ([xshift=+0.6cm]Vm.90) --  ++(0,0.6);
          \draw[mid arrow=0.65, btensorcolor] ([xshift=-0.6cm]Vm+1.90) -- ++(0,0.6);
          \draw[mid arrow=0.65, btensorcolor] ([xshift=+0.6cm]Vm+1.90) -- ++(0,0.6);
          \draw[mid arrow=0.65, very thin] ([xshift=-0.6cm]Wm-1.90) -- ++(0,0.6);
          \draw[mid arrow=0.65,semithick, btensorcolor] ([xshift=+0.6cm]Wm-1.90) -- ++(0,0.6);
        \draw[mid arrow=0.65, very thin, btensorcolor] ([xshift=-0.6cm]Wm.90) --  ++(0,0.6);
          \draw[mid arrow=0.65,semithick,btensorcolor] ([xshift=+0.6cm]Wm.90) -- ++(0,0.6);
          \draw[mid arrow=0.65, very thin, btensorcolor] ([xshift=-0.6cm]Wm+1.90) -- ++(0,0.6);
          \draw[mid arrow=0.65,semithick] ([xshift=+0.6cm]Wm+1.90) -- ++(0,0.6);
          \draw[mid arrow=0.65,very thin] ([xshift=-0.6cm,yshift=-0.6cm]Vm-1.270) -- ++(0,0.6);
    \end{tikzpicture}
\end{align}
It can be seen that a local operator $O_{X}$, transformed by a unitary network $U = \text{Eval}(U_{Net})$, is localized at the base of the causal cone:
\begin{equation}
    u(O_{X}) = UO_{X}U^{\dagger} \in \mathcal{A}_{\cup_{x \in X}B[U_{Net}^{cc}(i_{x})]}.
\end{equation}
where the base of causal cone is given by
\begin{equation}
    B[G_{cc}(O_{2m})] = \{ 2m-1, 2m,2m+1, 2m+2 \}.
\end{equation}

Sometimes, it is desirable for the network architecture to inherently exhibit the property of locality-preserving. This can be done by requiring that every causal cone of a vertex is bounded:
\begin{equation}
   \forall v \in V: V^{cc}(v) \subseteq B_{U_{Net}}(v,R)
\end{equation}
for some radius $R$, where 
\begin{equation}
     \begin{aligned}
         B_{U_{Net}}(v,R) = \{ u\in V: d_{U_{Net}}(v,u) \le R\}.
     \end{aligned}
\end{equation}
This property is related to the graph $G[U_{Net}]$ and does not depend on the value of the local unitary tensor at each vertex. Therefore, we can maintain the locality-preserving property in QCA representations by designing unitary network architectures. Margolus partitioning architecture in \eqref{fig:causal cone} meets this condition.

A bilayer unitary network serves as a counterexample to a locality-preserving architecture. In such a network, the causal cone of any source $i$ extends across the whole lattice $\Gamma$. As a result, a bilayer unitary network can implement a global unitary operation that violates the locality-preserving condition.

Despite its non-local architecture, a bilayer unitary network can still function as a locality-preserving implementation of a QCA. 
\begin{align} \label{fig:shift OBC}
    \tikzsetnextfilename{shift_UN_OBC}
    \begin{tikzpicture} 
    \node[draw, rectangle, rounded corners=2pt, minimum width=0.6cm, minimum height=0.6cm, fill=tensorcolor, thick] (a) at (0,0) {};
    \node[draw, rectangle, rounded corners=2pt, minimum width=0.6cm, minimum height=0.6cm, fill=tensorcolor, thick] (b) at (1.2,0) {};
    \node[draw, rectangle, rounded corners=2pt, minimum width=0.6cm, minimum height=0.6cm, fill=tensorcolor, thick] (c) at (2.4,0) {};
    \node[draw, rectangle, rounded corners=2pt, minimum width=0.6cm, minimum height=0.6cm, fill=tensorcolor, thick] (d) at (3.6,0) {};
    \draw[mid arrow=0.65] (a.90) -- ++(0,0.6);
    \draw[mid arrow=0.65] ([yshift=-0.6cm]a.270) -- ++(0,0.6);
    \draw[mid arrow=0.65] (b.90) -- ++(0,0.6);
    \draw[mid arrow=0.65] ([yshift=-0.6cm]b.270) -- ++(0,0.6);
    \draw[mid arrow=0.65] (c.90) -- ++(0,0.6);
    \draw[mid arrow=0.65] ([yshift=-0.6cm]c.270) -- ++(0,0.6);
    \draw[mid arrow=0.65] (d.90) -- ++(0,0.6);
    \draw[mid arrow=0.65] ([yshift=-0.6cm]d.270) -- ++(0,0.6);
    \draw[mid arrow=0.65] (a.0) -- (b.180);
    \draw[mid arrow=0.65] (b.0) -- (c.180);
    \draw[mid arrow=0.65] (c.0) -- (d.180);
    \draw[mid arrow=0.5] (d.0) --  ++(0.6,0);
    \draw[mid arrow=0.5] ([xshift=-0.6cm]a.180) --  (a.180);
    \draw (a.180) to [out = 0,in=270] (a.90);
    \draw (b.180) to [out = 0,in=270] (b.90);
    \draw (c.180) to [out = 0,in=270] (c.90);
    \draw (d.180) to [out = 0,in=270] (d.90);
    \draw (a.270) to [out = 90,in=180] (a.0);
    \draw (b.270) to [out = 90,in=180] (b.0);
    \draw (c.270) to [out = 90,in=180] (c.0);
    \draw (d.270) to [out = 90,in=180] (d.0);
\end{tikzpicture}
\end{align}
For example, the above bilayer unitary network provides a locality-preserving implementation for a shift operation. It is a locality-preserving implementation since its information flows only locally. Therefore, an intrinsic definition is necessary to determine when a unitary network implementation maintains locality.

\begin{definition} [(Approximate) locality-preserving unitary network]
\label{def: (Approximate) locality-preserving unitary network}
    A unitary network $U_{Net}$ is locality preserving if there is some radius $R>0$, such that any sub-network $U^{S}_{Net} \subseteq U_{Net}$  forms a locality-preserving unitary operator with respect to the distance defined on the network $d_{U^{S}_{Net}}(s,t)$:
    \begin{equation}
         u^{S}_{Net}(O_{s}^{in})  \in  \mathcal{A}^{out}_{\bar{B}(s,R)} \quad \text{for} \ O_{s}^{in} \in \mathcal{A}_{s}^{in},
     \end{equation}
     where $\bar{B}(s,R)$ is the closed ball in graph $G[U^{S}_{Net}]$:
     \begin{equation}
     \begin{aligned}
         \bar{B}(s,R) = \{t \in V^{S}: d_{U_{Net}^{S}}(s,t) \le R\},
     \end{aligned}
     \end{equation}
    $\mathcal{A}_{s}^{in}$ is the algebra on the external input Hilbert space of vertex $s$, and $\mathcal{A}_{t}^{out}$ is the algebra on the external output Hilbert space of vertex $t$.

    Similarly, a unitary network $U_{Net}$ is approximately locality preserving if all its sub-networks are approximately locality preserving with respect to the distance $d_{U^{S}_{Net}}(s,t)$:

     \begin{equation}
        u^{S}_{Net}(O_{s}^{in}) \overset{f(r)}{\in} \mathcal{A}^{out}_{\bar{B}(s,R)}.
    \end{equation}
where $f (r)$ is some positive function that satisfied $\lim_{r\rightarrow\infty} f(r)=0$.
\end{definition}

Intuitively, a locality-preserving unitary network forbids non-local information flow. This is illustrated by the following characteristic of such unitary networks: 
For two vertices $s,t \in V$, if $d_{U_{Net}}(s,t) > R$, then for any sub-network $U_{Net}^{S}$ from $s$ to $t$, 
     \begin{equation}
         \boldsymbol{S}(u^{S}_{Net}(\mathcal{A}_{s}^{in}), \mathcal{A}_{t}^{out}) = I_{t}^{out},  
     \end{equation}
 where $\boldsymbol{S}$ means the support algebra, $I_{t}^{out}$ is the identity.

\section{Non-local unitaries and Unitary networks} 
\label{Sec: Non-local unitaries and Unitary networks}

This section provides examples showing that unitary networks with finite bond dimensions can capture non-local unitaries, which are unitary operators that violate the locality-preserving condition. The stacked XY circuit exemplifies an ALPU with exponentially decaying tails. The stacked CNOT circuit and 1D Kramers-Wannier transformation are non-local unitary operators, which do not display decaying tails.

\subsection{Stacked CNOT circuit}

A quantum circuit of infinite depth is typically non-local. Here, we present an example of an SQC, yet it can be represented by a unitary network with finite bond dimensions. 

Consider the CNOT gate $C_{n}$ that operates on the $n$ th and $n+1$ th qubits. The dynamic of the algebra of qubits is given by \cite{Nielsen_Chuang_2010}
\begin{equation}
\label{Eq:CNOT algebra}
    \begin{aligned}
        &C_{n} X_{n} C_{n} = X_{n} X_{n+1}, \\
          &C_{n} Z_{n} C_{n} = Z_{n},  \\
        &C_{n} X_{n+1} C_{n} = X_{n+1}, \\
        &C_{n} Z_{n+1} C_{n} = Z_{n} Z_{n+1}.    
    \end{aligned}
\end{equation}
The transformation of $Y_{n}$ can be derived from the transformations of $X_{n}$ and $Z_{n}$, with $X_{n}, Y_{n}, Z_{n}$ being Pauli matrices on the $n$th qubit, and $C_{n}^{-1} = C_{n}^{\dagger} = C_{n}$.

\begin{definition} [Stacked CNOT circuit]
The stacked CNOT circuit is defined by the limit of $StC_{N}$, achieved through the repeated application of CNOT gates:
\begin{equation}
\begin{aligned}
StC_{N} &=  \prod_{n=N}^{-N} C_{n} = C_{-N}\cdots C_{n-1} C_{n} C_{n+1} \cdots  C_{N},\\
    StC &= \lim_{N \rightarrow +\infty} StC_{N} = \cdots C_{n-1} C_{n} C_{n+1} \cdots.
\end{aligned}
\end{equation} 
The CNOT gates $C_{n}$ and $C_{n+1}$ do not commute, so we specify they are performed in the reverse order ($C_{n+1}$ is performed before $C_{n}$). 
The diagram of a stacked CNOT circuit is given below:
\begin{align}
\tikzsetnextfilename{stacked_CNOT}
\begin{tikzpicture}
    \filldraw[black] (-0.6,0) circle (2pt) ;  
    \filldraw[black] (0,0) circle (2pt) ;
    \filldraw[black] (0.6,0) circle (2pt);
    \filldraw[black] (1.2,0) circle (2pt);
    \draw (-0.6, -0.6) -- (-0.6, 1.2);
    \draw (0, -0.6) -- (0, 1.2);
    \draw (0.6, -0.6) -- (0.6, 1.2);
    \draw (1.2, -0.6) -- (1.2, 1.2);
    \node[draw, circle, minimum size=0.3cm] (a) at (-0.6,0.6) {};
    \node[draw, circle, minimum size=0.3cm] (b) at (0,0.6) {};
    \node[draw, circle, minimum size=0.3cm] (c) at (0.6,0.6) {};
    \node[draw, circle, minimum size=0.3cm] (d) at (1.2,0.6) {};
    \draw (a.180) -- (a.0);
    \draw (b.180) -- (b.0);
    \draw (c.180) -- (c.0);
    \draw (d.180) -- (d.0);
    \draw[dashed] (-1.2,0) to [out=0,in=180] (a.180);
    \draw (-0.6,0) to [out=0,in=180] (b.180);
    \draw (0,0) to [out=0,in=180] (c.180);
    \draw (0.6,0) to [out=0,in=180] (d.180);
    \draw[dashed] (1.2,0) to [out=0,in=180] ([xshift=0.6cm]d.180);
\end{tikzpicture}
\end{align}
\end{definition}

By iteratively using Eq.~(\ref{Eq:CNOT algebra}), we may get the algebra of $StC_{N}$:
\begin{equation}
\label{Eq:StC algebra}
\begin{aligned}
   & StC_{N} X_{n} StC_{N}^{\dagger} = X_{n} X_{n+1} \quad  \text{for} \ n \le  N-1, \\
   & StC_{N} X_{N} StC_{N}^{\dagger} = X_{N}, \\
   & StC_{N} Z_{n} StC_{N}^{\dagger} = Z_{-N} \cdots Z_{n},  \quad \text{for} \ -N \le  n \le  N.
\end{aligned}
\end{equation}
The stacked CNOT circuit is non-local, as the supports of $u_{StC}(Y_{n})$ and $u_{StC}(Z_{n})$ indefinitely extend in the thermodynamic limit as $N \rightarrow \infty$. A reverse alignment of the stacked-CNOT can give a different global unitary operation, which can be found in Appendix \ref{Append: Reversed Stacked-CNOT}. In an infinite OBC scenario, the stacked CNOT circuit can be represented by a bilayer unitary network, as illustrated below. 
\begin{align}
\tikzsetnextfilename{stacked_CNOT_UN}
\begin{tikzpicture}
    \node[draw, rectangle, rounded corners=2pt, minimum width=1cm, minimum height=0.6cm, fill=tensorcolor, thick] (a1) at (0,0) {};
    \node[draw, rectangle, rounded corners=2pt, minimum width=1cm, minimum height=0.6cm, fill=tensorcolor, thick] (b1) at (1.6,0) {};
    \node[draw, rectangle, rounded corners=2pt, minimum width=1cm, minimum height=0.6cm, fill=tensorcolor, thick] (c1) at (3.2,0) {};
    \node[draw, rectangle, rounded corners=2pt, minimum width=1cm, minimum height=0.6cm, fill=tensorcolor, thick] (d1) at (4.8,0) {};
    \node[draw, rectangle, rounded corners=2pt, minimum width=1cm, minimum height=0.6cm, fill=tensorcolor, thick] (a2) at (0,1.2) {};
    \node[draw, rectangle, rounded corners=2pt, minimum width=1cm, minimum height=0.6cm, fill=tensorcolor, thick] (b2) at (1.6,1.2) {};
    \node[draw, rectangle, rounded corners=2pt, minimum width=1cm, minimum height=0.6cm, fill=tensorcolor, thick] (c2) at (3.2,1.2) {};
    \node[draw, rectangle, rounded corners=2pt, minimum width=1cm, minimum height=0.6cm, fill=tensorcolor, thick] (d2) at (4.8,1.2) {};
    \draw[mid arrow=0.65] (a2.90) -- ++(0,0.6);
    \draw[mid arrow=0.65] ([yshift=-0.6cm]a1.270) -- ++(0,0.6);
    \draw[mid arrow=0.65] (b2.90) -- ++(0,0.6);
    \draw[mid arrow=0.65] ([yshift=-0.6cm]b1.270) -- ++(0,0.6);
    \draw[mid arrow=0.65] (c2.90) -- ++(0,0.6);
    \draw[mid arrow=0.65] ([yshift=-0.6cm]c1.270) -- ++(0,0.6);
    \draw[mid arrow=0.65] (d2.90) -- ++(0,0.6);
    \draw[mid arrow=0.65] ([yshift=-0.6cm]d1.270) -- ++(0,0.6);
    \draw[mid arrow=0.65] (a1.0) -- (b1.180);
    \draw[mid arrow=0.65] (b1.0) -- (c1.180);
    \draw[mid arrow=0.65] (c1.0) -- (d1.180);
    \draw[mid arrow=0.65] (b2.180) -- (a2.0);
    \draw[mid arrow=0.65] (c2.180) -- (b2.0);
    \draw[mid arrow=0.65] (d2.180) -- (c2.0);
    \draw[mid arrow=0.65] (a1.90) -- (a2.270);
    \draw[mid arrow=0.65] (b1.90) -- (b2.270);
    \draw[mid arrow=0.65] (c1.90) -- (c2.270);
    \draw[mid arrow=0.65] (d1.90) -- (d2.270);
    \draw (a1.270) to [out=90,in=180] (a1.0);
    \draw (b1.270) to [out=90, in=180] (b1.0);
    \draw (c1.270) to [out=90, in=180] (c1.0);
    \draw (d1.270) to [out=90, in=180] (d1.0);
    \draw (a1.180) to [out=0, in=270] (a1.90);
    \draw (b1.180) to [out=0, in=270] (b1.90);
    \draw (c1.180) to [out=0, in=270] (c1.90);
    \draw (d1.180) to [out=0, in=270] (d1.90);
    \draw[mid arrow=0.65] ([xshift=-0.6cm]a1.180) -- ++(0.6,0);
    \draw[mid arrow=0.65] (d1.0) -- ++(0.6,0);
    \draw[mid arrow=0.65] ([xshift=0.6cm]d2.0) -- ++(-0.6,0);
    \draw[mid arrow=0.65] (a2.180) -- ++(-0.6,0);
    \draw (d2.0) -- (d2.180);
    \draw (c2.0) -- (c2.180);
    \draw (b2.0) -- (b2.180);
    \draw (a2.0) -- (a2.180);
    \filldraw[black] ([xshift=-0.2cm]a2) circle (1pt) ; 
    \filldraw[black] ([xshift=-0.2cm]b2) circle (1pt) ; 
    \filldraw[black] ([xshift=-0.2cm]c2) circle (1pt) ; 
    \filldraw[black] ([xshift=-0.2cm]d2) circle (1pt) ; 
    \node[draw, circle, minimum size=0.3cm] (a3) at ([xshift=0.2cm]a2) {};
    \node[draw, circle, minimum size=0.3cm] (b3) at ([xshift=0.2cm]b2) {};
    \node[draw, circle, minimum size=0.3cm] (c3) at ([xshift=0.2cm]c2) {};
    \node[draw, circle, minimum size=0.3cm] (d3) at ([xshift=0.2cm]d2) {};
     \draw (a3.270) -- (a3.90);
    \draw (b3.270) -- (b3.90);
    \draw (c3.270) -- (c3.90);
    \draw (d3.270) -- (d3.90);
    \draw (a3.90) to [out=90, in=270] (a2.90);
    \draw (b3.90) to [out=90, in=270] (b2.90);
    \draw (c3.90) to [out=90, in=270] (c2.90);
    \draw (d3.90) to [out=90, in=270] (d2.90);
    \draw (a2.270) to [out=90, in=0 ]([xshift=-0.2cm]a2); 
    \draw (b2.270) to [out=90, in=0 ]([xshift=-0.2cm]b2); 
    \draw (c2.270) to [out=90, in=0 ]([xshift=-0.2cm]c2); 
    \draw (d2.270) to [out=90, in=0 ]([xshift=-0.2cm]d2); 
\end{tikzpicture}
\end{align}
The bottom layer of this unitary network is simply a shift, and its top layer is composed of horizontally linked CNOT gates. 

Next, we explore the definition of a stacked-CNOT circuit on a PBC system. Enforcing translational invariance makes the causal order of each CNOT gate ambiguous \cite{schumacher2004reversiblequantumcellularautomata}: it is not possible to determine which CNOT gate is performed first:
\begin{align}
\tikzsetnextfilename{stacked_CNOT_PBC}
\begin{tikzpicture}
    \filldraw[black] (-0.6,0) circle (2pt) ;  
    \filldraw[black] (0,0) circle (2pt) ;
    \filldraw[black] (0.6,0) circle (2pt);
    \filldraw[black] (1.2,0) circle (2pt);
    \draw (-0.6, -0.6) -- (-0.6, 1.2);
    \draw (0, -0.6) -- (0, 1.2);
    \draw (0.6, -0.6) -- (0.6, 1.2);
    \draw (1.2, -0.6) -- (1.2, 1.2);
    \node[draw, circle, minimum size=0.3cm] (a) at (-0.6,0.6) {};
    \node[draw, circle, minimum size=0.3cm] (b) at (0,0.6) {};
    \node[draw, circle, minimum size=0.3cm] (c) at (0.6,0.6) {};
    \node[draw, circle, minimum size=0.3cm] (d) at (1.2,0.6) {};
    \draw (a.180) -- (a.0);
    \draw (b.180) -- (b.0);
    \draw (c.180) -- (c.0);
    \draw (d.180) -- (d.0);
    \draw (-0.6,0) to [out=0,in=180] (b.180);
    \draw (0,0) to [out=0,in=180] (c.180);
    \draw (0.6,0) to [out=0,in=180] (d.180);
    \draw (1.2,0) -- ++(0.3,0) -- ++(0,-0.2) -- ++(-2.4,0) -- ++ (0,0.8) -- (a.180);
\end{tikzpicture}
\end{align}
As a result, the global operator is not unitary. This can be seen from the transformation of states:
\begin{equation}
\begin{aligned}
    &|0000 \rangle \overset{StC}{\longrightarrow} |0000\rangle, \\
    &|1111 \rangle \overset{StC}{\longrightarrow} |0000\rangle.
\end{aligned}
\end{equation}
On the other hand, breaking translational invariance and allowing alternative boundaries can yield a valid SQC applicable to both PBC and finite OBC systems:
\begin{align}
\tikzsetnextfilename{stacked_CNOT_PBC_SQC}
\begin{tikzpicture}
    \filldraw[black] (-0.6,0) circle (2pt) ;  
    \filldraw[black] (0,0) circle (2pt) ;
    \filldraw[black] (0.6,0) circle (2pt);
    \draw (-0.6, -0.6) -- (-0.6, 1.2);
    \draw (0, -0.6) -- (0, 1.2);
    \draw (0.6, -0.6) -- (0.6, 1.2);
    \draw (1.2, -0.6) -- (1.2, 1.2);
    \node[draw, circle, minimum size=0.3cm] (b) at (0,0.6) {};
    \node[draw, circle, minimum size=0.3cm] (c) at (0.6,0.6) {};
    \node[draw, circle, minimum size=0.3cm] (d) at (1.2,0.6) {};
    \draw (b.180) -- (b.0);
    \draw (c.180) -- (c.0);
    \draw (d.180) -- (d.0);
    \draw (-0.6,0) to [out=0,in=180] (b.180);
    \draw (0,0) to [out=0,in=180] (c.180);
    \draw (0.6,0) to [out=0,in=180] (d.180);
\end{tikzpicture}
\end{align}
This SQC corresponds to the unitary network as:
\begin{align}
\tikzsetnextfilename{stacked_CNOT_UN_SQC}
\begin{tikzpicture}
    \node[draw, rectangle, rounded corners=2pt, minimum width=1cm, minimum height=0.6cm, fill=tensorcolor, thick] (a1) at (0,0) {};
    \node[draw, rectangle, rounded corners=2pt, minimum width=1cm, minimum height=0.6cm, fill=tensorcolor, thick] (b1) at (1.6,0) {};
    \node[draw, rectangle, rounded corners=2pt, minimum width=1cm, minimum height=0.6cm, fill=tensorcolor, thick] (c1) at (3.2,0) {};
    \node[draw, rectangle, rounded corners=2pt, minimum width=1cm, minimum height=0.6cm, fill=tensorcolor, thick] (d1) at (4.8,0) {};
    \node[draw, rectangle, rounded corners=2pt, minimum width=1cm, minimum height=0.6cm, fill=tensorcolor, thick] (a2) at (0,1.2) {};
    \node[draw, rectangle, rounded corners=2pt, minimum width=1cm, minimum height=0.6cm, fill=tensorcolor, thick] (b2) at (1.6,1.2) {};
    \node[draw, rectangle, rounded corners=2pt, minimum width=1cm, minimum height=0.6cm, fill=tensorcolor, thick] (c2) at (3.2,1.2) {};
    \node[draw, rectangle, rounded corners=2pt, minimum width=1cm, minimum height=0.6cm, fill=tensorcolor, thick] (d2) at (4.8,1.2) {};
    \draw[mid arrow=0.65] (a2.90) -- ++(0,0.6);
    \draw[mid arrow=0.65] ([yshift=-0.6cm]a1.270) -- ++(0,0.6);
    \draw[mid arrow=0.65] (b2.90) -- ++(0,0.6);
    \draw[mid arrow=0.65] ([yshift=-0.6cm]b1.270) -- ++(0,0.6);
    \draw[mid arrow=0.65] (c2.90) -- ++(0,0.6);
    \draw[mid arrow=0.65] ([yshift=-0.6cm]c1.270) -- ++(0,0.6);
    \draw[mid arrow=0.65] (d2.90) -- ++(0,0.6);
    \draw[mid arrow=0.65] ([yshift=-0.6cm]d1.270) -- ++(0,0.6);
    \draw[mid arrow=0.65] (a1.0) -- (b1.180);
    \draw[mid arrow=0.65] (b1.0) -- (c1.180);
    \draw[mid arrow=0.65] (c1.0) -- (d1.180);
    \draw[mid arrow=0.65] (b2.180) -- (a2.0);
    \draw[mid arrow=0.65] (c2.180) -- (b2.0);
    \draw[mid arrow=0.65] (d2.180) -- (c2.0);
    \draw[mid arrow=0.65] (b1.90) -- (b2.270);
    \draw[mid arrow=0.65] (c1.90) -- (c2.270);
    \draw[mid arrow=0.65] (d1.90) -- (d2.270);
    \draw[mid arrow=0.65] ([xshift=0.2cm]d1.90) -- ([xshift=0.2cm]d2.270);
    \draw (a1.270) to [out=90,in=180] (a1.0);
    \draw (b1.270) to [out=90, in=180] (b1.0);
    \draw (c1.270) to [out=90, in=180] (c1.0);
    \draw (d1.270) to [out=90, in=270] ([xshift=0.2cm]d1.90);
    \draw (b1.180) to [out=0, in=270] (b1.90);
    \draw (c1.180) to [out=0, in=270] (c1.90);
    \draw (d1.180) to [out=0, in=270] (d1.90);
    \draw (c2.0) -- (c2.180);
    \draw (b2.0) -- (b2.180);
    \draw (a2.0) to [out=180, in=270] (a2.90);
    \filldraw[black] ([xshift=-0.2cm]b2) circle (1pt) ; 
    \filldraw[black] ([xshift=-0.2cm]c2) circle (1pt) ; 
    \filldraw[black] ([xshift=-0.2cm]d2) circle (1pt) ; 
    \node[draw, circle, minimum size=0.3cm] (b3) at ([xshift=0.2cm]b2) {};
    \node[draw, circle, minimum size=0.3cm] (c3) at ([xshift=0.2cm]c2) {};
    \node[draw, circle, minimum size=0.3cm] (d3) at ([xshift=0.2cm]d2) {};
    \draw (b3.270) -- (b3.90);
    \draw (c3.270) -- (c3.90);
    \draw (d3.270) -- (d3.90);
    \draw (b3.90) to [out=90, in=270] (b2.90);
    \draw (c3.90) to [out=90, in=270] (c2.90);
    \draw (d3.90) to [out=90, in=270] (d2.90);
    \draw (b2.270) to [out=90, in=0 ]([xshift=-0.2cm]b2); 
    \draw (c2.270) to [out=90, in=0 ]([xshift=-0.2cm]c2); 
    \draw (d2.270) to [out=90, in=0 ]([xshift=-0.2cm]d2); 
    \draw ([xshift=0.2cm]d2.270) -- (d3.270);
    \draw (d3.0) -- (d2.180);
\end{tikzpicture}
\end{align}
It can be viewed as either a finite OBC system or a PBC system with a cut. Note that its boundary behavior differs from the bulk.

\subsection{1D Kramers-Wannier transformation}
Kramers-Wannier (KW) transformation relates the paramagnetic and ferromagnetic phases of the transverse-field Ising chain. Consider the transverse field Ising model in 1D:
\begin{equation}
    H = -J \sum_{n} Z_{n}Z_{n+1} - B \sum_{n} X_{n},
\end{equation}
where $J>0, B>0$. For $B > J$, the ground state is in the symmetric phase, while for $B < J$, it is in the symmetry breaking phase. The 1D KW transformation is given by a map of Pauli operators \cite{aasen2016topological, cao2023subsystem}:
\begin{equation}
\label{eq: KW-transformation}
    X_{n} \rightarrow Z_{n}Z_{n+1}, \ Z_{n} \rightarrow \prod_{m \le n} X_{n},  \ Z_{n-1}Z_{n} \rightarrow X_{n}.
\end{equation}
Under KW transformation the function of $J$ and $B$ exchange. Therefore, KW transformation maps the ground state of the symmetric phase $|\Psi_{SP} \rangle$ to the ground state of the symmetry breaking phase $|\Psi_{SB} \rangle$ and vice versa. By comparing (\ref{eq: KW-transformation}) and (\ref{Eq:StC algebra}), we observe that the KW transformation can be achieved with a stacked-CNOT followed by a Hadamard gate $H$ on each qubit. Therefore, it can be represented by a bilayer unitary network: 
\begin{align}
\tikzsetnextfilename{KW-transformation}
\begin{tikzpicture}
    \node[draw, rectangle, rounded corners=2pt, minimum width=1cm, minimum height=0.6cm, fill=tensorcolor, thick] (a1) at (0,0) {};
    \node[draw, rectangle, rounded corners=2pt, minimum width=1cm, minimum height=0.6cm, fill=tensorcolor, thick] (b1) at (1.6,0) {};
    \node[draw, rectangle, rounded corners=2pt, minimum width=1cm, minimum height=0.6cm, fill=tensorcolor, thick] (c1) at (3.2,0) {};
    \node[draw, rectangle, rounded corners=2pt, minimum width=1cm, minimum height=0.6cm, fill=tensorcolor, thick] (d1) at (4.8,0) {};
    \node[draw, rectangle, rounded corners=2pt, minimum width=1cm, minimum height=0.6cm, fill=tensorcolor, thick] (a2) at (0,1.2) {};
    \node[draw, rectangle, rounded corners=2pt, minimum width=1cm, minimum height=0.6cm, fill=tensorcolor, thick] (b2) at (1.6,1.2) {};
    \node[draw, rectangle, rounded corners=2pt, minimum width=1cm, minimum height=0.6cm, fill=tensorcolor, thick] (c2) at (3.2,1.2) {};
    \node[draw, rectangle, rounded corners=2pt, minimum width=1cm, minimum height=0.6cm, fill=tensorcolor, thick] (d2) at (4.8,1.2) {};
    \node[draw, rectangle, rounded corners=2pt, minimum width=0.6cm, minimum height=0.6cm, fill=tensorcolor, thick] (a4) at (0,2.4) {$H$};
    \node[draw, rectangle, rounded corners=2pt, minimum width=0.6cm, minimum height=0.6cm, fill=tensorcolor, thick] (b4) at (1.6,2.4) {$H$};
    \node[draw, rectangle, rounded corners=2pt, minimum width=0.6cm, minimum height=0.6cm, fill=tensorcolor, thick] (c4) at (3.2,2.4) {$H$};
    \node[draw, rectangle, rounded corners=2pt, minimum width=0.6cm, minimum height=0.6cm, fill=tensorcolor, thick] (d4) at (4.8,2.4) {$H$};
    \draw[mid arrow=0.65] (a2.90) -- ++(0,0.6);
    \draw[mid arrow=0.65] ([yshift=-0.6cm]a1.270) -- ++(0,0.6);
    \draw[mid arrow=0.65] (b2.90) -- ++(0,0.6);
    \draw[mid arrow=0.65] ([yshift=-0.6cm]b1.270) -- ++(0,0.6);
    \draw[mid arrow=0.65] (c2.90) -- ++(0,0.6);
    \draw[mid arrow=0.65] ([yshift=-0.6cm]c1.270) -- ++(0,0.6);
    \draw[mid arrow=0.65] (d2.90) -- ++(0,0.6);
    \draw[mid arrow=0.65] ([yshift=-0.6cm]d1.270) -- ++(0,0.6);
    \draw[mid arrow=0.65] (a1.0) -- (b1.180);
    \draw[mid arrow=0.65] (b1.0) -- (c1.180);
    \draw[mid arrow=0.65] (c1.0) -- (d1.180);
    \draw[mid arrow=0.65] (b2.180) -- (a2.0);
    \draw[mid arrow=0.65] (c2.180) -- (b2.0);
    \draw[mid arrow=0.65] (d2.180) -- (c2.0);
    \draw[mid arrow=0.65] (a1.90) -- (a2.270);
    \draw[mid arrow=0.65] (b1.90) -- (b2.270);
    \draw[mid arrow=0.65] (c1.90) -- (c2.270);
    \draw[mid arrow=0.65] (d1.90) -- (d2.270);
    \draw[mid arrow=0.65] (a4.90) -- ++(0,0.6);
    \draw[mid arrow=0.65] (b4.90) -- ++(0,0.6);
    \draw[mid arrow=0.65] (c4.90) -- ++(0,0.6);
    \draw[mid arrow=0.65] (d4.90) -- ++(0,0.6);
    \draw (a1.270) to [out=90,in=180] (a1.0);
    \draw (b1.270) to [out=90, in=180] (b1.0);
    \draw (c1.270) to [out=90, in=180] (c1.0);
    \draw (d1.270) to [out=90, in=180] (d1.0);
    \draw (a1.180) to [out=0, in=270] (a1.90);
    \draw (b1.180) to [out=0, in=270] (b1.90);
    \draw (c1.180) to [out=0, in=270] (c1.90);
    \draw (d1.180) to [out=0, in=270] (d1.90);
    \draw[mid arrow=0.65] ([xshift=-0.6cm]a1.180) -- ++(0.6,0);
    \draw[mid arrow=0.65] (d1.0) -- ++(0.6,0);
    \draw[mid arrow=0.65] ([xshift=0.6cm]d2.0) -- ++(-0.6,0);
    \draw[mid arrow=0.65] (a2.180) -- ++(-0.6,0);
    \draw (d2.0) -- (d2.180);
    \draw (c2.0) -- (c2.180);
    \draw (b2.0) -- (b2.180);
    \draw (a2.0) -- (a2.180);
    \filldraw[black] ([xshift=-0.2cm]a2) circle (1pt) ; 
    \filldraw[black] ([xshift=-0.2cm]b2) circle (1pt) ; 
    \filldraw[black] ([xshift=-0.2cm]c2) circle (1pt) ; 
    \filldraw[black] ([xshift=-0.2cm]d2) circle (1pt) ; 
    \node[draw, circle, minimum size=0.3cm] (a3) at ([xshift=0.2cm]a2) {};
    \node[draw, circle, minimum size=0.3cm] (b3) at ([xshift=0.2cm]b2) {};
    \node[draw, circle, minimum size=0.3cm] (c3) at ([xshift=0.2cm]c2) {};
    \node[draw, circle, minimum size=0.3cm] (d3) at ([xshift=0.2cm]d2) {};
     \draw (a3.270) -- (a3.90);
    \draw (b3.270) -- (b3.90);
    \draw (c3.270) -- (c3.90);
    \draw (d3.270) -- (d3.90);
    \draw (a3.90) to [out=90, in=270] (a2.90);
    \draw (b3.90) to [out=90, in=270] (b2.90);
    \draw (c3.90) to [out=90, in=270] (c2.90);
    \draw (d3.90) to [out=90, in=270] (d2.90);
    \draw (a2.270) to [out=90, in=0 ]([xshift=-0.2cm]a2); 
    \draw (b2.270) to [out=90, in=0 ]([xshift=-0.2cm]b2); 
    \draw (c2.270) to [out=90, in=0 ]([xshift=-0.2cm]c2); 
    \draw (d2.270) to [out=90, in=0 ]([xshift=-0.2cm]d2); 
\end{tikzpicture}
\end{align}
where Hadamard gates $H$ may be integrated within the CNOT blocks.

The MPO implementations of the KW transformation are presented in \cite{bridgeman2017anomalies, Lootens2023, Tantivasadakarn2024}. An SQC realization of the 1D KW transformation was given in \cite{Chen_2024}.

\subsection{Approximately locality preserving unitaries}

Below we show that a unitary network can effectively represent ALPU, featuring exponentially decaying tails:
\begin{equation}
    f(r) \sim e^{-r/\xi},
\end{equation}
where $\xi$ is the characteristic length.

A XY entangling gate \cite{kempe2001encoded,kempe2001encoded,echternach2001universal, schuch2003natural, abrams2019implementation} on two qubits is defined as 
\begin{equation}
\begin{aligned}
    XY(\theta) &= \exp[i\frac{\theta}{2}(\frac{X_{1}\otimes X_{2}+Y_{1}\otimes Y_{2}}{2})] \\
    &= 
\left(
\begin{array}{cccc}
   1  & 0 & 0 & 0 \\
   0  & c & is & 0 \\
   0  & is & c & 0 \\
   0 & 0 & 0 & 1
\end{array}
\right),
\end{aligned}
\end{equation}
where $c=\cos (\frac{\theta}{2}), s=\sin(\frac{\theta}{2})$. $X_{1}$ and $Y_{1}$ are Pauli operators on the first qubit, $X_{2}$ and $Y_{2}$ are Pauli operators on the second qubit. The algebraic dynamics of qubits operated by the XY gate is described by:
\begin{equation}
\begin{aligned}
    X_{1} \overset{XY(\theta)}{\longrightarrow} &\cos{\theta} \cdot X_{1} + \sin{\theta} \cdot Z_{1}Y_{2}, \\
     Y_{1} \overset{XY(\theta)}{\longrightarrow} &\cos{\theta} \cdot Y_{1} - \sin{\theta} \cdot Z_{1}X_{2}, \\
     Z_{1} \overset{XY(\theta)}{\longrightarrow} &
     Z_{1} + (\cos2\theta-1)(\frac{Z_{1}-Z_{2}}{2})\\
     &+ \sin2\theta \cdot (\frac{Y_{1}X_{2}-X_{1}Y_{2}}{2})\\
     =&\cos^2\theta \cdot Z_{1} + \sin^{2}\theta \cdot Z_{2}\\
     &+ \sin \theta \cos \theta \cdot (Y_{1}X_{2}-X_{1}Y_{2}).\\
\end{aligned}
\end{equation}

\begin{definition} [Stacked-XY circuit]
    The stacked-XY circuit is defined by the limit of $StXY_{N}$, achieved through the repeated application of XY gates:
\begin{equation}
\begin{aligned}
StXY_{N}(\theta) &=  \prod_{n=N}^{-N} XY_{n} = XY_{-N}(\theta)\cdots  XY_{n}(\theta) \cdots  XY_{N}(\theta)\\
    StXY(\theta) &= \lim_{N \rightarrow +\infty} StXY_{N}(\theta)\\
    &= \cdots XY_{n-1}(\theta) \cdot XY_{n}(\theta) \cdot XY_{n+1}(\theta) \cdots
\end{aligned}
\end{equation} 
\end{definition}
The dynamic of algebra under the Stacked-XY circuit is complicated, but one can verify that when $\sin\theta<1$, a stacked-XY circuit $StXY(\theta)$ is an ALPU with an exponentially-decaying tail:
\begin{equation}
   f(r) \sim e^{-r/\xi},
   \quad \xi = \frac{a}{-\ln (\sqrt{1-\cos^{4} \theta})},
\end{equation}
where $a$ is the lattice constant. A stacked-XY circuit can be represented by a unitary network as below:
\begin{align}
\tikzsetnextfilename{stacked_XY_UN}
\begin{tikzpicture}
    \node[draw, rectangle, rounded corners=2pt, minimum width=1cm, minimum height=0.6cm, fill=tensorcolor, thick] (a1) at (0,0) {};
    \node[draw, rectangle, rounded corners=2pt, minimum width=1cm, minimum height=0.6cm, fill=tensorcolor, thick] (b1) at (1.6,0) {};
    \node[draw, rectangle, rounded corners=2pt, minimum width=1cm, minimum height=0.6cm, fill=tensorcolor, thick] (c1) at (3.2,0) {};
    \node[draw, rectangle, rounded corners=2pt, minimum width=1cm, minimum height=0.6cm, fill=tensorcolor, thick] (d1) at (4.8,0) {};
    \node[draw, rectangle, rounded corners=2pt, minimum width=1cm, minimum height=0.6cm, font=\scriptsize, fill=tensorcolor, thick] (a2) at (0,1.2) {$XY(\theta)$};
    \node[draw, rectangle, rounded corners=2pt, minimum width=1cm, minimum height=0.6cm, font=\scriptsize, fill=tensorcolor, thick] (b2) at (1.6,1.2) {$XY(\theta)$};
    \node[draw, rectangle, rounded corners=2pt, minimum width=1cm, minimum height=0.6cm, font=\scriptsize, fill=tensorcolor, thick] (c2) at (3.2,1.2) {$XY(\theta)$};
    \node[draw, rectangle, rounded corners=2pt, minimum width=1cm, minimum height=0.6cm, font=\scriptsize, fill=tensorcolor, thick] (d2) at (4.8,1.2) {$XY(\theta)$};
    \draw[mid arrow=0.65] (a2.90) -- ++(0,0.6);
    \draw[mid arrow=0.65] ([yshift=-0.6cm]a1.270) -- ++(0,0.6);
    \draw[mid arrow=0.65] (b2.90) -- ++(0,0.6);
    \draw[mid arrow=0.65] ([yshift=-0.6cm]b1.270) -- ++(0,0.6);
    \draw[mid arrow=0.65] (c2.90) -- ++(0,0.6);
    \draw[mid arrow=0.65] ([yshift=-0.6cm]c1.270) -- ++(0,0.6);
    \draw[mid arrow=0.65] (d2.90) -- ++(0,0.6);
    \draw[mid arrow=0.65] ([yshift=-0.6cm]d1.270) -- ++(0,0.6);
    \draw[mid arrow=0.65] (a1.0) -- (b1.180);
    \draw[mid arrow=0.65] (b1.0) -- (c1.180);
    \draw[mid arrow=0.65] (c1.0) -- (d1.180);
    \draw[mid arrow=0.65] (b2.180) -- (a2.0);
    \draw[mid arrow=0.65] (c2.180) -- (b2.0);
    \draw[mid arrow=0.65] (d2.180) -- (c2.0);
    \draw[mid arrow=0.65] (a1.90) -- (a2.270);
    \draw[mid arrow=0.65] (b1.90) -- (b2.270);
    \draw[mid arrow=0.65] (c1.90) -- (c2.270);
    \draw[mid arrow=0.65] (d1.90) -- (d2.270);
    \draw (a1.270) to [out=90,in=180] (a1.0);
    \draw (b1.270) to [out=90, in=180] (b1.0);
    \draw (c1.270) to [out=90, in=180] (c1.0);
    \draw (d1.270) to [out=90, in=180] (d1.0);
    \draw (a1.180) to [out=0, in=270] (a1.90);
    \draw (b1.180) to [out=0, in=270] (b1.90);
    \draw (c1.180) to [out=0, in=270] (c1.90);
    \draw (d1.180) to [out=0, in=270] (d1.90);
    \draw[mid arrow=0.65] ([xshift=-0.6cm]a1.180) -- ++(0.6,0);
    \draw[mid arrow=0.65] (d1.0) -- ++(0.6,0);
    \draw[mid arrow=0.65] ([xshift=0.6cm]d2.0) -- ++(-0.6,0);
    \draw[mid arrow=0.65] (a2.180) -- ++(-0.6,0);
\end{tikzpicture}
\end{align}

\section{Information flow and the GNVW index}
\label{Sec: Information flow and the GNVW index}

This section introduces the idea of the information flow. In the unitary network, each leg carries an information flow with a magnitude given by the logarithm of its Hilbert space dimension and follows the direction of the leg. In a local unitary tensor, information flows in and out equally, maintaining a conservation of information.

In 1D unitary networks, the net information flow through any surface remains constant. Moreover, the net entropy flow in locality-preserving unitary network representations naturally captures the property of physical operators, mirroring the GNVW index \cite{Gross_2012}.

\subsection{Information flow and information conservation} \label{subsec: Information flow}

The local unitary nature of a unitary network allows for the definition of an additional structure on its directed graph.
\begin{definition} [Flow network]
A flow network (or network flows) \cite{manber1989introduction} is a directed graph $G = (V, E)$ that includes specific vertices: sources $In =\{i_{1}, i_{2}, \cdots \}$ and sinks $Out =\{ o_{1}, o_2,\cdots \}$, along with a flow defined within $G$.

A flow in $G$ is a function $f: E \rightarrow \mathbb{R}$ satisfying flow conservation condition: for each $u \in V \backslash (In \cup Out)$, 
\begin{equation}
    \sum_{(u,v)\in E}f(u,v) = \sum_{(v,u) \in E} f(v,u).
\end{equation}
In certain contexts, each edge is assigned a nonnegative capacity $c(u, v)$, which is irrelevant for our discussion.
\end{definition}

\begin{proposition} [Unitary network is a flow network]
A unitary network $U_{Net}$ is a flow network, with graph $G[U_{Net}]$ and flow function $f: E \rightarrow \mathbb{R}$:
\begin{equation}
    f(e) =  \log_{d} \dim \mathcal{H}_{e}.
\end{equation}
\end{proposition}
\begin{proof}
    For a local unitary tensor $v \in V$, unitarity suggests that its outgoing Hilbert space $\mathcal{H}^{out}_{v}$ must be isomorphic to the incoming Hilbert space $\mathcal{H}^{in}_{v}$:  $\mathcal{H}^{out}_{v} \cong \mathcal{H}^{in}_{v}$. Suppose $v$ has multiple incoming and outgoing legs, each carries a Hilbert space. The Hilbert space $\mathcal{H}_{in}$ ($\mathcal{H}_{out}$) is decomposed into tensor product of the Hilbert spaces associated with different legs: 
\begin{equation}
    \begin{aligned}
    \mathcal{H}^{out}_{v} &= \bigotimes_{e = (v,w) \in E } \mathcal{H}_{e}, \\ 
    \mathcal{H}^{in}_{v} &= \bigotimes_{e=(w,v)\in E} \mathcal{H}_{e}.
    \end{aligned}
\end{equation}
The Hilbert space dimensions hold the following relation:
\begin{equation}
\begin{aligned}
     \prod_{e=(v,w)\in E} \dim \mathcal{H}_{e} &=\dim \mathcal{H}^{out}_{v} \\
     &= \dim \mathcal{H}^{in}_{v} \\
     &= \prod_{e=(w,v)\in E} \dim \mathcal{H}_{e}
\end{aligned}
\end{equation}
Taking the logarithms gives 
\begin{equation}
    \sum_{e=(v,w)\in E} \log_{d}\dim \mathcal{H}_{e} = \sum _{e=(w,v)\in E} \log_{d} \dim \mathcal{H}_{e}
\end{equation}
\end{proof}

The flow $f(e) = \log_{d} \dim \mathcal{H}_{e}$ mirrors the entanglement entropy $S(\rho)$ associated with a maximally mixed state within $\mathcal{H}_{e}$. However, this interpretation is unsatisfactory, particularly for unitary operations acting on pure states. In Ref.~\cite{Roncaglia_2019}, a coherent entropy $S_{c}$ was introduced to measure the coherent information carried by a quantum state, which is given by
\begin{equation}
    S_{c} (\rho) = \log_{d} \dim \mathcal{H} - S(\rho).
\end{equation}
In addition, a conservation law of quantum information (under unitary time-evolution) was discussed in Ref.~\cite{Roncaglia_2019}. To illustrate, the total information in a bipartite state always equals the sum of the coherent information of each part and the mutual information between the subsystems:
\begin{equation}
    S_{c}[\rho_{AB}] = S_{c}[\rho_{A}] +S_{c}[\rho_{B}] + I_{A:B},
\end{equation}
where $I_{A:B} = S[\rho_{A}] +S[\rho_{B}] - S[\rho_{AB}]$ is the mutual information between the two systems $A$ and $B$. This inspired us to identify $f(e) = \log_{d} \dim \mathcal{H}_{e}$ as the maximum coherent information allowed for the Hilbert space $\mathcal{H}_{e}$, denoted as $S_{c}[\mathcal{H}_{e}]$, achieved by a pure state on $\mathcal{H}_{e}$. Therefore, $f(e)$ will be referred to as the information flow.

\begin{definition} [Information flow]
\label{Def: Information flow}
For a unitary network $U_{Net}$ and its associated graph $G[U_{Net}] = (V,E)$, the information flow is defined as a flow function $f_{I}:E \rightarrow  \mathbb{R}$:
\begin{equation}
    f_{I}(e) \coloneq S_{c}[\mathcal{H}_{e}] = \log_{d} \dim \mathcal{H}_{e},
\end{equation}
which satisfies the conservation law on each local unitary tensor: for each $u \in V \backslash (In \cup Out)$, 
\begin{equation}
\label{Eq: Information flow conservation}
    \sum_{(u,v)\in E}f_{I}(u,v) = \sum_{(v,u) \in E} f_{I}(v,u).
\end{equation}
\end{definition}
As discussed in \cite{Roncaglia_2019}, the conservation of information is a manifestation of unitarity. Information flow then represents the propagation of information in unitary networks.

\subsection{Bilayer unitary networks and the conservation of information}
A bilayer unitary network is structured in space-time, with its legs oriented horizontally or vertically. Here, the horizontal direction denotes space, while the vertical direction indicates time. Recall that a vector in $D=d+1$-dimensional space-time is given by $j^{\mu} = (j^{0}, \vec{j}) = (\rho, \vec{j})$, where the temporal component of flux is identified as density. The conservation law in space-time is given by
\begin{equation} \label{Eq:conservation law}
    \partial_{t} \rho + \nabla \cdot \vec{j} = 0 .
\end{equation}

In a bilayer unitary network, the horizontal and vertical legs serve distinct functions. Vertical legs extend upward through time, keeping the information localized. The information flow of a vertical leg can be viewed as a charge of information $Q_{I} = S_{c}$:
\begin{equation}
    f_{I} (e^{\perp}) = \int_{V} \rho dV =Q_{I},
\end{equation}
In contrast, the horizontal legs connect to adjacent local tensors, demonstrating the transmission of information. Information flow in a horizontal leg is actually an information transfer:
\begin{equation}
     f_{I} (e^{\parallel})  = \int_{T} dt \int_{S} \vec{j} \cdot d\vec{S} = \Delta Q_{I},
\end{equation}
The conservation of information flow  (\ref{Eq: Information flow conservation}) serves as a discrete version of (\ref{Eq:conservation law}): each local unitary tensor represents a spacetime point with conserved information flow.

The information flow is essentially the redistribution of information. In typical bilayer unitary networks, the input and output Hilbert spaces at the same site are equivalent, so the information distribution remains unchanged, imposing an additional constraint on the information flow. This perspective on information flow helps in understanding the GNVW index of QCAs.

\subsection{Net information flow for 1D unitary network}
We define net information flow across a surface and demonstrate that it is an invariance of a 1D unitary network, unaffected by the vertical surface chosen. Subsequently, we will illustrate that the net information flow aligns with the GNVW index \cite{Gross_2012}.

\begin{definition} [Net information flow across a surface]
    Consider a unitary network $U_{Net}$ with the associated graph $G[U_{Net}] = (V,E)$. For an orientable surface $\Sigma$, the net information flow across this surface is determined by the difference between the outgoing information and the incoming information: \begin{equation}
        f_{I}[\Sigma] = \sum_{e \ \text{points out}} f_{I}(e) - \sum_{e \ \text{points in}} f_{I}(e)
    \end{equation}
    where `in' and `out' pertain to the surface $\Sigma$, and $f_{I}[e]$ represents the information flow of edge $e$, as defined in Def. \ref{Def: Information flow}.
\end{definition}

We analyze the net information flow in a (1+1)-dimensional bilayer unitary network:
 \begin{align} 
    \tikzsetnextfilename{vertical_cuts}
    \begin{tikzpicture}
        \draw[dashed] (-0.6, -1.2) -- (-0.6, 2.4) node[above] {$\Sigma_{1}$};
        \draw[dashed] (0.6, -1.2) -- (0.6, 2.4)  node[above] {$\Sigma_{2}$};
        \node[draw, rectangle, rounded corners=2pt, minimum width=0.6cm, minimum height=0.6cm, fill=tensorcolor, thick](A) at (-1.2, 0) {};
        \node[draw, rectangle, rounded corners=2pt, minimum width=0.6cm, minimum height=0.6cm, fill=tensorcolor, thick](B) at (0,0) {};
        \node[draw, rectangle, rounded corners=2pt, minimum width=0.6cm, minimum height=0.6cm, fill=tensorcolor, thick](C) at (1.2, 0) {};
        \node[draw, rectangle, rounded corners=2pt, minimum width=0.6cm, minimum height=0.6cm, fill=tensorcolor, thick](D) at (-1.2, 1.2) {};
        \node[draw, rectangle, rounded corners=2pt, minimum width=0.6cm, minimum height=0.6cm, fill=tensorcolor, thick](E) at (0, 1.2) {};
        \node[draw, rectangle, rounded corners=2pt, minimum width=0.6cm, minimum height=0.6cm, fill=tensorcolor, thick](F) at (1.2, 1.2) {};
        \draw[mid arrow=0.65] (D.90) -- ++(0,0.6);
        \draw[mid arrow=0.65] (E.90) -- ++(0,0.6);
        \draw[mid arrow=0.65] (F.90) -- ++(0,0.6);
        \draw[mid arrow=0.65] ([yshift=-0.6cm]A.270) -- (A.270);
        \draw[mid arrow=0.65] ([yshift=-0.6cm]B.270) -- (B.270);
        \draw[mid arrow=0.65] ([yshift=-0.6cm]C.270) -- (C.270);
        \draw[mid arrow=0.65] ([xshift=-0.6cm]A.180) -- (A.180);
        \draw[mid arrow=0.65] (A.0) -- (B.180);
        \draw[mid arrow=0.65] (B.0) -- (C.180);
        \draw[mid arrow=0.65] (C.0) -- ++(0.6,0);
        \draw[mid arrow=0.65] ([xshift=0.6cm]F.0) -- (F.0);
        \draw[mid arrow=0.65] (F.180) -- (E.0);
        \draw[mid arrow=0.65] (E.180) -- (D.0);
        \draw[mid arrow=0.65] (D.180) -- ([xshift=-0.6cm]D.180);
        \draw[mid arrow=0.65] (A.90) -- (D.270);
        \draw[mid arrow=0.65] (B.90) -- (E.270);
        \draw[mid arrow=0.65] (C.90) -- (F.270);
    \end{tikzpicture}
\end{align}
 A vertical surface, $\Sigma_{1}$ or $\Sigma_{2}$, represents a spatial point that extends along the time axis. For any vertical surface, only two legs cross it: one points to the right, labeled $e_{R}$, and the other points to the left, labeled $e_{L}$. The net information flow of surface $\Sigma$ can be expressed as:
\begin{equation}
            f_{I}[\Sigma] = f_{I}(e_{R}) - f_{I}(e_{L}).
        \end{equation}
 In fact, this net information flow is independent of the chosen vertical surface.
\begin{properties} [Net information flow does not depend on the chosen vertical surface] 
\label{prop: uniform net information flow}
In a (1+1) dimensional bilayer unitary network $U_{Net}$, the net information flow $f_{I}[\Sigma]$ does not depend on the location of the vertical surface. Therefore, for any two vertical surfaces $\Sigma_{1}$ and $\Sigma_{2}$,
\begin{equation}
    f_{I}[\Sigma_{1}] = f_{I}[\Sigma_{2}].
\end{equation}
\end{properties}
\begin{proof}
Applying the conservation law of information for the subsystem enclosed by $\Sigma_{1}$ and $\Sigma_{2}$:
\begin{equation}
\begin{aligned}
    f_{I}[\Sigma_{1}] - f_{I}[\Sigma_{2}] &= 
    \Delta S_{c} \\
    &=  S_{c}[\mathcal{H}_{phys,out}] -S_{c}[\mathcal{H}_{phys,in}] \\
    &= 0
\end{aligned}
\end{equation}
From the second line to the third line, we implicitly assume that the incoming and outgoing physical Hilbert spaces for the same sites are isomorphic. This assumption generally holds; however, for synthetic situations where it does not, see Appendix \ref{Append: Non-uniform net information flow in 1D system}.

Simply put, the global unitary operator does not alter the distribution of the information (coherent entropy) $S_{c}$ on each site. Therefore, according to the conservation law, the net flux through any closed surface must be zero. In a one-dimensional space, a closed surface consists of two points, corresponding to two vertical surfaces within a (1+1)-dimensional unitary network, thereby enforcing a uniform flow.
\end{proof}

Using the property above, net information flow for a 1D unitary network can be consistently defined.
\begin{definition} [Net information flow for unitary network]
The net information flow for a unitary network $U_{Net}$ is defined as
\begin{equation}
    f_{I}[U_{Net}] \coloneq f_{I}[\Sigma],
\end{equation}
where $\Sigma$ is any vertical surface.
\end{definition}

\subsection{Concatenation of unitary networks}

In \cite{Gross_2012}, GNVW index is used as a classification of 1D QCAs: An equivalence class for QCAs is given by $U_{1} \sim U_{2}$ if there exists a crossover QCA $U_{12}$ that behaves as $U_{1}$ for $x \le a$ and as $U_{2}$ for $x \ge b$. We demonstrate that the net information flow exhibits a similar property. 
\begin{proposition} [Concatenation of two unitary networks]
\label{prop: Concatenation of two unitary networks}
Two unitary networks $U^{1}_{Net}$ and $U^{2}_{Net}$ (not necessarily locality-preserving) can be concatenated and form a crossover unitary network  $U_{Net_{12}}$ if and only if they have identical net information flow,
\begin{equation}
    f_{I}[U^{1}_{Net}] = f_{I}[U^{2}_{Net}].
\end{equation}
\end{proposition}

\begin{proof}

For two unitary networks $U^{1}_{Net}$ and $U^{2}_{Net}$ (not necessarily locality-preserving) with identical net information flow, the crossover unitary network $U_{Net_{12}}$ can be constructed as follows:
\begin{equation}
\label{eq: crossver of UN}
\tikzsetnextfilename{crossover_of_UN}
    \begin{tikzpicture}
         \node[draw, rectangle, rounded corners=2pt, minimum width=0.4cm, minimum height=0.4cm, fill=tensorcolor, thick] (a1) at (-2.4,0) {};
        \node[draw, rectangle, rounded corners=2pt, minimum width=0.4cm, minimum height=0.4cm, fill=tensorcolor, thick] (b1) at (-1.6,0) {};
        \node[draw, rectangle, rounded corners=2pt, minimum width=0.4cm, minimum height=0.4cm, fill=tensorcolor, thick] (c1) at (-0.8,0) {};
        \node[draw, rectangle, rounded corners=2pt, minimum width=0.4cm, minimum height=0.4cm, fill=btensorcolor, thick] (d1) at (0,0) {};
        \node[draw, rectangle, rounded corners=2pt, minimum width=0.4cm, minimum height=0.4cm, fill=tensorcolor, thick] (e1) at (0.8,0) {};
        \node[draw, rectangle, rounded corners=2pt, minimum width=0.4cm, minimum height=0.4cm, fill=tensorcolor, thick] (f1) at (1.6,0) {};
        \node[draw, rectangle, rounded corners=2pt, minimum width=0.4cm, minimum height=0.4cm, fill=tensorcolor, thick] (g1) at (2.4,0) {};
        \node[draw, rectangle, rounded corners=2pt, minimum width=0.4cm, minimum height=0.4cm, fill=tensorcolor, thick] (a2) at (-2.4,0.8) {};
        \node[draw, rectangle, rounded corners=2pt, minimum width=0.4cm, minimum height=0.4cm, fill=tensorcolor, thick] (b2) at (-1.6,0.8) {};
        \node[draw, rectangle, rounded corners=2pt, minimum width=0.4cm, minimum height=0.4cm, fill=tensorcolor, thick] (c2) at (-0.8,0.8) {};
        \node[draw, rectangle, rounded corners=2pt, minimum width=0.4cm, minimum height=0.4cm, fill=btensorcolor, thick] (d2) at (0,0.8) {};
        \node[draw, rectangle, rounded corners=2pt, minimum width=0.4cm, minimum height=0.4cm, fill=tensorcolor, thick] (e2) at (0.8,0.8) {};
        \node[draw, rectangle, rounded corners=2pt, minimum width=0.4cm, minimum height=0.4cm, fill=tensorcolor, thick] (f2) at (1.6,0.8) {};
        \node[draw, rectangle, rounded corners=2pt, minimum width=0.4cm, minimum height=0.4cm, fill=tensorcolor, thick] (g2) at (2.4,0.8) {};
        \draw[mid arrow = 0.65] (a1.90) -- (a2.270); 
        \draw[mid arrow = 0.65] (b1.90) -- (b2.270); 
        \draw[mid arrow = 0.65] (c1.90) -- (c2.270); 
        \draw[mid arrow = 0.65] (d1.90) -- (d2.270); 
        \draw[mid arrow = 0.65] (e1.90) -- (e2.270); 
        \draw[mid arrow = 0.65] (f1.90) -- (f2.270); 
        \draw[mid arrow = 0.65] (g1.90) -- (g2.270); 
        \draw[mid arrow = 0.65] ([xshift=-0.4cm]a1.180) -- (a1.180); 
        \draw[mid arrow = 0.65] (a1.0) -- (b1.180); 
        \draw[mid arrow = 0.65] (b1.0) -- (c1.180); 
        \draw[mid arrow = 0.65] (c1.0) -- (d1.180); 
        \draw[mid arrow = 0.65] (e1.180) -- (d1.0); 
        \draw[mid arrow = 0.65] (f1.180) -- (e1.0); 
        \draw[mid arrow = 0.65] (g1.180) -- (f1.0); 
        \draw[mid arrow = 0.65] ([xshift=0.4cm]g1.0) -- (g1.0); 
        \draw[mid arrow = 0.65] (a2.180) -- ++(-0.4,0); 
        \draw[mid arrow = 0.65] (b2.180) -- (a2.0); 
        \draw[mid arrow = 0.65] (b2.180) -- (a2.0); 
        \draw[mid arrow = 0.65] (c2.180) -- (b2.0);
        \draw[mid arrow = 0.65] (d2.180) -- (c2.0);
        \draw[mid arrow = 0.65] (d2.0) -- (e2.180);
        \draw[mid arrow = 0.65] (e2.0) -- (f2.180);
        \draw[mid arrow = 0.65] (f2.0) -- (g2.180);
        \draw[mid arrow = 0.65] (g2.0) -- ++(0.4,0);
        \draw[dashed] (-0.4,-0.8) -- (-0.4,1.6);
        \draw[dashed] (0.4,-0.8) -- (0.4,1.6);
        \draw [mid arrow=0.65] ([yshift=-0.4cm]a1.270) -- (a1.270);
        \draw [mid arrow=0.65] ([yshift=-0.4cm]b1.270) -- (b1.270);
        \draw [mid arrow=0.65] ([yshift=-0.4cm]c1.270) -- (c1.270);
        \draw [mid arrow=0.65] ([yshift=-0.4cm]d1.270) -- (d1.270);
        \draw [mid arrow=0.65] ([yshift=-0.4cm]e1.270) -- (e1.270);
        \draw [mid arrow=0.65] ([yshift=-0.4cm]f1.270) -- (f1.270);
        \draw [mid arrow=0.65] ([yshift=-0.4cm]g1.270) -- (g1.270);
        \draw [mid arrow=0.65] (a2.90) -- ++(0,0.4);
        \draw [mid arrow=0.65] (b2.90) -- ++(0,0.4);
        \draw [mid arrow=0.65] (c2.90) -- ++(0,0.4);
        \draw [mid arrow=0.65] (d2.90) -- ++(0,0.4);
        \draw [mid arrow=0.65] (e2.90) -- ++(0,0.4);
        \draw [mid arrow=0.65] (f2.90) -- ++(0,0.4);
        \draw [mid arrow=0.65] (g2.90) -- ++(0,0.4);
         \node at (-1.6,-1) {Truncated $U_{1}$};
        \node at (1.6,-1) {Truncated $U_{2}$};
    \end{tikzpicture}
\end{equation}
In general, if one side has bond dimensions greater than that of the other, the blue-colored middle area must be large enough to absorb the additional information flow. By adopting the middle canonical forms as in \eqref{eq: crossver of UN}, the central section can be any local unitary tensors.

Conversely, if $f_{I}[U^{1}_{Net}] \ne f_{I}[U^{2}_{Net}]$, information flow conservation in the middle section forbids the concatenation.
\end{proof}

\subsection{Intrinsic net information flow}
\label{subsec: Intrinsic net information flow}
In general, the same physical state or operator can admit multiple tensor-network representations. In the above, we have defined a flow index for a unitary network. Here, we discuss under what condition does the flow index becomes independent of unitary network implementation and therefore can be viewed as an intrinsic property of the unitary operator itself.

As an example, consider the following infinite unitary network: 

\begin{equation}
\label{eq: redundant flow}
    \begin{aligned}
        \begin{array}{c}
        \tikzsetnextfilename{UN_bond_decoupled}
        \begin{tikzpicture}
              \node[draw, rectangle, rounded corners=2pt, minimum width=0.6cm, minimum height=0.6cm, fill=tensorcolor, thick](A1) at (-1.2, 0) {$A$};
        \node[draw, rectangle, rounded corners=2pt, minimum width=0.6cm, minimum height=0.6cm, fill=tensorcolor, thick](B1) at (0,0) {$A$};
        \node[draw, rectangle, rounded corners=2pt, minimum width=0.6cm, minimum height=0.6cm, fill=tensorcolor, thick](C1) at (1.2, 0) {$A$};
        \node[draw, rectangle, rounded corners=2pt, minimum width=0.6cm, minimum height=0.6cm, fill=tensorcolor, thick](A2) at (-1.2, 1.2) {};
        \node[draw, rectangle, rounded corners=2pt, minimum width=0.6cm, minimum height=0.6cm, fill=tensorcolor, thick](B2) at (0, 1.2) {};
        \node[draw, rectangle, rounded corners=2pt, minimum width=0.6cm, minimum height=0.6cm, fill=tensorcolor, thick](C2) at (1.2, 1.2) {};
        \draw[mid arrow=0.65] (A2.90) -- ++(0,0.6);
        \draw[mid arrow=0.65] (B2.90) -- ++(0,0.6);
        \draw[mid arrow=0.65] (C2.90) -- ++(0,0.6);
        \draw[mid arrow=0.65] ([yshift=-0.6cm]A1.270) -- (A1.270);
        \draw[mid arrow=0.65] ([yshift=-0.6cm]B1.270) -- (B1.270);
        \draw[mid arrow=0.65] ([yshift=-0.6cm]C1.270) -- (C1.270);
        \draw[mid arrow=0.65] ([xshift=-0.6cm]A1.180) -- (A1.180);
        \draw[mid arrow=0.65] (A1.0) -- (B1.180);
        \draw[mid arrow=0.65] (B1.0) -- (C1.180);
        \draw[mid arrow=0.65] (C1.0) -- ++(0.6,0);
        \draw[mid arrow=0.65, red!70] ([xshift=0.6cm]C2.0) -- (C2.0);
        \draw[mid arrow=0.65, red!70] (C2.180) -- (B2.0);
        \draw[mid arrow=0.65,red!70] (B2.180) -- (A2.0);
        \draw[mid arrow=0.65,red!70] (A2.180) -- ++(-0.6,0);
        \draw[mid arrow=0.65] (A1.90) -- (A2.270);
        \draw[mid arrow=0.65] (B1.90) -- (B2.270);
        \draw[mid arrow=0.65] (C1.90) -- (C2.270);
        \draw[red!70] (A2.0) -- (A2.180);
        \draw[red!70] (B2.0) -- (B2.180);
        \draw[red!70] (C2.0) -- (C2.180);
        \draw (A2.270) -- (A2.90);
        \draw (B2.270) -- (B2.90);
        \draw (C2.270) -- (C2.90);
        \end{tikzpicture}
        \end{array}
    \end{aligned}
\end{equation}
In so far as the flow index is concerned, the red arrow contributes a leftward flow which compensates the rightward flow in the bottom row. However, by design the red arrow is completely decoupled from the physical legs, and therefore, independent of the Hilbert space dimension attached to the red arrow, the unitary network leads to the same physical unitary operator. In these situations, we say the red arrow contributes to redundant net information flow. 

As will be detailed in Appendix \ref{Append: Redundant net information flow of finite PBC unitary networks}, a finite PBC unitary network representation shows redundant net information flow when the information loops unnecessarily around the system. Moreover, an infinite OBC unitary network with redundant net information flow arises from the thermodynamic limit of a finite PBC unitary network exhibiting the same feature.

These discussions highlight the fact that the flow index we have defined is a property of the unitary network, and may not be an intrinsic index attached to the unitary operator it represents.

We now argue that, when we endow the unitary network with the locality-preserving condition, then the flow index becomes intrinsic. As locality-preserving unitary networks give rise to QCAs, the assertion can be phrased in terms of the following proposition:

\begin{proposition}
\label{prop: same net information flow for LU-rep}
Any two locality-preserving implementations  $U_{Net}^{1}$ and $U_{Net}^{2}$ of the same QCA $U$ must share the same net information flow:
\begin{equation}
    f_{I}[U_{Net}^{1}] =  f_{I}[U_{Net}^{2}]
\end{equation}
\end{proposition}

\begin{proof}

We aim to show that any two locality-preserving unitary networks $U_{Net}^{1}$ and $U_{Net}^{2}$ implementation for the same QCA $U$ can be connected by a local unitary tensor $U_{M}$: 
\begin{equation}
\label{eq: concatenate UN reps}
\tikzsetnextfilename{concatenate_locality_preserving_implementations}
\begin{tikzpicture}
     \node[draw, rectangle, rounded corners=2pt, minimum width=1.2cm, minimum height=0.8cm, fill=tensorcolor, thick](A) at (-1.8, 0) {$U_{Net}^{1, L}$};
     \node[draw, rectangle, rounded corners=2pt, minimum width=1.2cm, minimum height=0.8cm, fill=tensorcolor, thick](M) at (0, 0) {$U_{M}$};
     \node[draw, rectangle, rounded corners=2pt, minimum width=1.2cm, minimum height=0.8cm, fill=tensorcolor, thick](B) at (1.8, 0) {$U_{Net}^{2, R}$};
        \draw[mid arrow=0.65] ([xshift=-0.4cm,yshift=-0.6cm]A.270) -- ([xshift=-0.4cm]A.270);
        \draw[mid arrow=0.65] ([yshift=-0.6cm]A.270) -- (A.270);
        \draw[mid arrow=0.65] ([xshift=0.4cm,yshift=-0.6cm]A.270) -- ([xshift=0.4cm]A.270);
        \draw[mid arrow=0.65] (A.90) -- ++(0,0.6);
        \draw[mid arrow=0.65] ([xshift=-0.4cm]A.90) -- ++(0,0.6);
        \draw[mid arrow=0.65] ([xshift=0.4cm]A.90) -- ++(0,0.6);

        \draw[mid arrow=0.65] ([xshift=-0.4cm,yshift=-0.6cm]M.270) -- ([xshift=-0.4cm]M.270);
        \draw[mid arrow=0.65] ([yshift=-0.6cm]M.270) -- (M.270);
        \draw[mid arrow=0.65] ([xshift=0.4cm,yshift=-0.6cm]M.270) -- ([xshift=0.4cm]M.270);
        \draw[mid arrow=0.65] (M.90) -- ++(0,0.6);
        \draw[mid arrow=0.65] ([xshift=-0.4cm]M.90) -- ++(0,0.6);
        \draw[mid arrow=0.65] ([xshift=0.4cm]M.90) -- ++(0,0.6);

        \draw[mid arrow=0.65] ([xshift=-0.4cm,yshift=-0.6cm]B.270) -- ([xshift=-0.4cm]B.270);
        \draw[mid arrow=0.65] ([yshift=-0.6cm]B.270) -- (B.270);
        \draw[mid arrow=0.65] ([xshift=0.4cm,yshift=-0.6cm]B.270) -- ([xshift=0.4cm]B.270);
        \draw[mid arrow=0.65] (B.90) -- ++(0,0.6);
        \draw[mid arrow=0.65] ([xshift=-0.4cm]B.90) -- ++(0,0.6);
        \draw[mid arrow=0.65] ([xshift=0.4cm]B.90) -- ++(0,0.6);

        \draw[mid arrow=0.65] 
        ([yshift=-0.2cm]A.0) -- node[below]{$L, in$}([yshift=-0.2cm]M.180) ;
        \draw[mid arrow=0.65] ([yshift=0.2cm]M.180) -- node[above]{$L, out$}([yshift=0.2cm]A.0);

         \draw[mid arrow=0.65] 
        ([yshift=0.2cm]M.0) -- node[above]{$R, out$}([yshift=0.2cm]B.180);
        \draw[mid arrow=0.65] ([yshift=-0.2cm]B.180) -- node[below]{$R, in$}([yshift=-0.2cm]M.0);
        
        \draw[mid arrow=0.65] ([xshift=-0.6cm,yshift=-0.2cm]A.180) -- ++(0.6,0);
         \draw[mid arrow=0.65] ([yshift=0.2cm]A.180) -- ++(-0.6,0);
         \draw[mid arrow=0.65] ([yshift=0.2cm]B.0) -- ++(0.6,0);
         \draw[mid arrow=0.65] ([xshift=0.6cm, yshift=-0.2cm]B.0) -- ++(-0.6,0);
\end{tikzpicture}
\end{equation}
Then according to Proposition \ref{prop: Concatenation of two unitary networks}, $U_{Net}^{1}$ and $U_{Net}^{2}$ must share the same net information flow.

Assume unitary network representations $U_{Net}^{1}$ and $U_{Net}^{2}$ are locality-preserving with respective radii $R_{1}$ and $R_{2}$. We require $U_{M}$ to be supported on a middle region $X$ with $\text{diam}(X) > R_{1} + R_{2}$. Before concatenating $U_{Net}^{1}$ with $U_{Net}^{2}$, we cut $U_{Net}^{1}$ into $U_{Net}^{1, L}$ and $U_{Net}^{1, R}$,
\begin{equation}
\tikzsetnextfilename{CUT_UN_implementation}
\begin{tikzpicture}
     \node[draw, rectangle, rounded corners=2pt, minimum width=1.2cm, minimum height=0.8cm, fill=tensorcolor, thick](A) at (-1.8, 0) {$U_{Net}^{1, L}$};
     \node[draw, rectangle, rounded corners=2pt, minimum width=1.2cm, minimum height=0.8cm, fill=tensorcolor, thick](B) at (0, 0) {$U_{Net}^{1, R}$};
        \draw[mid arrow=0.65] ([xshift=-0.4cm,yshift=-0.6cm]A.270) -- ([xshift=-0.4cm]A.270);
        \draw[mid arrow=0.65] ([yshift=-0.6cm]A.270) -- (A.270);
        \draw[mid arrow=0.65] ([xshift=0.4cm,yshift=-0.6cm]A.270) -- ([xshift=0.4cm]A.270);
        \draw[mid arrow=0.65] (A.90) -- ++(0,0.6);
        \draw[mid arrow=0.65] ([xshift=-0.4cm]A.90) -- ++(0,0.6);
        \draw[mid arrow=0.65] ([xshift=0.4cm]A.90) -- ++(0,0.6);

        \draw[mid arrow=0.65] ([xshift=-0.4cm,yshift=-0.6cm]B.270) -- ([xshift=-0.4cm]B.270);
        \draw[mid arrow=0.65] ([yshift=-0.6cm]B.270) -- (B.270);
        \draw[mid arrow=0.65] ([xshift=0.4cm,yshift=-0.6cm]B.270) -- ([xshift=0.4cm]B.270);
        \draw[mid arrow=0.65] (B.90) -- ++(0,0.6);
        \draw[mid arrow=0.65] ([xshift=-0.4cm]B.90) -- ++(0,0.6);
        \draw[mid arrow=0.65] ([xshift=0.4cm]B.90) -- ++(0,0.6);
        
        \draw[mid arrow=0.65] 
        ([yshift=-0.2cm]A.0) -- node[below]{$L, in$}([yshift=-0.2cm]B.180) ;
        \draw[mid arrow=0.65] ([yshift=0.2cm]B.180) -- node[above]{$L, out$}([yshift=0.2cm]A.0);

        \draw[mid arrow=0.65] ([xshift=-0.6cm,yshift=-0.2cm]A.180) -- ++(0.6,0);
         \draw[mid arrow=0.65] ([yshift=0.2cm]A.180) -- ++(-0.6,0);
         \draw[mid arrow=0.65] ([yshift=0.2cm]B.0) -- ++(0.6,0);
         \draw[mid arrow=0.65] ([xshift=0.6cm, yshift=-0.2cm]B.0) -- ++(-0.6,0);
\end{tikzpicture}
\end{equation}
We also split $U_{Net}^{2}$ into $U_{Net}^{2, L}$ and $U_{Net}^{2, R}$. We aim to join $U_{Net}^{1, L}$ with $U_{Net}^{2, R}$ by a middle $U_{M}$. In the construction of $U_{M}$, we actually need information of $U_{Net}^{1, R}$ and $U_{Net}^{2, L}$.

We explicitly construct $U_{M}$ now. If $U_{M}$ exists, it will feature horizontal bond legs in addition to physical legs. We merge horizontal bond legs into four legs $(L,in), (L,out), (R,in), (R,out)$ as depicted in the \eqref{eq: concatenate UN reps}. Moreover, we label its input and output physical Hilbert spaces as $(phy,in)$ and $(phy,out)$. Our approach avoids specifying the causal order of the horizontal legs, allowing for more flexible unitary network architectures.
Define subalgebras of $\mathcal{A}_{phy,in}$ by supporting algebra:
\begin{equation}
\begin{aligned}
   & \mathcal{B}^{\prime}_{L} = \boldsymbol{S}((u_{Net}^{1,R})^{-1}(\mathcal{A}_{L,out}), \mathcal{A}_{phy, in}), \\
   & \mathcal{B}^{\prime}_{R} = \boldsymbol{S}((u_{Net}^{2,L})^{-1}(\mathcal{A}_{R,out}), \mathcal{A}_{phy, in}).
\end{aligned} 
\end{equation}
We choose minimum matrix algebras ${B}_{L} \cong M_{n_{L}}$ and ${B}_{R} \cong M_{n_{R}}$, such that $\mathcal{B}^{\prime}_{L} \subseteq {B}_{L}$ and $\mathcal{B}^{\prime}_{R} \subseteq {B}_{R}$. Since the middle region has a size $L > R_{1} + R_{2}$, the locality-preserving property of $u_{Net}^{1,R}$  and $u_{Net}^{2,L}$ guarantees $[\mathcal{B}_{L}, \mathcal{B}_{R}] = 0 $. Therefore, we can decompose $\mathcal{A}_{phy,in}$ as 
\begin{equation}
    \mathcal{A}_{phy,in} = \mathcal{B}_{L} \otimes \mathcal{B}_{M} \otimes \mathcal{B}_{R},
\end{equation}
We define the action of $U_{M}$ as below:
\begin{equation}
\label{eq: def UM}
\begin{aligned}
    & u_{M}(\mathcal{A}_{L,in}) \coloneq u_{Net}^{1,R}(\mathcal{A}_{L,in}) \subseteq \mathcal{A}_{phy,out} \otimes \mathcal{A}_{L,out}, \\
    & u_{M}(\mathcal{B}_{L}) \coloneq u_{Net}^{1,R}(\mathcal{B}_{L}) \subseteq \mathcal{A}_{phy,out}\otimes \mathcal{A}_{L,out}, \\
    & u_{M}(\mathcal{A}_{R,in}) \coloneq u_{Net}^{2,L}(\mathcal{A}_{R,in}) \subseteq \mathcal{A}_{phy,out} \otimes \mathcal{A}_{R,out}, \\
    & u_{M}(\mathcal{B}_{R}) \coloneq u_{Net}^{2,L}(\mathcal{B}_{R}) \subseteq \mathcal{A}_{phy,out}\otimes \mathcal{A}_{R,out}, \\
    & u_{M}(\mathcal{B}_{M}) \coloneq u(\mathcal{B}_{M}) \subseteq \mathcal{A}_{phy,out}.
\end{aligned}
\end{equation}
Since $U_{Net}^{1}$ and $U_{Net}^{2}$ both implement the same QCA $U$, we have 
\begin{equation}
    u_{M}(\mathcal{B}_{M})  \coloneq u(\mathcal{B}_{M}) =  u_{Net}^{1,R}(\mathcal{B}_{M}) = u_{Net}^{2,L}(\mathcal{B}_{M}).
\end{equation} 
The mutual commutation of result algebras are therefore assured.  The physical QCA $U$ also ensures the surjection. Therefore, we find the unitary tensor $U_{M}$ connecting $U_{Net}^{1}$ and $U_{Net}^{2}$. $U_{Net}^{1}$ and $U_{Net}^{2}$ must share the same net information flow.
\end{proof}

We here argue that the inherent characteristic exhibited by the net information flow within a locality-preserving unitary network representation precisely corresponds to the "net flow of information" discussed in ref.~\cite{Gross_2012, Kitaev_2006}. We illustrate this by proving that the following identity.

\begin{proposition} [GNVW index is net information flow]
For a one-dimensional QCA $u$ represented by a locality preserving unitary network $U_{Net}$, the net information flow $f_{I}[U_{Net}]$ equals the logarithm of the GNVW index:
\begin{equation}
    \log_{d} I_{GNVW}(u) = -f_{I}[U_{Net}] .
\end{equation}
\end{proposition}
\begin{proof}
    Consider the locality-preserving unitary network representation of a QCA given in (\ref{fig:QCA Bilayer}). We concentrate on the local unitary tensor $L_{m}$ from (\ref{fig:Wm}):
    \begin{align}
    \tikzsetnextfilename{info_flow_GNVW}
    \begin{tikzpicture}
        \draw[dashed] (1.2, -1.6) -- (1.2, 1.6) node[right]{$\Sigma$};
        \node[draw, rectangle, rounded corners=2pt, minimum width=0.6cm, minimum height=1.2cm, fill=tensorcolor, thick] (Lm)   {$L_{m}$};
        \draw[mid arrow=0.65] (Lm.90) -- node[left]{$\mathcal{H}_{\mathcal{B}_{2m}}$} ++(0,0.6);
         \draw[mid arrow=0.65] ([yshift=-0.6cm]Lm.270) -- node[left]{$\mathcal{H}_{\mathcal{A}_{2m}}$} ++(0,0.6);
         \draw[mid arrow=0.65] ([yshift=-0.3cm]Lm.0) -- node[below]{$\mathcal{H}_{R}$} ++(0.6,0);
         \draw[mid arrow=0.65] ([xshift=0.6cm,yshift=0.3cm]Lm.0) --node[above]{$\mathcal{H}_{L}$} ++(-0.6,0);
    \end{tikzpicture}
    \end{align}
The conservation law of information for $L_{m}$ dictates
\begin{equation}
S_{c}[\mathcal{H}_{\mathcal{B}_{2m}}] + S_{c}[\mathcal{H}_{R}] = S_{c}[\mathcal{H}_{\mathcal{A}_{2m}}] + S_{c}[\mathcal{H}_{L}].
\end{equation}
The logarithm of GNVW index can be expressed as:
\begin{equation}
\begin{aligned}
     \log_{d} I_{GNVW}(u) &= S_{c}[\mathcal{H}_{\mathcal{B}_{2m}}] - S_{c}[\mathcal{H}_{\mathcal{A}_{2m}}]\\
     &=  S_{c}[\mathcal{H}_{L}] - S_{c}[\mathcal{H}_{R}] \\
     & = -f_{I}[\Sigma] \\
     & = -f_{I}[U_{Net}],
\end{aligned}
\end{equation}
where the negative sign arises from the convention of selecting the positive direction of a surface. Thus, according to Proposition \ref{prop: same net information flow for LU-rep}, all unitary network representations that preserve locality adhere to this identity.
\end{proof}

An extension of intrinsic net information flow to the case of APLUs can be found in Appendix~\ref{Append: Intrinsic net information flow for approximately locality-preserving unitary networks}.

\section{Sequential quantum circuits and Unitary networks}
\label{Sec: UN and SQC}

A \textit{Sequential Quantum Circuit} \cite{schon2005sequential, schon2007sequential, banuls2008sequentially, wei2022sequential, Chen_2024} (SQC) is defined as a quantum circuit with local unitary gates, where each site is acted upon a finite number of gates \cite{Chen_2024}. SQCs are shown to be able to represent all QCAs and various non-local transformations to connect states of distinct gapped phases for PBC systems \cite{Chen_2024}. This section demonstrates the close relationship between unitary networks and SQCs. Essentially, a unitary network is an SQC with an additional shift, and any SQC can be transformed into a unitary network with finite bond dimensions.

\subsection{Unitary networks as SQCs}

Refs.~\cite{pollmann2016efficient, wahl2017efficient, haghshenas2022variational, Styliaris_2025} employ quantum circuit architectures to build tensor networks. A finite OBC unitary network can be regarded as a quantum circuit, particularly when each leg possesses a dimension of $\dim \mathcal{H}_{e} = d^{n}$, where
$d = \dim \mathcal{H}^{\text{qudit}}\in \mathbb{N}^+$ is any positive integer representing the dimension of a qudit Hilbert space, and $n \in \mathbb{N}^+$ is another positive integer. In such cases, the Hilbert space of the leg corresponds to the Hilbert space of $n$ qudits:
\begin{equation}
    \mathcal{H}_{e} \sim \bigotimes_{i=1}^{n} \mathcal{H}_{i}^{\text{qudit}}
\end{equation}
To highlight the link between unitary networks and quantum circuits, we will split the legs with dimension $d^{n}$ into $n$ legs of dimension $d$ in the diagrams.

As a first step, we will demonstrate that a local unitary tensor can be regarded as a quantum gate operating on qudits:
\begin{align}
\label{Eq: tensor=gate}
    \begin{array}{c}   \tikzsetnextfilename{local_unitary_tensor_is_quantum_gate__tensor}
        \begin{tikzpicture}
            \node[draw, rectangle, rounded corners=2pt, minimum width=0.6cm, minimum height=0.6cm, fill=tensorcolor, thick] (U)  {$U$};
            \draw[mid arrow=0.65,very thick] (U.0) --  node[above]{$mn$} ++(0.6,0) ;
             \draw[mid arrow=0.65,very thick] ([xshift=-0.6cm]U.180) --  node[above]{$ij$} ++(0.6,0) ;
             \draw[mid arrow=0.65, red!70] (U.90) --  node[right]{$l$} ++(0,0.6) ;
             \draw[mid arrow=0.65, red!70] ([yshift=-0.6cm]U.270) --  node[right]{$k$} ++(0,0.6) ;
        \end{tikzpicture}
    \end{array}
    =
   \begin{array}{c}
   \tikzsetnextfilename{local_unitary_tensor_is_quantum_gate__tensor_leg_splitting}
        \begin{tikzpicture}
            \node[draw, rectangle, rounded corners=2pt, minimum width=0.6cm, minimum height=0.6cm, fill=tensorcolor, thick] (U)  {$U$};
            \draw[mid arrow=0.65] ([yshift=-0.1cm]U.0) --  node[below]{$n$} ++(0.6,0) ;
            \draw[mid arrow=0.65] ([yshift=0.1cm]U.0) --  node[above]{$m$} ++(0.6,0) ;
             \draw[mid arrow=0.65] ([xshift=-0.6cm, yshift=0.1cm]U.180) --  node[above]{$i$} ++(0.6,0) ;
             \draw[mid arrow=0.65] ([xshift=-0.6cm,yshift=-0.1cm]U.180) --  node[below]{$j$} ++(0.6,0) ;
             \draw[mid arrow=0.65, red!70] (U.90) --  node[right]{$l$} ++(0,0.6) ;
             \draw[mid arrow=0.65,red!70] ([yshift=-0.6cm]U.270) --  node[right]{$k$} ++(0,0.6) ;
        \end{tikzpicture}
    \end{array}
        = 
    \begin{array}{c}
    \tikzsetnextfilename{local_unitary_tensor_is_quantum_gate__quantum_gate}
        \begin{tikzpicture}
             \node[draw, rectangle, rounded corners=2pt, minimum width=1.2cm, minimum height=0.6cm, fill=tensorcolor, thick] (U){$\Tilde{U}$};
             \draw[mid arrow=0.65, red!70] (-0.3, 0.3) -- (-0.3, 0.9)
            node[draw=none,fill=none,font=\scriptsize] at (-0.3, 1.1) {$l$};
             \draw[mid arrow=0.65] (0, 0.3) -- (0, 0.9)
            node[draw=none,fill=none,font=\scriptsize] at (0, 1.1) {$m$};
            \draw[mid arrow=0.65] (0.3, 0.3) -- (0.3, 0.9)
            node[draw=none,fill=none,font=\scriptsize] at (0.3, 1.1) {$n$};
            \draw[mid arrow=0.65] (-0.3, -0.9) -- (-0.3, -0.3)
            node[draw=none,fill=none,font=\scriptsize] at (-0.3, -1.1) {$i$};
             \draw[mid arrow=0.65] (0, -0.9) -- (0, -0.3)
            node[draw=none,fill=none,font=\scriptsize] at (0, -1.1) {$j$};
            \draw[mid arrow=0.65,red!70] (0.3, -0.9) -- (0.3, -0.3)
            node[draw=none,fill=none,font=\scriptsize] at (0.3, -1.1) {$k$};
        \end{tikzpicture} 
    \end{array}.
    \end{align}
As illustrated in the diagram, converting a local unitary tensor into a quantum gate involves repositioning all outgoing legs to the top of the tensor and all incoming legs to the bottom. A specific order needs to be determined for the legs of the quantum gate $\tilde{U}$. In one-dimensional scenarios, the outgoing (incoming) legs are naturally ordered from top to bottom and from left to right. 

The vertical legs $l$ and $k$ of the unitary tensor $U$ are painted red. Although $l$ and $k$ represent outgoing and incoming vertical legs positioned at the same location in $U$, they correspond to different qudits in the quantum gate $\tilde{U}$. This is allowed since the intermediate physical Hilbert space in the quantum circuit does not represent the final outcome and can be regarded as a type of bond space. Note that the bond dimension $D$ of a local unitary tensor determines the width of the corresponding quantum gate.

A unitary network meeting the following conditions can be considered a quantum circuit:
\begin{proposition} [Unitary network as a quantum circuit]
A unitary network $U_{Net}$ can be consistently transformed into a quantum circuit if it meets the following criteria:
\begin{enumerate} [label=(\roman*)]
    \item It is defined on a finite OBC system.
    \item Each of its legs has a dimension of $d^{n}$.
    \item Its graph $G[U_{Net}]$ is a DAG.
\end{enumerate}
\end{proposition}
\begin{proof}
The conversion resembles the method demonstrated in Eq.~(\ref{Eq:DAG=circuit}). First, we perform a topological sort on the DAG to establish the sequence of the local tensors. Each local tensor is then treated as a quantum gate and executed sequentially. If two lines intersect, a swap gate might be introduced to handle the crossing.
\end{proof}

The following diagram illustrates how to transform a bilayer unitary network into a quantum circuit:
\begin{align} 
\begin{array}{c}
\tikzsetnextfilename{unitary_network_is_quantum_circuit__UN}
    \begin{tikzpicture}
        \node[draw, rectangle, rounded corners=2pt, minimum width=0.6cm, minimum height=0.6cm, fill=tensorcolor, thick](A) at (-1.2, 0) {1};
        \node[draw, rectangle, rounded corners=2pt, minimum width=0.6cm, minimum height=0.6cm, fill=tensorcolor, thick](B) {2};
        \node[draw, rectangle, rounded corners=2pt, minimum width=0.6cm, minimum height=0.6cm, fill=tensorcolor, thick](C) at (1.2, 0) {3};
        \node[draw, rectangle, rounded corners=2pt, minimum width=0.6cm, minimum height=0.6cm, fill=tensorcolor, thick](D) at (-1.2, 1.2) {6};
        \node[draw, rectangle, rounded corners=2pt, minimum width=0.6cm, minimum height=0.6cm, fill=tensorcolor, thick](E) at (0, 1.2) {5};
        \node[draw, rectangle, rounded corners=2pt, minimum width=0.6cm, minimum height=0.6cm, fill=tensorcolor, thick](F) at (1.2, 1.2) {4};
        \draw[mid arrow=0.65] ([xshift=-0.1cm]D.90) -- ++(0,0.6);
        \draw[mid arrow=0.65] ([xshift=0.1cm]D.90) -- ++(0,0.6);
        \draw[mid arrow=0.65] ([xshift=-0.1cm]E.90) -- ++(0,0.6);
        \draw[mid arrow=0.65] ([xshift=0.1cm]E.90) -- ++(0,0.6);
        \draw[mid arrow=0.65] ([xshift=-0.1cm]F.90) -- ++(0,0.6);
        \draw[mid arrow=0.65] ([xshift=0.1cm]F.90) -- ++(0,0.6);
        \draw[mid arrow=0.65] ([xshift=-0.1cm,yshift=-0.6cm]A.270) -- ++(0,0.6);
        \draw[mid arrow=0.65] ([xshift=0.1cm,yshift=-0.6cm]A.270) -- ++(0,0.6);
        \draw[mid arrow=0.65] ([xshift=-0.1cm,yshift=-0.6cm]B.270) -- ++(0,0.6);
        \draw[mid arrow=0.65] ([xshift=0.1cm,yshift=-0.6cm]B.270) -- ++(0,0.6);
        \draw[mid arrow=0.65] ([xshift=-0.1cm,yshift=-0.6cm]C.270) -- ++(0,0.6);
        \draw[mid arrow=0.65] ([xshift=0.1cm,yshift=-0.6cm]C.270) -- ++(0,0.6);
        \draw[mid arrow=0.65] (A.0) -- (B.180);
        \draw[mid arrow=0.65] ([yshift=-0.1cm]B.0) -- ([yshift=-0.1cm]C.180);
        \draw[mid arrow=0.65] ([yshift=0.1cm]B.0) -- ([yshift=0.1cm]C.180);
        \draw[mid arrow=0.65] ([yshift=-0.1cm]F.180) -- ([yshift=-0.1cm]E.0);
        \draw[mid arrow=0.65] ([yshift=0.1cm]F.180) -- ([yshift=0.1cm]E.0);
        \draw[mid arrow=0.65] (E.180) -- (D.0);
        \draw[mid arrow=0.65] (A.90) -- (D.270);
        \draw[mid arrow=0.65] (B.90) -- (E.270);
        \draw[mid arrow=0.65] ([xshift=-0.15cm]C.90) -- ([xshift=-0.15cm]F.270);
        \draw[mid arrow=0.65] ([xshift=-0.05cm]C.90) -- ([xshift=-0.05cm]F.270);
        \draw[mid arrow=0.65] ([xshift=0.05cm]C.90) -- ([xshift=0.05cm]F.270);
        \draw[mid arrow=0.65] ([xshift=0.15cm]C.90) -- ([xshift=0.15cm]F.270);
    \end{tikzpicture}
\end{array} = 
\begin{array}{c}
\tikzsetnextfilename{unitary_network_is_quantum_circuit__quantum_circuit}
\begin{tikzpicture}
    \node[draw, rectangle, rounded corners=2pt, minimum width=0.6cm, minimum height=0.6cm, fill=tensorcolor, thick](A) at (0, 0) {1};
    \draw ([xshift=-0.2cm,yshift=-0.4cm]A.270) -- ++(0,0.4);
    \draw ([xshift=0.2cm,yshift=-0.4cm]A.270) -- ++(0,0.4);
    \draw ([xshift=-0.2cm]A.90) -- ++(0,4.4);
    \draw ([xshift=0.2cm]A.90) -- ++(0,0.4);
     \node[draw, rectangle, rounded corners=2pt, minimum width=1.0cm, minimum height=0.6cm, fill=tensorcolor, thick](B) at (0.6, 1) {2};
    \draw ([xshift=0.4cm,yshift=-1.4cm]B.270) -- ++(0,1.4);
    \draw ([yshift=-1.4cm]B.270) -- ++(0,1.4);
    \draw ([xshift=-0.4cm]B.90) -- ++(0,2.4);
    \draw (B.90) -- ++(0,0.4);
    \draw ([xshift=0.4cm]B.90) -- ++(0,0.4);
    \node[draw, rectangle, rounded corners=2pt, minimum width=1.4cm, minimum height=0.6cm, fill=tensorcolor, thick](C) at (1.2, 2) {3};
    \draw ([xshift=-0.6cm]C.90) -- ++(0,0.4);
    \draw ([xshift=-0.2cm]C.90) -- ++(0,0.4);
    \draw ([xshift=0.2cm]C.90) -- ++(0,0.4);
    \draw ([xshift=0.6cm]C.90) -- ++(0,0.4);
    \draw ([xshift=0.2cm,yshift=-2.4cm]C.270) -- ++(0,2.4);
    \draw ([xshift=0.6cm,yshift=-2.4cm]C.270) -- ++(0,2.4);
    \node[draw, rectangle, rounded corners=2pt, minimum width=1.4cm, minimum height=0.6cm, fill=tensorcolor, thick](F) at (1.2, 3) {4};
    \draw ([xshift=-0.6cm]F.90) -- ++(0,0.4);
    \draw ([xshift=-0.2cm]F.90) -- ++(0,0.4);
    \draw ([xshift=0.2cm]F.90) -- ++(0,2.4);
    \draw ([xshift=0.6cm]F.90) -- ++(0,2.4);
    \node[draw, rectangle, rounded corners=2pt, minimum width=1.0cm, minimum height=0.6cm, fill=tensorcolor, thick](E) at (0.6, 4) {5};
    \draw ([xshift=-0.4cm]E.90) -- ++(0,0.4);
    \draw (E.90) -- ++(0,1.4);
    \draw ([xshift=0.4cm]E.90) -- ++(0,1.4);
    \node[draw, rectangle, rounded corners=2pt, minimum width=0.6cm, minimum height=0.6cm, fill=tensorcolor, thick](D) at (0, 5) {6};
    \draw ([xshift=-0.2cm]D.90) -- ++(0,0.4);
    \draw ([xshift=0.2cm]D.90) -- ++(0,0.4);
\end{tikzpicture}
\end{array}
\end{align}

Unitary networks defined in a finite-size system, when viewed as quantum circuits, exhibit linear depth and sequential architecture. Moreover, each site will be acted upon by a finite number of gates. Quantum circuits with these characteristics are exactly SQCs.

\subsection{SQCs as unitary networks}
\label{Subsec: SQCs as unitary networks}

In this subsection, we show that any SQC, featuring local unitarity and DAGs, can be regarded as a unitary network with finite bond dimension. 

\begin{proposition} [SQC as a unitary network]
Any SQC can be converted into a multi-layer unitary network with finite bond dimensions.
\end{proposition}
\begin{proof}
An SQC is, by definition, made up of local unitary gates, each acting on a finite number of sites.
Starting from an SQC as below:
\begin{align}
\tikzsetnextfilename{SQC}
\begin{tikzpicture}
    \node[draw, rectangle, rounded corners=2pt, minimum width=0.6cm, minimum height=0.6cm, fill=tensorcolor, thick](A) at (0, 0) {1};
    \draw ([xshift=-0.2cm,yshift=-0.4cm]A.270) -- ++(0,0.4);
    \draw ([xshift=0.2cm,yshift=-0.4cm]A.270) -- ++(0,0.4);
    \draw ([xshift=-0.2cm]A.90) -- ++(0,2.4);
    \draw ([xshift=0.2cm]A.90) -- ++(0,0.4);
     \node[draw, rectangle, rounded corners=2pt, minimum width=1.0cm, minimum height=0.6cm, fill=tensorcolor, thick](B) at (0.6, 1) {2};
    \draw ([xshift=0.4cm,yshift=-1.4cm]B.270) -- ++(0,1.4);
    \draw ([yshift=-1.4cm]B.270) -- ++(0,1.4);
    \draw ([xshift=-0.4cm]B.90) -- ++(0,1.4);
    \draw (B.90) -- ++(0,0.4);
    \draw ([xshift=0.4cm]B.90) -- ++(0,0.4);
    \node[draw, rectangle, rounded corners=2pt, minimum width=1.0cm, minimum height=0.6cm, fill=tensorcolor, thick](C) at (1.0, 2) {3};
    \draw ([xshift=-0.4cm]C.90) -- ++(0,0.4);
    \draw (C.90) -- ++(0,0.4);
    \draw ([xshift=0.4cm]C.90) -- ++(0,0.4);
    \draw ([xshift=0.4cm,yshift=-2.4cm]C.270) -- ++(0,2.4);
\end{tikzpicture}
\end{align}
In this case, a maximum of $n=2$ gates operate on each site. We may decompose each gate into a bilayer unitary network as in \eqref{Eq: decompose into BUN}. This will yield a unitary network, whose number of layers increases linearly with system size.
\begin{equation}
\tikzsetnextfilename{SQC_UN}
    \begin{tikzpicture}
         \node[draw, rectangle, rounded corners=2pt, minimum width=0.4cm, minimum height=0.4cm, fill=tensorcolor, thick] (a1) at (-2.4,0) {};
        \node[draw, rectangle, rounded corners=2pt, minimum width=0.4cm, minimum height=0.4cm, fill=tensorcolor, thick] (b1) at (-1.6,0) {};
         \node[draw, rectangle, rounded corners=2pt, minimum width=0.4cm, minimum height=0.4cm, fill=tensorcolor, thick] (a2) at (-2.4,0.8) {};
        \node[draw, rectangle, rounded corners=2pt, minimum width=0.4cm, minimum height=0.4cm, fill=tensorcolor, thick] (b2) at (-1.6,0.8) {};
        \node[draw, rectangle, rounded corners=2pt, minimum width=0.4cm, minimum height=0.4cm, fill=tensorcolor, thick] (a3) at (-1.6,1.6) {};
        \node[draw, rectangle, rounded corners=2pt, minimum width=0.4cm, minimum height=0.4cm, fill=tensorcolor, thick] (b3) at (-0.8,1.6) {};
        \node[draw, rectangle, rounded corners=2pt, minimum width=0.4cm, minimum height=0.4cm, fill=tensorcolor, thick] (c3) at (0, 1.6) {};
        \node[draw, rectangle, rounded corners=2pt, minimum width=0.4cm, minimum height=0.4cm, fill=tensorcolor, thick] (a4) at (-1.6,2.4) {};
        \node[draw, rectangle, rounded corners=2pt, minimum width=0.4cm, minimum height=0.4cm, fill=tensorcolor, thick] (b4) at (-0.8,2.4) {};
        \node[draw, rectangle, rounded corners=2pt, minimum width=0.4cm, minimum height=0.4cm, fill=tensorcolor, thick] (c4) at (0, 2.4) {};
         \node[draw, rectangle, rounded corners=2pt, minimum width=0.4cm, minimum height=0.4cm, fill=tensorcolor, thick] (a5) at (-0.8,3.2) {};
        \node[draw, rectangle, rounded corners=2pt, minimum width=0.4cm, minimum height=0.4cm, fill=tensorcolor, thick] (b5) at (0,3.2) {};
        \node[draw, rectangle, rounded corners=2pt, minimum width=0.4cm, minimum height=0.4cm, fill=tensorcolor, thick] (c5) at (0.8, 3.2) {};
        \node[draw, rectangle, rounded corners=2pt, minimum width=0.4cm, minimum height=0.4cm, fill=tensorcolor, thick] (a6) at (-0.8,4) {};
        \node[draw, rectangle, rounded corners=2pt, minimum width=0.4cm, minimum height=0.4cm, fill=tensorcolor, thick] (b6) at (0,4) {};
        \node[draw, rectangle, rounded corners=2pt, minimum width=0.4cm, minimum height=0.4cm, fill=tensorcolor, thick] (c6) at (0.8, 4) {};
        \draw[mid arrow=0.65] ([yshift=-0.4cm]a1.270) -- (a1.270);
        \draw[mid arrow=0.65] ([yshift=-0.4cm]b1.270) -- (b1.270);
        \draw[mid arrow=0.65] (a1.0) -- (b1.180);
         \draw[mid arrow=0.65] (a1.90) -- (a2.270);
         \draw[mid arrow=0.65] (b1.90) -- (b2.270);
         \draw[mid arrow=0.65] (b2.180) -- (a2.0);
         \draw[mid arrow=0.65] (a2.90) -- ++(0,3.6);
         \draw[mid arrow=0.65] (b2.90) -- (a3.270);
         \draw[mid arrow=0.65] ([yshift=-2cm]b3.270) -- (b3.270);
         \draw[mid arrow=0.65] ([yshift=-2cm]c3.270) -- (c3.270);
         \draw[mid arrow=0.65] ([yshift=-3.6cm]c5.270) -- (c5.270);
         \draw[mid arrow=0.65] (a3.90) -- (a4.270);
         \draw[mid arrow=0.65] (b3.90) -- (b4.270);
         \draw[mid arrow=0.65] (c3.90) -- (c4.270);
         \draw[mid arrow=0.65] (a5.90) -- (a6.270);
         \draw[mid arrow=0.65] (b5.90) -- (b6.270);
         \draw[mid arrow=0.65] (c5.90) -- (c6.270);
         \draw[mid arrow=0.65] (b4.90) -- (a5.270);
         \draw[mid arrow=0.65] (c4.90) -- (b5.270);
         \draw[mid arrow=0.65] (a6.90) -- ++(0,0.4);
         \draw[mid arrow=0.65] (b6.90) -- ++(0,0.4);
         \draw[mid arrow=0.65] (c6.90) -- ++(0,0.4);
         \draw[mid arrow=0.65] (a4.90) -- ++(0,2);
         \draw[mid arrow=0.65] (a3.0) -- (b3.180);
         \draw[mid arrow=0.65] (b3.0) -- (c3.180);
         \draw[mid arrow=0.65] (c4.180) -- (b4.0);
         \draw[mid arrow=0.65] (b4.180) -- (a4.0);
          \draw[mid arrow=0.65] (a5.0) -- (b5.180);
         \draw[mid arrow=0.65] (b5.0) -- (c5.180);
         \draw[mid arrow=0.65] (c6.180) -- (b6.0);
         \draw[mid arrow=0.65] (b6.180) -- (a6.0);
    \end{tikzpicture}
\end{equation}
We can vertically compress this unitary network, leading to a $2n=4$-layer unitary network. 
\begin{equation}
\tikzsetnextfilename{SQC_UN_compressed}
    \begin{tikzpicture}
         \node[draw, rectangle, rounded corners=2pt, minimum width=0.4cm, minimum height=0.4cm, fill=tensorcolor, thick] (a1) at (-2.4,0) {};
        \node[draw, rectangle, rounded corners=2pt, minimum width=0.4cm, minimum height=0.4cm, fill=tensorcolor, thick] (b1) at (-1.6,0) {};
         \node[draw, rectangle, rounded corners=2pt, minimum width=0.4cm, minimum height=0.4cm, fill=tensorcolor, thick] (a2) at (-2.4,0.8) {};
        \node[draw, rectangle, rounded corners=2pt, minimum width=0.4cm, minimum height=0.4cm, fill=tensorcolor, thick] (b2) at (-1.6,0.8) {};
        \node[draw, rectangle, rounded corners=2pt, minimum width=0.4cm, minimum height=0.4cm, fill=tensorcolor, thick] (a3) at (-1.6,1.6) {};
        \node[draw, rectangle, rounded corners=2pt, minimum width=0.4cm, minimum height=0.4cm, fill=tensorcolor, thick] (b3) at (-0.8,0) {};
        \node[draw, rectangle, rounded corners=2pt, minimum width=0.4cm, minimum height=0.4cm, fill=tensorcolor, thick] (c3) at (0, 0) {};
        \node[draw, rectangle, rounded corners=2pt, minimum width=0.4cm, minimum height=0.4cm, fill=tensorcolor, thick] (a4) at (-1.6,2.4) {};
        \node[draw, rectangle, rounded corners=2pt, minimum width=0.4cm, minimum height=0.4cm, fill=tensorcolor, thick] (b4) at (-0.8,0.8) {};
        \node[draw, rectangle, rounded corners=2pt, minimum width=0.4cm, minimum height=0.4cm, fill=tensorcolor, thick] (c4) at (0, 0.8) {};
         \node[draw, rectangle, rounded corners=2pt, minimum width=0.4cm, minimum height=0.4cm, fill=tensorcolor, thick] (a5) at (-0.8,1.6) {};
        \node[draw, rectangle, rounded corners=2pt, minimum width=0.4cm, minimum height=0.4cm, fill=tensorcolor, thick] (b5) at (0,1.6) {};
        \node[draw, rectangle, rounded corners=2pt, minimum width=0.4cm, minimum height=0.4cm, fill=tensorcolor, thick] (c5) at (0.8, 0) {};
        \node[draw, rectangle, rounded corners=2pt, minimum width=0.4cm, minimum height=0.4cm, fill=tensorcolor, thick] (a6) at (-0.8,2.4) {};
        \node[draw, rectangle, rounded corners=2pt, minimum width=0.4cm, minimum height=0.4cm, fill=tensorcolor, thick] (b6) at (0,2.4) {};
        \node[draw, rectangle, rounded corners=2pt, minimum width=0.4cm, minimum height=0.4cm, fill=tensorcolor, thick] (c6) at (0.8, 0.8) {};
        \draw[mid arrow=0.65] ([yshift=-0.4cm]a1.270) -- (a1.270);
        \draw[mid arrow=0.65] ([yshift=-0.4cm]b1.270) -- (b1.270);
        \draw[mid arrow=0.65] (a1.0) -- (b1.180);
         \draw[mid arrow=0.65] (a1.90) -- (a2.270);
         \draw[mid arrow=0.65] (b1.90) -- (b2.270);
         \draw[mid arrow=0.65] (b2.180) -- (a2.0);
         \draw[mid arrow=0.65] (a2.90) -- ++(0,2);
         \draw[mid arrow=0.65] (b2.90) -- (a3.270);
         \draw[mid arrow=0.65] ([yshift=-0.4cm]b3.270) -- (b3.270);
         \draw[mid arrow=0.65] ([yshift=-0.4cm]c3.270) -- (c3.270);
         \draw[mid arrow=0.65] ([yshift=-0.4cm]c5.270) -- (c5.270);
         \draw[mid arrow=0.65] (a3.90) -- (a4.270);
         \draw[mid arrow=0.65] (b3.90) -- (b4.270);
         \draw[mid arrow=0.65] (c3.90) -- (c4.270);
         \draw[mid arrow=0.65] (a5.90) -- (a6.270);
         \draw[mid arrow=0.65] (b5.90) -- (b6.270);
         \draw[mid arrow=0.65] (c5.90) -- (c6.270);
         \draw[mid arrow=0.65] (b4.90) -- (a5.270);
         \draw[mid arrow=0.65] (c4.90) -- (b5.270);
         \draw[mid arrow=0.65] (a6.90) -- ++(0,0.4);
         \draw[mid arrow=0.65] (b6.90) -- ++(0,0.4);
         \draw[mid arrow=0.65] (c6.90) -- ++(0,2);
         \draw[mid arrow=0.65] (a4.90) -- ++(0,0.4);
         \draw[mid arrow=0.65] (a3.0) -- (b3.180);
         \draw[mid arrow=0.65] (b3.0) -- (c3.180);
         \draw[mid arrow=0.65] (c4.180) -- (b4.0);
         \draw[mid arrow=0.65] (b4.180) -- (a4.0);
          \draw[mid arrow=0.65] (a5.0) -- (b5.180);
         \draw[mid arrow=0.65] (b5.0) -- (c5.180);
         \draw[mid arrow=0.65] (c6.180) -- (b6.0);
         \draw[mid arrow=0.65] (b6.180) -- (a6.0);
    \end{tikzpicture}
\end{equation}
In doing so, interpreting the vertical direction as temporal is no longer possible due to the presence of downward arrows. Nonetheless, this is irrelevant as it won't impact the global unitarity, with causal order preserved in the DAG structure. Legs possessing a vertical component can still be considered as bond legs.

We will next demonstrate that this unitary network has a finite bond dimension. Each bilayer unitary network, decomposed from a local gate, possesses a bond dimension $D_{gate}$ such that $D_{gate} < d^{2m}$, where $d$ is the Hilbert space dimension per site, and $m$ represents the maximum number of sites acted by the local gates. With up to $n$ gates acting on each site, the total bond dimension $D_{circuit}$ satisfies $D_{circuit} \le D_{gate}^{n} = d^{2mn}$, and is thus finite.
\end{proof}

\subsection{Unitary networks and SQCs: differences and equivalence}

In discussing the net information flow, we differentiate between a physical unitary operator and its unitary network representations. Similarly to unitary networks, a quantum circuit serves as a representation rather than an operator. This arises from the fact that a physical operator can be realized using different circuits.

For finite OBC systems, we have demonstrated that unitary networks can be seen as SQCs. SQCs, with their local unitary gates and DAG structures, can also be regarded unitary networks. In finite 1D OBC systems, unitary networks and SQCs are interchangeable and thus equivalent.

For infinite OBC systems and PBC systems, unitary networks extend SQCs because they naturally combine quantum circuit and shift operations. In Section \ref{Sec: Information flow and the GNVW index}, we showed that a unitary network can sustain a non-zero net information flow. In contrast, SQCs, when representing physical unitary operators, allow only zero information flow, as detailed in the Appendix \ref{Append: SQC has zero net information flow}.

In Subsections \ref{subsec: Intrinsic net information flow}, we demonstrated that both infinite OBC and PBC systems allow for shift operations implemented with zero information flow, thus achievable by SQCs. Consequently, SQCs and unitary networks possess equivalent representability. However, the SQC implementation of the shift operation introduces an unnecessary loop of information flow in PBC and extra propagating modes that do not influence the bulk in OBC system. These redundant information flows increase the bond dimensions, and thus reduce the efficiency of SQC implementation compared to unitary network representations.

\section{Reduce the Bond dimensions}
\label{Sec: reduce bond dimension}
A global unitary can be represented by different bilayer unitary networks, demonstrating diverse implementations for a global unitary tensor. In many scenarios, it is helpful to identify the most efficient unitary network representation that minimizes computational resources. 

For general unitary operators, it is challenging to find an efficient unitary network representation. In this section, we offer an algorithm to decompose a Fermionic Gaussian unitary operator into an efficient unitary network.

\subsection{Defining the problem}

A single global unitary operation can be represented by several different bilayer unitary networks. Take, for instance, the expression of an identity operation on a 3-qudit system using a bilayer unitary network; there are at least two ways to achieve this:
\begin{align}
\tikzsetnextfilename{two_unitary_network_for_identity}
\begin{tikzpicture}
    \begin{scope}[local bounding box=graph_a]
                \node[draw, rectangle, rounded corners=2pt, minimum width=0.6cm, minimum height=0.6cm, fill=tensorcolor, thick](A) at (-1.2, 0) {$I_{A}$};
        \node[draw, rectangle, rounded corners=2pt, minimum width=0.6cm, minimum height=0.6cm, fill=tensorcolor, thick](B) {$I_{B}$};
        \node[draw, rectangle, rounded corners=2pt, minimum width=0.6cm, minimum height=0.6cm, fill=tensorcolor, thick](C) at (1.2, 0) {$I_{C}$};
        \node[draw, rectangle, rounded corners=2pt, minimum width=0.6cm, minimum height=0.6cm, fill=tensorcolor, thick](D) at (-1.2, 1.2) {$I_{D}$};
        \node[draw, rectangle, rounded corners=2pt, minimum width=0.6cm, minimum height=0.6cm, fill=tensorcolor, thick](E) at (0, 1.2) {$I_{E}$};
        \node[draw, rectangle, rounded corners=2pt, minimum width=0.6cm, minimum height=0.6cm, fill=tensorcolor, thick](F) at (1.2, 1.2) {$I_{F}$};
        \draw[mid arrow=0.65] (D.90) -- ++(0,0.6);
        \draw[mid arrow=0.65] (E.90) -- ++(0,0.6);
        \draw[mid arrow=0.65] (F.90) -- ++(0,0.6);
        \draw[mid arrow=0.65] ([yshift=-0.6cm]A.270) -- (A.270);
        \draw[mid arrow=0.65] ([yshift=-0.6cm]B.270) -- (B.270);
        \draw[mid arrow=0.65] ([yshift=-0.6cm]C.270) -- (C.270);
        \draw[mid arrow=0.65] (A.0) -- (B.180);
        \draw[mid arrow=0.65] ([yshift=-0.1cm]B.0) -- ([yshift=-0.1cm]C.180);
        \draw[mid arrow=0.65] ([yshift=0.1cm]B.0) -- ([yshift=0.1cm]C.180);
        \draw[mid arrow=0.65] ([yshift=-0.1cm]F.180) -- ([yshift=-0.1cm]E.0);
        \draw[mid arrow=0.65] ([yshift=0.1cm]F.180) -- ([yshift=0.1cm]E.0);
        \draw[mid arrow=0.65] (E.180) -- (D.0);
        \draw[mid arrow=0.65] (C.90) -- (F.270);
        \draw[mid arrow=0.65] ([xshift=-0.2cm]C.90) -- ([xshift=-0.2cm]F.270);
        \draw[mid arrow=0.65] ([xshift=0.2cm]C.90) -- ([xshift=0.2cm]F.270);
    \end{scope}
    \begin{scope} [shift={(4, 0)}, local bounding box=graph_b]
                \node[draw, rectangle, rounded corners=2pt, minimum width=0.6cm, minimum height=0.6cm, fill=tensorcolor, thick](A) at (-1.2, 0) {$I_{A}$};
        \node[draw, rectangle, rounded corners=2pt, minimum width=0.6cm, minimum height=0.6cm, fill=tensorcolor, thick](B) {$I_{B}$};
        \node[draw, rectangle, rounded corners=2pt, minimum width=0.6cm, minimum height=0.6cm, fill=tensorcolor, thick](C) at (1.2, 0) {$I_{C}$};
        \node[draw, rectangle, rounded corners=2pt, minimum width=0.6cm, minimum height=0.6cm, fill=tensorcolor, thick](D) at (-1.2, 1.2) {$I_{D}$};
        \node[draw, rectangle, rounded corners=2pt, minimum width=0.6cm, minimum height=0.6cm, fill=tensorcolor, thick](E) at (0, 1.2) {$I_{E}$};
        \node[draw, rectangle, rounded corners=2pt, minimum width=0.6cm, minimum height=0.6cm, fill=tensorcolor, thick](F) at (1.2, 1.2) {$I_{F}$};
        \draw[mid arrow=0.65] (D.90) -- ++(0,0.6);
        \draw[mid arrow=0.65] (E.90) -- ++(0,0.6);
        \draw[mid arrow=0.65] (F.90) -- ++(0,0.6);
        \draw[mid arrow=0.65] ([yshift=-0.6cm]A.270) -- (A.270);
        \draw[mid arrow=0.65] ([yshift=-0.6cm]B.270) -- (B.270);
        \draw[mid arrow=0.65] ([yshift=-0.6cm]C.270) -- (C.270);
        \draw[mid arrow=0.65] (A.90) -- (D.270);
        \draw[mid arrow=0.65] (B.90) -- (E.270);
        \draw[mid arrow=0.65] (C.90) -- (F.270);
    \end{scope}
    \node at (graph_a.north west)  (label_a)  {(a)}; 
    \node  at (graph_b.north west) (label_b) {(b)}; 
\end{tikzpicture}
\end{align}
In the diagram presented above, each leg possesses a dimension of $d$. When there are $n$ legs linking two local tensors, the bond dimension becomes $d^{n}$. (b) intuitively seems more efficient than (a). By counting the legs, we see that (b) exhibits smaller bond dimensions than (a).

As the local unitary tensors in the preceding diagrams are identity tensors, we can alternatively represent the above unitary networks as follows:
\begin{align}
\tikzsetnextfilename{two_info_path_for_identity}
\begin{tikzpicture}
    \begin{scope}[shift={(-1, 0)}, local bounding box=graph_a]
    \draw[mid arrow=0.3, mid arrow=0.5, mid arrow =0.7] (0,0) -- (0, 1) -- (2.2, 1) -- (2.2, 2) -- (0,2) -- (0,3);
    \draw[mid arrow=0.3, mid arrow=0.5, mid arrow=0.7] (1.2,0) -- (1.2, 0.8) -- (2.3, 0.8) -- (2.3, 2.2) -- (1.2,2.2) -- (1.2,3);
    \draw[mid arrow=0.5] (2.4,0) -- (2.4,3);
    \end{scope}
    \begin{scope} [shift={(4, 0)}, local bounding box=graph_b]
    \draw[mid arrow=0.6] (-1.2,0) -- (-1.2,3);
    \draw[mid arrow=0.6] (0,0) -- (0,3);
    \draw[mid arrow=0.6] (1.2,0) -- (1.2,3);
    \end{scope}
    \node at ([xshift=-0.2cm, yshift=0.2cm]graph_a.north west) {(a)}; 
    \node at ([xshift=-0.2cm, yshift=0.2cm]graph_b.north west) {(b)}; 
\end{tikzpicture}
\end{align}
The degrees of freedom and their corresponding information $S_{c}[\mathcal{H}]$ in (a) travel a longer distance than in (b). In fact, with a distance measure in the unitary network as described in Def.\ref{def: UN distance}, we can define a cost function on its edges:
\begin{definition} [Cost function]
For a unitary network $U_{Net}$ and its asscociate graph $G[U_{Net}] = (V,E)$, we define a cost function $C_{I}$ for a directed edge $e_{ij} \in E$ as:
\begin{equation}
    C_{I}[e_{ij}] \coloneq S_{c}[\mathcal{H}_{e_{ij}}] \cdot |e_{ij}|,  
\end{equation}
where $S_{c}[\mathcal{H}_{e_{ij}}]$ denotes the coherent entropy of the Hilbert space associated with $e_{ij}$, and $|e_{ij}|$ represents the length of $e_{ij}$. 

The total cost $C_{I}[U_{Net}]$ for the whole unitary network is defined as:
\begin{equation}
    C_{I}[U_{Net}] \coloneq \sum_{e_{ij} \in E} C_{I}[e_{ij}].
\end{equation}
\end{definition}

With the cost function $C_{I}$, we are ready to define the optimized unitary network for a given unitary operator $\tilde{U}$.
\begin{definition} [Optimized unitary network]
\label{def: Optimized unitary network}
    Given a target global unitary $\tilde{U}$ and a graph for unitary network $G[U_{Net}] = (V,E)$, one may define an optimized unitary network by the following optimization formulation: 
\begin{equation}
\begin{aligned}
  & \text{Minimize:} \quad  C_{I}[U_{Net}], \\
  & \text{subject to:} \quad \text{the contraction of} \ U_{Net} = \tilde{U}. 
\end{aligned}
\end{equation}
\end{definition}

The total cost $C_{I}[U_{Net}]$ has a clear physical meaning. When confined to manipulating single-site unitary operations, realizing unitary operators with larger support requires transferring qudits between sites through physical transport or quantum teleportation. Hence, the transfer of qudits is considered a resource. In a unitary network defined on a lattice system, the total cost $C_{I}[U_{Net}]$ measures the required qudit transfers between neighboring sites. Optimizing the UN identifies the minimal qudit transfer required for implementing the global unitary.

Moreover, in a bilinear unitary network, the total cost $C_{I}[U_{Net}]$ is equal to the logarithm of the product of the bond dimensions:
\begin{equation}
    \sum_{e_{ij} \in E} C_{I}[e_{ij}] = a \cdot\sum_{<m,n>} \log_{d} D_{mn} = a \log_{d} (\prod_{<m,n>} D_{mn}),
\end{equation}
where $a$ represents the lattice constant, and $D_{mn}$ denotes the bond dimension between adjacent sites $m$ and $n$.

\subsection{Fermionic Gaussian Unitary and Fermionic Gaussian Unitary Network}

Throughout the majority of this paper, our focus is restricted to the application of unitary networks within qudit systems. In this section, we briefly analyze the Fermionic Gaussian unitary. By calculating everything at the mode level, we can disregard the complexities of fermionic statistics.

Consider a Fermionic system with $N$ Fermion modes. We denote the Hilbert space for $N$ Fermion modes as 
\begin{equation}
    \mathcal{H}_{f}^{\wedge N} = \mathcal{H}_{f}\wedge\mathcal{H}_{f} \wedge \cdots \wedge \mathcal{H}_{f} \quad (N \ \text{times}),
\end{equation}
where $\mathcal{H}_{f} \cong \mathbb{C}^{2}$ is the Hilbert space of a single Fermion mode, and $\wedge$ is the wedge product. The general unitary group on $\mathcal{H}_{f}^{\wedge N}$ is denoted as $U(\mathcal{H}_{f}^{\wedge N}) \subseteq U(2^{N})$. If we limit the unitary to a Fermionic Gaussian unitary with particle number conservation, it can be represented by a single-particle unitary, also referred to as mode unitary $U(N)$. Within this section, we represent a mode unitary by $U \in U(N)$ (without a hat), whereas the many-body unitary operator is denoted by $\hat{U} \in U(2^{N})$ (with hat).
\begin{definition} [Projective group representations of $U(N)$]
There is a projective representation $\hat{\rho}: U(N) \rightarrow U(\mathcal{H}_{f}^{\wedge N}) \subseteq U(2^{N})$ \cite{Hackl_2021}, given by:
\begin{equation}
\begin{aligned}
    &\hat{\rho}(e^{ih}) = \pm e^{i\hat{H}}, \\
    & \hat{H} = h_{ab} \hat{c}^{\dagger}_{a} \hat{c}_{b},
\end{aligned}
\end{equation}
which is defined up to an overall sign.  
This projective representation $\rho$ provides a group homomorphism between unitary mode $U$ and many-body unitary $\hat{U}$ such that:
\begin{equation}
    \begin{aligned}
       & \hat{\rho}(U) = \pm \hat{U},\\
      &\hat{\rho}(U_{1})\hat{\rho}(U_{2}) = \pm \hat{\rho}(U_{1}U_{2}).
    \end{aligned}
\end{equation}
\end{definition}

As we will discuss shortly, the homomorphism $U \rightarrow \hat{U}$ between mode unitary and many-body unitary can be extended to a homomorphism between mode unitary network and many-body unitary network $U_{Net} \rightarrow \hat{U}_{Net}$. 

First, we will discuss the visualization of a  mode unitary network.

\begin{definition} [Mode unitary matrix]
A mode unitary $U$, or $U^{a}_{b}$, is a unitary matrix (rank-2 tensor). Illustrate in this manner, with each outgoing leg corresponding to a column in $U$, while each incoming leg indicates a row in $U$.
\begin{align}
\tikzsetnextfilename{mode_unitary}
\begin{tikzpicture}
    \begin{scope}  [local bounding box=graph_a]
        \node[draw, rectangle, rounded corners=2pt, minimum width=1.2cm, minimum height=0.6cm, fill=tensorcolor, thick] (hat_U){$\hat{U}$};
        \draw[mid arrow=0.6] ([xshift=-0.4cm]hat_U.90) -- ++(0,0.4);
        \draw[mid arrow=0.6] (hat_U.90) -- ++(0,0.4);
        \draw[mid arrow=0.6] ([xshift=0.4cm]hat_U.90) -- ++(0,0.4);
        \draw[mid arrow=0.6] ([xshift=-0.4cm,yshift=-0.4cm]hat_U.270) -- ++(0,0.4);
        \draw[mid arrow=0.6] ([yshift=-0.4cm]hat_U.270) -- ++(0,0.4);
        \draw[mid arrow=0.6] ([xshift=0.4cm,yshift=-0.4cm]hat_U.270) -- ++(0,0.4);
    \end{scope}
    \begin{scope}  [shift={(4, 0)}, local bounding box=graph_b]
        \node[draw, rectangle, rounded corners=2pt, minimum width=1.2cm, minimum height=0.6cm, fill=tensorcolor, thick] (U){$U$};
        \draw[mid arrow=0.6] (U.90) -- ++(0,0.4);
        \draw[mid arrow=0.6] ([xshift=0.4cm]U.90) -- ++(0,0.4);
        \draw[mid arrow=0.6] ([xshift=-0.4cm]U.90) -- ++(0,0.4);
        \draw[mid arrow=0.6] ([xshift=-0.4cm,yshift=-0.4cm]U.270) -- ++(0,0.4);
        \draw[mid arrow=0.6] ([yshift=-0.4cm]U.270) -- ++(0,0.4);
        \draw[mid arrow=0.6] ([xshift=0.4cm,yshift=-0.4cm]U.270) -- ++(0,0.4);
    \end{scope}
    \node  at ([xshift=-0.2cm, yshift=0.2cm]graph_a.north west) {(a)};
    \node at ([xshift=-0.2cm, yshift=0.2cm]graph_b.north west) {(b)};
\end{tikzpicture}
\end{align}
\end{definition}
In (a), we depict a Fermionic Gaussian unitary tensor $\hat{U}$ acting on $N$ Fermion modes. It features $N$ incoming and $N$ outgoing legs, each leg carries a Hilbert space $\mathcal{H}_{f} \cong \mathbb{C}^{2}$. Therefore, $\hat{U}$ is a rank-$2N$ tensor. 

In (b), we depict the corresponding mode unitary $U$, with $\hat{\rho}(U) = \hat{U}$. Though the diagram in (b) looks identical to that of (a), it signifies a mode unitary $U \in U(N)$ with a distinct interpretation. With $N$ incoming and $N$ outgoing legs, $U$ remains a matrix (a rank-2 tensor). The presence of $N$ outgoing legs implies that $U$ consists of $N$ columns. Or, in other words, $a = 1, 2, \ldots, N$ in $U^{a}_{b}$.

In order for a set of mode unitary matrices $\{U_{1}, U_{2},\cdots, U_{n}\}$ to form a unitary network $U_{Net}$, we need to know the contraction rules of the mode unitary matrices.

\begin{proposition} [Computation rules for mode unitary network] The following are the computational rules for a unitary network:
\begin{enumerate} [label=(\roman*)]
    \item Parallel:  When two mode unitary matrix exist in parallel, the total unitary matrix is given by $U_{1} \oplus U_{2}$:
    \begin{align}
        \begin{array}{c}
        \tikzsetnextfilename{direct_sum_mode_unitary}
        \begin{tikzpicture}
        \begin{scope}  [local bounding box=graph_a, shift={(0,0)}]
        \node[draw, rectangle, rounded corners=2pt, minimum width=1.2cm, minimum height=0.6cm, fill=tensorcolor, thick] (U){$U_{1}$};
        \draw[mid arrow=0.6] ([xshift=-0.4cm]U.90) -- ++(0,0.4);
        \draw[mid arrow=0.6] (U.90) -- ++(0,0.4);
        \draw[mid arrow=0.6] ([xshift=0.4cm]U.90) -- ++(0,0.4);
        \draw[mid arrow=0.6] ([xshift=-0.4cm,yshift=-0.4cm]U.270) -- ++(0,0.4);
        \draw[mid arrow=0.6] ([yshift=-0.4cm]U.270) -- ++(0,0.4);
        \draw[mid arrow=0.6] ([xshift=0.4cm,yshift=-0.4cm]U.270) -- ++(0,0.4);
    \end{scope}
    \begin{scope}  [local bounding box=graph_b, shift={(1.6,0)}]
        \node[draw, rectangle, rounded corners=2pt, minimum width=1.2cm, minimum height=0.6cm, fill=tensorcolor, thick] (U){$U_{2}$};
        \draw[mid arrow=0.6] ([xshift=-0.4cm]U.90) -- ++(0,0.4);
        \draw[mid arrow=0.6] (U.90) -- ++(0,0.4);
        \draw[mid arrow=0.6] ([xshift=0.4cm]U.90) -- ++(0,0.4);
        \draw[mid arrow=0.6] ([xshift=-0.4cm,yshift=-0.4cm]U.270) -- ++(0,0.4);
        \draw[mid arrow=0.6] ([yshift=-0.4cm]U.270) -- ++(0,0.4);
        \draw[mid arrow=0.6] ([xshift=0.4cm,yshift=-0.4cm]U.270) -- ++(0,0.4);
    \end{scope}
    \end{tikzpicture}
    \end{array} = U_{1} \oplus U_{2} =  \left( \begin{array}{cc}
    U_{1}  & 0\\
      0  &  U_{2}
   \end{array}\right)
    \end{align} 
    \item Contraction: When one or more pairs of legs are contracted, the result is given by:
   \begin{equation}
\begin{aligned}
    \begin{array}{c}
    \tikzsetnextfilename{contraction_mode_unitary}
    \begin{tikzpicture}
        \node[draw, rectangle, rounded corners=2pt, minimum width=1cm, minimum height=0.6cm, fill=tensorcolor, thick] (U1){$U_{1}$};
        \node[draw, rectangle, rounded corners=2pt, minimum width=1cm, minimum height=0.6cm, fill=tensorcolor, thick] (U2) at (0.6,-1.2) {$U_{2}$};
        \draw[mid arrow=0.65] ([xshift=-0.3cm]U1.90) -- ++ (0,0.6);
        \draw[mid arrow=0.65] ([xshift=0.3cm]U1.90) -- ++ (0,0.6);
        \draw[mid arrow=0.65] ([xshift=-0.3cm,yshift=-0.6cm]U1.270) --  node[left] {$A$} ++(0,0.6);
        \draw[mid arrow=0.65] ([xshift=0.3cm,yshift=-0.6cm]U1.270) -- node[left] {$B$} ++(0,0.6);
         \draw[mid arrow=0.65] ([xshift=0.9cm,yshift=-0.6cm]U1.270) -- node[left] {$C$} ++(0,0.6);
          \draw[mid arrow=0.65] ([xshift=-0.3cm,yshift=-0.6cm]U2.270) -- ++ (0,0.6);
        \draw[mid arrow=0.65] ([xshift=0.3cm,yshift=-0.6cm]U2.270) -- ++ (0,0.6);
        \draw[mid arrow=0.65, densely dashed] ([shift={(-0.3,-1.8)}]U1.270) -- ++ (0,1.2);
        \draw[mid arrow=0.65, densely dashed] ([shift={(0.3,0.6)}]U2.90) -- ++ (0,1.2);
    \end{tikzpicture}
    \end{array} &= 
  U_{1} \times_{B} U_{2} \\
  &= (U_{1} \oplus I_{C}) \times (I_{A} \oplus U_{2})\\
  &=    \left(
     \begin{array}{c|c|c}
     U^{A}_{1}  & U^{B}_{1} & 0 \\
     \hline
     0 & 0 & I_{C}
   \end{array} 
\right) \left(
     \begin{array}{c|c}
     I_{A}  & 0  \\
     \hline
     0 & U^{B}_{2} \\
     \hline
     0 & U^{C}_{2}
   \end{array} 
\right) \\
&= \left(
     \begin{array}{c|c}
     U_{1}^{A}  & U_{1}^{B}U_{2}^{B}  \\
     \hline
     0 & U^{C}_{2} 
   \end{array} 
\right)
\end{aligned}
\end{equation}
where $U_{1} \times_{B} U_{2}$ means only leg $B$ is contracted. Observe that the matrix multiplication $U_{1}U_{2}$ can be viewed as a special case of contraction, wherein every incoming leg of $U_{1}$ is paired with each outgoing leg of $U_{2}$.
\end{enumerate}
\end{proposition}

\begin{definition} [Mode unitary network]
    At present, we limit our focus to mode unitary networks specifically for cases of DAG. Within these cases, a clear causal sequence for local unitary tensors can be identified. A mode unitary network is constructed by contracting mode unitary matrices according to the rules mentioned earlier.
\end{definition}

\begin{proposition} [Homomorphism between many-body unitary network and mode unitary network] 
A (projective) homomorphism can be established between the many-body Fermionic Gaussian unitary network and the mode unitary network, given that the graph $G[U_{Net}]$ is a DAG:
\begin{equation}
    \hat{\rho}(U_{Net}) = \pm \hat{U}_{Net}.
\end{equation}
This can be seen from the fact that
\begin{equation}
\begin{aligned}
   & \hat{\rho}(U_{1} \oplus U_{2}) = \hat{\rho}({U})_{1} \wedge \hat{\rho}({U}_{2}), \\
   &\hat{\rho}(U_{1} \times_{B} U_{2}) = \hat{\rho}(U_{1}) \times_{B} \hat{\rho}(U_{2}).
\end{aligned}
\end{equation}
\end{proposition}

In summary, as long as the graph lacks loops, we can repeatedly apply mode unitary contraction to build a mode unitary network.  This mode unitary network can then be transformed into a many-body unitary network. The transformation involves converting mode unitaries into many-body unitaries and translating mode matrix contractions into many-body tensor contractions.

\subsection{Cosine-sine Decomposition}

Considering a $N \times N$ unitary matrix $U \in \text{SU}(N)$. Any unitary matrix $U$ can be written into a 2-by-2 block matrix:
\begin{equation}
U = \left(
     \begin{array}{cc}
      U^{Aa}_{p\times q}  & U^{Ab}_{p \times (N-q)} \\
      U^{Ba}_{(N-p)\times q}  &  U^{Bb}_{(N-p)\times (N-q)}
   \end{array}
\right), 
\end{equation}
where we divide the outgoing basis into subsets $A$ and $B$, and split the incoming basis into $a$ and $b$, satisfying:
\begin{equation}
\begin{aligned}
    &|A| = p, \\
    &|a| = q, \\
   &|A|+|B| = |a| + |b| = N .
\end{aligned}
\end{equation}
For convenience we assume $q < p$ and $q+p < N$.
The \textit{cosine-sine decomposition} (CSD) of $U$ is given by \cite{davis1969some,davis1970rotation}
\begin{equation}
\begin{aligned}
     U &=  VRW^{\dagger} \\ &
     = \left(
     \begin{array}{cc}
    V^{A}  & 0\\
      0  &  V^{B}
   \end{array}
\right) 
\left(
     \begin{array}{c|ccc}
      C  & -S & 0 & 0\\
      0  &  0 & I_{p-q} & 0 \\
      \hline
      S & C & 0& 0    \\
       0& 0 & 0 & I_{N-p-q}           
   \end{array}
\right) 
\left(
     \begin{array}{cc}
      W^{a}  & 0\\
      0  &  W^{b}
   \end{array}
\right) ^{\dagger} ,
\end{aligned}
\end{equation}
where $V^{A}, V^{B}, W^{A},W^{B}$ are unitary matrices. 
\begin{equation}
    C = \text{diag}(c_{1},\cdots, c_{q}), \ S = \text{diag}(s_{1},\cdots, s_{q}),
\end{equation}
where $c_{i}$ and $s_{i}$ are real numbers satisfying
\begin{equation}
    c_{i}^{2} + s_{i}^{2} = 1.
\end{equation}
For a more detailed introduction of CSD, please refer to \cite{paige1994history,sutton2009computing}.

\subsection{Decompose a Fermionic Gaussian unitary}

The CSD offers a method to break down a global mode unitary matrix into a more efficient bilayer unitary network. Consider the scenario where the bipartition of both the input and output Hilbert spaces is identical:
\begin{equation}
    A = a, \ B =b.
\end{equation}
The CSD can be written as 
\begin{equation}
\begin{aligned}
    U &= VRW^{\dagger} \\
    &= \left(
     \begin{array}{cc}
    V_{A}  & 0\\
      0  &  V_{B}
   \end{array}
\right) 
\left(
     \begin{array}{cc|cc}
      I_{a}  & 0 & 0 & 0\\
      0  &  C & -S & 0 \\
      \hline
      0 & S & C& 0    \\
       0& 0 & 0 & I_{b}           
   \end{array}
\right) 
\left(
     \begin{array}{cc}
      W_{A}^{\dagger}  & 0\\
      0  &  W_{B}^{\dagger}
   \end{array}
\right).
\end{aligned}
\end{equation}
This can be visualized as:
\begin{align}
\begin{array}{c}
\tikzsetnextfilename{CSD_global_unitary}
    \begin{tikzpicture}
         \node[draw, rectangle, rounded corners=2pt, minimum width=2.4cm, minimum height=0.6cm, fill=tensorcolor, thick] (U){$U$};
         \draw[mid arrow=0.65] ([xshift=-1cm]U.90) -- ++(0,0.4);
         \draw[mid arrow=0.65] ([xshift=-0.6cm]U.90) -- ++(0,0.4);
         \draw[mid arrow=0.65] ([xshift=-0.2cm]U.90) -- ++(0,0.4);
         \draw[mid arrow=0.65] ([xshift=0.2cm]U.90) -- ++(0,0.4);
         \draw[mid arrow=0.65] ([xshift=0.6cm]U.90) -- ++(0,0.4);
         \draw[mid arrow=0.65] ([xshift=1cm]U.90) -- ++(0,0.4);
         \draw[mid arrow=0.65] 
         ([xshift=-1cm,yshift=-0.4cm]U.270) -- ++(0,0.4);
         \draw[mid arrow=0.65] 
         ([xshift=-0.6cm,yshift=-0.4cm]U.270) -- ++(0,0.4);
         \draw[mid arrow=0.65] 
         ([xshift=-0.2cm,yshift=-0.4cm]U.270) -- ++(0,0.4);
         \draw[mid arrow=0.65] 
         ([xshift=0.2cm,yshift=-0.4cm]U.270) -- ++(0,0.4);
         \draw[mid arrow=0.65] 
         ([xshift=0.6cm,yshift=-0.4cm]U.270) -- ++(0,0.4);
         \draw[mid arrow=0.65] 
         ([xshift=1cm,yshift=-0.4cm]U.270) -- ++(0,0.4);
    \end{tikzpicture}
\end{array}
= 
\begin{array}{c}
\tikzsetnextfilename{CSD_after_decomposition}
\begin{tikzpicture}
    \node[draw, rectangle, rounded corners=2pt, minimum width=0.2cm, minimum height=0.6cm, dashed] (Ia) at (-1.8, 0) {} node[left] at (Ia.180) {$I_{a}$};
    \node[draw, rectangle, rounded corners=2pt, minimum width=1cm, minimum height=0.6cm, dashed] (Ib) at (0.2, 0) {} node[right] at (Ib.0) {$I_{b}$};
    \node[draw, rectangle, rounded corners=2pt, minimum width=1.2cm, minimum height=0.6cm, fill=tensorcolor, thick] (R) at (-1, 0) {$R$};
    \node[draw, rectangle, rounded corners=2pt, minimum width=0.8cm, minimum height=0.6cm, fill=tensorcolor, thick] (VA) at (-1.6, 1.2) {$V_{A}$};
     \node[draw, rectangle, rounded corners=2pt, minimum width=1.6cm, minimum height=0.6cm, fill=tensorcolor, thick] (VB) at (0, 1.2) {$V_{B}$};
     \node[draw, rectangle, rounded corners=2pt, minimum width=0.8cm, minimum height=0.6cm, fill=tensorcolor, thick] (WA) at (-1.6, -1.2) {$W_{A}^{\dagger}$};
     \node[draw, rectangle, rounded corners=2pt, minimum width=1.6cm, minimum height=0.6cm, fill=tensorcolor, thick] (WB) at (0, -1.2) {$W_{B}^{\dagger}$};
     \draw[mid arrow=0.65] ([xshift=-0.2cm]VA.90) -- ++(0,0.6);
     \draw[mid arrow=0.65] ([xshift=0.2cm]VA.90) -- ++(0,0.6);
    \draw[mid arrow=0.65] ([xshift=-0.6cm]VB.90) -- ++(0,0.6);
     \draw[mid arrow=0.65] ([xshift=-0.2cm]VB.90) -- ++(0,0.6);
     \draw[mid arrow=0.65] ([xshift=0.2cm]VB.90) -- ++(0,0.6);
     \draw[mid arrow=0.65] ([xshift=0.6cm]VB.90) -- ++(0,0.6);
     \draw[mid arrow=0.65] ([xshift=-0.2cm,yshift=-0.6cm]WA.270) -- ++(0,0.6);
     \draw[mid arrow=0.65] ([xshift=0.2cm,yshift=-0.6cm]WA.270) -- ++(0,0.6);
     \draw[mid arrow=0.65] ([xshift=-0.6cm,yshift=-0.6cm]WB.270) -- ++(0,0.6);
     \draw[mid arrow=0.65] ([xshift=-0.2cm,yshift=-0.6cm]WB.270) -- ++(0,0.6);
     \draw[mid arrow=0.65] ([xshift=0.2cm,yshift=-0.6cm]WB.270) -- ++(0,0.6);
     \draw[mid arrow=0.65] ([xshift=0.6cm,yshift=-0.6cm]WB.270) -- ++(0,0.6);
     \draw[mid arrow=0.65] ([xshift=-0.2cm]WA.90) -- ([xshift=-0.2cm]VA.270);
     \draw[mid arrow=0.65] ([xshift=-0.2cm]WB.90) -- ([xshift=-0.2cm]VB.270);
     \draw[mid arrow=0.65] ([xshift=0.2cm]WB.90) -- ([xshift=0.2cm]VB.270);
     \draw[mid arrow=0.65] ([xshift=0.6cm]WB.90) -- ([xshift=0.6cm]VB.270);
     \draw[mid arrow=0.65] ([xshift=0.4cm]R.90) -- ([xshift=-0.6cm]VB.270);
     \draw[mid arrow=0.65] ([xshift=-0.4cm]R.90) -- ([xshift=0.2cm]VA.270);
      \draw[mid arrow=0.65] ([xshift=0.2cm]WA.90)  -- ([xshift=-0.4cm]R.270);
      \draw[mid arrow=0.65] ([xshift=-0.6cm]WB.90) -- ([xshift=0.4cm]R.270) ;
\end{tikzpicture}
\end{array}
\end{align}

The CSD identifies canonical bases $\{w_{i}^{A}\}$, $\{w_{i}^{B}\}$, $\{v^{A}_{i}\}$, and $\{v^{B}_{i}\}$ which correspond to subsystems A and B, ensuring they are paired in alignment, so that
\begin{equation}
    \begin{aligned}
        U\bm{w}^{A}_{i} &= c_{i}\bm{v}^{A}_{i} + s_{i}\bm{v}^{B}_{i}, \\
        U\bm{w}^{B}_{i} &= -s_{i}\bm{v}^{A}_{i} + c_{i}\bm{v}^{B}_{i}.
    \end{aligned}
\end{equation}
Observe that we now express the $R$ matrix as
\begin{equation}
    R = \left(
     \begin{array}{cc|cc}
      I_{a}  & 0 & 0 & 0\\
      0  &  C & -S & 0 \\
      \hline
      0 & S & C& 0    \\
       0& 0 & 0 & I_{b}           
   \end{array}
\right)  = I_{a} \oplus R^{reduced} \oplus I_{b}.
\end{equation}
Within the $R$ matrix, the presence of a $c_{i} = 1$ implies that the transformation of modes is localized in subsystem $A$ and $B$:
\begin{equation}
    \begin{aligned}
        U\bm{w}^{A}_{i} &= \bm{v}^{A}_{i},  \\
        U\bm{w}^{B}_{i} &=  \bm{v}^{B}_{i}.
    \end{aligned}
\end{equation}
For this particular pair of modes, the transformation does not occur across subsystems $A$ and $B$. The entries where $c_{i} = 1$ can be isolated and arranged into identity matrices $I_{a}$ and $I_{b}$. The process of extracting $c_{i}=1$ from matrix $R$ can be continued until $c_{i}<1$ for all pairs in $R$. This leads to the formation of a rotation matrix $R^{reduced}$ that precisely captures the rotational relationship between modes in $A$ and $B$.

Now we discuss how to decompose a global unitary matrix into mode unitary network. After performing CSD on $U$, get local unitary matrices $\{V_{A}, W_{A}^{\dagger},V_{B}, W_{B}^{\dagger},R \}$. we may contract $\{R, V_{B}, W_{B}^{\dagger}\}$ to form a new unitary $R^\prime$: 
\begin{align}
    \begin{array}{c}
    \tikzsetnextfilename{decomposition_CSD}
    \begin{tikzpicture}
    \node[draw, rectangle, rounded corners=2pt, minimum width=1.2cm, minimum height=0.6cm, fill=tensorcolor, thick] (R) at (-1, 0) {$R$};
    \node[draw, rectangle, rounded corners=2pt, minimum width=0.8cm, minimum height=0.6cm, fill=tensorcolor, thick] (VA) at (-1.6, 1.2) {$V_{A}$};
     \node[draw, rectangle, rounded corners=2pt, minimum width=1.6cm, minimum height=0.6cm, fill=tensorcolor, thick] (VB) at (0, 1.2) {$V_{B}$};
     \node[draw, rectangle, rounded corners=2pt, minimum width=0.8cm, minimum height=0.6cm, fill=tensorcolor, thick] (WA) at (-1.6, -1.2) {$W_{A}^{\dagger}$};
     \node[draw, rectangle, rounded corners=2pt, minimum width=1.6cm, minimum height=0.6cm, fill=tensorcolor, thick] (WB) at (0, -1.2) {$W_{B}^{\dagger}$};
     \draw[mid arrow=0.65] ([xshift=-0.2cm]VA.90) -- ++(0,0.6);
     \draw[mid arrow=0.65] ([xshift=0.2cm]VA.90) -- ++(0,0.6);
    \draw[mid arrow=0.65] ([xshift=-0.6cm]VB.90) -- ++(0,0.6);
     \draw[mid arrow=0.65] ([xshift=-0.2cm]VB.90) -- ++(0,0.6);
     \draw[mid arrow=0.65] ([xshift=0.2cm]VB.90) -- ++(0,0.6);
     \draw[mid arrow=0.65] ([xshift=0.6cm]VB.90) -- ++(0,0.6);
     \draw[mid arrow=0.65] ([xshift=-0.2cm,yshift=-0.6cm]WA.270) -- ++(0,0.6);
     \draw[mid arrow=0.65] ([xshift=0.2cm,yshift=-0.6cm]WA.270) -- ++(0,0.6);
     \draw[mid arrow=0.65] ([xshift=-0.6cm,yshift=-0.6cm]WB.270) -- ++(0,0.6);
     \draw[mid arrow=0.65] ([xshift=-0.2cm,yshift=-0.6cm]WB.270) -- ++(0,0.6);
     \draw[mid arrow=0.65] ([xshift=0.2cm,yshift=-0.6cm]WB.270) -- ++(0,0.6);
     \draw[mid arrow=0.65] ([xshift=0.6cm,yshift=-0.6cm]WB.270) -- ++(0,0.6);
     \draw[mid arrow=0.65] ([xshift=-0.2cm]WA.90) -- ([xshift=-0.2cm]VA.270);
     \draw[mid arrow=0.65] ([xshift=-0.2cm]WB.90) -- ([xshift=-0.2cm]VB.270);
     \draw[mid arrow=0.65] ([xshift=0.2cm]WB.90) -- ([xshift=0.2cm]VB.270);
     \draw[mid arrow=0.65] ([xshift=0.6cm]WB.90) -- ([xshift=0.6cm]VB.270);
     \draw[mid arrow=0.65] ([xshift=0.4cm]R.90) -- ([xshift=-0.6cm]VB.270);
     \draw[mid arrow=0.65, red!70] ([xshift=-0.4cm]R.90) -- ([xshift=0.2cm]VA.270);
      \draw[mid arrow=0.65, red!70] ([xshift=0.2cm]WA.90)  -- ([xshift=-0.4cm]R.270);
      \draw[mid arrow=0.65] ([xshift=-0.6cm]WB.90) -- ([xshift=0.4cm]R.270) ;
    \end{tikzpicture}
    \end{array}
    =
     \begin{array}{c}
     \tikzsetnextfilename{decomposition_UN}
    \begin{tikzpicture}
        \node[draw, rectangle, rounded corners=2pt, minimum width=0.8cm, minimum height=0.6cm, fill=tensorcolor, thick] (VA) at (-0.6, 0.6) {$V_{A}$};
        \node[draw, rectangle, rounded corners=2pt, minimum width=0.8cm, minimum height=0.6cm, fill=tensorcolor, thick] (WA) at (-0.6, -0.6) {$W_{A}^{\dagger}$};
        \node[draw, rectangle, rounded corners=2pt, minimum width=1.6cm, minimum height=1.8cm, fill=tensorcolor, thick] (Up) at (1.2, 0) {$U^\prime$};
        \draw[mid arrow=0.65] ([xshift=-0.2cm]VA.90) -- ++(0,0.6);
        \draw[mid arrow=0.65] ([xshift=0.2cm]VA.90) -- ++(0,0.6);
        \draw[mid arrow=0.65]
        ([xshift=-0.2cm,yshift=-0.6cm]WA.270) -- ++(0,0.6);
        \draw[mid arrow=0.65]
        ([xshift=0.2cm,yshift=-0.6cm]WA.270) -- ++(0,0.6);
        \draw[mid arrow=0.65] ([xshift=-0.6cm]Up.90) -- ++(0,0.6);
        \draw[mid arrow=0.65] ([xshift=-0.2cm]Up.90) -- ++(0,0.6);
        \draw[mid arrow=0.65] ([xshift=0.2cm]Up.90) -- ++(0,0.6);
        \draw[mid arrow=0.65] ([xshift=0.6cm]Up.90) -- ++(0,0.6);
        \draw[mid arrow=0.65]
        ([xshift=-0.6cm,yshift=-0.6cm]Up.270) -- ++(0,0.6);
        \draw[mid arrow=0.65]
        ([xshift=-0.2cm,yshift=-0.6cm]Up.270) -- ++(0,0.6);
        \draw[mid arrow=0.65]
        ([xshift=0.2cm,yshift=-0.6cm]Up.270) -- ++(0,0.6);
        \draw[mid arrow=0.65]
        ([xshift=0.6cm,yshift=-0.6cm]Up.270) -- ++(0,0.6);
        \draw[mid arrow=0.65]
        (WA.90) -- (VA.270);
        \draw[mid arrow=0.65,red!70]
        (WA.0) -- ([yshift=-0.6cm]Up.180);
        \draw[mid arrow=0.65,red!70]
        ([yshift=0.6cm]Up.180) -- (VA.0);
    \end{tikzpicture}
    \end{array}.
\end{align}
Notice that the new unitary $U^\prime$ is in defined on smaller number modes:
\begin{equation}
    U^\prime \in U(M), \quad M \le N.
\end{equation}
The equality holds when no local transformed modes are found in $A$. The bond dimension pertains to the number of modes that undergo non-local transformations and mixed between subsystem $A$ and $B$.

Through repeated application of the CSD, a unitary can be broken down into a bilayer unitary network, ensuring the smallest possible bond dimension at each stage.  

The bond dimensions can be further reduced by assigning a value of 1 to some $c_{i} < 1$, enabling additional local mode transformations. In this way, the unitary network $U_{Net}$ serves as an approximation of the global unitary $U$.

\section{Conclusions and Outlook}

In this paper, we introduce a tensor network architecture called the unitary network, designed to represent global unitary operations. The unitarity of local tensors, which is simple to enforce, combined with the directed acyclic nature of the graph ensures that the corresponding global operator is unitary. Unitary networks are capable of representing a wider class than uniform-MPUs, enabling the representation of unitaries that do not preserve locality.

Moreover, the structure of the unitary network naturally provides a flow index, which can be interpreted as the flow of information. With locality-preserving unitary network representation for a 1D QCA, the net information flow aligns with the GNVW index. However, for a unitary operator $U$ that is not locality-preserving, its implementation is inherently non-local. In such cases, it is not straightforward how to define an intrinsic net entropy flow.

We demonstrated the connection between unitary networks and SQCs. SQCs, incorporating local unitary gates and a DAG structure, can be regarded as unitary networks. For finite OBC systems, we demonstrate that an SQC and a unitary network can be transformed into each other. For PBC and infinite OBC systems, SQCs and unitary networks have equivalent representability. Unitary networks, which permit non-zero information flow, can more efficiently represent global unitary operations with non-zero net information flow.

We explored the decomposition of a specified unitary into a more efficient unitary network, marked by reduced bond dimensions. For Fermionic Gaussian unitaries, CSD provides a method for effective decomposition. Nevertheless, when it comes to general many-body unitaries, the algorithm for decomposition requires further investigation.

Our work opens new directions for future studies.
First, we can incorporate symmetries into our unitary networks. As with the classification of symmetry-protected topological phases, it is well known that symmetry enriches the classification of MPUs~\cite{cirac2017matrix, Gong_PRL_2020, Carolyn_SPTinv}.
In Refs.~\cite{Gong_PRL_2020, Carolyn_SPTinv}, topological indices are also introduced.
It is interesting to investigate how the net information flow introduced in this work relates to these topological indices in the presence of symmetries. 
Second, unitary network architectures, such as bilayer unitary networks and structures similar to the Margolus partitioning scheme, can be extended to higher than one dimension, potentially encompassing a broad range of unitary operations in higher dimensions. The properties of these higher-dimensional unitary networks will be explored in future work. Lastly, the conditions for global unitarity for a unitary network are easy to impose and therefore this representation could serve as a promising way to find tensor-network approximation of a desired global unitary operators through variational optimization.

\begin{acknowledgments}
    We thank Zongping Gong for helpful discussions.
    We thank the Yukawa Institute for Theoretical Physics at Kyoto University for fruitful discussions during the YITP workshop YITP-T-24-03 on ``Recent Developments and Challenges in Topological Phases'' and YITP-I-25-02 on ``Recent Developments and Challenges in Tensor Networks: Algorithms, Applications to Science, and Rigorous Theories.''
    This work is supported by the Hong Kong Research Grants Council (GRF 16308822), the Croucher Foundation (CIA23SC01), and the Fei Chi En Education and Research Fund.
    S.O.~was supported by RIKEN Special Postdoctoral Researchers Program, RIKEN Quantum, and KAKENHI Grant No.~25K17322 from the Japan Society for the Promotion of Science (JSPS).
\end{acknowledgments}

\appendix

\section{Splitting or merging legs}
\label{Append: Splitting or merging legs}

Like with any tensor, we can split or merge the legs attached to a unitary tensor. Legs can only be merged if they share the same source and destination. Consider the following example:
\begin{align}
    \begin{array}{c}
    \tikzsetnextfilename{legs_merged}
    \begin{tikzpicture}
        \node[draw, rectangle, rounded corners=2pt, minimum width=0.6cm, minimum height=0.6cm, fill=tensorcolor, thick] (U) {$U$};
        \draw[mid arrow=0.65] ([xshift=-0.6cm]U.180) -- node[above] {$i$} (U.180);
         \draw[mid arrow=0.65] ([yshift=-0.6cm]U.270) -- node[left] {$j$} (U.270);
         \draw[mid arrow=0.65,very thick] (U.90) -- node[left] {$mn$} ++(0,0.6);
    \end{tikzpicture}
    \end{array}=
    \begin{array}{c}
    \tikzsetnextfilename{after_splitting}
    \begin{tikzpicture}
        \node[draw, rectangle, rounded corners=2pt, minimum width=0.6cm, minimum height=0.6cm, fill=tensorcolor, thick] (U) {$U$};
        \draw[mid arrow=0.65] ([xshift=-0.6cm]U.180) -- node[above] {$i$} (U.180);
         \draw[mid arrow=0.65] ([yshift=-0.6cm]U.270) -- node[left] {$j$} (U.270);
         \draw[mid arrow=0.65] ([xshift=-0.1cm]U.90) -- node[left] {$m$} ++(0,0.6);
         \draw[mid arrow=0.65] ([xshift=0.1cm]U.90) -- node[right] {$n$} ++(0,0.6);
    \end{tikzpicture}
    \end{array},
\end{align}
leg $mn$ is split into leg $m$ and leg $n$. This split operation can be denoted as:
\begin{equation}
    U^{(mn)}_{ij} = U^{mn}_{ij}.
\end{equation}
Splitting a leg means decomposing the Hilbert space carried by the leg. In the above example:
\begin{equation}
    \mathcal{H}_{mn} = \mathcal{H}_{m} \otimes \mathcal{H}_{n},
\end{equation}
where $\mathcal{H}_{mn}$, $\mathcal{H}_{m}$ and  $\mathcal{H}_{n}$ are Hilbert space carried by leg $(mn)$, $m$ and $n$ separately. The relationship concerning the dimension of Hilbert spaces is expressed by
\begin{equation}
    \dim \mathcal{H}_{mn} = \dim \mathcal{H}_{m} \cdot \dim \mathcal{H}_{n}.
\end{equation}

\section{Unitary tensor contraction}
\label{Append: Unitary tensor contraction}

The process of contracting local unitary tensors is akin to the contraction of general tensors. Nevertheless, there is an additional restriction for contracting unitary tensors: the incoming leg must be contracted with an outgoing leg. Below is an example of unitary tensor contraction:
\begin{align}
\begin{array}{c}
\tikzsetnextfilename{UN_before_contraction}
\begin{tikzpicture}
     \node[draw, rectangle, rounded corners=2pt, minimum width=0.6cm, minimum height=0.6cm, fill=tensorcolor, thick] (A) {$A$};
      \node[draw, rectangle, rounded corners=2pt, minimum width=0.6cm, minimum height=0.6cm, fill=tensorcolor, thick] (B) at (1.2,0) {$B$};
    \draw[mid arrow=0.65](A.90) -- node[right] {$i$} ++(0,0.6) ;
    \draw[mid arrow=0.65](A.0) -- node[above] {$j$} (B.180);
    \draw[mid arrow=0.65] ([xshift=-0.6cm]A.180) -- node[above] {$k$} (A.180);
    \draw[mid arrow=0.65] ([yshift=-0.6cm]A.270) -- node[right] {$l$} (A.270);
    \draw[mid arrow=0.65](B.90) -- node[right] {$m$} ++(0,0.6);
    \draw[mid arrow=0.65] (B.0) -- node[above] {$n$} ++(0.6, 0);
    \draw[mid arrow=0.65] ([yshift=-0.6cm]B.270) -- node[right] {$h$} (B.270);
\end{tikzpicture}
\end{array} = 
\begin{array}{c}
\tikzsetnextfilename{UN_contraction}
\begin{tikzpicture}
     \node[draw, rectangle, rounded corners=2pt, minimum width=0.8cm, minimum height=0.6cm, fill=tensorcolor, thick] (A) {$A$};
      \node[draw, rectangle, rounded corners=2pt, minimum width=0.8cm, minimum height=0.6cm, fill=tensorcolor, thick] (B) at (0.4,1.2) {$B$};
    \draw[mid arrow=0.65]([xshift=-0.2cm]A.90) -- node[left] {$i$} ++(0,0.6) ;
    \draw[mid arrow=0.65]([xshift=0.2cm]A.90) -- node[right] {$j$} ([xshift=-0.2cm]B.270);
    \draw[mid arrow=0.65] ([xshift=-0.2cm,yshift=-0.6cm]A.270) -- node[left] {$k$} ([xshift=-0.2cm]A.270);
    \draw[mid arrow=0.65] ([xshift=0.2cm,yshift=-0.6cm]A.270) -- node[right] {$l$} ([xshift=0.2cm]A.270);
    \draw[mid arrow=0.65]([xshift=-0.2cm]B.90) -- node[left] {$m$} ++(0,0.6);
    \draw[mid arrow=0.65] ([xshift=0.2cm]B.90) -- node[right] {$n$} ++(0, 0.6);
    \draw[mid arrow=0.65] ([xshift=0.2cm,yshift=-0.6cm]B.270) -- node[right] {$h$} ([xshift=0.2cm]B.270);
\end{tikzpicture}
\end{array}
=
\begin{array}{c}
\tikzsetnextfilename{UN_after_contraction}
\begin{tikzpicture}
     \node[draw, rectangle, rounded corners=2pt, minimum width=0.8cm, minimum height=0.6cm, fill=tensorcolor, thick] (C) {$C$};
    \draw[mid arrow=0.65]([xshift=-0.2cm]C.90) -- node[left] {$i$} ++(0,0.6);
    \draw[mid arrow=0.65] ([xshift=-0.6cm]C.180) -- node[above] {$k$} (C.180);
    \draw[mid arrow=0.65] ([yshift=-0.6cm,xshift=-0.2cm]C.270) -- node[left] {$l$} ([xshift=-0.2cm]C.270);
    \draw[mid arrow=0.65]([xshift=0.2cm]C.90) -- node[right] {$m$} ++(0,0.6);
    \draw[mid arrow=0.65] (C.0) -- node[above] {$n$} ++(0.6, 0);
    \draw[mid arrow=0.65] ([xshift=0.2cm, yshift=-0.6cm]C.270) -- node[right] {$h$} ([xshift=0.2cm]C.270);
\end{tikzpicture}
\end{array}.
\end{align}
The above diagram reads:
\begin{equation}
    A^{ij}_{kl} B^{mn}_{jh} = C^{imn}_{klh},
\end{equation}
where the Einstein summation convention is assumed, in which any index that appears exactly twice in the equation is summed over. Observe that in the example given above, $C^{imn}_{klh}$ remains a local unitary tensor because the contraction of $A$ and $B$ can be seen as performing the local unitary $A$ first, followed by the unitary $B$. 

However, the tensor contraction between local unitary tensors will not produce a unitary tensor in general. Considering the following cases:
\begin{align}
\begin{array}{c}
\tikzsetnextfilename{UN_before_contraction_general}
\begin{tikzpicture}
     \node[draw, rectangle, rounded corners=2pt, minimum width=0.6cm, minimum height=0.6cm, fill=tensorcolor, thick] (A) {$A$};
      \node[draw, rectangle, rounded corners=2pt, minimum width=0.6cm, minimum height=0.6cm, fill=tensorcolor, thick] (B) at (1.2,0) {$B$};
    \draw[mid arrow=0.65](A.90) -- node[right] {$i$} ++(0,0.6);
    \draw[mid arrow=0.65]([yshift=0.1cm]A.0) -- node[above] {$j$} ([yshift=0.1cm]B.180);
    \draw[mid arrow=0.65]([yshift=-0.1cm]B.180) --  node[below] {$f$} ([yshift=-0.1cm]A.0);
    \draw[mid arrow=0.65] ([xshift=-0.6cm]A.180) -- node[above] {$k$} (A.180);
    \draw[mid arrow=0.65] ([yshift=-0.6cm]A.270) -- node[right] {$l$} (A.270);
    \draw[mid arrow=0.65](B.90) -- node[right] {$m$} ++(0,0.6);
    \draw[mid arrow=0.65] (B.0) -- node[above] {$n$} ++(0.6, 0);
    \draw[mid arrow=0.65] ([yshift=-0.6cm]B.270) -- node[right] {$h$} (B.270);
\end{tikzpicture}
\end{array} = 
\begin{array}{c}
\tikzsetnextfilename{UN_after_contraction_general}
\begin{tikzpicture}
     \node[draw, rectangle, rounded corners=2pt, minimum width=0.8cm, minimum height=0.6cm, fill=tensorcolor, thick] (C) {$C$};
    \draw[mid arrow=0.65]([xshift=-0.2cm]C.90) -- node[left] {$i$} ++(0,0.6) ;
    \draw[mid arrow=0.65] ([xshift=-0.6cm]C.180) -- node[above] {$k$} (C.180);
    \draw[mid arrow=0.65] ([yshift=-0.6cm,xshift=-0.2cm]C.270) -- node[left] {$l$} ([xshift=-0.2cm]C.270);
    \draw[mid arrow=0.65]([xshift=0.2cm]C.90) -- node[right] {$m$} ++(0,0.6);
    \draw[mid arrow=0.65] (C.0) -- node[above] {$n$} ++(0.6, 0);
    \draw[mid arrow=0.65] ([xshift=0.2cm, yshift=-0.6cm]C.270) -- node[right] {$h$} ([xshift=0.2cm]C.270);
\end{tikzpicture}
\end{array}.
\end{align}
We can still perform the tensor contraction:
\begin{equation}
    A^{ij}_{klf} B^{fmn}_{jh} = C^{imn}_{klh}.
\end{equation}
In the above case, $C$ is typically not a unitary tensor. The reason lies in the contraction of $j: A \rightarrow B$ and $f: B \rightarrow A$, as there is no time sequence established for the tensors $A$ and $B$. We demonstrate in Proposition \ref{prop: Unitarity of unitary network} that a unitary network generates a global unitary tensor if there is no directed loop within the graph.

\section{Reversed Stacked-CNOT}
\label{Append: Reversed Stacked-CNOT}
An alternative of stacked CNOT is the reverse alignment of it, where the "not" operation precedes the "control" for each qubit:
\begin{align}
\tikzsetnextfilename{stacked_CNOT_reversed}
\begin{tikzpicture}
    \filldraw[black] (-0.6,0) circle (2pt) ;  
    \filldraw[black] (0,0) circle (2pt) ;
    \filldraw[black] (0.6,0) circle (2pt);
    \filldraw[black] (1.2,0) circle (2pt);
    \draw (-0.6, -1.2) -- (-0.6, 0.6);
    \draw (0, -1.2) -- (0, 0.6);
    \draw (0.6, -1.2) -- (0.6, 0.6);
    \draw (1.2, -1.2) -- (1.2, 0.6);
    \node[draw, circle, minimum size=0.3cm] (a) at (-0.6,-0.6) {};
    \node[draw, circle, minimum size=0.3cm] (b) at (0,-0.6) {};
    \node[draw, circle, minimum size=0.3cm] (c) at (0.6,-0.6) {};
    \node[draw, circle, minimum size=0.3cm] (d) at (1.2,-0.6) {};
    \draw (a.180) -- (a.0);
    \draw (b.180) -- (b.0);
    \draw (c.180) -- (c.0);
    \draw (d.180) -- (d.0);
    \draw[dashed] (-1.2,0) to [out=0,in=180] (a.180);
    \draw (-0.6,0) to [out=0,in=180] (b.180);
    \draw (0,0) to [out=0,in=180] (c.180);
    \draw (0.6,0) to [out=0,in=180] (d.180);
    \draw[dashed] (1.2,0) to [out=0,in=180] ([xshift=0.6cm]d.180);
\end{tikzpicture}
\end{align}
\begin{equation}
\begin{aligned}
StC_{N}^{rev} &=  \prod_{n=-N}^{N} C_{n} = C_{N}\cdots C_{n+1} C_{n} C_{n-1} \cdots  C_{-N}\\
    StC^{rev} &= \lim_{N \rightarrow +\infty} StC_{N}^{rev} = \cdots C_{n+1} C_{n} C_{n-1} \cdots
\end{aligned}
\end{equation}
Here we simply list its algebra dynamics:
\begin{equation}
\begin{aligned}
   & StC_{N}^{rev} X_{n} StC_{N}^{rev \dagger} = X_{n} \cdots X_{N}, \\
   & StC_{N}^{rev} Y_{n} StC_{N}^{rev\dagger} = Z_{n-1}Y_{n}X_{n+1} \cdots X_{N},  \\
   & StC_{N}^{rev} Z_{n} StC_{N}^{rev\dagger} = Z_{n-1}  Z_{n}.
\end{aligned}
\end{equation}

\section{Margolus Partitioning is not a quantum circuit}
\label{Append: Margolus Partitioning is not a quantum circuit}

In the main section, we demonstrated that the Margolus partitioning approach of a QCA forms a unitary network. From the diagram \eqref{fig:margolus partitioning}, it is attempting to recognize the Margolus partitioning scheme of QCA as an FDLU circuit. However, this is not true in general. It is well recognized that QCAs with a non-trivial GNVW index are not representable by an FDLU circuit \cite{Gross_2012}. The key lies in the fact that $\mathcal{B}_{2m}$ and $\mathcal{B}_{2m+1}$ can differ in dimension from $\mathcal{A}_{2m}$.

Consider the following example: At each site, we define $\mathcal{A}_{n} \cong \mathcal{A}_{\text{qudit}}^{\otimes2}$, representing the algebra of two qudits. Moreover, we assert $\mathcal{B}_{2m} \cong \mathcal{A}_{\text{qudit}}$, indicating a single qubit algebra, and $\mathcal{B}_{2m+1} \cong \mathcal{A}_{\text{qudit}}^{\otimes 3}$, for the algebra of three qudits.

\begin{align}
\tikzsetnextfilename{margolus_partitioning_shift}
    \begin{tikzpicture}
        \node[draw, rectangle, rounded corners=2pt, minimum width=1.0cm, minimum height=0.6cm, thick, dashed, font=\scriptsize]  at (-0.6, -1.8) {$A_{2m}$};
        \node[draw, rectangle, rounded corners=2pt, minimum width=1.0cm, minimum height=0.6cm,  thick, dashed, font=\scriptsize]  at (0.6, -1.8) {$A_{2m+1}$};
        \node[draw, rectangle, rounded corners=2pt, minimum width=1.0cm, minimum height=0.6cm, thick, dashed, font=\scriptsize]  at (-1.2, 1.8) {$A_{2m}$};
        \node[draw, rectangle, rounded corners=2pt, minimum width=1.0cm, minimum height=0.6cm,  thick, dashed, font=\scriptsize]  at (0, 1.8) {$A_{2m+1}$};
      \node[draw, rectangle, rounded corners=2pt, minimum width=2.2cm, minimum height=0.6cm, fill=tensorcolor, thick] (Wm-1) at (-2.4, -0.6) {$W_{m-1}$};
         \node[draw, rectangle, rounded corners=2pt, minimum width=2.2cm, minimum height=0.6cm, fill=tensorcolor, thick] (Wm) at (0, -0.6) {$W_{m}$};
         \node[draw, rectangle, rounded corners=2pt, minimum width=2.2cm, minimum height=0.6cm, fill=tensorcolor, thick] (Wm+1) at (2.4, -0.6) {$W_{m+1}$};
         \node[draw, rectangle, rounded corners=2pt, minimum width=2.2cm, minimum height=0.6cm, fill=tensorcolor, thick] (Vm-1) at (-4.2, 0.6) {$V_{m-1}$};
         \node[draw, rectangle, rounded corners=2pt, minimum width=2.2cm, minimum height=0.6cm, fill=tensorcolor, thick] (Vm) at (-1.8, 0.6) {$V_{m}$};
         \node[draw, rectangle, rounded corners=2pt, minimum width=2.2cm, minimum height=0.6cm, fill=tensorcolor, thick] (Vm+1) at (0.6, 0.6) {$V_{m+1}$};
         \draw 
         ([xshift=-0.9cm,yshift=-0.6cm]Wm-1.270) -- ++(0,0.6);
         \draw 
         ([xshift=-0.3cm,yshift=-0.6cm]Wm-1.270) -- ++(0,0.6);
         \draw
         ([xshift=0.3cm,yshift=-0.6cm]Wm-1.270) -- ++(0,0.6);
         \draw 
         ([xshift=0.9cm,yshift=-0.6cm]Wm-1.270) -- ++(0,0.6);
         \draw 
         ([xshift=-0.9cm,yshift=-0.6cm]Wm.270) -- ++(0,0.6);
         \draw 
         ([xshift=-0.3cm,yshift=-0.6cm]Wm.270) -- ++(0,0.6);
         \draw
         ([xshift=0.3cm,yshift=-0.6cm]Wm.270) -- ++(0,0.6);
         \draw 
         ([xshift=0.9cm,yshift=-0.6cm]Wm.270) -- ++(0,0.6);
         \draw 
         ([xshift=-0.9cm,yshift=-0.6cm]Wm+1.270) -- ++(0,0.6);
         \draw 
         ([xshift=-0.3cm,yshift=-0.6cm]Wm+1.270) -- ++(0,0.6);
         \draw
         ([xshift=0.3cm,yshift=-0.6cm]Wm+1.270) -- ++(0,0.6);
         \draw 
         ([xshift=0.9cm,yshift=-0.6cm]Wm+1.270) -- ++(0,0.6);
         \draw 
         ([xshift=-0.9cm]Wm-1.90) -- ++(0,0.6);
         \draw 
         ([xshift=-0.3cm]Wm-1.90) -- ++(0,0.6);
         \draw 
         ([xshift=0.3cm]Wm-1.90) -- ++(0,0.6);
         \draw 
         ([xshift=0.9cm]Wm-1.90) -- ++(0,0.6);
          \draw 
         ([xshift=-0.9cm]Wm.90) -- ++(0,0.6);
         \draw 
         ([xshift=-0.3cm]Wm.90) -- ++(0,0.6);
         \draw 
         ([xshift=0.3cm]Wm.90) -- ++(0,0.6);
         \draw 
         ([xshift=0.9cm]Wm.90) -- ++(0,0.6);
          \draw 
         ([xshift=-0.9cm]Wm+1.90) -- ++(0,0.6);
         \draw 
         ([xshift=-0.3cm]Wm+1.90) -- ++(0,0.6);
         \draw 
         ([xshift=0.3cm]Wm+1.90) -- ++(0,0.6);
         \draw 
         ([xshift=0.9cm]Wm+1.90) -- ++(0,0.6);
         \draw 
         ([xshift=-0.9cm]Vm-1.90) -- ++(0,0.6);
         \draw 
         ([xshift=-0.3cm]Vm-1.90) -- ++(0,0.6);
         \draw 
         ([xshift=0.3cm]Vm-1.90) -- ++(0,0.6);
         \draw 
         ([xshift=0.9cm]Vm-1.90) -- ++(0,0.6);
          \draw 
         ([xshift=-0.9cm]Vm.90) -- ++(0,0.6);
         \draw 
         ([xshift=-0.3cm]Vm.90) -- ++(0,0.6);
         \draw 
         ([xshift=0.3cm]Vm.90) -- ++(0,0.6);
         \draw 
         ([xshift=0.9cm]Vm.90) -- ++(0,0.6);
          \draw 
         ([xshift=-0.9cm]Vm+1.90) -- ++(0,0.6);
         \draw 
         ([xshift=-0.3cm]Vm+1.90) -- ++(0,0.6);
         \draw 
         ([xshift=0.3cm]Vm+1.90) -- ++(0,0.6);
         \draw 
         ([xshift=0.9cm]Vm+1.90) -- ++(0,0.6);
    \end{tikzpicture}
\end{align}
Treating $V_{m}$ and $W_{m}$ as quantum gates, as we did in (\ref{Eq: tensor=gate}), reveals that $\mathcal{A}_{2m}$ in the input and output aligns with different qudit pairs. Indeed, as illustrated in the preceding diagram, $\mathcal{A}_{2m}$ shifts one qubit leftward in the output. To restore $\mathcal{A}_{2m}$ to the initial position, an additional shift operation is required, which cannot be performed by an FDLU circuit. 

Consequently, the Margolus partitioning of QCA can be described as an FDLU circuit complemented by a shift. A unitary network surpasses a quantum circuit in an infinite system as its local unitary tensors act as local unitary gates while the links facilitate unrestricted Hilbert space transfer. It can be seen as an extension of a quantum circuit.

\section{Non-uniform net information flow in 1D system}

\label{Append: Non-uniform net information flow in 1D system}

In proving that the net information flow is independent of the chosen surface for Properties \ref{prop: uniform net information flow}, we implicitly assume that the incoming and outgoing Hilbert space at each site are equivalent. Although this assumption is generally valid for most physical situations, a unitary network can exemplify exceptions.

In a one-dimensional system, a non-uniform information flow signifies a change in information density distribution. Consider a hypothetical scenario where such an information redistribution takes place:
\begin{align}
\tikzsetnextfilename{non_uniform_information_flow}
    \begin{tikzpicture}
        \draw[dashed] (-0.7, -1.2) -- (-0.7,1.2) node[above] {$\Sigma_{1}$};
        \draw[dashed] (0.7, -1.2) -- (0.7,1.2) node[above] {$\Sigma_{2}$};
        \node[draw, rectangle, rounded corners=2pt, minimum width=0.8cm, minimum height=0.8cm, fill=tensorcolor, thick] (a) at (-1.4,0) {};
        \node[draw, rectangle, rounded corners=2pt, minimum width=0.8cm, minimum height=0.8cm, fill=tensorcolor, thick] (b) at (0,0) {};
        \node[draw, rectangle, rounded corners=2pt, minimum width=0.8cm, minimum height=0.8cm, fill=tensorcolor, thick] (c) at (1.4,0) {};
        \draw[mid arrow=0.65] ([xshift=-0.6cm]a.180) -- ++(0.6,0);
        \draw[mid arrow=0.65] ([xshift=-0.2cm,yshift=-0.6cm]a.270) -- ++(0,0.6);
        \draw[mid arrow=0.65,red!70] ([xshift=0.2cm,yshift=-0.6cm, ]a.270) -- ++(0,0.6);
        \draw[mid arrow=0.65] (a.90) -- ++(0,0.6);
        \draw[mid arrow=0.65] ([yshift=-0.6cm]b.270) -- ++(0,0.6);
        \draw[mid arrow=0.65,red!70] ([xshift=-0.2cm]b.90) -- ++(0,0.6);
        \draw[mid arrow=0.65] ([xshift=0.2cm]b.90) -- ++(0,0.6);
         \draw[mid arrow=0.65] ([yshift=-0.6cm]c.270) -- ++(0,0.6);
        \draw[mid arrow=0.65] (c.90) -- ++(0,0.6);
        \draw[mid arrow=0.65] (c.0) -- ++(0.6,0);
        \draw[mid arrow=0.65,red!70] ([yshift=0.2cm]a.0) -- ([yshift=0.2cm]b.180);
        \draw[mid arrow=0.65] ([yshift=-0.2cm]a.0) -- ([yshift=0-.2cm]b.180);
         \draw[mid arrow=0.65] (b.0) -- (c.180);
    \end{tikzpicture}
\end{align}
In the diagram provided above, the trajectory of an impurity mode and its corresponding Hilbert space are marked by the red-highlighted legs. The net information flow differs for $\Sigma_{1}$ and $\Sigma_{2}$:
\begin{equation}
    f_{I}[\Sigma_{1}] \ne f_{I}[\Sigma_{2}].
\end{equation}
The dependency of net information flow on the chosen surface prevents a consistent definition of a GNVW index.

\section{Redundant net information flow of finite PBC unitary networks}
\label{Append: Redundant net information flow of finite PBC unitary networks}

This appendix demonstrates how various network implementations of a PBC unitary operator can lead to distinct net information flows. To illustrate, consider two distinct unitary network implementations crafted for the PBC shift operation in the following:
\begin{enumerate} [label=(\roman*)]
    \item A locality-preserving unitary network with non-trivial net information flow $f_{I} = 1$ (This can correspond to the index $\log_{d} I_{GNVW}$ additive in $\mathbb{Z}_{L} = \mathbb{Z}/L\mathbb{Z}$, where $L$ is the system size) : 
    \begin{align} \label{fig:shift}
    \tikzsetnextfilename{shift_PBC}
    \begin{tikzpicture} 
    \node[draw, rectangle, rounded corners=2pt, minimum width=0.6cm, minimum height=0.6cm, fill=tensorcolor, thick] (a) at (0,0) {};
    \node[draw, rectangle, rounded corners=2pt, minimum width=0.6cm, minimum height=0.6cm, fill=tensorcolor, thick] (b) at (1.2,0) {};
    \node[draw, rectangle, rounded corners=2pt, minimum width=0.6cm, minimum height=0.6cm, fill=tensorcolor, thick] (c) at (2.4,0) {};
    \node[draw, rectangle, rounded corners=2pt, minimum width=0.6cm, minimum height=0.6cm, fill=tensorcolor, thick] (d) at (3.6,0) {};
    \draw[mid arrow=0.65] (a.90) -- ++(0,0.6);
    \draw[mid arrow=0.65] ([yshift=-0.6cm]a.270) -- ++(0,0.6);
    \draw[mid arrow=0.65] (b.90) -- ++(0,0.6);
    \draw[mid arrow=0.65] ([yshift=-0.6cm]b.270) -- ++(0,0.6);
    \draw[mid arrow=0.65] (c.90) -- ++(0,0.6);
    \draw[mid arrow=0.65] ([yshift=-0.6cm]c.270) -- ++(0,0.6);
    \draw[mid arrow=0.65] (d.90) -- ++(0,0.6);
    \draw[mid arrow=0.65] ([yshift=-0.6cm]d.270) -- ++(0,0.6);
    \draw[mid arrow=0.65] (a.0) -- (b.180);
    \draw[mid arrow=0.65] (b.0) -- (c.180);
    \draw[mid arrow=0.65] (c.0) -- (d.180);
    \draw[mid arrow=0.5] (d.0) --  ++(0.3,0)-- ++(0,-0.5) -- ++(-4.8,0) -- ++(0,0.5) --(a.180);
    \draw (a.180) to [out = 0,in=270] (a.90);
    \draw (b.180) to [out = 0,in=270] (b.90);
    \draw (c.180) to [out = 0,in=270] (c.90);
    \draw (d.180) to [out = 0,in=270] (d.90);
    \draw (a.270) to [out = 90,in=180] (a.0);
    \draw (b.270) to [out = 90,in=180] (b.0);
    \draw (c.270) to [out = 90,in=180] (c.0);
    \draw (d.270) to [out = 90,in=180] (d.0);
\end{tikzpicture}
\end{align}

\item An SQC with $N-1$ layers of SWAP gates \cite{Po_2016}, where $N$ is the number of sites:
\begin{align}
\tikzsetnextfilename{shift_quantum_circuit}
    \begin{tikzpicture}
        \draw[rounded corners] (0.9,0) -- (0.9,0.6) -- (0.3,1.2) -- (0.3,1.8) -- (-0.3,2.4) -- (-0.3,3.0) -- (-0.9,3.6) -- (-0.9,4.2);
        \draw[rounded corners] (0.3,0) -- (0.3,0.6) -- (0.9,1.2) -- (0.9,4.2);
         \draw[rounded corners] (-0.3,0) -- (-0.3,1.8) -- (0.3,2.4) -- (0.3,4.2);
         \draw[rounded corners] (-0.9,0) -- (-0.9,3.0) -- (-0.3,3.6) -- (-0.3,4.2);
    \end{tikzpicture}
\end{align}
This quantum circuit exhibits linear depth $O(L)$. As a representation (rather than a physical operator) it violates the locality-preserving property. This quantum circuit can also be represented by a unitary network as below:
\begin{align}
\label{fig:SQC_shift}
\tikzsetnextfilename{shift_quantum_circuit_UN}
\begin{tikzpicture}
    \node[draw, rectangle, rounded corners=2pt, minimum width=0.6cm, minimum height=0.6cm, fill=tensorcolor, thick] (a1) at (0,0) {};
    \node[draw, rectangle, rounded corners=2pt, minimum width=0.6cm, minimum height=0.6cm, fill=tensorcolor, thick] (b1) at (1.2,0) {};
    \node[draw, rectangle, rounded corners=2pt, minimum width=0.6cm, minimum height=0.6cm, fill=tensorcolor, thick] (c1) at (2.4,0) {};
    \node[draw, rectangle, rounded corners=2pt, minimum width=0.6cm, minimum height=0.6cm, fill=tensorcolor, thick] (d1) at (3.6,0) {};
    \node[draw, rectangle, rounded corners=2pt, minimum width=0.6cm, minimum height=0.6cm, fill=tensorcolor, thick] (a2) at (0,1.2) {};
    \node[draw, rectangle, rounded corners=2pt, minimum width=0.6cm, minimum height=0.6cm, fill=tensorcolor, thick] (b2) at (1.2,1.2) {};
    \node[draw, rectangle, rounded corners=2pt, minimum width=0.6cm, minimum height=0.6cm, fill=tensorcolor, thick] (c2) at (2.4,1.2) {};
    \node[draw, rectangle, rounded corners=2pt, minimum width=0.6cm, minimum height=0.6cm, fill=tensorcolor, thick] (d2) at (3.6,1.2) {};
    \draw[mid arrow=0.65] (a2.90) -- ++(0,0.6);
    \draw[mid arrow=0.65] ([yshift=-0.6cm]a1.270) -- ++(0,0.6);
    \draw[mid arrow=0.65] (b2.90) -- ++(0,0.6);
    \draw[mid arrow=0.65] ([yshift=-0.6cm]b1.270) -- ++(0,0.6);
    \draw[mid arrow=0.65] (c2.90) -- ++(0,0.6);
    \draw[mid arrow=0.65] ([yshift=-0.6cm]c1.270) -- ++(0,0.6);
    \draw[mid arrow=0.65] (d2.90) -- ++(0,0.6);
    \draw[mid arrow=0.65] ([yshift=-0.6cm]d1.270) -- ++(0,0.6);
    \draw[mid arrow=0.65] (a1.0) -- (b1.180);
    \draw[mid arrow=0.65] (b1.0) -- (c1.180);
    \draw[mid arrow=0.65] (c1.0) -- (d1.180);
    \draw[mid arrow=0.65] (b2.180) -- (a2.0);
    \draw[mid arrow=0.65] (c2.180) -- (b2.0);
    \draw[mid arrow=0.65] (d2.180) -- (c2.0);
    \draw[mid arrow=0.65] (b1.90) -- (b2.270);
    \draw[mid arrow=0.65] (c1.90) -- (c2.270);
    \draw[mid arrow=0.65] ([xshift=-0.1cm]d1.90) -- ([xshift=-0.1cm]d2.270);
    \draw[mid arrow=0.65] ([xshift=0.1cm]d1.90) -- ([xshift=0.1cm]d2.270);
    \draw (a1.270) to [out=90,in=180] (a1.0);
    \draw (b1.270) to [out=90, in=180] (b1.0);
    \draw (c1.270) to [out=90, in=180] (c1.0);
    \draw (b1.180) to [out=0, in=270] (b1.90);
    \draw (c1.180) to [out=0, in=270] (c1.90);
    \draw (d1.180) to [out=0, in=270] ([xshift=-0.1cm]d1.90);
    \draw (a2.0) to [out=180, in=270] (a2.90);
    \draw (b2.0) -- (b2.180);
    \draw (b2.270) -- (b2.90);
    \draw (c2.0) -- (c2.180);
    \draw (c2.270) -- (c2.90);
    \draw (d1.270) to  ([xshift=0.1cm]d1.90);
    \draw ([xshift=0.1cm]d2.270) to  [out=90, in=0] (d2.180);
    \draw ([xshift=-0.1cm]d2.270) -- (d2.90);
\end{tikzpicture}
\end{align}
As seen in the diagram, the rightward information flow is neutralized by the leftward flow. Therefore, the net information flow is zero: $f_{I} = 0$, which is different from the net information flow in (\ref{fig:shift}). 
\end{enumerate}
Non-unique net information flow arises in PBC 1D systems because information can loop multiple times before reaching its destination, corresponding to various possible implementations.  In fact, a identity operation in a PBC system can be performed by a unitary network with net information flow $f_{I}=1$:
\begin{align}
\label{eq: un identity non zero information flow}
\tikzsetnextfilename{UN_identity_non_zero_information_flow}
\begin{tikzpicture}
    \node[draw, rectangle, rounded corners=2pt, minimum width=0.6cm, minimum height=0.6cm, fill=tensorcolor, thick] (a2) at (0,1.2) {};
    \node[draw, rectangle, rounded corners=2pt, minimum width=0.6cm, minimum height=0.6cm, fill=tensorcolor, thick] (b2) at (1.2,1.2) {};
    \node[draw, rectangle, rounded corners=2pt, minimum width=0.6cm, minimum height=0.6cm, fill=tensorcolor, thick] (c2) at (2.4,1.2) {};
    \node[draw, rectangle, rounded corners=2pt, minimum width=0.6cm, minimum height=0.6cm, fill=tensorcolor, thick] (d2) at (3.6,1.2) {};
    \draw[mid arrow=0.65] (a2.90) -- ++(0,0.6);
    \draw[mid arrow=0.65] ([yshift=-0.6cm]a2.270) -- ++(0,0.6);
    \draw[mid arrow=0.65] (b2.90) -- ++(0,0.6);
    \draw[mid arrow=0.65] ([yshift=-0.6cm]b2.270) -- ++(0,0.6);
    \draw[mid arrow=0.65] (c2.90) -- ++(0,0.6);
    \draw[mid arrow=0.65] ([yshift=-0.6cm]c2.270) -- ++(0,0.6);
    \draw[mid arrow=0.65] (d2.90) -- ++(0,0.6);
    \draw[mid arrow=0.65] ([yshift=-0.6cm]d2.270) -- ++(0,0.6);
    \draw[mid arrow=0.65] (b2.180) -- (a2.0);
    \draw[mid arrow=0.65] (c2.180) -- (b2.0);
    \draw[mid arrow=0.65] (d2.180) -- (c2.0);
    \draw (a2.0) to [out=180, in=270] (a2.90);
    \draw (b2.0) -- (b2.180);
    \draw (b2.270) -- (b2.90);
    \draw (c2.0) -- (c2.180);
    \draw (c2.270) -- (c2.90);
    \draw (d2.0) -- (d2.180);
    \draw (d2.270) -- (d2.90);
    \draw (a2.270) to [out=90, in=0] (a2.180);
    \draw[mid arrow=0.5] (a2.180) -- ++(-0.3,0) --++(0,-0.4) -- ++(4.8,0) -- ++(0,0.4) -- (d2.0);
\end{tikzpicture}.
\end{align}

Notice that the difference between two unitary network implementations \eqref{fig:shift} and \eqref{fig:SQC_shift} is twofold:
\begin{enumerate} [label=(\roman*)]
\item  Their behaviors differ at the left and right boundaries. While we are considering a PBC system, such boundaries appear at the end points of sequential SWAP gates.
\item In the bulk, \eqref{fig:SQC_shift} includes an additional left-traveling propagating mode.
\end{enumerate}
In a finite system, this decoupled propagating mode acts as an information transfer from the right to the left boundary. The net information flow of the bulk unitary operation is $f_{I}=1$, balanced by the decoupled boundary information flow of $f_{I}=-1$. 

In the thermodynamic limit, the left and right boundaries are effectively at infinity. Information transfer between two boundaries lacks physical meaning and becomes irrelevant. The only remaining difference is the redundant left-traveling mode in the bulk, relating to our earlier discussion on the non-unique net information flow for the infinite OBC scenario in Subsection \ref{subsec: Intrinsic net information flow}.

In the above examples, we observe that although the global unitary $U$ preserves locality, its implementation $U_{Net}$ may not. When information circulates in cycles (in a PBC system) or redundant modes travel throughout the bulk (in an infinite OBC system), they invoke a redundant and non-local information flow.

\section{Intrinsic net information flow for approximately locality-preserving unitary networks}
\label{Append: Intrinsic net information flow for approximately locality-preserving unitary networks}

We expect that the uniqueness of net information flow can be generalized to approximately locality-preserving unitary network implementations of ALPUs. The net information flow in these unitary network implementations should align with the index defined for ALPUs~\cite{ranard2022converse} (See Def.~\ref{def: Index for ALPUs}).

\begin{proposition}
Any two approximately locality-preserving unitary network implementations  $U_{Net}^{1}$ and $U_{Net}^{2}$ of the same ALPU must share the same net information flow:
\begin{equation}
    f_{I}[U_{Net}^{1}] =  f_{I}[U_{Net}^{2}]
\end{equation}
\end{proposition}

\begin{proof}
The proof parallels the QCA case, but now we are dealing with an ALPU $U$. Splitting a unitary network representation $U_{Net}^{i}$ into $U_{Net}^{i, L}$ and $U_{Net}^{i, R}$ results in both parts being ALPUs as well.

We employ sequences of QCAs $\{\beta_{j}\}_{j=1}^{\infty}$, $\{\beta_{Net,j}^{1,R}\}_{j=1}^{\infty}$, and $\{\beta_{Net,j}^{2,L}\}_{j=1}^{\infty}$ to approximate the ALPUs $u$, $u_{Net}^{1, R}$, and $u_{Net}^{2, L}$ as stated in Proposition \ref{prop: QCA approximations of an ALPU}, while retaining $u_{Net}^{1, L}$ and $u_{Net}^{2, R}$ as ALPUs. For each $j$, we can define a local tensor $U_{M,j}$ exactly as in \eqref{eq: def UM}, simply replacing $u$, $u_{Net}^{1,R}$, and $u_{Net}^{2,L}$ by 
QCAs $\beta_{j}$, $\beta_{Net,j}^{1,R}$ and $\beta_{Net,j}^{2,L}$. This time, $U_{M,j}$ may not be unitary because, in general:
\begin{equation}
    \beta_{j}(\mathcal{B}_{M}) \ne  \beta_{Net,j}^{1,R}(\mathcal{B}_{M}) \ne \beta_{Net,j}^{2,L}(\mathcal{B}_{M}).
\end{equation} 
As $j$ increases, the radii of $\beta_{Net}^{1,R}$ and $\beta_{Net}^{2,L}$ expand, requiring $U_{M,j}$ to be supported on a subsystem $X_{j}$ with a larger diameter $\text{diam}(X_{j}) =  4j$. At the same time the difference between  $\beta_{j}(\mathcal{B}_{M})$,   $\beta_{Net,j}^{1,R}(\mathcal{B}_{M})$ and $ \beta_{Net,j}^{2,L}(\mathcal{B}_{M})$ becomes smaller:
\begin{equation}
\begin{aligned}
     \| (\beta_{Net,j}^{1,R} - \beta_{j})|_{\mathcal{B}_{M}} \| &\le 
     \| (\beta_{Net,j}^{1,R} - u_{Net,}^{1,R})|_{\mathcal{B}_{M}} \| \\
     &+ \| ( u_{Net,}^{1,R} - u)|_{\mathcal{B}_{M}} \| +  \| (  u - \beta_{j})|_{\mathcal{B}_{M}} \| \\
    &\le 2C_{f} \cdot f(j) \cdot \frac{ \left\lceil\text{diam}(X_{j}) \right \rceil}{j}
\end{aligned}
\end{equation}
In the limit $j \rightarrow \infty$ we have
\begin{equation}
    \lim_{j \rightarrow \infty}\beta_{j}(\mathcal{B}_{M}) = \lim_{j \rightarrow \infty}  \beta_{Net,j}^{1,R}(\mathcal{B}_{M}) = \lim_{j \rightarrow \infty} \beta_{Net,j}^{2,L}(\mathcal{B}_{M}).
\end{equation}
Therefore, we find that a connecting tensor $U_{M,j}$, which has diameter $4j$, becomes a unitary tensor when $j \rightarrow \infty$. If $U_{Net}^{1}$ and $U_{Net}^{2}$ possess different net information flows, finding such $U_{M}$ is impossible even in the thermodynamic limit, due to the mismatch of incoming and outgoing Hilbert space dimensions.
\end{proof}

We will provide a qualitative and intuitive argument here. Starting from an approximately locality-preserving unitary network implementation $U_{Net}^{0}$ of an ALPU $U$, we perturb it to change the implementation, without changing the physical unitary operator $U$ it represents.

If the unitary operator remains unchanged, the perturbation-induced extra information flow that carries observable algebras must return to its original location. For local perturbations, the backward information flow always offsets the forward information flow. This shows that, similar to the GNVW index, net information flow acts as a topological invariant and is intrinsic for QCAs.

When non-local perturbation is allowed, the net information flow can change only if some information flow loops around the PBC system as shown in \eqref{eq: un identity non zero information flow}, or in the thermodynamic limit, information transfers from left to right infinity, as shown in \eqref{eq: redundant flow}. Both types of perturbation described do not decay with distance; thus, they cannot be realized by an approximately locality-preserving unitary network. Therefore, we expect that for different approximately locality-preserving unitary network implementations of the same ALPU, the net information flow is unique.

\section{SQC has zero net information flow}
\label{Append: SQC has zero net information flow}

An FDLU circuit has a trivial GNVW index \cite{Gross_2012}. Instead of being the physical operator itself, a quantum circuit implements a unitary operator, since different circuits can realize the same global unitary operator. In this sense, quantum circuits are representations of physical operators, just like unitary networks. 

By considering quantum gates as local unitary tensors, an FDLU unitary circuit can be reinterpreted as a unitary network with zero net information flow ($f_{I} = 0$). In \ref{Subsec: SQCs as unitary networks}, we establish that an SQC with linear depth can be interpreted as a unitary network with finite bond dimensions. In this appendix, we demonstrate that the zero net information flow persists in the unitary network representation of SQC.

\begin{proposition} [Quantum circuit has zero net information flow] \label{prop: Quantum circuit has zero net information flow}
    For a quantum circuit of local unitary gates, regardless of whether the circuit has finite or infinite depth, the net information flow is zero:
    \begin{equation}
        f_{I} = 0.
    \end{equation}
\end{proposition}

\begin{proof}
Let us begin by proving that a local unitary gate has zero net information flow through any vertical surface. A quantum gate can be bipartite vertically, enabling its decomposition into a unitary network formed by two local tensors:
\begin{align}
    \begin{array}{c}
\tikzsetnextfilename{bisected_quantum_gate}
    \begin{tikzpicture}
        \node[draw, rectangle, rounded corners=2pt, minimum width=1cm, minimum height=0.6cm, fill=tensorcolor, thick] (a) {};
        \draw ([xshift=-0.3cm]a.90) -- ++(0,0.6);
        \draw ([xshift=0.3cm]a.90) -- ++(0,0.6);
        \draw ([xshift=-0.3cm,yshift=-0.6cm]a.270) -- ++(0,0.6) ;
        \draw ([xshift=0.3cm,yshift=-0.6cm]a.270) -- ++(0,0.6) ;
        \draw[dashed] (0,-1) -- (0,1);
    \end{tikzpicture}
    \end{array} = 
    \begin{array}{c}
    \tikzsetnextfilename{UN_with_zero_information_flow}
    \begin{tikzpicture}
        \node[draw, rectangle, rounded corners=2pt, minimum width=0.6cm, minimum height=0.6cm, fill=tensorcolor, thick] (a) {};
        \node[draw, rectangle, rounded corners=2pt, minimum width=0.6cm, minimum height=0.6cm, fill=tensorcolor, thick] (b) at (1.2,0) {};
        \draw[mid arrow = 0.65] (a.90) -- ++(0,0.6);
        \draw[mid arrow = 0.65] (b.90) -- ++(0,0.6);
        \draw[mid arrow = 0.65] ([yshift=-0.6cm]a.270) -- ++(0,0.6) ;
        \draw[mid arrow = 0.65] ([yshift=-0.6cm]b.270) -- ++(0,0.6) ;
        \draw[mid arrow = 0.65] ([yshift=-0.1cm]a.0) -- ([yshift=-0.1cm]b.180);
        \draw[mid arrow = 0.65] ([yshift=0.1cm]b.180) -- ([yshift=0.1cm]a.0);
    \end{tikzpicture}
    \end{array}
\end{align}
Examine either local tensor. Based on the conservation law of information (Def. \ref{Def: Information flow}) of the local tensor and the isomorphism of the incoming and outgoing physical Hilbert spaces $\dim \mathcal{H}_{phy, out} = \dim \mathcal{H}_{phy, in}$, it follows that the net information flow between the two local unitary tensors is zero.

Now, let us examine the quantum circuit. In an SQC, similar to FDLU circuit with depth $D$, a bisection will intersect at most $D$ local unitary gates. When a gate is bisected, it produces a zero net information flow $f_{I}[\Sigma] = 0$ across the surface. Taking the limit as $D$ approaches infinity to create an infinite-depth circuit does not alter the result. 
\begin{equation}
    \lim _{D \rightarrow \infty} f_{I}[\Sigma] = 0.
\end{equation}
\end{proof}

%bib file
\bibliography{ref}

\begin{thebibliography}{72}%
\makeatletter
\providecommand \@ifxundefined [1]{%
 \@ifx{#1\undefined}
}%
\providecommand \@ifnum [1]{%
 \ifnum #1\expandafter \@firstoftwo
 \else \expandafter \@secondoftwo
 \fi
}%
\providecommand \@ifx [1]{%
 \ifx #1\expandafter \@firstoftwo
 \else \expandafter \@secondoftwo
 \fi
}%
\providecommand \natexlab [1]{#1}%
\providecommand \enquote  [1]{``#1''}%
\providecommand \bibnamefont  [1]{#1}%
\providecommand \bibfnamefont [1]{#1}%
\providecommand \citenamefont [1]{#1}%
\providecommand \href@noop [0]{\@secondoftwo}%
\providecommand \href [0]{\begingroup \@sanitize@url \@href}%
\providecommand \@href[1]{\@@startlink{#1}\@@href}%
\providecommand \@@href[1]{\endgroup#1\@@endlink}%
\providecommand \@sanitize@url [0]{\catcode `\\12\catcode `\$12\catcode
  `\&12\catcode `\#12\catcode `\^12\catcode `\_12\catcode `\%12\relax}%
\providecommand \@@startlink[1]{}%
\providecommand \@@endlink[0]{}%
\providecommand \url  [0]{\begingroup\@sanitize@url \@url }%
\providecommand \@url [1]{\endgroup\@href {#1}{\urlprefix }}%
\providecommand \urlprefix  [0]{URL }%
\providecommand \Eprint [0]{\href }%
\providecommand \doibase [0]{https://doi.org/}%
\providecommand \selectlanguage [0]{\@gobble}%
\providecommand \bibinfo  [0]{\@secondoftwo}%
\providecommand \bibfield  [0]{\@secondoftwo}%
\providecommand \translation [1]{[#1]}%
\providecommand \BibitemOpen [0]{}%
\providecommand \bibitemStop [0]{}%
\providecommand \bibitemNoStop [0]{.\EOS\space}%
\providecommand \EOS [0]{\spacefactor3000\relax}%
\providecommand \BibitemShut  [1]{\csname bibitem#1\endcsname}%
\let\auto@bib@innerbib\@empty
%</preamble>
\bibitem [{\citenamefont {White}(1992)}]{White_DMRG}%
  \BibitemOpen
  \bibfield  {author} {\bibinfo {author} {\bibfnamefont {S.~R.}\ \bibnamefont
  {White}},\ }\bibfield  {title} {\bibinfo {title} {{Density matrix formulation
  for quantum renormalization groups}},\ }\href
  {https://doi.org/10.1103/PhysRevLett.69.2863} {\bibfield  {journal} {\bibinfo
   {journal} {Phys. Rev. Lett.}\ }\textbf {\bibinfo {volume} {69}},\ \bibinfo
  {pages} {2863} (\bibinfo {year} {1992})}\BibitemShut {NoStop}%
\bibitem [{\citenamefont {Verstraete}\ and\ \citenamefont
  {Cirac}(2006)}]{verstraete2006matrix}%
  \BibitemOpen
  \bibfield  {author} {\bibinfo {author} {\bibfnamefont {F.}~\bibnamefont
  {Verstraete}}\ and\ \bibinfo {author} {\bibfnamefont {J.~I.}\ \bibnamefont
  {Cirac}},\ }\bibfield  {title} {\bibinfo {title} {Matrix product states
  represent ground states faithfully},\ }\href
  {https://doi.org/10.1103/PhysRevB.73.094423} {\bibfield  {journal} {\bibinfo
  {journal} {Phys. Rev. B}\ }\textbf {\bibinfo {volume} {73}},\ \bibinfo
  {pages} {094423} (\bibinfo {year} {2006})}\BibitemShut {NoStop}%
\bibitem [{\citenamefont {Hastings}(2007)}]{hastings2007area}%
  \BibitemOpen
  \bibfield  {author} {\bibinfo {author} {\bibfnamefont {M.~B.}\ \bibnamefont
  {Hastings}},\ }\bibfield  {title} {\bibinfo {title} {An area law for
  one-dimensional quantum systems},\ }\href
  {https://doi.org/10.1088/1742-5468/2007/08/P08024} {\bibfield  {journal}
  {\bibinfo  {journal} {Journal of Statistical Mechanics: Theory and
  Experiment}\ }\textbf {\bibinfo {volume} {2007}},\ \bibinfo {pages} {P08024}
  (\bibinfo {year} {2007})}\BibitemShut {NoStop}%
\bibitem [{\citenamefont {Chen}\ \emph
  {et~al.}(2011{\natexlab{a}})\citenamefont {Chen}, \citenamefont {Gu},\ and\
  \citenamefont {Wen}}]{Chen_MPS_SPT_2011}%
  \BibitemOpen
  \bibfield  {author} {\bibinfo {author} {\bibfnamefont {X.}~\bibnamefont
  {Chen}}, \bibinfo {author} {\bibfnamefont {Z.-C.}\ \bibnamefont {Gu}},\ and\
  \bibinfo {author} {\bibfnamefont {X.-G.}\ \bibnamefont {Wen}},\ }\bibfield
  {title} {\bibinfo {title} {{Classification of gapped symmetric phases in
  one-dimensional spin systems}},\ }\href
  {https://doi.org/10.1103/PhysRevB.83.035107} {\bibfield  {journal} {\bibinfo
  {journal} {Phys. Rev. B}\ }\textbf {\bibinfo {volume} {83}},\ \bibinfo
  {pages} {035107} (\bibinfo {year} {2011}{\natexlab{a}})}\BibitemShut
  {NoStop}%
\bibitem [{\citenamefont {Schuch}\ \emph {et~al.}(2011)\citenamefont {Schuch},
  \citenamefont {P\'erez-Garc\'{\i}a},\ and\ \citenamefont
  {Cirac}}]{Schuch_MPS_SPT_2011}%
  \BibitemOpen
  \bibfield  {author} {\bibinfo {author} {\bibfnamefont {N.}~\bibnamefont
  {Schuch}}, \bibinfo {author} {\bibfnamefont {D.}~\bibnamefont
  {P\'erez-Garc\'{\i}a}},\ and\ \bibinfo {author} {\bibfnamefont
  {I.}~\bibnamefont {Cirac}},\ }\bibfield  {title} {\bibinfo {title}
  {{Classifying quantum phases using matrix product states and projected
  entangled pair states}},\ }\href {https://doi.org/10.1103/PhysRevB.84.165139}
  {\bibfield  {journal} {\bibinfo  {journal} {Phys. Rev. B}\ }\textbf {\bibinfo
  {volume} {84}},\ \bibinfo {pages} {165139} (\bibinfo {year}
  {2011})}\BibitemShut {NoStop}%
\bibitem [{\citenamefont {Cirac}\ \emph {et~al.}(2021)\citenamefont {Cirac},
  \citenamefont {P\'erez-Garc\'{\i}a}, \citenamefont {Schuch},\ and\
  \citenamefont {Verstraete}}]{Cirac_RMP}%
  \BibitemOpen
  \bibfield  {author} {\bibinfo {author} {\bibfnamefont {J.~I.}\ \bibnamefont
  {Cirac}}, \bibinfo {author} {\bibfnamefont {D.}~\bibnamefont
  {P\'erez-Garc\'{\i}a}}, \bibinfo {author} {\bibfnamefont {N.}~\bibnamefont
  {Schuch}},\ and\ \bibinfo {author} {\bibfnamefont {F.}~\bibnamefont
  {Verstraete}},\ }\bibfield  {title} {\bibinfo {title} {{Matrix product states
  and projected entangled pair states: Concepts, symmetries, theorems}},\
  }\href {https://doi.org/10.1103/RevModPhys.93.045003} {\bibfield  {journal}
  {\bibinfo  {journal} {Rev. Mod. Phys.}\ }\textbf {\bibinfo {volume} {93}},\
  \bibinfo {pages} {045003} (\bibinfo {year} {2021})}\BibitemShut {NoStop}%
\bibitem [{\citenamefont {Harper}\ and\ \citenamefont
  {Roy}(2017)}]{Harper_chiral_Floquet}%
  \BibitemOpen
  \bibfield  {author} {\bibinfo {author} {\bibfnamefont {F.}~\bibnamefont
  {Harper}}\ and\ \bibinfo {author} {\bibfnamefont {R.}~\bibnamefont {Roy}},\
  }\bibfield  {title} {\bibinfo {title} {{Floquet Topological Order in
  Interacting Systems of Bosons and Fermions}},\ }\href
  {https://doi.org/10.1103/PhysRevLett.118.115301} {\bibfield  {journal}
  {\bibinfo  {journal} {Phys. Rev. Lett.}\ }\textbf {\bibinfo {volume} {118}},\
  \bibinfo {pages} {115301} (\bibinfo {year} {2017})}\BibitemShut {NoStop}%
\bibitem [{\citenamefont {Po}\ \emph {et~al.}(2016)\citenamefont {Po},
  \citenamefont {Fidkowski}, \citenamefont {Morimoto}, \citenamefont {Potter},\
  and\ \citenamefont {Vishwanath}}]{Po_2016}%
  \BibitemOpen
  \bibfield  {author} {\bibinfo {author} {\bibfnamefont {H.~C.}\ \bibnamefont
  {Po}}, \bibinfo {author} {\bibfnamefont {L.}~\bibnamefont {Fidkowski}},
  \bibinfo {author} {\bibfnamefont {T.}~\bibnamefont {Morimoto}}, \bibinfo
  {author} {\bibfnamefont {A.~C.}\ \bibnamefont {Potter}},\ and\ \bibinfo
  {author} {\bibfnamefont {A.}~\bibnamefont {Vishwanath}},\ }\bibfield  {title}
  {\bibinfo {title} {{Chiral Floquet Phases of Many-Body Localized Bosons}},\
  }\href {https://doi.org/10.1103/PhysRevX.6.041070} {\bibfield  {journal}
  {\bibinfo  {journal} {Phys. Rev. X}\ }\textbf {\bibinfo {volume} {6}},\
  \bibinfo {pages} {041070} (\bibinfo {year} {2016})}\BibitemShut {NoStop}%
\bibitem [{\citenamefont {Chen}\ \emph
  {et~al.}(2011{\natexlab{b}})\citenamefont {Chen}, \citenamefont {Liu},\ and\
  \citenamefont {Wen}}]{chen20122d}%
  \BibitemOpen
  \bibfield  {author} {\bibinfo {author} {\bibfnamefont {X.}~\bibnamefont
  {Chen}}, \bibinfo {author} {\bibfnamefont {Z.-X.}\ \bibnamefont {Liu}},\ and\
  \bibinfo {author} {\bibfnamefont {X.-G.}\ \bibnamefont {Wen}},\ }\bibfield
  {title} {\bibinfo {title} {{Two-dimensional symmetry-protected topological
  orders and their protected gapless edge excitations}},\ }\href
  {https://doi.org/10.1103/PhysRevB.84.235141} {\bibfield  {journal} {\bibinfo
  {journal} {Phys. Rev. B}\ }\textbf {\bibinfo {volume} {84}},\ \bibinfo
  {pages} {235141} (\bibinfo {year} {2011}{\natexlab{b}})}\BibitemShut
  {NoStop}%
\bibitem [{\citenamefont {{\c S}ahino{\u g}lu}\ \emph
  {et~al.}(2021)\citenamefont {{\c S}ahino{\u g}lu}, \citenamefont
  {Williamson}, \citenamefont {Bultinck}, \citenamefont {Mari{\"e}n},
  \citenamefont {Haegeman}, \citenamefont {Schuch},\ and\ \citenamefont
  {Verstraete}}]{sahinoglu2150characterizing}%
  \BibitemOpen
  \bibfield  {author} {\bibinfo {author} {\bibfnamefont {M.~B.}\ \bibnamefont
  {{\c S}ahino{\u g}lu}}, \bibinfo {author} {\bibfnamefont {D.}~\bibnamefont
  {Williamson}}, \bibinfo {author} {\bibfnamefont {N.}~\bibnamefont
  {Bultinck}}, \bibinfo {author} {\bibfnamefont {M.}~\bibnamefont
  {Mari{\"e}n}}, \bibinfo {author} {\bibfnamefont {J.}~\bibnamefont
  {Haegeman}}, \bibinfo {author} {\bibfnamefont {N.}~\bibnamefont {Schuch}},\
  and\ \bibinfo {author} {\bibfnamefont {F.}~\bibnamefont {Verstraete}},\
  }\bibfield  {title} {\bibinfo {title} {{Characterizing Topological Order with
  Matrix Product Operators}},\ }\href
  {https://doi.org/10.1007/s00023-020-00992-4} {\bibfield  {journal} {\bibinfo
  {journal} {Annales Henri Poincar{\'e}}\ }\textbf {\bibinfo {volume} {22}},\
  \bibinfo {pages} {563} (\bibinfo {year} {2021})}\BibitemShut {NoStop}%
\bibitem [{\citenamefont {Bultinck}\ \emph {et~al.}(2017)\citenamefont
  {Bultinck}, \citenamefont {Mari{\"e}n}, \citenamefont {Williamson},
  \citenamefont {{\c S}ahino{\u g}lu}, \citenamefont {Haegeman},\ and\
  \citenamefont {Verstraete}}]{bultinck2017anyons}%
  \BibitemOpen
  \bibfield  {author} {\bibinfo {author} {\bibfnamefont {N.}~\bibnamefont
  {Bultinck}}, \bibinfo {author} {\bibfnamefont {M.}~\bibnamefont
  {Mari{\"e}n}}, \bibinfo {author} {\bibfnamefont {D.}~\bibnamefont
  {Williamson}}, \bibinfo {author} {\bibfnamefont {M.}~\bibnamefont {{\c
  S}ahino{\u g}lu}}, \bibinfo {author} {\bibfnamefont {J.}~\bibnamefont
  {Haegeman}},\ and\ \bibinfo {author} {\bibfnamefont {F.}~\bibnamefont
  {Verstraete}},\ }\bibfield  {title} {\bibinfo {title} {{Anyons and matrix
  product operator algebras}},\ }\href
  {https://doi.org/https://doi.org/10.1016/j.aop.2017.01.004} {\bibfield
  {journal} {\bibinfo  {journal} {Annals of Physics}\ }\textbf {\bibinfo
  {volume} {378}},\ \bibinfo {pages} {183} (\bibinfo {year}
  {2017})}\BibitemShut {NoStop}%
\bibitem [{\citenamefont {Okada}\ and\ \citenamefont
  {Tachikawa}(2024)}]{Okada2024}%
  \BibitemOpen
  \bibfield  {author} {\bibinfo {author} {\bibfnamefont {M.}~\bibnamefont
  {Okada}}\ and\ \bibinfo {author} {\bibfnamefont {Y.}~\bibnamefont
  {Tachikawa}},\ }\bibfield  {title} {\bibinfo {title} {{Noninvertible
  Symmetries Act Locally by Quantum Operations}},\ }\href
  {https://doi.org/10.1103/PhysRevLett.133.191602} {\bibfield  {journal}
  {\bibinfo  {journal} {Phys. Rev. Lett.}\ }\textbf {\bibinfo {volume} {133}},\
  \bibinfo {pages} {191602} (\bibinfo {year} {2024})}\BibitemShut {NoStop}%
\bibitem [{\citenamefont {Lootens}\ \emph {et~al.}(2023)\citenamefont
  {Lootens}, \citenamefont {Delcamp}, \citenamefont {Ortiz},\ and\
  \citenamefont {Verstraete}}]{Lootens2023}%
  \BibitemOpen
  \bibfield  {author} {\bibinfo {author} {\bibfnamefont {L.}~\bibnamefont
  {Lootens}}, \bibinfo {author} {\bibfnamefont {C.}~\bibnamefont {Delcamp}},
  \bibinfo {author} {\bibfnamefont {G.}~\bibnamefont {Ortiz}},\ and\ \bibinfo
  {author} {\bibfnamefont {F.}~\bibnamefont {Verstraete}},\ }\bibfield  {title}
  {\bibinfo {title} {{Dualities in One-Dimensional Quantum Lattice Models:
  Symmetric Hamiltonians and Matrix Product Operator Intertwiners}},\ }\href
  {https://doi.org/10.1103/PRXQuantum.4.020357} {\bibfield  {journal} {\bibinfo
   {journal} {PRX Quantum}\ }\textbf {\bibinfo {volume} {4}},\ \bibinfo {pages}
  {020357} (\bibinfo {year} {2023})}\BibitemShut {NoStop}%
\bibitem [{\citenamefont {Lootens}\ \emph {et~al.}(2024)\citenamefont
  {Lootens}, \citenamefont {Delcamp},\ and\ \citenamefont
  {Verstraete}}]{Lootens2024}%
  \BibitemOpen
  \bibfield  {author} {\bibinfo {author} {\bibfnamefont {L.}~\bibnamefont
  {Lootens}}, \bibinfo {author} {\bibfnamefont {C.}~\bibnamefont {Delcamp}},\
  and\ \bibinfo {author} {\bibfnamefont {F.}~\bibnamefont {Verstraete}},\
  }\bibfield  {title} {\bibinfo {title} {{Dualities in One-Dimensional Quantum
  Lattice Models: Topological Sectors}},\ }\href
  {https://doi.org/10.1103/PRXQuantum.5.010338} {\bibfield  {journal} {\bibinfo
   {journal} {PRX Quantum}\ }\textbf {\bibinfo {volume} {5}},\ \bibinfo {pages}
  {010338} (\bibinfo {year} {2024})}\BibitemShut {NoStop}%
\bibitem [{\citenamefont {Lootens}\ \emph {et~al.}(2025)\citenamefont
  {Lootens}, \citenamefont {Delcamp}, \citenamefont {Williamson},\ and\
  \citenamefont {Verstraete}}]{Lootens2025}%
  \BibitemOpen
  \bibfield  {author} {\bibinfo {author} {\bibfnamefont {L.}~\bibnamefont
  {Lootens}}, \bibinfo {author} {\bibfnamefont {C.}~\bibnamefont {Delcamp}},
  \bibinfo {author} {\bibfnamefont {D.}~\bibnamefont {Williamson}},\ and\
  \bibinfo {author} {\bibfnamefont {F.}~\bibnamefont {Verstraete}},\ }\bibfield
   {title} {\bibinfo {title} {{Low-Depth Unitary Quantum Circuits for Dualities
  in One-Dimensional Quantum Lattice Models}},\ }\href
  {https://doi.org/10.1103/PhysRevLett.134.130403} {\bibfield  {journal}
  {\bibinfo  {journal} {Phys. Rev. Lett.}\ }\textbf {\bibinfo {volume} {134}},\
  \bibinfo {pages} {130403} (\bibinfo {year} {2025})}\BibitemShut {NoStop}%
\bibitem [{\citenamefont {Fukusumi}\ and\ \citenamefont
  {Iino}(2021)}]{Fukusumi2021}%
  \BibitemOpen
  \bibfield  {author} {\bibinfo {author} {\bibfnamefont {Y.}~\bibnamefont
  {Fukusumi}}\ and\ \bibinfo {author} {\bibfnamefont {S.}~\bibnamefont
  {Iino}},\ }\bibfield  {title} {\bibinfo {title} {Open spin chain realization
  of a topological defect in a one-dimensional ising model: Boundary and bulk
  symmetry},\ }\href {https://doi.org/10.1103/PhysRevB.104.125418} {\bibfield
  {journal} {\bibinfo  {journal} {Phys. Rev. B}\ }\textbf {\bibinfo {volume}
  {104}},\ \bibinfo {pages} {125418} (\bibinfo {year} {2021})}\BibitemShut
  {NoStop}%
\bibitem [{\citenamefont {Ashkenazi}\ and\ \citenamefont
  {Zohar}(2022)}]{Ashkenazi2022}%
  \BibitemOpen
  \bibfield  {author} {\bibinfo {author} {\bibfnamefont {S.}~\bibnamefont
  {Ashkenazi}}\ and\ \bibinfo {author} {\bibfnamefont {E.}~\bibnamefont
  {Zohar}},\ }\bibfield  {title} {\bibinfo {title} {Duality as a feasible
  physical transformation for quantum simulation},\ }\href
  {https://doi.org/10.1103/PhysRevA.105.022431} {\bibfield  {journal} {\bibinfo
   {journal} {Phys. Rev. A}\ }\textbf {\bibinfo {volume} {105}},\ \bibinfo
  {pages} {022431} (\bibinfo {year} {2022})}\BibitemShut {NoStop}%
\bibitem [{\citenamefont {Tantivasadakarn}\ \emph {et~al.}(2024)\citenamefont
  {Tantivasadakarn}, \citenamefont {Thorngren}, \citenamefont {Vishwanath},\
  and\ \citenamefont {Verresen}}]{Tantivasadakarn2024}%
  \BibitemOpen
  \bibfield  {author} {\bibinfo {author} {\bibfnamefont {N.}~\bibnamefont
  {Tantivasadakarn}}, \bibinfo {author} {\bibfnamefont {R.}~\bibnamefont
  {Thorngren}}, \bibinfo {author} {\bibfnamefont {A.}~\bibnamefont
  {Vishwanath}},\ and\ \bibinfo {author} {\bibfnamefont {R.}~\bibnamefont
  {Verresen}},\ }\bibfield  {title} {\bibinfo {title} {{Long-Range Entanglement
  from Measuring Symmetry-Protected Topological Phases}},\ }\href
  {https://doi.org/10.1103/PhysRevX.14.021040} {\bibfield  {journal} {\bibinfo
  {journal} {Phys. Rev. X}\ }\textbf {\bibinfo {volume} {14}},\ \bibinfo
  {pages} {021040} (\bibinfo {year} {2024})}\BibitemShut {NoStop}%
\bibitem [{\citenamefont {Aasen}\ \emph {et~al.}(2022)\citenamefont {Aasen},
  \citenamefont {Wang},\ and\ \citenamefont {Hastings}}]{Aasen2022}%
  \BibitemOpen
  \bibfield  {author} {\bibinfo {author} {\bibfnamefont {D.}~\bibnamefont
  {Aasen}}, \bibinfo {author} {\bibfnamefont {Z.}~\bibnamefont {Wang}},\ and\
  \bibinfo {author} {\bibfnamefont {M.~B.}\ \bibnamefont {Hastings}},\
  }\bibfield  {title} {\bibinfo {title} {Adiabatic paths of hamiltonians,
  symmetries of topological order, and automorphism codes},\ }\href
  {https://doi.org/10.1103/PhysRevB.106.085122} {\bibfield  {journal} {\bibinfo
   {journal} {Phys. Rev. B}\ }\textbf {\bibinfo {volume} {106}},\ \bibinfo
  {pages} {085122} (\bibinfo {year} {2022})}\BibitemShut {NoStop}%
\bibitem [{\citenamefont {Bravyi}\ \emph {et~al.}(2022)\citenamefont {Bravyi},
  \citenamefont {Kim}, \citenamefont {Kliesch},\ and\ \citenamefont
  {Koenig}}]{bravyi2022adaptiveconstantdepthcircuitsmanipulating}%
  \BibitemOpen
  \bibfield  {author} {\bibinfo {author} {\bibfnamefont {S.}~\bibnamefont
  {Bravyi}}, \bibinfo {author} {\bibfnamefont {I.}~\bibnamefont {Kim}},
  \bibinfo {author} {\bibfnamefont {A.}~\bibnamefont {Kliesch}},\ and\ \bibinfo
  {author} {\bibfnamefont {R.}~\bibnamefont {Koenig}},\ }\href
  {https://arxiv.org/abs/2205.01933} {\bibinfo {title} {Adaptive constant-depth
  circuits for manipulating non-abelian anyons}} (\bibinfo {year} {2022}),\
  \Eprint {https://arxiv.org/abs/2205.01933} {arXiv:2205.01933 [quant-ph]}
  \BibitemShut {NoStop}%
\bibitem [{\citenamefont {Tantivasadakarn}\ \emph {et~al.}(2023)\citenamefont
  {Tantivasadakarn}, \citenamefont {Vishwanath},\ and\ \citenamefont
  {Verresen}}]{Tantivasadakarn2023}%
  \BibitemOpen
  \bibfield  {author} {\bibinfo {author} {\bibfnamefont {N.}~\bibnamefont
  {Tantivasadakarn}}, \bibinfo {author} {\bibfnamefont {A.}~\bibnamefont
  {Vishwanath}},\ and\ \bibinfo {author} {\bibfnamefont {R.}~\bibnamefont
  {Verresen}},\ }\bibfield  {title} {\bibinfo {title} {Hierarchy of topological
  order from finite-depth unitaries, measurement, and feedforward},\ }\href
  {https://doi.org/10.1103/PRXQuantum.4.020339} {\bibfield  {journal} {\bibinfo
   {journal} {PRX Quantum}\ }\textbf {\bibinfo {volume} {4}},\ \bibinfo {pages}
  {020339} (\bibinfo {year} {2023})}\BibitemShut {NoStop}%
\bibitem [{\citenamefont {Fechisin}\ \emph {et~al.}(2025)\citenamefont
  {Fechisin}, \citenamefont {Tantivasadakarn},\ and\ \citenamefont
  {Albert}}]{Fechisin2025}%
  \BibitemOpen
  \bibfield  {author} {\bibinfo {author} {\bibfnamefont {C.}~\bibnamefont
  {Fechisin}}, \bibinfo {author} {\bibfnamefont {N.}~\bibnamefont
  {Tantivasadakarn}},\ and\ \bibinfo {author} {\bibfnamefont {V.~V.}\
  \bibnamefont {Albert}},\ }\bibfield  {title} {\bibinfo {title}
  {{Noninvertible Symmetry-Protected Topological Order in a Group-Based Cluster
  State}},\ }\href {https://doi.org/10.1103/PhysRevX.15.011058} {\bibfield
  {journal} {\bibinfo  {journal} {Phys. Rev. X}\ }\textbf {\bibinfo {volume}
  {15}},\ \bibinfo {pages} {011058} (\bibinfo {year} {2025})}\BibitemShut
  {NoStop}%
\bibitem [{\citenamefont {Okuda}\ \emph {et~al.}(2024)\citenamefont {Okuda},
  \citenamefont {Parayil~Mana},\ and\ \citenamefont {Sukeno}}]{Okuda2024}%
  \BibitemOpen
  \bibfield  {author} {\bibinfo {author} {\bibfnamefont {T.}~\bibnamefont
  {Okuda}}, \bibinfo {author} {\bibfnamefont {A.}~\bibnamefont
  {Parayil~Mana}},\ and\ \bibinfo {author} {\bibfnamefont {H.}~\bibnamefont
  {Sukeno}},\ }\bibfield  {title} {\bibinfo {title} {Anomaly inflow, dualities,
  and quantum simulation of abelian lattice gauge theories induced by
  measurements},\ }\href {https://doi.org/10.1103/PhysRevResearch.6.043018}
  {\bibfield  {journal} {\bibinfo  {journal} {Phys. Rev. Res.}\ }\textbf
  {\bibinfo {volume} {6}},\ \bibinfo {pages} {043018} (\bibinfo {year}
  {2024})}\BibitemShut {NoStop}%
\bibitem [{\citenamefont {Aasen}\ \emph {et~al.}(2016)\citenamefont {Aasen},
  \citenamefont {Mong},\ and\ \citenamefont {Fendley}}]{aasen2016topological}%
  \BibitemOpen
  \bibfield  {author} {\bibinfo {author} {\bibfnamefont {D.}~\bibnamefont
  {Aasen}}, \bibinfo {author} {\bibfnamefont {R.~S.~K.}\ \bibnamefont {Mong}},\
  and\ \bibinfo {author} {\bibfnamefont {P.}~\bibnamefont {Fendley}},\
  }\bibfield  {title} {\bibinfo {title} {{Topological defects on the lattice:
  I. The Ising model}},\ }\href
  {https://doi.org/10.1088/1751-8113/49/35/354001} {\bibfield  {journal}
  {\bibinfo  {journal} {Journal of Physics A: Mathematical and Theoretical}\
  }\textbf {\bibinfo {volume} {49}},\ \bibinfo {pages} {354001} (\bibinfo
  {year} {2016})}\BibitemShut {NoStop}%
\bibitem [{\citenamefont {Cao}\ \emph {et~al.}(2023)\citenamefont {Cao},
  \citenamefont {Li}, \citenamefont {Yamazaki},\ and\ \citenamefont
  {Zheng}}]{cao2023subsystem}%
  \BibitemOpen
  \bibfield  {author} {\bibinfo {author} {\bibfnamefont {W.}~\bibnamefont
  {Cao}}, \bibinfo {author} {\bibfnamefont {L.}~\bibnamefont {Li}}, \bibinfo
  {author} {\bibfnamefont {M.}~\bibnamefont {Yamazaki}},\ and\ \bibinfo
  {author} {\bibfnamefont {Y.}~\bibnamefont {Zheng}},\ }\bibfield  {title}
  {\bibinfo {title} {{Subsystem non-invertible symmetry operators and
  defects}},\ }\href {https://doi.org/10.21468/SciPostPhys.15.4.155} {\bibfield
   {journal} {\bibinfo  {journal} {SciPost Phys.}\ }\textbf {\bibinfo {volume}
  {15}},\ \bibinfo {pages} {155} (\bibinfo {year} {2023})}\BibitemShut
  {NoStop}%
\bibitem [{\citenamefont {Verstraete}\ \emph {et~al.}(2004)\citenamefont
  {Verstraete}, \citenamefont {Garc\'{\i}a-Ripoll},\ and\ \citenamefont
  {Cirac}}]{Verstraete_MPDO_2004}%
  \BibitemOpen
  \bibfield  {author} {\bibinfo {author} {\bibfnamefont {F.}~\bibnamefont
  {Verstraete}}, \bibinfo {author} {\bibfnamefont {J.~J.}\ \bibnamefont
  {Garc\'{\i}a-Ripoll}},\ and\ \bibinfo {author} {\bibfnamefont {J.~I.}\
  \bibnamefont {Cirac}},\ }\bibfield  {title} {\bibinfo {title} {{Matrix
  Product Density Operators: Simulation of Finite-Temperature and Dissipative
  Systems}},\ }\href {https://doi.org/10.1103/PhysRevLett.93.207204} {\bibfield
   {journal} {\bibinfo  {journal} {Phys. Rev. Lett.}\ }\textbf {\bibinfo
  {volume} {93}},\ \bibinfo {pages} {207204} (\bibinfo {year}
  {2004})}\BibitemShut {NoStop}%
\bibitem [{\citenamefont {Pirvu}\ \emph {et~al.}(2010)\citenamefont {Pirvu},
  \citenamefont {Murg}, \citenamefont {Cirac},\ and\ \citenamefont
  {Verstraete}}]{pirvu2010matrix}%
  \BibitemOpen
  \bibfield  {author} {\bibinfo {author} {\bibfnamefont {B.}~\bibnamefont
  {Pirvu}}, \bibinfo {author} {\bibfnamefont {V.}~\bibnamefont {Murg}},
  \bibinfo {author} {\bibfnamefont {J.~I.}\ \bibnamefont {Cirac}},\ and\
  \bibinfo {author} {\bibfnamefont {F.}~\bibnamefont {Verstraete}},\ }\bibfield
   {title} {\bibinfo {title} {Matrix product operator representations},\ }\href
  {https://doi.org/10.1088/1367-2630/12/2/025012} {\bibfield  {journal}
  {\bibinfo  {journal} {New Journal of Physics}\ }\textbf {\bibinfo {volume}
  {12}},\ \bibinfo {pages} {025012} (\bibinfo {year} {2010})}\BibitemShut
  {NoStop}%
\bibitem [{\citenamefont {Ignacio~Cirac}\ \emph {et~al.}(2017)\citenamefont
  {Ignacio~Cirac}, \citenamefont {Perez-Garcia}, \citenamefont {Schuch},\ and\
  \citenamefont {Verstraete}}]{cirac2017matrix}%
  \BibitemOpen
  \bibfield  {author} {\bibinfo {author} {\bibfnamefont {J.}~\bibnamefont
  {Ignacio~Cirac}}, \bibinfo {author} {\bibfnamefont {D.}~\bibnamefont
  {Perez-Garcia}}, \bibinfo {author} {\bibfnamefont {N.}~\bibnamefont
  {Schuch}},\ and\ \bibinfo {author} {\bibfnamefont {F.}~\bibnamefont
  {Verstraete}},\ }\bibfield  {title} {\bibinfo {title} {Matrix product
  unitaries: structure, symmetries, and topological invariants},\ }\href
  {https://doi.org/10.1088/1742-5468/aa7e55} {\bibfield  {journal} {\bibinfo
  {journal} {Journal of Statistical Mechanics: Theory and Experiment}\ }\textbf
  {\bibinfo {volume} {2017}},\ \bibinfo {pages} {083105} (\bibinfo {year}
  {2017})}\BibitemShut {NoStop}%
\bibitem [{\citenamefont {\ifmmode \mbox{\c{S}}\else
  \c{S}\fi{}ahino\ifmmode~\breve{g}\else \u{g}\fi{}lu}\ \emph
  {et~al.}(2018)\citenamefont {\ifmmode \mbox{\c{S}}\else
  \c{S}\fi{}ahino\ifmmode~\breve{g}\else \u{g}\fi{}lu}, \citenamefont {Shukla},
  \citenamefont {Bi},\ and\ \citenamefont {Chen}}]{_ahino_lu_2018}%
  \BibitemOpen
  \bibfield  {author} {\bibinfo {author} {\bibfnamefont {M.~B.}\ \bibnamefont
  {\ifmmode \mbox{\c{S}}\else \c{S}\fi{}ahino\ifmmode~\breve{g}\else
  \u{g}\fi{}lu}}, \bibinfo {author} {\bibfnamefont {S.~K.}\ \bibnamefont
  {Shukla}}, \bibinfo {author} {\bibfnamefont {F.}~\bibnamefont {Bi}},\ and\
  \bibinfo {author} {\bibfnamefont {X.}~\bibnamefont {Chen}},\ }\bibfield
  {title} {\bibinfo {title} {Matrix product representation of locality
  preserving unitaries},\ }\href {https://doi.org/10.1103/PhysRevB.98.245122}
  {\bibfield  {journal} {\bibinfo  {journal} {Phys. Rev. B}\ }\textbf {\bibinfo
  {volume} {98}},\ \bibinfo {pages} {245122} (\bibinfo {year}
  {2018})}\BibitemShut {NoStop}%
\bibitem [{\citenamefont {Piroli}\ and\ \citenamefont
  {Cirac}(2020)}]{Piroli_2020}%
  \BibitemOpen
  \bibfield  {author} {\bibinfo {author} {\bibfnamefont {L.}~\bibnamefont
  {Piroli}}\ and\ \bibinfo {author} {\bibfnamefont {J.~I.}\ \bibnamefont
  {Cirac}},\ }\bibfield  {title} {\bibinfo {title} {{Quantum Cellular Automata,
  Tensor Networks, and Area Laws}},\ }\href
  {https://doi.org/10.1103/PhysRevLett.125.190402} {\bibfield  {journal}
  {\bibinfo  {journal} {Phys. Rev. Lett.}\ }\textbf {\bibinfo {volume} {125}},\
  \bibinfo {pages} {190402} (\bibinfo {year} {2020})}\BibitemShut {NoStop}%
\bibitem [{\citenamefont {Schumacher}\ and\ \citenamefont
  {Werner}(2004)}]{schumacher2004reversiblequantumcellularautomata}%
  \BibitemOpen
  \bibfield  {author} {\bibinfo {author} {\bibfnamefont {B.}~\bibnamefont
  {Schumacher}}\ and\ \bibinfo {author} {\bibfnamefont {R.~F.}\ \bibnamefont
  {Werner}},\ }\href@noop {} {\bibinfo {title} {Reversible quantum cellular
  automata}} (\bibinfo {year} {2004}),\ \Eprint
  {https://arxiv.org/abs/quant-ph/0405174} {arXiv:quant-ph/0405174 [quant-ph]}
  \BibitemShut {NoStop}%
\bibitem [{\citenamefont {Farrelly}(2020)}]{Farrelly_2020}%
  \BibitemOpen
  \bibfield  {author} {\bibinfo {author} {\bibfnamefont {T.}~\bibnamefont
  {Farrelly}},\ }\bibfield  {title} {\bibinfo {title} {A review of quantum
  cellular automata},\ }\href {https://doi.org/10.22331/q-2020-11-30-368}
  {\bibfield  {journal} {\bibinfo  {journal} {Quantum}\ }\textbf {\bibinfo
  {volume} {4}},\ \bibinfo {pages} {368} (\bibinfo {year} {2020})}\BibitemShut
  {NoStop}%
\bibitem [{\citenamefont {Arrighi}(2019)}]{arrighi2019overview}%
  \BibitemOpen
  \bibfield  {author} {\bibinfo {author} {\bibfnamefont {P.}~\bibnamefont
  {Arrighi}},\ }\bibfield  {title} {\bibinfo {title} {An overview of quantum
  cellular automata},\ }\href {https://doi.org/10.1007/s11047-019-09762-6}
  {\bibfield  {journal} {\bibinfo  {journal} {Natural Computing}\ }\textbf
  {\bibinfo {volume} {18}},\ \bibinfo {pages} {885} (\bibinfo {year}
  {2019})}\BibitemShut {NoStop}%
\bibitem [{\citenamefont {Gross}\ \emph {et~al.}(2012)\citenamefont {Gross},
  \citenamefont {Nesme}, \citenamefont {Vogts},\ and\ \citenamefont
  {Werner}}]{Gross_2012}%
  \BibitemOpen
  \bibfield  {author} {\bibinfo {author} {\bibfnamefont {D.}~\bibnamefont
  {Gross}}, \bibinfo {author} {\bibfnamefont {V.}~\bibnamefont {Nesme}},
  \bibinfo {author} {\bibfnamefont {H.}~\bibnamefont {Vogts}},\ and\ \bibinfo
  {author} {\bibfnamefont {R.~F.}\ \bibnamefont {Werner}},\ }\bibfield  {title}
  {\bibinfo {title} {Index theory of one dimensional quantum walks and cellular
  automata},\ }\href {https://doi.org/10.1007/s00220-012-1423-1} {\bibfield
  {journal} {\bibinfo  {journal} {Communications in Mathematical Physics}\
  }\textbf {\bibinfo {volume} {310}},\ \bibinfo {pages} {419} (\bibinfo {year}
  {2012})}\BibitemShut {NoStop}%
\bibitem [{\citenamefont {Lieb}\ and\ \citenamefont
  {Robinson}(1972)}]{lieb1972finite}%
  \BibitemOpen
  \bibfield  {author} {\bibinfo {author} {\bibfnamefont {E.~H.}\ \bibnamefont
  {Lieb}}\ and\ \bibinfo {author} {\bibfnamefont {D.~W.}\ \bibnamefont
  {Robinson}},\ }\bibfield  {title} {\bibinfo {title} {The finite group
  velocity of quantum spin systems},\ }\href
  {https://doi.org/10.1007/BF01645779} {\bibfield  {journal} {\bibinfo
  {journal} {Communications in Mathematical Physics}\ }\textbf {\bibinfo
  {volume} {28}},\ \bibinfo {pages} {251} (\bibinfo {year} {1972})}\BibitemShut
  {NoStop}%
\bibitem [{\citenamefont {Ranard}\ \emph {et~al.}(2022)\citenamefont {Ranard},
  \citenamefont {Walter},\ and\ \citenamefont
  {Witteveen}}]{ranard2022converse}%
  \BibitemOpen
  \bibfield  {author} {\bibinfo {author} {\bibfnamefont {D.}~\bibnamefont
  {Ranard}}, \bibinfo {author} {\bibfnamefont {M.}~\bibnamefont {Walter}},\
  and\ \bibinfo {author} {\bibfnamefont {F.}~\bibnamefont {Witteveen}},\
  }\bibfield  {title} {\bibinfo {title} {{A Converse to Lieb--Robinson Bounds
  in One Dimension Using Index Theory}},\ }\href
  {https://doi.org/10.1007/s00023-022-01193-x} {\bibfield  {journal} {\bibinfo
  {journal} {Annales Henri Poincar{\'e}}\ }\textbf {\bibinfo {volume} {23}},\
  \bibinfo {pages} {3905} (\bibinfo {year} {2022})}\BibitemShut {NoStop}%
\bibitem [{\citenamefont {Ferris}\ and\ \citenamefont
  {Vidal}(2012)}]{ferris2012perfect}%
  \BibitemOpen
  \bibfield  {author} {\bibinfo {author} {\bibfnamefont {A.~J.}\ \bibnamefont
  {Ferris}}\ and\ \bibinfo {author} {\bibfnamefont {G.}~\bibnamefont {Vidal}},\
  }\bibfield  {title} {\bibinfo {title} {Perfect sampling with unitary tensor
  networks},\ }\href {https://doi.org/10.1103/PhysRevB.85.165146} {\bibfield
  {journal} {\bibinfo  {journal} {Phys. Rev. B}\ }\textbf {\bibinfo {volume}
  {85}},\ \bibinfo {pages} {165146} (\bibinfo {year} {2012})}\BibitemShut
  {NoStop}%
\bibitem [{\citenamefont {Pollmann}\ \emph {et~al.}(2016)\citenamefont
  {Pollmann}, \citenamefont {Khemani}, \citenamefont {Cirac},\ and\
  \citenamefont {Sondhi}}]{pollmann2016efficient}%
  \BibitemOpen
  \bibfield  {author} {\bibinfo {author} {\bibfnamefont {F.}~\bibnamefont
  {Pollmann}}, \bibinfo {author} {\bibfnamefont {V.}~\bibnamefont {Khemani}},
  \bibinfo {author} {\bibfnamefont {J.~I.}\ \bibnamefont {Cirac}},\ and\
  \bibinfo {author} {\bibfnamefont {S.~L.}\ \bibnamefont {Sondhi}},\ }\bibfield
   {title} {\bibinfo {title} {{Efficient variational diagonalization of fully
  many-body localized Hamiltonians}},\ }\href
  {https://doi.org/10.1103/PhysRevB.94.041116} {\bibfield  {journal} {\bibinfo
  {journal} {Phys. Rev. B}\ }\textbf {\bibinfo {volume} {94}},\ \bibinfo
  {pages} {041116} (\bibinfo {year} {2016})}\BibitemShut {NoStop}%
\bibitem [{\citenamefont {Wahl}\ \emph {et~al.}(2017)\citenamefont {Wahl},
  \citenamefont {Pal},\ and\ \citenamefont {Simon}}]{wahl2017efficient}%
  \BibitemOpen
  \bibfield  {author} {\bibinfo {author} {\bibfnamefont {T.~B.}\ \bibnamefont
  {Wahl}}, \bibinfo {author} {\bibfnamefont {A.}~\bibnamefont {Pal}},\ and\
  \bibinfo {author} {\bibfnamefont {S.~H.}\ \bibnamefont {Simon}},\ }\bibfield
  {title} {\bibinfo {title} {{Efficient Representation of Fully Many-Body
  Localized Systems Using Tensor Networks}},\ }\href
  {https://doi.org/10.1103/PhysRevX.7.021018} {\bibfield  {journal} {\bibinfo
  {journal} {Phys. Rev. X}\ }\textbf {\bibinfo {volume} {7}},\ \bibinfo {pages}
  {021018} (\bibinfo {year} {2017})}\BibitemShut {NoStop}%
\bibitem [{\citenamefont {Liu}\ \emph {et~al.}(2019)\citenamefont {Liu},
  \citenamefont {Ran}, \citenamefont {Wittek}, \citenamefont {Peng},
  \citenamefont {Garc{\'\i}a}, \citenamefont {Su},\ and\ \citenamefont
  {Lewenstein}}]{liu2019machine}%
  \BibitemOpen
  \bibfield  {author} {\bibinfo {author} {\bibfnamefont {D.}~\bibnamefont
  {Liu}}, \bibinfo {author} {\bibfnamefont {S.-J.}\ \bibnamefont {Ran}},
  \bibinfo {author} {\bibfnamefont {P.}~\bibnamefont {Wittek}}, \bibinfo
  {author} {\bibfnamefont {C.}~\bibnamefont {Peng}}, \bibinfo {author}
  {\bibfnamefont {R.~B.}\ \bibnamefont {Garc{\'\i}a}}, \bibinfo {author}
  {\bibfnamefont {G.}~\bibnamefont {Su}},\ and\ \bibinfo {author}
  {\bibfnamefont {M.}~\bibnamefont {Lewenstein}},\ }\bibfield  {title}
  {\bibinfo {title} {Machine learning by unitary tensor network of hierarchical
  tree structure},\ }\href {https://doi.org/10.1088/1367-2630/ab31ef}
  {\bibfield  {journal} {\bibinfo  {journal} {New Journal of Physics}\ }\textbf
  {\bibinfo {volume} {21}},\ \bibinfo {pages} {073059} (\bibinfo {year}
  {2019})}\BibitemShut {NoStop}%
\bibitem [{\citenamefont {Haghshenas}(2021)}]{haghshenas2021optimization}%
  \BibitemOpen
  \bibfield  {author} {\bibinfo {author} {\bibfnamefont {R.}~\bibnamefont
  {Haghshenas}},\ }\bibfield  {title} {\bibinfo {title} {Optimization schemes
  for unitary tensor-network circuit},\ }\href
  {https://doi.org/10.1103/PhysRevResearch.3.023148} {\bibfield  {journal}
  {\bibinfo  {journal} {Phys. Rev. Res.}\ }\textbf {\bibinfo {volume} {3}},\
  \bibinfo {pages} {023148} (\bibinfo {year} {2021})}\BibitemShut {NoStop}%
\bibitem [{\citenamefont {Haghshenas}\ \emph {et~al.}(2022)\citenamefont
  {Haghshenas}, \citenamefont {Gray}, \citenamefont {Potter},\ and\
  \citenamefont {Chan}}]{haghshenas2022variational}%
  \BibitemOpen
  \bibfield  {author} {\bibinfo {author} {\bibfnamefont {R.}~\bibnamefont
  {Haghshenas}}, \bibinfo {author} {\bibfnamefont {J.}~\bibnamefont {Gray}},
  \bibinfo {author} {\bibfnamefont {A.~C.}\ \bibnamefont {Potter}},\ and\
  \bibinfo {author} {\bibfnamefont {G.~K.-L.}\ \bibnamefont {Chan}},\
  }\bibfield  {title} {\bibinfo {title} {{Variational Power of Quantum Circuit
  Tensor Networks}},\ }\href {https://doi.org/10.1103/PhysRevX.12.011047}
  {\bibfield  {journal} {\bibinfo  {journal} {Phys. Rev. X}\ }\textbf {\bibinfo
  {volume} {12}},\ \bibinfo {pages} {011047} (\bibinfo {year}
  {2022})}\BibitemShut {NoStop}%
\bibitem [{\citenamefont {Styliaris}\ \emph
  {et~al.}(2025{\natexlab{a}})\citenamefont {Styliaris}, \citenamefont
  {Trivedi}, \citenamefont {Perez-Garcia},\ and\ \citenamefont
  {Cirac}}]{Styliaris_2025}%
  \BibitemOpen
  \bibfield  {author} {\bibinfo {author} {\bibfnamefont {G.}~\bibnamefont
  {Styliaris}}, \bibinfo {author} {\bibfnamefont {R.}~\bibnamefont {Trivedi}},
  \bibinfo {author} {\bibfnamefont {D.}~\bibnamefont {Perez-Garcia}},\ and\
  \bibinfo {author} {\bibfnamefont {J.~I.}\ \bibnamefont {Cirac}},\ }\bibfield
  {title} {\bibinfo {title} {Matrix-product unitaries: Beyond quantum cellular
  automata},\ }\href {https://doi.org/10.22331/q-2025-02-25-1645} {\bibfield
  {journal} {\bibinfo  {journal} {Quantum}\ }\textbf {\bibinfo {volume} {9}},\
  \bibinfo {pages} {1645} (\bibinfo {year} {2025}{\natexlab{a}})}\BibitemShut
  {NoStop}%
\bibitem [{\citenamefont {Styliaris}\ \emph
  {et~al.}(2025{\natexlab{b}})\citenamefont {Styliaris}, \citenamefont
  {Trivedi},\ and\ \citenamefont {Cirac}}]{styliaris2025MPU}%
  \BibitemOpen
  \bibfield  {author} {\bibinfo {author} {\bibfnamefont {G.}~\bibnamefont
  {Styliaris}}, \bibinfo {author} {\bibfnamefont {R.}~\bibnamefont {Trivedi}},\
  and\ \bibinfo {author} {\bibfnamefont {J.~I.}\ \bibnamefont {Cirac}},\ }\href
  {https://arxiv.org/abs/2508.08160} {\bibinfo {title} {{Quantum Circuit
  Complexity of Matrix-Product Unitaries}}} (\bibinfo {year}
  {2025}{\natexlab{b}}),\ \Eprint {https://arxiv.org/abs/2508.08160}
  {arXiv:2508.08160 [quant-ph]} \BibitemShut {NoStop}%
\bibitem [{\citenamefont {Sch\"on}\ \emph {et~al.}(2005)\citenamefont
  {Sch\"on}, \citenamefont {Solano}, \citenamefont {Verstraete}, \citenamefont
  {Cirac},\ and\ \citenamefont {Wolf}}]{schon2005sequential}%
  \BibitemOpen
  \bibfield  {author} {\bibinfo {author} {\bibfnamefont {C.}~\bibnamefont
  {Sch\"on}}, \bibinfo {author} {\bibfnamefont {E.}~\bibnamefont {Solano}},
  \bibinfo {author} {\bibfnamefont {F.}~\bibnamefont {Verstraete}}, \bibinfo
  {author} {\bibfnamefont {J.~I.}\ \bibnamefont {Cirac}},\ and\ \bibinfo
  {author} {\bibfnamefont {M.~M.}\ \bibnamefont {Wolf}},\ }\bibfield  {title}
  {\bibinfo {title} {{Sequential Generation of Entangled Multiqubit States}},\
  }\href {https://doi.org/10.1103/PhysRevLett.95.110503} {\bibfield  {journal}
  {\bibinfo  {journal} {Phys. Rev. Lett.}\ }\textbf {\bibinfo {volume} {95}},\
  \bibinfo {pages} {110503} (\bibinfo {year} {2005})}\BibitemShut {NoStop}%
\bibitem [{\citenamefont {Sch\"on}\ \emph {et~al.}(2007)\citenamefont
  {Sch\"on}, \citenamefont {Hammerer}, \citenamefont {Wolf}, \citenamefont
  {Cirac},\ and\ \citenamefont {Solano}}]{schon2007sequential}%
  \BibitemOpen
  \bibfield  {author} {\bibinfo {author} {\bibfnamefont {C.}~\bibnamefont
  {Sch\"on}}, \bibinfo {author} {\bibfnamefont {K.}~\bibnamefont {Hammerer}},
  \bibinfo {author} {\bibfnamefont {M.~M.}\ \bibnamefont {Wolf}}, \bibinfo
  {author} {\bibfnamefont {J.~I.}\ \bibnamefont {Cirac}},\ and\ \bibinfo
  {author} {\bibfnamefont {E.}~\bibnamefont {Solano}},\ }\bibfield  {title}
  {\bibinfo {title} {{Sequential generation of matrix-product states in cavity
  QED}},\ }\href {https://doi.org/10.1103/PhysRevA.75.032311} {\bibfield
  {journal} {\bibinfo  {journal} {Phys. Rev. A}\ }\textbf {\bibinfo {volume}
  {75}},\ \bibinfo {pages} {032311} (\bibinfo {year} {2007})}\BibitemShut
  {NoStop}%
\bibitem [{\citenamefont {Ba\~nuls}\ \emph {et~al.}(2008)\citenamefont
  {Ba\~nuls}, \citenamefont {P\'erez-Garc\'{\i}a}, \citenamefont {Wolf},
  \citenamefont {Verstraete},\ and\ \citenamefont
  {Cirac}}]{banuls2008sequentially}%
  \BibitemOpen
  \bibfield  {author} {\bibinfo {author} {\bibfnamefont {M.~C.}\ \bibnamefont
  {Ba\~nuls}}, \bibinfo {author} {\bibfnamefont {D.}~\bibnamefont
  {P\'erez-Garc\'{\i}a}}, \bibinfo {author} {\bibfnamefont {M.~M.}\
  \bibnamefont {Wolf}}, \bibinfo {author} {\bibfnamefont {F.}~\bibnamefont
  {Verstraete}},\ and\ \bibinfo {author} {\bibfnamefont {J.~I.}\ \bibnamefont
  {Cirac}},\ }\bibfield  {title} {\bibinfo {title} {Sequentially generated
  states for the study of two-dimensional systems},\ }\href
  {https://doi.org/10.1103/PhysRevA.77.052306} {\bibfield  {journal} {\bibinfo
  {journal} {Phys. Rev. A}\ }\textbf {\bibinfo {volume} {77}},\ \bibinfo
  {pages} {052306} (\bibinfo {year} {2008})}\BibitemShut {NoStop}%
\bibitem [{\citenamefont {Chen}\ \emph {et~al.}(2024)\citenamefont {Chen},
  \citenamefont {Dua}, \citenamefont {Hermele}, \citenamefont {Stephen},
  \citenamefont {Tantivasadakarn}, \citenamefont {Vanhove},\ and\ \citenamefont
  {Zhao}}]{Chen_2024}%
  \BibitemOpen
  \bibfield  {author} {\bibinfo {author} {\bibfnamefont {X.}~\bibnamefont
  {Chen}}, \bibinfo {author} {\bibfnamefont {A.}~\bibnamefont {Dua}}, \bibinfo
  {author} {\bibfnamefont {M.}~\bibnamefont {Hermele}}, \bibinfo {author}
  {\bibfnamefont {D.~T.}\ \bibnamefont {Stephen}}, \bibinfo {author}
  {\bibfnamefont {N.}~\bibnamefont {Tantivasadakarn}}, \bibinfo {author}
  {\bibfnamefont {R.}~\bibnamefont {Vanhove}},\ and\ \bibinfo {author}
  {\bibfnamefont {J.-Y.}\ \bibnamefont {Zhao}},\ }\bibfield  {title} {\bibinfo
  {title} {Sequential quantum circuits as maps between gapped phases},\ }\href
  {https://doi.org/10.1103/PhysRevB.109.075116} {\bibfield  {journal} {\bibinfo
   {journal} {Phys. Rev. B}\ }\textbf {\bibinfo {volume} {109}},\ \bibinfo
  {pages} {075116} (\bibinfo {year} {2024})}\BibitemShut {NoStop}%
\bibitem [{\citenamefont {Bratteli}\ and\ \citenamefont
  {Robinson}(2012)}]{bratteli2012operator}%
  \BibitemOpen
  \bibfield  {author} {\bibinfo {author} {\bibfnamefont {O.}~\bibnamefont
  {Bratteli}}\ and\ \bibinfo {author} {\bibfnamefont {D.~W.}\ \bibnamefont
  {Robinson}},\ }\href@noop {} {\emph {\bibinfo {title} {Operator algebras and
  quantum statistical mechanics: Volume 1: C*-and W*-Algebras. Symmetry Groups.
  Decomposition of States}}}\ (\bibinfo  {publisher} {Springer Science \&
  Business Media},\ \bibinfo {year} {2012})\BibitemShut {NoStop}%
\bibitem [{\citenamefont
  {Zanardi}(2001)}]{zanardi2002stabilizationquantuminformationunified}%
  \BibitemOpen
  \bibfield  {author} {\bibinfo {author} {\bibfnamefont {P.}~\bibnamefont
  {Zanardi}},\ }\bibinfo {title} {{Stabilization of Quantum Information: A
  Unified Dynamical-Algebraic Approach}},\ in\ \href
  {https://doi.org/10.1007/978-1-4615-1245-5_35} {\emph {\bibinfo {booktitle}
  {Macroscopic Quantum Coherence and Quantum Computing}}},\ \bibinfo {editor}
  {edited by\ \bibinfo {editor} {\bibfnamefont {D.~V.}\ \bibnamefont {Averin}},
  \bibinfo {editor} {\bibfnamefont {B.}~\bibnamefont {Ruggiero}},\ and\
  \bibinfo {editor} {\bibfnamefont {P.}~\bibnamefont {Silvestrini}}}\ (\bibinfo
   {publisher} {Springer US},\ \bibinfo {address} {Boston, MA},\ \bibinfo
  {year} {2001})\ pp.\ \bibinfo {pages} {351--357}\BibitemShut {NoStop}%
\bibitem [{\citenamefont {Perez-Delgado}\ and\ \citenamefont
  {Cheung}(2005)}]{perez2005models}%
  \BibitemOpen
  \bibfield  {author} {\bibinfo {author} {\bibfnamefont {C.~A.}\ \bibnamefont
  {Perez-Delgado}}\ and\ \bibinfo {author} {\bibfnamefont {D.}~\bibnamefont
  {Cheung}},\ }\href {https://arxiv.org/abs/quant-ph/0508164} {\bibinfo {title}
  {Models of quantum cellular automata}} (\bibinfo {year} {2005}),\ \Eprint
  {https://arxiv.org/abs/quant-ph/0508164} {arXiv:quant-ph/0508164 [quant-ph]}
  \BibitemShut {NoStop}%
\bibitem [{\citenamefont {P\'erez-Delgado}\ and\ \citenamefont
  {Cheung}(2007)}]{perez2007local}%
  \BibitemOpen
  \bibfield  {author} {\bibinfo {author} {\bibfnamefont {C.~A.}\ \bibnamefont
  {P\'erez-Delgado}}\ and\ \bibinfo {author} {\bibfnamefont {D.}~\bibnamefont
  {Cheung}},\ }\bibfield  {title} {\bibinfo {title} {Local unitary quantum
  cellular automata},\ }\href {https://doi.org/10.1103/PhysRevA.76.032320}
  {\bibfield  {journal} {\bibinfo  {journal} {Phys. Rev. A}\ }\textbf {\bibinfo
  {volume} {76}},\ \bibinfo {pages} {032320} (\bibinfo {year}
  {2007})}\BibitemShut {NoStop}%
\bibitem [{\citenamefont {Toffoli}\ and\ \citenamefont
  {Margolus}(1990)}]{toffoli1990invertible}%
  \BibitemOpen
  \bibfield  {author} {\bibinfo {author} {\bibfnamefont {T.}~\bibnamefont
  {Toffoli}}\ and\ \bibinfo {author} {\bibfnamefont {N.~H.}\ \bibnamefont
  {Margolus}},\ }\bibfield  {title} {\bibinfo {title} {Invertible cellular
  automata: a review},\ }\href@noop {} {\bibfield  {journal} {\bibinfo
  {journal} {Physica D: Nonlinear Phenomena}\ }\textbf {\bibinfo {volume}
  {45}},\ \bibinfo {pages} {229} (\bibinfo {year} {1990})}\BibitemShut
  {NoStop}%
\bibitem [{\citenamefont {Manber}(1989)}]{manber1989introduction}%
  \BibitemOpen
  \bibfield  {author} {\bibinfo {author} {\bibfnamefont {U.}~\bibnamefont
  {Manber}},\ }\href@noop {} {\emph {\bibinfo {title} {Introduction to
  algorithms: a creative approach}}},\ Vol.\ \bibinfo {volume} {142}\ (\bibinfo
   {publisher} {Addison-Wesley Reading, MA},\ \bibinfo {year}
  {1989})\BibitemShut {NoStop}%
\bibitem [{\citenamefont {Evenbly}\ and\ \citenamefont
  {Vidal}(2011)}]{evenbly2011tensor}%
  \BibitemOpen
  \bibfield  {author} {\bibinfo {author} {\bibfnamefont {G.}~\bibnamefont
  {Evenbly}}\ and\ \bibinfo {author} {\bibfnamefont {G.}~\bibnamefont
  {Vidal}},\ }\bibfield  {title} {\bibinfo {title} {{Tensor Network States and
  Geometry}},\ }\href {https://doi.org/10.1007/s10955-011-0237-4} {\bibfield
  {journal} {\bibinfo  {journal} {Journal of Statistical Physics}\ }\textbf
  {\bibinfo {volume} {145}},\ \bibinfo {pages} {891} (\bibinfo {year}
  {2011})}\BibitemShut {NoStop}%
\bibitem [{\citenamefont {Arrighi}\ \emph {et~al.}(2008)\citenamefont
  {Arrighi}, \citenamefont {Nesme},\ and\ \citenamefont
  {Werner}}]{arrighi2008one}%
  \BibitemOpen
  \bibfield  {author} {\bibinfo {author} {\bibfnamefont {P.}~\bibnamefont
  {Arrighi}}, \bibinfo {author} {\bibfnamefont {V.}~\bibnamefont {Nesme}},\
  and\ \bibinfo {author} {\bibfnamefont {R.}~\bibnamefont {Werner}},\
  }\bibfield  {title} {\bibinfo {title} {One-dimensional quantum cellular
  automata over finite, unbounded configurations},\ }in\ \href@noop {} {\emph
  {\bibinfo {booktitle} {Language and Automata Theory and Applications: Second
  International Conference, LATA 2008, Tarragona, Spain, March 13-19, 2008.
  Revised Papers 2}}}\ (\bibinfo {organization} {Springer},\ \bibinfo {year}
  {2008})\ pp.\ \bibinfo {pages} {64--75}\BibitemShut {NoStop}%
\bibitem [{\citenamefont {Nielsen}\ and\ \citenamefont
  {Chuang}(2010)}]{Nielsen_Chuang_2010}%
  \BibitemOpen
  \bibfield  {author} {\bibinfo {author} {\bibfnamefont {M.~A.}\ \bibnamefont
  {Nielsen}}\ and\ \bibinfo {author} {\bibfnamefont {I.~L.}\ \bibnamefont
  {Chuang}},\ }\href@noop {} {\emph {\bibinfo {title} {Quantum Computation and
  Quantum Information: 10th Anniversary Edition}}}\ (\bibinfo  {publisher}
  {Cambridge University Press},\ \bibinfo {year} {2010})\BibitemShut {NoStop}%
\bibitem [{\citenamefont {Bridgeman}\ and\ \citenamefont
  {Williamson}(2017)}]{bridgeman2017anomalies}%
  \BibitemOpen
  \bibfield  {author} {\bibinfo {author} {\bibfnamefont {J.~C.}\ \bibnamefont
  {Bridgeman}}\ and\ \bibinfo {author} {\bibfnamefont {D.~J.}\ \bibnamefont
  {Williamson}},\ }\bibfield  {title} {\bibinfo {title} {{Anomalies and
  entanglement renormalization}},\ }\href
  {https://doi.org/10.1103/PhysRevB.96.125104} {\bibfield  {journal} {\bibinfo
  {journal} {Phys. Rev. B}\ }\textbf {\bibinfo {volume} {96}},\ \bibinfo
  {pages} {125104} (\bibinfo {year} {2017})}\BibitemShut {NoStop}%
\bibitem [{\citenamefont {Kempe}\ \emph {et~al.}(2001)\citenamefont {Kempe},
  \citenamefont {Bacon}, \citenamefont {DiVincenzo},\ and\ \citenamefont
  {Whaley}}]{kempe2001encoded}%
  \BibitemOpen
  \bibfield  {author} {\bibinfo {author} {\bibfnamefont {J.}~\bibnamefont
  {Kempe}}, \bibinfo {author} {\bibfnamefont {D.}~\bibnamefont {Bacon}},
  \bibinfo {author} {\bibfnamefont {D.~P.}\ \bibnamefont {DiVincenzo}},\ and\
  \bibinfo {author} {\bibfnamefont {K.~B.}\ \bibnamefont {Whaley}},\ }\bibfield
   {title} {\bibinfo {title} {Encoded universality from a single physical
  interaction},\ }\href {https://dl.acm.org/doi/abs/10.5555/2016994.2017000}
  {\bibfield  {journal} {\bibinfo  {journal} {Quantum Info. Comput.}\ }\textbf
  {\bibinfo {volume} {1}},\ \bibinfo {pages} {33} (\bibinfo {year}
  {2001})}\BibitemShut {NoStop}%
\bibitem [{\citenamefont {Echternach}\ \emph {et~al.}(2001)\citenamefont
  {Echternach}, \citenamefont {Williams}, \citenamefont {Dultz}, \citenamefont
  {Delsing}, \citenamefont {Braunstein},\ and\ \citenamefont
  {Dowling}}]{echternach2001universal}%
  \BibitemOpen
  \bibfield  {author} {\bibinfo {author} {\bibfnamefont {P.}~\bibnamefont
  {Echternach}}, \bibinfo {author} {\bibfnamefont {C.~P.}\ \bibnamefont
  {Williams}}, \bibinfo {author} {\bibfnamefont {S.~C.}\ \bibnamefont {Dultz}},
  \bibinfo {author} {\bibfnamefont {P.}~\bibnamefont {Delsing}}, \bibinfo
  {author} {\bibfnamefont {S.}~\bibnamefont {Braunstein}},\ and\ \bibinfo
  {author} {\bibfnamefont {J.~P.}\ \bibnamefont {Dowling}},\ }\bibfield
  {title} {\bibinfo {title} {{Universal quantum gates for single cooper pair
  box based quantum computing}},\ }\href
  {https://dl.acm.org/doi/10.5555/2016994.2017010} {\bibfield  {journal}
  {\bibinfo  {journal} {Quantum Info. Comput.}\ }\textbf {\bibinfo {volume}
  {1}},\ \bibinfo {pages} {143} (\bibinfo {year} {2001})}\BibitemShut {NoStop}%
\bibitem [{\citenamefont {Schuch}\ and\ \citenamefont
  {Siewert}(2003)}]{schuch2003natural}%
  \BibitemOpen
  \bibfield  {author} {\bibinfo {author} {\bibfnamefont {N.}~\bibnamefont
  {Schuch}}\ and\ \bibinfo {author} {\bibfnamefont {J.}~\bibnamefont
  {Siewert}},\ }\bibfield  {title} {\bibinfo {title} {{Natural two-qubit gate
  for quantum computation using the $\mathrm{XY}$ interaction}},\ }\href
  {https://doi.org/10.1103/PhysRevA.67.032301} {\bibfield  {journal} {\bibinfo
  {journal} {Phys. Rev. A}\ }\textbf {\bibinfo {volume} {67}},\ \bibinfo
  {pages} {032301} (\bibinfo {year} {2003})}\BibitemShut {NoStop}%
\bibitem [{\citenamefont {Abrams}\ \emph {et~al.}(2020)\citenamefont {Abrams},
  \citenamefont {Didier}, \citenamefont {Johnson}, \citenamefont {Silva},\ and\
  \citenamefont {Ryan}}]{abrams2019implementation}%
  \BibitemOpen
  \bibfield  {author} {\bibinfo {author} {\bibfnamefont {D.~M.}\ \bibnamefont
  {Abrams}}, \bibinfo {author} {\bibfnamefont {N.}~\bibnamefont {Didier}},
  \bibinfo {author} {\bibfnamefont {B.~R.}\ \bibnamefont {Johnson}}, \bibinfo
  {author} {\bibfnamefont {M.~P.~d.}\ \bibnamefont {Silva}},\ and\ \bibinfo
  {author} {\bibfnamefont {C.~A.}\ \bibnamefont {Ryan}},\ }\bibfield  {title}
  {\bibinfo {title} {{Implementation of $XY$ entangling gates with a single
  calibrated pulse}},\ }\href {https://doi.org/10.1038/s41928-020-00498-1}
  {\bibfield  {journal} {\bibinfo  {journal} {Nature Electronics}\ }\textbf
  {\bibinfo {volume} {3}},\ \bibinfo {pages} {744} (\bibinfo {year}
  {2020})}\BibitemShut {NoStop}%
\bibitem [{\citenamefont {Roncaglia}(2019)}]{Roncaglia_2019}%
  \BibitemOpen
  \bibfield  {author} {\bibinfo {author} {\bibfnamefont {M.}~\bibnamefont
  {Roncaglia}},\ }\bibfield  {title} {\bibinfo {title} {On the conservation of
  information in quantum physics},\ }\href
  {https://doi.org/10.1007/s10701-019-00304-9} {\bibfield  {journal} {\bibinfo
  {journal} {Foundations of Physics}\ }\textbf {\bibinfo {volume} {49}},\
  \bibinfo {pages} {1278} (\bibinfo {year} {2019})}\BibitemShut {NoStop}%
\bibitem [{\citenamefont {Kitaev}(2006)}]{Kitaev_2006}%
  \BibitemOpen
  \bibfield  {author} {\bibinfo {author} {\bibfnamefont {A.}~\bibnamefont
  {Kitaev}},\ }\bibfield  {title} {\bibinfo {title} {Anyons in an exactly
  solved model and beyond},\ }\href {https://doi.org/10.1016/j.aop.2005.10.005}
  {\bibfield  {journal} {\bibinfo  {journal} {Annals of Physics}\ }\textbf
  {\bibinfo {volume} {321}},\ \bibinfo {pages} {2} (\bibinfo {year}
  {2006})}\BibitemShut {NoStop}%
\bibitem [{\citenamefont {Wei}\ \emph {et~al.}(2022)\citenamefont {Wei},
  \citenamefont {Malz},\ and\ \citenamefont {Cirac}}]{wei2022sequential}%
  \BibitemOpen
  \bibfield  {author} {\bibinfo {author} {\bibfnamefont {Z.-Y.}\ \bibnamefont
  {Wei}}, \bibinfo {author} {\bibfnamefont {D.}~\bibnamefont {Malz}},\ and\
  \bibinfo {author} {\bibfnamefont {J.~I.}\ \bibnamefont {Cirac}},\ }\bibfield
  {title} {\bibinfo {title} {{Sequential Generation of Projected Entangled-Pair
  States}},\ }\href {https://doi.org/10.1103/PhysRevLett.128.010607} {\bibfield
   {journal} {\bibinfo  {journal} {Phys. Rev. Lett.}\ }\textbf {\bibinfo
  {volume} {128}},\ \bibinfo {pages} {010607} (\bibinfo {year}
  {2022})}\BibitemShut {NoStop}%
\bibitem [{\citenamefont {Hackl}\ and\ \citenamefont
  {Bianchi}(2021)}]{Hackl_2021}%
  \BibitemOpen
  \bibfield  {author} {\bibinfo {author} {\bibfnamefont {L.}~\bibnamefont
  {Hackl}}\ and\ \bibinfo {author} {\bibfnamefont {E.}~\bibnamefont
  {Bianchi}},\ }\bibfield  {title} {\bibinfo {title} {{Bosonic and fermionic
  Gaussian states from K{\"a}hler structures}},\ }\href
  {https://doi.org/10.21468/SciPostPhysCore.4.3.025} {\bibfield  {journal}
  {\bibinfo  {journal} {SciPost Phys. Core}\ }\textbf {\bibinfo {volume} {4}},\
  \bibinfo {pages} {025} (\bibinfo {year} {2021})}\BibitemShut {NoStop}%
\bibitem [{\citenamefont {Davis}\ and\ \citenamefont
  {Kahan}(1969)}]{davis1969some}%
  \BibitemOpen
  \bibfield  {author} {\bibinfo {author} {\bibfnamefont {C.}~\bibnamefont
  {Davis}}\ and\ \bibinfo {author} {\bibfnamefont {W.~M.}\ \bibnamefont
  {Kahan}},\ }\bibfield  {title} {\bibinfo {title} {Some new bounds on
  perturbation of subspaces},\ }\href@noop {} {\bibfield  {journal} {\bibinfo
  {journal} {Bulletin of the American Mathematical Society}\ }\textbf {\bibinfo
  {volume} {75}},\ \bibinfo {pages} {863} (\bibinfo {year} {1969})}\BibitemShut
  {NoStop}%
\bibitem [{\citenamefont {Davis}\ and\ \citenamefont
  {Kahan}(1970)}]{davis1970rotation}%
  \BibitemOpen
  \bibfield  {author} {\bibinfo {author} {\bibfnamefont {C.}~\bibnamefont
  {Davis}}\ and\ \bibinfo {author} {\bibfnamefont {W.~M.}\ \bibnamefont
  {Kahan}},\ }\bibfield  {title} {\bibinfo {title} {{The Rotation of
  Eigenvectors by a Perturbation. III}},\ }\href
  {http://www.jstor.org/stable/2949580} {\bibfield  {journal} {\bibinfo
  {journal} {SIAM Journal on Numerical Analysis}\ }\textbf {\bibinfo {volume}
  {7}},\ \bibinfo {pages} {1} (\bibinfo {year} {1970})}\BibitemShut {NoStop}%
\bibitem [{\citenamefont {Paige}\ and\ \citenamefont
  {Wei}(1994)}]{paige1994history}%
  \BibitemOpen
  \bibfield  {author} {\bibinfo {author} {\bibfnamefont {C.}~\bibnamefont
  {Paige}}\ and\ \bibinfo {author} {\bibfnamefont {M.}~\bibnamefont {Wei}},\
  }\bibfield  {title} {\bibinfo {title} {{History and generality of the CS
  decomposition}},\ }\href
  {https://doi.org/https://doi.org/10.1016/0024-3795(94)90446-4} {\bibfield
  {journal} {\bibinfo  {journal} {Linear Algebra and its Applications}\
  }\textbf {\bibinfo {volume} {208-209}},\ \bibinfo {pages} {303} (\bibinfo
  {year} {1994})}\BibitemShut {NoStop}%
\bibitem [{\citenamefont {Sutton}(2009)}]{sutton2009computing}%
  \BibitemOpen
  \bibfield  {author} {\bibinfo {author} {\bibfnamefont {B.~D.}\ \bibnamefont
  {Sutton}},\ }\bibfield  {title} {\bibinfo {title} {{Computing the complete CS
  decomposition}},\ }\href {https://doi.org/10.1007/s11075-008-9215-6}
  {\bibfield  {journal} {\bibinfo  {journal} {Numerical Algorithms}\ }\textbf
  {\bibinfo {volume} {50}},\ \bibinfo {pages} {33} (\bibinfo {year}
  {2009})}\BibitemShut {NoStop}%
\bibitem [{\citenamefont {Gong}\ \emph {et~al.}(2020)\citenamefont {Gong},
  \citenamefont {S\"underhauf}, \citenamefont {Schuch},\ and\ \citenamefont
  {Cirac}}]{Gong_PRL_2020}%
  \BibitemOpen
  \bibfield  {author} {\bibinfo {author} {\bibfnamefont {Z.}~\bibnamefont
  {Gong}}, \bibinfo {author} {\bibfnamefont {C.}~\bibnamefont {S\"underhauf}},
  \bibinfo {author} {\bibfnamefont {N.}~\bibnamefont {Schuch}},\ and\ \bibinfo
  {author} {\bibfnamefont {J.~I.}\ \bibnamefont {Cirac}},\ }\bibfield  {title}
  {\bibinfo {title} {{Classification of Matrix-Product Unitaries with
  Symmetries}},\ }\href {https://doi.org/10.1103/PhysRevLett.124.100402}
  {\bibfield  {journal} {\bibinfo  {journal} {Phys. Rev. Lett.}\ }\textbf
  {\bibinfo {volume} {124}},\ \bibinfo {pages} {100402} (\bibinfo {year}
  {2020})}\BibitemShut {NoStop}%
\bibitem [{\citenamefont {Zhang}(2023)}]{Carolyn_SPTinv}%
  \BibitemOpen
  \bibfield  {author} {\bibinfo {author} {\bibfnamefont {C.}~\bibnamefont
  {Zhang}},\ }\bibfield  {title} {\bibinfo {title} {Topological invariants for
  symmetry-protected topological phase entanglers},\ }\href
  {https://doi.org/10.1103/PhysRevB.107.235104} {\bibfield  {journal} {\bibinfo
   {journal} {Phys. Rev. B}\ }\textbf {\bibinfo {volume} {107}},\ \bibinfo
  {pages} {235104} (\bibinfo {year} {2023})}\BibitemShut {NoStop}%
\end{thebibliography}%
\end{document}